\documentclass[pointlessnumbers,
twoside,
11pt
]{scrbook}

\usepackage[utf8x]{inputenc}
\usepackage{xcolor}
\usepackage{listings}
\usepackage[color=black]{ocaml}
\usepackage{amssymb}
\usepackage{amsmath}
\usepackage{amsthm}
\usepackage{textcomp}
\usepackage{pifont}
\usepackage{bbding}
\makeatletter

 \newcommand*\l@preface[2]{%
  \ifnum \c@tocdepth >\m@ne
   \addpenalty{-\@highpenalty}%
   \vskip 1.0em \@plus\p@
   \setlength\@tempdima{1.5em}%
   \if@tocleft
   \ifx\toc@l@number\@empty\else
   \setlength\@tempdima{0\toc@l@number}%
   \fi
   \fi
   \begingroup
    \parindent \z@ \rightskip \@pnumwidth
    \parfillskip -\@pnumwidth
    \leavevmode \itshape
    \advance\leftskip\@tempdima
    \hskip -\leftskip
    #1\nobreak\hfil \nobreak\hb@xt@\@pnumwidth{\hss{\it page} {\rm #2}}\par
    \penalty\@highpenalty
   \endgroup
  \fi
 }

 \newcommand*\l@bibidx[2]{%
  \ifnum \c@tocdepth >\m@ne
   \addpenalty{-\@highpenalty}%
   \vskip 1.0em \@plus\p@
   \setlength\@tempdima{1.5em}%
   \if@tocleft
   \ifx\toc@l@number\@empty\else
   \setlength\@tempdima{0\toc@l@number}%
   \fi
   \fi
   \begingroup
    \parindent \z@ \rightskip \@pnumwidth
    \parfillskip -\@pnumwidth
    \leavevmode \itshape
    \advance\leftskip\@tempdima
    \hskip -\leftskip
    #1\nobreak\hfil \nobreak\hb@xt@\@pnumwidth{\hss {\rm #2}}\par
    \penalty\@highpenalty
   \endgroup
  \fi
 }

 \renewcommand*{\@@makechapterhead}[1]{%
  \chapterheadstartvskip
  \begingroup
   \normalfont\sectfont\size@chapter
   \setlength{\parindent}{\z@}\setlength{\parfillskip}{\z@ \@plus 1fil}%
   \begin{center}
    \chapterformat \\[0.5cm]
    #1
   \end{center}
   \nobreak\chapterheadendvskip%
  \endgroup
 }

 \renewcommand*{\@@makeschapterhead}[1]{%
  \chapterheadstartvskip
  \begingroup
   \normalfont\sectfont\size@chapter
   \begin{center}
    #1
   \end{center}
   \nobreak\chapterheadendvskip%
  \endgroup
 }

 \renewcommand{\@dotsep}{1000}

 \def\cleardoublepage{%
  \clearpage
  \ifodd\c@page\else \hbox{}\thispagestyle{empty}\newpage\fi
 }

\makeatother

\definecolor{dgreen}{rgb}{0.0,0.5,0.0}
\definecolor{dblue} {rgb}{0.2,0.2,0.7}
\definecolor{clorange}{rgb}{1.0,0.73,0.0}
\definecolor{clblue}{rgb}{0.38,0.38,0.83}
\definecolor{cllblue}{rgb}{0.9,0.9,0.93}
\definecolor{cldblue}{rgb}{0.0,0.0,0.5}

\lstdefinelanguage{OCaml}{%
 morekeywords={and,as,assert,asr,begin,class,%
  constraint,do,done,downto,else,end,%
  exception,external,false,for,fun,function,%
  functor,if,in,include,inherit,initializer,%
  land,lazy,let,lor,lsl,lsr,%
  lxor,match,method,mod,module,mutable,%
  new,object,of,open,or,private,%
  rec,sig,struct,then,to,true,%
  try,type,val,virtual,when,while,%
  with},%
 sensitive=true,%
 morecomment=[n]{(*}{*)},%
 morestring=[b]{"},%
}%

\lstdefinestyle{camlsrc}{%
 columns=fullflexible,
 basicstyle=\color{dgreen}\sffamily\normalsize,
 stringstyle=\ttfamily,
 commentstyle=\itshape,
 numbers=left,
 numberstyle=\tiny,
 tabsize=1,
 frame=lines,
 showstringspaces=true,
 numberbychapter=true,
 captionpos=b,
 xleftmargin=15pt,
 xrightmargin=15pt,
 language=OCaml,
 texcl=true,
}

\lstdefinestyle{caml}{%
 columns=fullflexible,
 basicstyle=\color{dgreen}\sffamily\normalsize\upshape,
 stringstyle=\ttfamily,
 commentstyle=\itshape,
 tabsize=1,
 language=OCaml,
 showstringspaces=true,
 xleftmargin=15pt,
 xrightmargin=15pt,
 language=OCaml,
 mathescape=true,
 numberstyle=\tiny,
}

\lstdefinestyle{bash}{%
 basicstyle=\ttfamily\normalsize,
 commentstyle=\ttfamily\normalsize,
 numbers=none,
 xleftmargin=15pt,
 xrightmargin=15pt,
 nolol=true,
 language=bash,
}

\theoremstyle{definition}

\newtheorem*{bsp*}{Beispiel}

\newtheorem*{aufgabe*}{Übungsaufgabe}

\newtheorem*{bemerkung*}{Bemerkung}
\newtheorem*{konvention*}{Konvention}

{\begin{proof}[Beweis]}{\end{proof}}
{\begin{proof}[Beweisskizze]}{\end{proof}}
{\begin{proof}[Beweisansatz]}{\end{proof}}

\usepackage{graphicx}
\usepackage[english]{babel}
\usepackage{makeidx}
\usepackage{scrpage2}
\usepackage{tikz}
\usetikzlibrary{arrows,positioning,backgrounds,fit}
\usepackage[numbers 
]{natbib}   

\newtheorem{theorem}{Theorem}[section]
\newtheorem*{mainth}{Main Theorem}
\newtheorem{corollary}{Corollary}[section]
\newtheorem{lemma}{Lemma}[section] 

\newtheorem*{claim*}{Claim}

\newtheorem{proposition}{Proposition}[section]
\theoremstyle{definition}      
\newtheorem{definition}{Definition}[section]
\newtheorem{remark}{Remark}[section]
\newtheorem{example}{Example}[section]

\def\abstractname{Abstract.}
\deffootnote{1em}{1em}{\textsuperscript{\thefootnotemark\ }}

\begin{document}
\frontmatter

\titlehead{{\Large University of Innsbruck \hfill Institute of Computer Science}}
\subject{\hrule\vskip1cm Cumulative Habilitation Thesis}
\title{%
Proof Theory at Work:\\
\vspace{.5cm}
Complexity Analysis of Term Rewrite Systems
}
\author{Georg Moser}
\date{August 2009 \vskip1cm \hrule}

\lowertitleback{%
\textbf{dedicated to Claudia}
}
\maketitle[1]
\tableofcontents
\cleardoublepage
\thispagestyle{empty}

\chapter*{Preface}

This cumulative habilitation thesis is based on the 
following publications studying the complexity of term rewrite systems. 
It goes without saying that for all included papers I wrote the manuscript.
\begin{enumerate}
\item G.~Moser and A.~Weiermann.
\newblock Relating derivation lengths with the slow-growing hierarchy directly.
\newblock In \emph{Proceedings of the 14th International Conference on
  Rewriting Techniques and Applications}, number 2706 in LNCS, pages 296--310, 
Springer Verlag, 2003.

\emph{Own Contribution:} Andreas Weiermann introduced me to the field of
complexity analysis of term rewriting systems and suggested the topic of the
paper. He made suggestions towards the proof methodology, but I was in charge to
fill in the details. Hence I crafted the central notions and theorems and
proved their correctness.

\medskip

\item T.~Arai and G.~Moser.
\newblock Proofs of termination of rewrite systems for polytime functions.
\newblock In \emph{Proceedings of the Annual Conference on Foundations of
  Software Technology and Theoretical Computer Science}, number 3821 in LNCS,
pages 529--540, Springer Verlag, 2005.

\emph{Own Contribution:} Toshiyasu Arai suggested a close study of a term-rewriting
characterisation of the polytime computable function given by Beckmann, Weiermann 
(Arch.~Math.~Log.~36(1):11--30, 1996). This study forms the basis of the paper. 
The established results were obtained jointly through iteratered revision of
the introduced concepts, theorems, and proofs.

\medskip

\item G.~Moser.
\newblock Derivational complexity of {K}nuth {B}endix orders revisited.
\newblock In \emph{Proceedings of the 13th International Conference on Logic
  for Programming Artificial Intelligence and Reasoning}, number 4246 in LNCS,
pages 75--89, Springer Verlag, 2006.

\medskip

\item M.~Avanzini and G.~Moser.
\newblock Complexity analysis by rewriting.
\newblock In \emph{Proceedings of the 9th International Symposium on Functional
  and Logic Programming}, number 4989 in LNCS, pages 130--146, 
Springer Verlag, 2008.

\emph{Own Contribution:} The theoretical contribution of this paper was
entirely my work. Martin Avanzini essentially implemented the technique and was
responsible to provide the experimental evidence.

\medskip
wi
\item G.~Moser and A.~Schnabl.
\newblock Proving quadratic derivational complexities using context dependent
  interpretations.
\newblock In \emph{Proceedings of the 19th International Conference on Rewrite
  Techniques and Applications},  number 5117 in LNCS, pages 276--290, 
Springer Verlag, 2008.

\emph{Own Contribution:} The paper provides an extension and clarification
of Andreas Schnabl's master thesis conducted under my supervision. Apart from
introducing Schnabl to the topic of the paper, I refined and smoothened 
essentially all introduced concepts and extended the theoretical contributions
of the paper. 

\medskip

\item N.~Hirokawa and G.~Moser.
\newblock Automated complexity analysis based on the dependency pair method.
\newblock In \emph{Proceedings of the 4th International Joint Conference on
  Automated Reasoning},  number 5195 in LNAI, pages 364--380, 
Springer Verlag, 2008.

\emph{Own Contribution:} I suggested the topic of complexity analysis to
Nao Hirokawa and extended initial ideas that were presented by him. For example the
focus shift towards the dependency pair method was my idea. The finally established
results were obtained jointly through iteratered revision of
the introduced concepts, theorems, and proofs.

\medskip

\item N.~Hirokawa and G.~Moser.
\newblock Complexity, graphs, and the dependency pair method.
\newblock In \emph{Proceedings of the International Conference on Logic for
  Programming Artificial Intelligence and Reasoning}, number 5330 in LNCS,
pages 652--666, Springer Verlag, 2008.

\emph{Own Contribution:} This paper extends the results presented in
paper (vi) and the initial idea of the extension was the result of joint
discussion. However I was in charge of filling in the details. 
Hence I crafted the central notions and theorems and proved their correctness. 
Apart from his contribution to the initial idea, Hirokawa's main contribution was the implementation
of the method and he was responsible for providing the experimental evidence.

\end{enumerate}
In order to keep a uniform presentation textual changes to these
papers have been necessary. Furthermore minor shortcomings have been repaired.

\bigskip\bigskip
\goodbreak

As an indication of the broadness of my research, below I also mention selected publications from the areas 
automated deduction, proof theory, and rewriting. 
\begin{enumerate}
\setcounter{enumi}{7}
\item G.~Moser. Some remarks on transfinite E-semantic trees and superposition.
In {\em Proceedings of the 1st International Workshop on First order
  Theorem Proving}, pages 97--103, 1997.

\medskip
\item C.~Ferm{\"u}ller and G.~Moser.
\newblock Have {\sc spass} with \ensuremath{{\mbox{\sl OCC1N}_g^=}}.
\newblock In {\em Proceedings of the 7th International Conference on Logic for
  Programming and Automated Reasoning}, number 1955 in LNCS, pages 114--130, 
Springer Verlag, 2000.

\medskip
\item G.~Moser. Ackermann's substitution method (remixed).
{\em Annals of Pure and Applied Logic}, 142\penalty0 (1--3):\penalty0 1--18, 2006.

\medskip
\item
M.~Baaz and G.~Moser. Herbrand's theorem and Term Induction.
{\em Archive of {M}athematical {L}ogic}, 45:\penalty0 447--503, 2006.
\medskip

\item N.~Dershowitz and G.~Moser.
The Hydra Battle Revisited.
In \emph{Rewriting, Computation and Proof; Essays Dedicated to Jean-Pierre Jouannaud on the Occasion of his 60th Birthday}, 
pages 1--27, 2007.

\medskip
\item G.~Moser, A.~Schnabl, and J.~Waldmann. Complexity Analysis of Term Rewriting Based on 
Matrix and Context Dependent Interpretations. In \emph{Proceedings of the Annual Conference on Foundations of
Software Technology and Theoretical Computer Science}, pages 304--315, 2008.
Creative-Commons-NC-ND licensed.
\end{enumerate}

\bigskip
\vfill

\subsection*{Acknowledgements}

I would like to take this opportunity to thank all my colleagues 
that contributed to this thesis. In particular my gratitude goes to
my co-authors Andreas Schabl, Andreas Weiermann, Martin Avanzini, Nao Hirokawa, 
and Toshiaysu Arai, who kindly allowed me to include our joint papers into this thesis.

\medskip
\noindent
Special thanks go to my past and present colleagues in the Computational Logic Group.
I want to thank Aart Middeldorp, 
Andreas Schnabl, Anna-Maria Scheiring, Christian Sternagel, Christian Vogt, 
Clemens Ballarin, Friedrich Neurauter, Harald Zankl, Martin Avanzini, Martin Korp, 
Martina Ingenhaeff, Mona Kornherr, Ren{\'e} Thiemann, Sandra Adelt, Simon Bailey, and Stefan Blom,
for providing a pleasant and inspiring working environment over the years.

\medskip
\noindent
Last but certainly not least, I want to express my gratitude towards my parents Helga Moser
and Gerhard Margreiter ($\dag$) who do and did their 
best in supporting me throughout my studies and research.

\vspace{1.5cm}

\noindent
Innsbruck, \today	\hfill Georg Moser

\mainmatter
\pagestyle{scrheadings}
\setheadsepline{0.5pt}
\automark[section]{chapter}
\lstset{language=OCaml}

\chapter{Introduction}

This thesis is concerned with investigations into the \emph{complexity of term rewriting systems}. 
Moreover the majority of the presented work deals with the \emph{automation} of 
such a complexity analysis. 
The aim of this introduction
is to present the main ideas in an easily accessible fashion to make the result presented accessible to
the general public. Necessarily some technical points are stated in an over-simplified
way. I kindly refer the knowledgeable reader to the technical summary as presented in Chapter~\ref{Summary}.

Since the advent of programming languages 
formalisms and tools have been thought and developed to express and
prove properties of programs. 
The general goal of the here presented work is to develop logical tools for
analysing the complexity of programs (automatically whenever possible). 

When reasoning about properties of programs, we need to fix the \emph{programs}, 
we aim to analyse. For a number of reasons this is a difficult choice. On one hand,
we would want our analysis to be directly transferable into applications, so that we indeed can
analyse existing software packages without further ado. On the other hand, we want our results 
to be as general as possible, hence our analysis should abstract from individual features
of programming languages. 

There is no decisive answer to these conflicting priorities, but for me the generality of the
obtained results appears more important. Successful investigation of abstract programs
can often be adapted to real-life contexts with ease, while the generalisation of tools and
methods invented in a specific setting to a more abstract level, may prove to be difficult
or even impossible. Still it is crucial that the applicability of the introduced general concepts is
not lost.

The most abstract formalism we may consider in computer science, 
are so-called \emph{abstract models of computation}. 
Such abstractions have been intensively studied by mathematicians and logicians at the beginning
of the $20^{\text{th}}$ century. Then the problem was to fix a suitable mathematical notion of 
\emph{computation}.
Several equivalent concepts like \emph{combinatory logic}, 
\emph{$\lambda$-calculus}, \emph{recursive functions}, \emph{register machines}, and 
\emph{Turing machines} have been put forward by Curry, Church, Gödel, Kleene, Turing and others. 
The central computational mechanism of Church's $\lambda$-calculus is \emph{rewriting} and bluntly
we could argue that the $\lambda$-calculus is simply \emph{rewriting} in disguise. 

Term rewriting is a conceptually simple, but powerful abstract model of computation. Let me 
describe this computation model in its most abstract and most simple form. We assert
a collection of widgets together with rules that govern the replacement of one widget by another.
For example, if we take the set of states of a computer as our collection of widgets and allow the 
replacement of one state by another, whenever the latter is reachable (in some sense) from the first, then
this constitutes a \emph{term rewrite system}. 
A fundamental property of rewrite systems is \emph{termination}, the non-existence
of endless rewrite steps, or replacements of one widget by another. Conclusively
strong techniques have been designed to ensure termination of term rewrite systems. In recent
years the emphasis shifted towards techniques that \emph{automatically} verify termination of 
a given term rewrite system.

Observe that despite its simplicity, rewriting is an abstract computation model that is equivalent to all notions
of computability mentioned above. For example any OCaml program is 
easily representable as a rewrite system. If this encoding is done carefully enough, then termination
techniques for rewriting become applicable to show termination of programs, sometimes even
fully automatically. 
The use of a functional programming language as an example may seem restrictive, as 
the representation of an OCaml program as a term rewrite system is typically simple.
However recent work aims at
the incorporation of imperative programming languages like Java or C. For example Java Bytecode
programs become applicable to this setting if a termination graph, representing the
program flow, is provided in a pre-processing step. The structural information of
these graphs can then be encoded as term rewrite systems. 

Once we have verified termination of a given term rewrite system, we have (perhaps automatically) established
a very important property. But in accordance with the legacy of rewriting as a computation
model, we should strive for more. For a terminating system, we can consider the following problem: Given
some widget, how many replacement steps can we perform till no more replacement is possible?
Termination assert that this problem is well-defined. 

This naturally entails investigations into the \emph{complexity} of term rewrite systems. 
A term rewrite system is considered of higher complexity, if the number of possible rewrite steps 
is larger. In other words the complexity of a rewrite system is measured through the maximal
possible computation steps possible with this abstract program.

The investigations into this problem are the topic of my thesis. 
In line with earlier results presented in the literature the study is performed as an analysis
of \emph{termination methods}. I.e., instead of directly seeking techniques to establish the
complexity of a given term rewrite system I have studied the complexity \emph{induced} by 
termination methods. These investigations cover well-established termination orders as
well as modern (automatable) termination techniques. 

Through this indirect study, a higher level of understanding is possible. Not only have I
provided new techniques to analyse the complexity of a given rewrite system, but at the same
time I have rendered insights into the expressivity of termination methods. These results
may lead to a new generation of termination provers for term rewrite system. Currently a 
termination prover, if successful, will simply output \emph{yes}. In the future, termination
provers can perhaps be identified with  complexity analysers: 
given a terminating rewrite system the laconic \emph{yes} of the prover is replaced
by expressive information on the complexity of the system. 

\chapter{Summary} \label{Summary}

{
\input{arxiv.sty}

\section{Status of Research}

As already mentioned in the introduction, 
term rewriting is a conceptually simple, but powerful abstract model of computation. The foundation
of rewriting is equational logic and \emph{term rewrite systems} (\emph{TRSs} for short) 
are conceivable as sets of directed equations.

To be a bit more formal, let $\FS$ denote a finite set of function symbols, i.e., a signature and
let $\VS$ denote a countable set of variables. Then the set of
terms over $\FS$ and $\VS$ is denoted as $\TA(\FS,\VS)$. 
A TRS $\RS$ is a set of rewrite rules $l \rew r$, where $l$ and
$r$ are terms.
The \emph{rewrite relation} $\rsrew{\RS}$  is the least binary
relation on the set of terms $\TA(\FS,\VS)$ containing $\RS$ such that
(i) if $s \rsrew{\RS} t$ and $\sigma$ a substitution, then
$s\sigma \rsrew{\RS} t\sigma$ holds, and
(ii) if $s \rsrew{\RS} t$, then for all $f \in \FS$: $f(\ldots, s,\ldots) \rsrew{\RS} f(\ldots ,t,\ldots)$ holds.
We sometimes write $\rew$, instead of $\rsrew{\RS}$, if no
confusion can arise from this. 

The implicit orientation of equations in rewrite systems naturally gives rise to computations, 
where a term is rewritten by successively replacing
subterms by equal terms until no further reduction is possible. 
Such a sequence of rewrite steps is also called a
derivation. Term rewriting forms a Turing complete model of computation, 
hence fundamental questions as for example 
termination of a given TRS, are undecidable in general. 
Furthermore term rewriting underlies much of 
declarative programming. As a special form of equational logic it has 
also found many applications in automated deduction and verification.

In this thesis, I will consider termination problems in term rewriting and
the complexity of term rewrite systems as measured by the maximal length of derivations.

\subsection{Termination in Rewriting}

In the area of term rewriting~\cite{BaaderNipkow:1998,Terese} powerful methods
have been introduced to establish termination of a given TRS $\RS$. 
Earlier research mainly concentrated on inventing suitable reduction orders---%
for example \emph{simplification orders}, see~\cite[Chapter 6]{Terese}---capable of
proving termination directly. 

As an example let us consider the \emph{multiset path order} 
(\emph{MPO} for short), cf.~\cite{Dershowitz:1982}. 
Let $>$ denote a strict partial order on the signature $\FS$. We call $>$ a \emph{precedence}. Let
$s$, $t$ be terms. For $s,t \not\in\VS$, we can write $s = f(s_1,\dots,s_m)$,
$t = g(t_1,\dots,t_n)$. We define $s \gmpo t$ if one of the 
following alternatives hold:
\begin{itemize}
\item there exists $i \in \{1,\dots,m\}$ such that $s_i \gmpo t$ or $s_i = t$,
\item $f > g$ and $s \gmpo t_j$ for all $1 \leqslant j \leqslant n$, or
\item $f = g$ and $\{s_1,\dots,s_m\} \gmpomul \{t_1,\dots,t_n\}$.
\end{itemize}
Here $\gmpomul$ denotes the multiset extension of $\gmpo$.

The definition of MPO entails, that if a TRS $\RS$ is compatible with
an MPO $\gmpo$, i.e., if ${\RS} \subseteq {\gmpo}$ then termination of $\RS$ follows. 
Clearly the reverse direction need not hold.
Consider the TRS $\RSa$ over the signature $\FS = \{\mf,\circ,\me\}$ 
taken from~\cite{Hofbauer:1991}:
\begin{align*}
      \mf(x) \circ (y \circ z) & \rew x \circ (\mf (\mf (y)) \circ z) 
      \tag{1} \\
      \mf(x) \circ (y \circ (z \circ w)) & \rew x \circ (z \circ (y \circ w)) 
      \tag{2} \\
      \mf(x) & \rew x
      \tag{3}
\end{align*}
The TRS $\RSa$ is terminating. However we cannot find a precedence $>$ such
that ${\RSa} \subseteq {\gmpo}$. This is due to
the fact that $\RSa$ encodes the Ackermann function as we will see in
the next section.
\label{trs:rsa}

In recent years the emphasis shifted towards transformation techniques like 
the \emph{dependency pair method}~\cite{ArtsGiesl:2000}, its extension the
\emph{dependency pair framework}~\cite{Thiemann:2007} or 
\emph{semantic labeling}~\cite{Zantema:1994,Zantema:1995}. 
The advantage---in particular of the dependency pair method---being that 
they are easily automatable.

Here we briefly recall the basics of the dependency pair method. Below on
page~\ref{p:SL} we will also recall the semantic labeling technique. In the presentation
of the dependency pair method we follow~\cite{ArtsGiesl:2000,HirokawaMiddeldorp:2005}. 
Let $t$ be a term. We set $t^\sharp \defsym t$ if $t \in \VS$, and 
$t^\sharp \defsym f^\sharp(t_1,\dots,t_n)$ if $t = f(\seq{t})$.
Here $f^\sharp$ is a new $n$-ary function symbol called 
\emph{dependency pair symbol}. For a signature
$\FS$, we define $\FS^\sharp = \FS \cup \{f^\sharp \mid f\in \FS\}$.
The set $\DP(\RS)$ of  \emph{dependency pairs} of a TRS $\RS$
is defined as the of pairs $l^\sharp \to u^{\sharp}$, where $l \to r \in \RS$ and
$u$ is a subterm of $r$, whose root symbol is a defined function symbol. Moreover
$u$ is not a proper subterm of $l$.

It is not difficult to see that a TRS $\RS$ is terminating if and only if 
there exists no infinite derivation of the following form 
\begin{equation*}
  t_1^\sharp \rssrew{\RS} t_2^\sharp \rsrew{\DP(\RS)} t_3^\sharp \rssrew{\RS} \ldots \tkom
\end{equation*}
where for all $i>0$, $t_i^\sharp$ is terminating with respect to $\RS$. 
This is a key observation in the formulation of the basic setting of the
dependency pair method. We obtain the following characterisation of
termination due to Arts and Giesl:
\label{p:DP}
\begin{itemize}
\item A TRS $\RS$ is terminating if and only if there exist
a reduction pair $(\gtrsim,\succ)$ such
that ${\DP(\RS)} \subseteq {\succ}$ and ${\RS} \subseteq {\gtrsim}$.
\end{itemize}
Here a \emph{reduction pair} $(\gtrsim, \succ)$ consists of a
rewrite preorder $\gtrsim$ and a compatible well-founded order $\succ$
which is closed under substitutions; compatibility means 
the inclusion ${\gtrsim \cdot \succ \cdot \gtrsim} \subseteq {\succ}$.

Another development is the use of automata 
techniques to prove termination~\cite{Geser2004,Geser2005,KM:2007}.
Moreover the technique to show termination by building an order-preserving mapping
into a well-founded domain has received renewed attention~\cite{Hofbauer2006,EWZ:2008,KW:2008}. 

These methods, among others, are used
in several recent software tools that aim to prove termination automatically. We mention
\aprove~\cite{Aprove1.2}, \cime~\cite{Cime:2003}, \jambox~\cite{EWZ:2008},  
\matchbox~\cite{Matchbox:2004}, \muterm~\cite{Muterm:2004}, \tpa~\cite{TPA:2006}, \tttt~\cite{TTTT:2008}.
In the termination competition (a subset of) these provers compete against
each other in order to prove termination of TRSs automatically, see
\begin{center}
\url{http://termcomp.uibk.ac.at} \tkom  
\end{center}
for this ongoing event. 

\subsection{Complexity of Rewrite Systems} \label{Complexity}

In order to assess the complexity of a TRS it is natural to look at the maximal
length of derivation sequences, a program that has already 
been suggested in~\cite{HofbauerLautemann:1989}. 
See also~\cite{CKS:1989} for a complementary study of the complexity of term rewrite systems. 

The \emph{derivational complexity function} 
with respect to a (terminating) TRS $\RS$ 
relates the length of a longest derivation sequence 
to the size of the initial term. 
Observe that the derivational complexity function is conceivable as 
a measure of proof complexity. Suppose an equational theory is 
representable as a \emph{convergent} (i.e., a confluent and terminating) TRS, then
rewriting to normal form induces an effective procedure to 
decide whether two terms are equal over a given equational theory.
Thus the derivational complexity with respect to a
convergent TRS amounts to the \emph{proof complexity} of this proof of identity.
 
In order to make further discussion more concrete, we present the central definitions. 
Let $\RS$ denote a finitely branching and terminating TRSs. 
The \emph{derivation length function} $\dheight(t,\rew)$ of a term $t$ with respect
to a rewrite relation $\rew$ is defined as 
\begin{equation*}
  \dheight(t,\rew) = \max \{n \mid 
  \exists (t_0,\dots,t_n) \colon t=t_0 \rew t_1 \rew \dots \rew t_n \}
\end{equation*}
To make the notion of derivation length independent of the choice of $t$ one defines the
\emph{derivational complexity function} (\emph{with respect to $\RS$}):
\begin{equation*}
  \Dc{\RS}(n) =\max(\{\dheight(t,\rsrew{\RS}) \mid \size{t} \leq n\}) \tkom
\end{equation*}
where $\size{\cdot}$ denotes a suitable term-complexity measure of $t$, 
e.g.\ the number of symbols in $t$. Observe that even for terminating and finitely
branching TRS $\RS$, the induced derivational complexity function $\Dc{\RS}$
is only well-defined, if either the signature $\FS$ is finite or $\RS$ is finite.
Hofbauer and Lautemann~\cite{HofbauerLautemann:1989} showed that for finite TRS
\begin{itemize}
\item a termination proof by polynomial interpretations implies a double-exponential
upper bound on the derivational complexity.
\end{itemize}
With respect to syntactically defined termination orders, 
Hofbauer~\cite{Hofbauer:1991,Hofbauer:1992} established that for finite TRSs
\begin{itemize}
\item a termination proof via the multiset path order implies that there is a primitive
    recursive bound on the derivational complexity. 
\end{itemize}
It is this result that explains why we cannot find a multiset path order
compatible with the TRS $\RSa$; recall that this TRS essentially encodes
the Ackermann function, which in turn implies that its derivational complexity
is at least the Ackermann function.

Weiermann~\cite{Weiermann:1995:tcs}, and
Lepper~\cite{Lepper:2001a} established that for finite TRSs
\begin{itemize}
\item a termination proof via the lexicographic path order (LPO for short) induces a multiple
    recursive bound on the derivational complexity, and
\item if termination is shown by the Knuth-Bendix order (KBO for short), then 
the derivational complexity function is a member of $\Ack(\bO(n),0)$, where $\Ack$ denotes the 
binary Ackermann function. 
\end{itemize}
In all mentioned cases the upper bounds are optimal, i.e., it is
possible to provide TRSs, whose derivational complexity function 
form a tight lower bound on the established upper bounds. 

For a specific TRS $\RS$ the mentioned results yield precise upper bounds 
on the derivational complexity of $\RS$, i.e.,
depending on $\RS$, one can compute exact upper bounds. 
Therefore, these results constitute an a priori
complexity analysis of TRSs provably terminating by polynomial interpretations, 
MPOs, LPOs, or KBOs. As term rewriting forms the basis of declarative programming, such
complexity results transcend naturally to (worst-case) complexity results on declarative programs.

It is well-known that all mentioned termination methods are incomparable in the sense that
there exist TRSs whose termination can be shown by one of the these techniques, but not
by any of the other. Still it is a deep and interesting question, how to characterise 
the strength of termination methods in the large. Dershowitz and Okada argued that the 
order type of the employed reduction order would constitute 
a suitable uniform measure, cf.~\cite{DershowitzOkada:1988}.
Clearly the complexity induced by a termination method serves equally well (or perhaps better)
as such a measure, cf.~\cite{HofbauerLautemann:1989}. 
See~\cite{Hofbauer:1991,Cichon:1992,Hofbauer:2001,Lepper:2002,Touzet:2002} 
for further reading on this subject.

In the remainder of this thesis I refer to investigations on (derivational) complexities
of TRSs as \emph{complexity analysis of term rewrite systems}.

To conclude this section, let me apply Lepper's result to the
motivating TRS $\RS$. I kindly refer the reader to Chapter~\ref{lpar06} for a formal
definition of KBO. Consider the precedence $>$ defined by $\mf > \circ > \me$ together
with the weight function $\Wei(\mf) = \Wei(\circ) = 0$ and $\Wei(\me) = 1$. Let $\gkbo$ be the
KBO\ induced by $>$ and $\Wei$. Then for all rules $l \rew r \in \RS$, 
$l \gkbo r$. 
As $\RS$ is compatible with $\gkbo$, we infer that the
derivational complexity of $\RS$ is bounded from above by the Ackermann function. Moreover
this upper bound is tight, cf.~\cite{Hofbauer:1991}. Hence we conclude that the
derivational complexity function $\Dc{\RSa}$ features the same growth rate as the
Ackermann function.

\subsection{Proof-Theoretic Analysis of Complexity Analysis} \label{ProofTheory}

The study of the length of derivations of TRSs has stirred some attention in proof theory, 
cf.~\cite{Weiermann:1995,Buchholz:1995,Arai:1998:trs,Friedman:2001}.
This proof-theoretic interest in termination proofs and their induced complexity
is not surprising. After all, the conceptual theme here is an old theme of
proof theory. We can suit a conception question dedicated to Kreisel to the present
context and ask:
\begin{quote}
  \emph{What more do we know from a termination proof, than the mere fact of termination?}
\end{quote}
While Hofbauer, Weiermann, and Lepper gave answers to this question for specific
instances of termination proofs it was Buchholz who delivered a direct proof-theoretic
analysis. In order to explain this result, we need a few definitions. 

As usual \emph{Peano Arithmetic} refers to the first-order axiomatisation of 
number theory and $I\Sigma_1$ is the fragment of Peano Arithmetic, where the
axiom of mathematical induction is restricted to existential induction formulas. 
The fragment $I\Sigma_1$ is relatively weak, but strong enough to prove the
totality of primitive recursive functions. 

Let $f\colon \N \to \N$ be a function on the naturals and let $\Godelnum{f(x) = y}$ 
denote a suitable chosen computation predicate for the function $f$. 
Then $f$ is called \emph{provably recursive} (in the theory $T$) if
$T\vdash \forall x \exists y \ \Godelnum{f(x) = y}$ holds, i.e., the 
totality of $f$ is provable in $T$.
Note that the provably recursive functions of the theory $I\Sigma_1$ are exactly
the primitive recursive functions.
Let $\succ$ denote a simplification order like MPO, LPO, or KBO and let
$W$ denote the \emph{accessible part} of $\succ$ on $\TA(\FS,\VS)$, i.e., 
\begin{equation*}
  W = \bigcap \{ X \subseteq \TA(\FS,\VS) \mid \forall t 
  ( \forall s (s \prec t  \impl s \in X) \impl t \in X ) \tpkt
\end{equation*}
Then well-foundedness of $\succ$ can be shown using (second-order) induction over $W$
(see~\cite{Buchholz:1995} but also~\cite{GL:2001}).
In proof one uses the axioms 
$(\dag): \forall t ( \forall s (s \prec t  \impl s \in X) \gdw t \in W)$
and
$(\ddag): \forall t \in W ( \forall s ( s \prec t \impl F(s) ) \impl F(t)) \impl \forall t \in W F(t)$,
together with the definition of $\succ$.

Based on this well-foundedness proof, Buchholz observes 
that Hofbauer's result is a consequence of the following meta-theorem, cf.~\cite{Buchholz:1995}.
\begin{itemize}
\item If $\succ$ is a primitive recursive relation on $\TA(\FS,\VS)$ such that
$I\Sigma_1 \vdash \Godelnum{s \succ t} \impl \Godelnum{s \gmpo t}$ and $W$ is a $\Sigma_1$-set
such that $I\Sigma_1$ proves axioms $(\dag)$ and $(\ddag)$ for 
all $\Sigma_1$-formulas, then well-foundedness of $\succ$ is $I\Sigma_1$-provable. 
\end{itemize}
To see this, Buchholz defines \emph{finite approximations} $\succ_k$ of $\succ$ so that
the assertions of the meta-theorem are fulfilled. Furthermore he proves that compatibility
of a TRS $\RS$ with $\gmpo$ implies that ${\RS} \subseteq {\succ_k}$ for some $k$
depending only on $\RS$. Conclusively the meta-theorem asserts that 
well-foundedness of $\succ_k$ is provable in $I\Sigma_1$. Hence
the derivational complexity function $\Dc{\RS}$ is contained
in the class of provably recursive functions of $I\Sigma_1$. And thus the
derivational complexity function $\Dc{\RS}$ is primitive recursive.

The definition of these approximations is surprisingly simple: It suffices to
guarantee that $\succ_k$ fulfils the definition of MPO above and additionally 
$s \succ_k t$ implies $\size{s} + k \geqslant \size{t}$, cf.~\cite{Buchholz:1995}.
A similar argument works for LPO, i.e., Hofbauer's and
Weiermann's result can both be obtained directly by a proof-theoretic analysis.
It is worthy of note that Buchholz's proof is more general than 
the combinatorial arguments of Hofbauer and Weiermann. First observe that in
any case, we cannot dispense the assumption that the 
set $\{ \size{r} \mid l \to r \in \RS\}$ is bounded.
However in~\cite{Hofbauer:1992,Weiermann:1995} interpretations into the natural number 
are employed that crucially rest on the cardinality of $\FS$ and 
the maximal arity of the symbols in $\FS$. In contrast to this~\cite{Buchholz:1995} 
makes only use of the finiteness of the signature.

Observe that a proof-theoretical analysis of KBO along this lines is not
possible. Although well-foundedness of finite approximations of KBO are
formalisable in Peano Arithmetic, the needed fragment is too strong to
provide an optimal complexity analysis. On the other hand proof theory can be used successfully
to extend Lepper's result (see Chapter~\ref{lpar06}).

The connection of rewriting and proof theory is also addressed in~\cite{Cichon:1992}
(see also Chapter~\ref{rta03}) where Cichon emphasises a connection between the
order type of a given reduction order $\succ$ and the induced derivational
complexity. More precisely, the so-called \emph{Cichon's principle}
can be formulated as follows.
\begin{quote}
\emph{The (worst-case) complexity of a rewrite system for which 
termination is provable using a
reduction order of order type $\alpha$ is eventually dominated 
by a function from the \emph{slow-growing hierarchy} along $\alpha$.} 
\end{quote}
Here the slow-growing hierarchy denotes a hierarchy of number theoretic
functions $\Slow{\alpha}$, indexed by transfinite ordinals $\alpha$, 
whose growth rate is relatively slow: for example $\Slow{\omega}(x)= x+1$,
where $\omega$ denotes the first limit ordinal. 
This function hierarchy is sometimes called \emph{point-wise} in the
literature, cf.~\cite{Gir87}.
See~\cite[Chapter~3]{Buss:1998} for further reading.

It ought to be stressed that this principle is false in general. 
According to Cichon's principle, for any simply terminating TRS $\RS$, the
derivational complexity function $\Dc{\RS}$ should be majorised by a 
multiple-recursive function~\cite{Peter:1967}. This however, is not true.
In~\cite{Touzet:1998} Touzet introduced
a rewrite system $\RS$ coding a restrained version of the (standard) 
Hydra Battle~\cite{KirbyParis:1982} such that 
no multiple-recursive function can majorise the derivational complexity function. 
Furthermore, $\RS$ is simply terminating. 

Thus the principle fails even for simply terminating rewrite systems. Motivated
by these negative results Touzet asserts that the \emph{Hardy hierarchy}, a
hierarchy of rather fast growing functions, index by ordinals is the right
tool to connect the order type and derivation lengths. This point is enforced
by later results due to Lepper, cf.~\cite{Lepper:2004}.

However, note that Cichon's principle is correct for two instances
of simplification orders mentioned above: MPO and LPO. Essentially this 
follows from the mentioned results by Hofbauer and Weiermann, 
cf.~\cite{Hofbauer:1992,Weiermann:1995}. Buchholz's proof-theoretic
analysis provides some explanation. Namely for $\RS$ 
compatible with MPO, or LPO, the termination proof
does not make full use of the order type of (the class of) MPOs or LPOs, but only in a 
\emph{point-wise} way. Note that the termination proof can even be formalised in a provability
relation that makes use of ordinals only in a point-wise way (see Arai~\cite{Arai:1998} for
a more precise account of this connection).

Let me conclude this section by mentioning that Cichon's principle underlies the 
open problem \# 23 in the list of open problems in rewriting (\emph{RTALooP} for short, 
see~\url{http://rtaloop.mancoosi.univ-paris-diderot.fr/}).
\begin{quote}
\emph{Must any termination order used for proving termination of the Battle of Hydra and Hercules-system have the Howard ordinal\footnote{The Howard-Bachmann ordinal is the proof theoretic ordinal of the arithmetical theory of one inductive definition, see~\cite[Chapter~3]{Buss:1998}. Note that the Howard-Bachmann ordinal easily dwarfs the proof-theoretical ordinal of Peano Arithmetic $\epsilon_0$.} 
as its order type?}  
\end{quote}
In~\cite{Moser:2009} (see also~\cite{DershowitzMoser:2007}) I resolve this problem
by answering it in the negative.

\section{Research Program}

The goal of my research is to make complexity analysis of term rewrite systems:
\begin{labeling}[~-]{\textbf{modern~-}}
\item[\textbf{modern}] by studying the complexities induced by modern termination techniques,

\item[\textbf{useful}] by establishing refinements of existing termination techniques 
  guaranteeing that the induced complexity is bounded by functions of 
  low computational complexity, for example polytime computable function,

\item[\textbf{broad}] by analysing the complexities of 
  higher-order rewrite systems and for TRSs based on particular rewrite strategies.
\end{labeling}
Further I want to ensure that the results that are developed in these 
three areas provide computable and precise bounds. 
To this avail I am working (together with
Avanzini and Schnabl) on a software tool: 
the \emph{Tyrolean Complexity Tool} (\tct\ for short) that 
incorporates the most powerful techniques to analyse the complexity of 
rewrite systems that are currently at hand. 

In order to test the competitive capability of \tct, a specialised category for complexity
analysers has been integrated into the international termination competition; see
\url{http://termcomp.uibk.ac.at} or Section~\ref{Conclusion}
for further details.

\subsection{Modern Termination Techniques}
\label{Modern}

Modern termination provers rarely employ base orders as those mentioned above directly. 
To the contrary almost all modern termination provers use variants of the dependency pair method to prove termination.
Hence, we cannot easily combine the result of a modern termination prover and the above results
to get insights on the complexity of a given terminating TRS~$\RS$.

To improve the situation, I investigated modern termination techniques, as are usually employed in 
termination provers, in particular I studied the complexities induced by the dependency pair method,
cf.~Section~\ref{E:modern}. 
As indicated below, a complexity analysis based on the dependency pair method
is and (in its full generality) remains a challenging task. See~\cite{MoserSchnabl:2009} 
for recent developments in the complexity analysis of this technique.
However, if we can obtain optimal complexity analysis results induced by modern techniques, 
we can significantly extend the expressivity of termination provers in general.

Consider a sorting algorithm $\Program$ like insertion
sort and its implementation in a functional programming language
like OCaml:
\begin{ocaml}
let rec insert x = function
  | [] -> [x]
  | y :: ys -> if x <= y then x :: y :: ys 
  else y :: insert x ys ;;
let rec sort = function
  | [] -> []
  | x :: xs -> insert x (sort xs) ;;
\end{ocaml}
It is not difficult to translate this program into a TRS $\RS$, such that termination of $\RS$ implies termination of $\Program$. With ease, termination of $\RS$ can be verified automatically.
If we can extend automatic termination proofs by expressive certificates on the 
complexity of $\RS$, we obtain an automatic complexity analysis on $\Program$.
I.e., the prover gives us in addition to the assertion that $\RS$ is terminating, 
an upper bound on the complexity of $\RS$ (and therefore of $\Program$).

Of course this goal requires theoretical and practical work: 
Firstly deep theoretical considerations on the complexity induced by 
modern termination techniques are necessary and secondly modern termination provers 
have to be extended suitably to render the sought certificates automatically. Finally
the complexity preservation of the transformation from the program $\Program$ into 
the TRS $\RS$ has to be established. 
Adapting transformation techniques as mentioned in the Introduction, it seems possible to extend 
this approach to imperative programming languages like Java or C without too much difficulties.
See~\cite{OBEG:2009,FGPSF:2009} for current work on the termination analysis of imperative programs via
rewriting. 

In the following I discuss the challenges of this endeavour for the key examples of the
dependency pair method and semantic labeling. In particular a complete analysis of the former is
of utmost importance as this technique has extended the termination proving power of automatic
tools significantly. 
To clarify my point, I briefly state the central observations and apply the dependency pair
method to the example given in the introduction. For further
information on the concepts and definitions employed, I kindly refer the reader 
to~\cite{ArtsGiesl:2000,HirokawaMiddeldorp:2005}; further refinements can be found e.g.~in~\cite{HirokawaMiddeldorp:2007,Giesl:2006}. 
\begin{itemize}
\item A TRS $\RS$ is terminating if and only if for every cycle $\C$ in the dependency graph $\DG(\RS)$
there are no $\C$-minimal rewrite sequences.
\item If there exists an argument filtering $\pi$ and a reduction pair $(\gtrsim,>)$ so that
  \linebreak
  ${\pi(\RS)} \subseteq {\gtrsim}$, ${\pi(\C)} \subseteq {\gtrsim \cup >}$, and ${\pi(\C)} \cap {>} \not= {\varnothing}$, then
  there are no $\C$-minimal rewrite sequences.
\end{itemize}
Note that this result is a refinement of the characterisation
of termination mentioned on page~\pageref{p:DP} above. The dependency graph $\DG(\RS)$ 
essentially plays the role of a call graph in program analysis.

Efficient implementations of the dependency pair method consider maximal cycles instead of cycles.%
\footnote{In the literature maximal cycles are sometimes called \emph{strongly connected components}. We use this
notion in its original graph-theoretic definitions later on, see Chapter~\ref{lpar08}. 
Hence I refrain from following this convention.} 
Moreover the stated criteria are applied recursively, 
by disregarding dependency pairs that are already strictly
decreasing, cf.~\cite{HirokawaMiddeldorp:2005}. 
Consider the TRS $\RSa$, defined on page~\pageref{trs:rsa}. In the first step the rules~(1)--(3) 
are extended by the TRS $\DP(\RSa)$:
\begin{align*}
  \mf(x) \circ^{\sharp} (y \circ z) & \rew x \circ^{\sharp} (\mf(\mf(y)) \circ z) \tag{4}\\
  \mf(x) \circ^{\sharp} (y \circ z) & \rew \mf(\mf(y)) \circ^{\sharp} y \tag{5} \\
  \mf(x) \circ^{\sharp} (y \circ z) & \rew \mf^{\sharp}(\mf(y)) \tag{6}\\
  \mf(x) \circ^{\sharp} (y \circ z) & \rew \mf^{\sharp}(y) \tag{7}\\
  \mf(x) \circ^{\sharp} (y \circ (z \circ w)) & \rew x \circ^{\sharp} (z \circ (y \circ w)) \tag{8}\\
  \mf(x) \circ^{\sharp} (y \circ (z \circ w)) & \rew  z \circ^{\sharp} (y \circ w) \tag{9}\\
  \mf(x) \circ^{\sharp} (y \circ (z \circ w)) & \rew y \circ^{\sharp} w \tag{10}
\end{align*}
The next step is to compute an approximated dependency graph for $\DP(\RSa)$, presented in 
Figure~\ref{fig:1}.
\begin{figure}
\label{fig:1}
 \centering
 \begin{tikzpicture}[xscale=1.5,yscale=1.5] 
   \node[shape=circle, draw] (4) at (5,1) {$4$} ;
   \node[shape=circle, draw] (5) at (3,3) {$5$} ;
   \node[shape=circle, draw] (6) at (1,3) {$6$} ;
   \node[shape=circle, draw] (7) at (7,2) {$7$} ;
   \node[shape=circle, draw] (8) at (5,3) {$8$} ;
   \node[shape=circle, draw] (9) at (3,1) {$9$} ;
   \node[shape=circle, draw] (10) at (1,1) {$10$} ;
   \draw[->] (10) edge [loop left] (10) ;
   \draw[->] (9) edge [loop below] (9) ;
   \draw[->] (8) edge [loop above] (8) ;
   \draw[->] (5) edge [loop above] (5) ;
   \draw[->] (4) edge [loop below] (4) ;
   \draw[->] (10) edge (6) ;   
   \draw[->] (10) .. controls (1,0.25) .. (5,0.25) 
   .. controls (7,0.25) and (7,1) .. (7) ;   
   \draw[->] (10) -- (9) ;
   \draw[->] (10) -- (5) ;
   \draw[->] (9) -- (5) ;
   \draw[->] (9) -- (8) ;
   \draw[->] (9) -- (7) ;
   \draw[->] (9) -- (4) ;
   \draw[->] (9) edge (6) ;
   \draw[->] (8) -- (9) ;
   \draw[->] (8) -- (5) ;
   \draw[->] (8) -- (7) ;
   \draw[->] (8) -- (10) ;
   \draw[->] (8) edge [bend right=45] (6) ;
   \draw[->] (8) -- (4) ;
   \draw[->] (5) -- (10) ;
   \draw[->] (5) -- (9) ;
   \draw[->] (5) -- (8) ;
   \draw[->] (5) edge (6) ;
   \draw[->] (5) -- (4) ;
   \draw[->] (5) -- (7) ;
   \draw[->] (4) -- (8) ;
   \draw[->] (4) -- (7) ;
\end{tikzpicture}
\caption{The approximated dependency graph $\DP(\RSa)$.}
\end{figure}
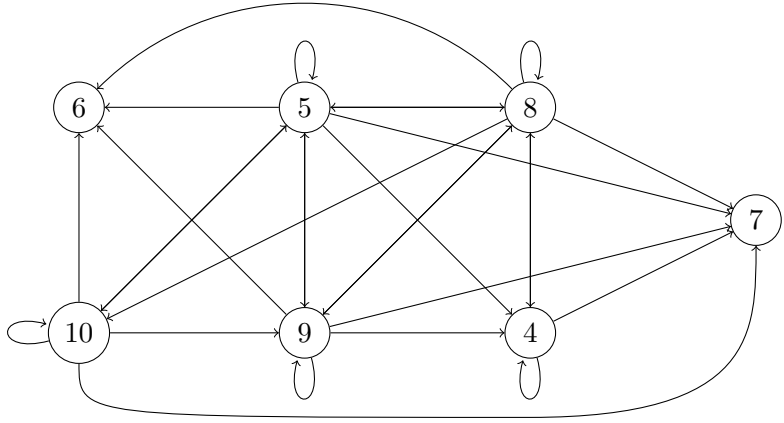

This graph contains only one maximal cycle comprising the rules $\{4,5,8,9,10\}$. 
By taking the polynomial interpretation $\mf_{\N}(x) = x$
and $\circ_{\N}(x,y) = \circ_{\N}^{\sharp}(x,y) = x + y + 1$ 
the rules in $\{1-4,8\}$ are weakly decreasing and
the rules in $\{5,9,10\}$ are strictly decreasing. The remaining maximal cycle $\{4, 8\}$ is
handled by the subterm criterion~\cite{HirokawaMiddeldorp:2007}.

This simple example should clarify the challenge of the dependency pair method in the
context of complexity analysis.
Recall that the TRS $\RSa$ (essentially) represents the binary Ackermann function. Hence the complexity
cannot be bounded by a primitive recursive function. However, the only information we can directly
gather from the given proof is the use of polynomial interpretations and the subterm criterion. Neither
of these methods is individually---i.e., as a basic termination technique---of sufficient strength to
yield an upper bound on the complexity of $\RS$. 

\subsubsection{Semantic Labeling}

Similar to the dependency pair method, semantic labeling is a transformation technique. 
Its central idea is to employ semantic information on the given TRS~$\RS$ to 
transform $\RS$ into a labeled TRS $\Rlab$ such that $\RS$ is terminating if $\Rlab$ is terminating. 
The obtained annotated TRS is typically larger and may even be infinite, but the structure may
get simpler. If this is indeed the case then TRSs whose termination proof is challenging 
can be handled with relatively simple methods. 

Let $\A$ be a model of the TRS $\RS$. A \emph{labeling} $\ell$ for $\A$ consists of a set of labels
$L_f$ together with mappings $\ell_f \colon A^n \to L_f$ for every $f \in \FS$, $f$ $n$-ary,
where $A$ is the domain of the model $\A$. For every assignment $\alpha \colon \VS \to A$, let 
$\lab{\alpha}$ denote a mapping from terms to terms defined as follows:
\begin{equation*}
  \lab{\alpha}(t) \defsym 
  \begin{cases}
    t & \text{if $t \in \VS$} \tkom \\
    f(\lab{\alpha}(t_1),\dots,\lab{\alpha}(t_n)) & \text{if $t=f(t_1,\dots,t_n)$ and $L_f = \varnothing$} \tkom \\
    f_a(\lab{\alpha}(t_1),\dots,\lab{\alpha}(t_n)) & \text{otherwise} \tpkt
  \end{cases}
\end{equation*}
The label $a$ in the last case is defined as $\ell_f(\eval{\alpha}{\A}{t_1},\dots,\eval{\alpha}{\A}{t_n})$,
where $\eval{\alpha}{\A}{t}$ denotes the evaluation of term $t$ with respect to the model $\A$ and 
the assignment $\alpha$.
The \emph{labeled} TRS $\Rlab$ is defined as
\begin{equation*}
  \{\lab{\alpha}(l) \rew \lab{\alpha}(r) \mid \text{$l \rew r \in \RS$ and $\alpha$ an assignment}\} \tpkt
\end{equation*}
Below I state the central result for one variant of semantic labeling, for further refinements 
see~\cite{Zantema:1994,Zantema:1995,Torpa:2005,KoprowskiZantema:2006,HirokawaMiddeldorp:2006}.
\begin{itemize}
\label{p:SL}
\item Let $\RS$ be a TRS, $\A$ a model for $\RS$, and
$\ell$ a labeling for $\A$. Then $\RS$ is terminating if and only if
the labelled TRS~$\Rlab$ (with respect to $\ell$) is terminating.
\end{itemize}
Note that the model $\A$ of the TRS $\RS$ is used to represent the semantic information
of $\RS$.

Semantic labeling turns out to be a promising candidate for complexity analysis. It is not
difficult to see that the derivation length of each term with respect to $\RS$ equals
its derivation length with respect to $\Rlab$. Therefore results on the complexity of
$\Rlab$ are transferable to the original system. 
The latter remains true, if refinements of the semantic labeling technique are used. 

Still there is plenty of room for research, as often the transformed system $\Rlab$ is infinite. 
However, the mentioned results on derivational complexities in Section~\ref{Complexity} 
do not necessarily carry over to infinite TRSs.  
Indeed in the case of a TRS compatible with either MPO or LPO it is easily verified
that the results become false for any computable upper bound, cf.~Section~\ref{E:modern}.

\subsubsection{Modern Direct Termination Techniques}

Let me briefly mention known results on the complexities of modern termination techniques 
that we have not yet treated. 
First, we consider the \emph{match-bound technique}~\cite{Geser2004,GHWZ07} 
a new method for automatically proving termination of left-linear term rewriting systems. 
In~\cite{GHWZ07} linear derivational complexity for linear match-bounded TRSs is established,
but exponential lower bounds exist for top-bounded TRSs.
This result extends to non-linear, 
but non-duplicating TRSs. For non-left-linear (non-duplicating) TRSs, 
the notion of match-boundedness has to be replaced by match-raise-boundedness. The
latter technique is introduced in~\cite{KM:2007}. 
Employing~\cite{HW:2004,KM:2007} it is not difficult to argue 
that any non-duplicating, right-linear, and match-raise bounded TRS 
induces at most linear derivational complexity. 
In the context of derivational complexity analysis 
the restriction to non-duplicating TRSs is harmless, as 
any duplicating TRS induces at least exponential derivational complexity, 
see Section~\ref{Established} for further details.

Secondly, consider the matrix interpretation method~\cite{Hofbauer2006,EWZ:2008}. 
In general the existence of a matrix interpretation for a given TRS $\RS$ 
induces exponential upper bounds on the derivational complexity function $\Dc{\RS}$. 
However, two specific instances of the matrix interpretation method have recently been
studied in the literature: the \emph{arctic matrix method} and \emph{triangular matrices}.

The arctic matrix method employs matrices over the arctic semi-ring that employs as
domain the set $\N \cup \{-\infty\}$ together with the operations maximisation
and addition, see~\cite{KW:2008}. This technique induces linear derivational
complexity for string rewrite systems%
\footnote{String rewrite systems are a specific class of
TRSs, such that all function symbols have unary arity.}
if employed as a direct termination technique, cf.~\cite{KW:2008}. 
On the other hand, triangular matrix interpretations restrict the form of matrices
(defined over the natural numbers $\N$ together with the usual operations) to
upper triangular form. In~\cite{MSW:2008} we establish that the induced derivational
complexity is polynomial, where the degree of the polynomial is the dimension
of the employed matrices.

\subsection{Low-Complexity Bounding Functions} \label{LowLevel}

The greatest hindrance in exploiting the majority of results on derivational complexities 
is the fact that the obtained upper bounds for general TRS are only of theoretical value. 
Even the smallest bound, i.e., the double-exponential bound, 
mentioned in Section~\ref{Complexity}, cannot be considered computationally feasible. 
Note that, while this upper bound is tight for 
the \emph{class} of polynomially terminating TRSs, it is not
difficult to find polynomially terminating TRSs whose derivational complexity functions grow much
slower than double-exponentially. In the same spirit we easily find TRSs compatible with
MPO, LPO, or KBO, respectively that do not exhibit the theoretical upper bound on the derivational complexity 
presented in Section~\ref{Complexity}.

With respect to polynomial interpretations, this 
observation led for example to the development of context-dependent interpretations, that allow
a finer analysis of the derivational complexity~\cite{Hofbauer:2001} and
with respect to LPO-termination this leads to a refined analysis of the above
mentioned result that LPO\ induces multiple recursive upper bounds, cf.~\cite{Arai:1998:trs}.
Therefore, one would want a more careful calibration of the results mentioned in
Section~\ref{Complexity}, so that the analysis of complexities of TRSs
becomes more versatile. Below I will take this further by 
striving for bounding functions on the complexity that are feasible, or at least 
belong to one of the lower classes in the polynomial hierarchy. 

This is a non-trivial task: Even in the case of linear termination, where 
we restrict the interpretation functions to linear polynomials, the derivational complexity is
optimally bounded by an exponential function, cf.~\cite{HofbauerLautemann:1989}. 
Moreover, although we can characterise
the class of \emph{polytime computable functions} (\emph{$\fptime$} for short) by carefully controlling the
way the successor symbols are interpreted, feasible upper bounds on the derivation length
requires new ideas, see for example~\cite{CichonLescanne:1992,BCMT:2001,Hofbauer:2001}
but also~\cite{fsttcs:2005} (Chapter~\ref{fsttcs05}) and~\cite{AM:2008} (Chapter~\ref{flops08}).

The quest for low-complexity bounding functions highlights a shortcoming of
the traditional notion of \emph{derivational complexity} that I will discuss now. 
While the derivational complexity function $\Dc{\RS}$ is well-motivated 
if we are mainly concerned with the strength of (direct) termination techniques 
(see~\cite{HofbauerLautemann:1989,Hofbauer:1991,Cichon:1992,Hofbauer:2001}), its
usability becomes more questionable in the wider perspective we take here.

Consider a TRS $\RS$, encoding a functional program $\Program$. Then 
the applicability of our results in program analysis hinges on the
fact that results on the complexity of $\RS$ are meaningful measures of the (runtime-)complexity
of the program $\Program$. 
For example, consider the version of insertion sort introduced in Section~\ref{Modern}.
Typically, we will not call the function \texttt{sort} iteratively, but 
\texttt{sort} will be given values as arguments. Consequently we are only
interested in the runtime complexity of such a function call, but the above
given definition of derivational complexity may overestimate
this complexity.
In particular, the following example indicates that the derivational complexity function
may overestimate the ``real'' complexity of a computed function for purely syntactic
reasons. Consider the TRS $\RSb$
\begin{align*}
\isempty(\nil) &\rew {\top} \\
\isempty({x} \cons {y}) &\rew {\bottom}\\
\hd({x} \cons {y}) &\rew x \\
\tl({x} \cons {y}) &\rew y\\
\app(x,y) &\rew \ifapp(x,y,x) \tag{$\ast$}\\
\ifapp(x,y,{u} \cons {v}) & \rew {u} \cons {\app(v,y)}  \tpkt
\end{align*}
Although the functions \emph{computed} by $\RSb$ are obviously feasible this
is not reflected in the derivational complexity of $\RS$. Consider rule $(\ast)$, which
I abbreviate as $C[x] \to D[x,x]$. 
Since the maximal derivation length starting with $C^n[x]$ equals $2^{n-1}$ for all $n > 0$,
$\RS$ admits (at least) exponential derivational complexity.
A possible solution how to overcome this obstacle, 
is discussed in Section~\ref{Established} below.
\label{trs:rsb}

\subsection{Strategies and Higher-Order Rewriting} \label{Broad}

Reduction strategies in rewriting and programming have attracted increasing attention within 
the last years. New types of reduction strategies have been invented and investigated, and 
new results on rewriting under particular strategies have been obtained. 

To be precise, I recall the central definitions, for further details 
see~\cite[Chapter~9]{Terese}.
A rewrite strategy for a TRS is a mapping $\strategy$ that assigns to every term
$s$ not in normal form a non-empty set of finite non-empty derivations starting from
$s$. We say that $s$ rewrites under the strategy $\strategy$ to the term $t$, if 
$s \rstrew{\RS} t \in \strategy(t)$. 
Typically strategies are defined by selecting the redexes which are to be contracted in
each step. Examples of such strategies are the \emph{leftmost outermost rewrite strategy}, where
always the leftmost outermost redex is selected. Likewise, the \emph{leftmost innermost strategy} contracts
the leftmost of the innermost redexes. Other examples of strategies are the \emph{parallel innermost}, 
\emph{parallel outermost}, the \emph{full substitution} and the \emph{call-by-need} strategy.

Strategies allow us to efficiently compute normal forms of weakly normalising rewrite systems. 
Thus considering complexities for TRSs governed by rewrite strategies immediately broadens 
the applicability of complexity investigations. 
The more pressing reason, why we want to investigate strategies is
that rewrite strategies allow far more efficient computations of normal forms. Through
strategies the best-case behaviour of a termination method, described as the shortest derivation
length, becomes accessible. Hence, considering strategies appears to be one step 
forward to obtain feasible upper bounds on the complexities of rewrite systems. 

Reduction strategies in rewriting are one way to broaden the applicability of 
complexity results. Another extension stems more directly from programming. 
Consider the following OCaml program $\Program$ encoding the definition of 
the higher-order function \texttt{map}:
\begin{ocaml}
let rec map f l = function
  | [] -> []
  | hd :: tl -> f hd :: map f tl;;
\end{ocaml}
Higher-order programs like $\Program$ can either be represented as 
\emph{$S$-expression rewrite systems}~\cite{Toyama:2004,Toyama:2008}
or as \emph{applicative systems}, employing a binary applicative symbol 
$\circ$, cf.~\cite{KKSV:1996,Middeldorp:1996,Giesl:2005frocos}.
While $S$-expression rewrite systems and in particular applicative systems
have been studied extensively, relative little effort has been spent to
prove termination of higher-order rewrite systems directly, 
see~\cite{Jouannaud:1998,Jouannaud:1999,Hamana:2005,Jouannaud:2006,Jouannaud:2006b,JR:2007}.

Currently the complexity analysis of higher-order systems via rewriting has not
yet attracted much attention. Although there is long established 
interest in the functional programming community to 
automatically verify complexity properties of programs, 
see for example~\cite{AKAL05,B02,GL02,R89}, no results along the
lines presented here can be found in the literature. Future
research will overcome this restriction as the applicability of complexity
analysis to rewrite systems in the context of (functional) programs is of utmost
importance to the sustainability of this research, see
Section~\ref{Future}.

\section{Contributions} \label{Established}

Above I define the derivation length function $\dheight(t,\rew)$ of term $t$ with respect
to a rewrite relation $\rew$ as the longest possible derivation (with respect to $\rew$)
starting with $t$. Based on the derivation length, Hofbauer and Lautemann defined
the derivational complexity function $\Dc{\RS}$ with respect to the (full)
rewrite relation $\rsrew{\RS}$. 
Instead I propose the following generalisation of this concept to arbitrary 
relations. Additionally this concept allows for better control on the
set of admitted start terms of a given computation.
Let $\RS$ be a TRS and ${T} \subseteq {\TA(\FS,\VS)}$ be a set of terms. 
The \emph{runtime complexity function with respect to a relation $\rew$ on $T$}
is defined as follows:
\begin{equation} \label{eq:rc}
\Rc{}(n, T, \rew) = \max\{ \dheight(t, \rew) \mid 
\text{$t \in T$ and $\size{t} \leqslant n$}\} \tpkt
\end{equation}

Based on this notion the derivational complexity function becomes
definable as follows: $\Dc{\RS}(n) = \Rc{}(n,\TA(\FS,\VS),\rsrew{\RS})$.
Currently four instances of~\eqref{eq:rc} are most prominent in research:
\begin{itemize}
\item the \emph{derivational complexity function} $\Dc{\RS}$, as defined above. 
\item the \emph{innermost derivational complexity function} 
  $\Dci{\RS}= \Rc{}(n,\TA(\FS,\VS),\irew{\RS})$. 
\item the \emph{runtime complexity function} 
  $\Rc{\RS} = \Rc{}(n, \TB, \rew)$.
\item the \emph{innermost runtime complexity function} 
  $\Rci{\RS} = \Rc{}(n, \TB, \ito)$.
\end{itemize}
Here $\irew{\RS}$ denotes the innermost rewrite relation with respect to $\RS$
and $\TB \subseteq \TA(\FS,\VS)$ denotes the set of constructor-based terms, 
cf.~\cite{BaaderNipkow:1998}. 
A constructor based term directly represents a function call with values as argument. Hence this notion 
corresponds nicely to the typical use of \emph{runtime complexity} in
the literature on functional programming, cf.~\cite{Bird:1998}. Note that 
the runtime complexity of a TRS extends the notion of the \emph{cost} of
(constructor based) term as introduced in~\cite{CKS:1989}.
Technically the broader definition has the advantage that the syntactic restriction to
non-duplicating TRSs mentioned in Section~\ref{LowLevel} can be overcome. Indeed,
as expected, the (innermost) runtime complexity function $\Rc{\RSb}$ of the TRS $\RSb$
given on page~\pageref{trs:rsb} is linear, 
and this can be verified automatically, see Chapter~\ref{ijcar08a}.

\subsection{Modern Termination Techniques} \label{E:modern}

The starting point of my research into the complexity of
rewrite systems was an investigation of Hofbauer's and Weiermann's results 
on the derivational complexity induced by MPO and LPO, see Section~\ref{Complexity}.
More precisely, in Chapter~\ref{rta03} a \emph{generalised system of fundamental sequences} 
is introduced  and its associated slow-growing hierarchy is defined. These notions provide
the tools to establish a modernised (and correct) treatment of Cichon's principle
for the simplification orders MPO and LPO. 

In order to state the central results precisely, I introduce some further definitions
(see~\cite[Chapter~4]{Buss:1998} for additional background information).
Let $\veblen$ denote the small Veblen ordinal~\cite{Schuette:1977} 
and let $\otyp(\succ)$ denote the order type of a well-founded relation~$\succ$.
It is well-known that $\veblen = \sup \{ \otyp(\succ) \mid \text{$\succ$ is a lexicographic path order}\}$,
cf.~\cite{Schmidt:1979}. 
Let $>$ denote a precedence on the signature $\FS$, 
let $\glpo$ denote the induced LPO and let $\int \colon \TA(\FS) \to \veblen$ denote
an interpretation from the set of ground terms into the ordinals less than $\veblen$. 
The central results of~\cite{MoserWeiermann:2003} (see Chapter~\ref{rta03}) can be paraphrased
as follows;
\begin{itemize}
\item[\eyecatch] There exists a generalised system of fundamental sequences for ordinals
below $\veblen$ that allows the definition of a point-wise relation $\gx{x}$.
Roughly speaking $\gx{x}$ denotes the descent along the $x^{th}$ branch of these
fundamental sequences.

\item[\eyecatch] If $\RS$ denotes a finite TRS compatible with $\glpo$, 
then there exists a number $k$, such that for any rule $l \rew r \in \RS$ and
any ground substitution $\rho$, we have $\int(l\rho) \gx{k} \int(r\rho)$.

\item[\eyecatch] There exists a slow-growing hierarchy of sub-recursive function 
$\bigcup_{\alpha < \veblen} \G_{\alpha}$ such that if 
$\alpha \gx{x} \beta$, then $\G_{\alpha}(x) > \G_{\beta}(x)$. 
\end{itemize}
As the hierarchy $\bigcup_{\alpha <\veblen} E(\G_{\alpha})$ characterises 
exactly the multiple-recursive functions, we re-obtain the above mentioned
result that LPO induces multiple-recursive derivational complexity.
(Here $E(f)$ denotes the elementary closure of function $f$.)

In subsequent research I generalised the introduced concepts suitably to
analyse the derivational complexity induced by the Knuth-Bendix order (see~\cite{Moser:2006lpar}).
This substantiated and clarified claims made in~\cite{MoserWeiermann:2003} that the 
provided concepts are genuinely related to the classification of the
complexity of rewrite systems for which termination is provable by a
simplification order.
In Chapter~\ref{lpar06} the derivational complexity of TRSs~$\RS$ compatible
with KBO is studied, where the signature of $\RS$ may be infinite. It is
shown that Lepper's result on the derivational complexity with respect
to finite TRS is essentially preserved, see~\cite{Moser:2006lpar} (cf.~Chapter~\ref{lpar06})
for further details.
\begin{itemize}
\item[\eyecatch] Let $\RS$ be a TRS based on a signature $\FS$ with bounded arities that is 
  compatible with a KBO $\gkbo$ and let some weak assumption on $\RS$ be fulfilled.
  Then for any term $t$:
  $\dheight(t,\rsrew{\RS}) \leqslant \Ack(2^{\bO(n)},0)$,
\end{itemize}
where the constant hidden in the big-Oh notation, 
depends only on syntactic properties of the function symbols in $t$, the TRS $\RS$
and the instance $\gkbo$ used. Note that $\FS$ need not be finite.
As a corollary to this result I re-obtain the $2$-recursive upper-bound on the derivational complexity of 
finite rewrite systems $\RS$ compatible with KBO. 

It seems worthy of note that the material presented in Chapter~\ref{lpar06} provides the first in-depth 
derivational complexity analysis of semantic labeling. 
Recall from Section~\ref{Modern} that the
central idea of semantic labeling is to transform the given TRS $\RS$ into
a system $\Rlab$ such that $\RS$ is terminating if and only if $\Rlab$ is terminating.
Furthermore showing termination of $\Rlab$ should be easier than showing termination of $\RS$.

As indicated, semantic information (i.e., a model of $\RS$) is used to define
the new system $\Rlab$. If this model is finite, then the complexity certificates for
$\Rlab$ are trivially transferable into complexity certificates for $\RS$. However,
often infinite models would be more suitable, which changes the picture
completely. See~\cite{KoprowskiZantema:2006,HirokawaMiddeldorp:2006,KoprowskiMiddeldorp:2007}
for further reading on semantic labeling with infinite models.

The main problem is that classic results on complexities of
simplification orders (see Section~\ref{Complexity}) not necessarily extend to
infinite signatures. It is not difficult to see that the complexity results on
MPO and LPO mentioned in Section~\ref{Complexity} cannot be extended to infinite
signatures, cf.~\cite{Middeldorp:1996}. 
On the other hand the above result shows that for KBO, complexity results
are transferable, even if the underlying model is infinite. Observe that
the weak restrictions mentioned, typically hold for systems obtained via the
semantic labeling transformation, see~\cite{Moser:2006lpar} or Chapter~\ref{lpar06}.

In Section~\ref{Modern} I indicated the challenges posed, if we
aim for a classification of the complexities of TRSs, whose termination is
shown by the dependency pair method. In order to tackle these difficulties 
recent efforts in this direction (see~\cite{HM:2008,HM:2008b}) concentrate on estimates 
for (innermost) runtime complexities. In this context we are most interested
in techniques that induce \emph{polynomial (innermost) runtime complexities}. 

In~\cite{HM:2008,HM:2008b} a variant of the dependency pair method for analysing runtime complexities has 
been introduced (cf.~Chapter~\ref{ijcar08a} and~\ref{lpar08}).
We show how natural improvements of the dependency pair method, like 
\emph{usable rules}, 
\emph{reduction pairs},  
\emph{argument filterings}, and \emph{dependency graphs} become applicable in this context.
More precisely, we have established a notion of dependency pairs, called \emph{weak dependency
pairs} that are applicable in the context of complexity analysis. This notion
provides us with the following method to analyse runtime complexity:
\begin{itemize}
\item[\eyecatch] Let $\RS$ be a TRS, let $\AA$ be a restricted polynomial interpretation,
essentially expressing a weight function, let $(\gtrsim,>)$ denote a reduction pair
(not necessarily based on $\AA$) that fulfils some additional conditions 
and let $\PP$ denote the set of weak dependency pairs of $\RS$ such that  
$\PP$ is non-duplicating.
Suppose the usable rules $\UU(\PP)$ of
$\PP$ are contained in $\gtrsim$ and ${\PP} \subseteq {>}$. 
Moreover, suppose ${\UU(\PP)} \subseteq {>_\A}$. 
Then the runtime complexity function $\Rc{\RS}$ with respect to $\RS$ 
depends linearly on the rank of the order $>$.
\end{itemize}
Here the \emph{rank} of a well-founded order is defined as usual.
Observe that it is very easy to verify the mentioned additional restriction 
on the reduction pair $(\gtrsim,>)$, if $(\gtrsim,>)$ is
based on a polynomial interpretation $\BB$, cf.~Chapter~\ref{ijcar08a}.
These results can be adapted for the special case of innermost rewriting. 
Here we replace the full rewrite relation $\rsrew{\RS}$ 
in the definition of runtime complexity by the innermost rewriting relation
$\irew{\RS}$. 
The established techniques are fully automatable and easy to implement. 

Let me reformulate this important result in a slightly more concrete setting.
Suppose $\AA$ is defined as above and assume $\BB$ denotes a polynomial interpretation,
fulfilling the restriction that constructors are interpreted as weights.
Then it is easy to see that if a TRS $\RS$ is compatible with such an interpretation
$\BB$ the runtime complexity of $\RS$ is polynomial (see Chapter~\ref{ijcar08a} but also~\cite{BCMT:2001}).
As a corollary to the above result we obtain:
\begin{itemize}
\item[\eyecatch] Let $\RS$ be a TRS, let $\PP$ be the set of
weak dependency pairs of $\RS$, and let $\AA$ and $\BB$ be defined
as above. 
Suppose $(\geqord{\BB},\gord{\BB})$ forms a reduction pair and in addition:
${\UU(\PP)} \subseteq {\geqord{\BB}}$ and ${\PP} \subseteq {\gord{\BB}}$,
where $\PP$ is supposed to be non-duplicating.
If $\UU(\PP) \subseteq {>_\A}$ then 
the runtime complexity function $\Rc{\RS}$ with respect to $\RS$ is polynomial.
\end{itemize}
This result significantly extends the analytic power of existing direct methods. 
Moreover this entails the first  method to analyse
the derivation length induced by the (standard) dependency pair method for innermost rewriting,
cf.~Chapter~\ref{ijcar08a}.

\subsection{Low-Complexity Bounding Functions} \label{E:LowComplexity}

As already observed in Section~\ref{LowLevel} it is not difficult to
find polynomially terminating TRSs, whose derivational complexity functions grow
significantly slower than double-exponentially. I.e., polynomial interpretations typically
overestimate the induced derivational complexity. In~\cite{Hofbauer:2001} Hofbauer introduced
context-dependent interpretations as a remedy. Consequently these interpretation provided
a starting point in the analysis of termination methods that induce polynomial derivational
complexity. Indeed in~\cite{MS:2008} (see Chapter~\ref{rta08}) such an analysis is conducted and a
new method to automatically conclude polynomial (even quadratic) derivational complexity
is given such that we obtain the following result:
\begin{itemize}
\item[\eyecatch] Let $\RS$ be compatible with a specific restriction of a context-dependent interpretation, called 
\emph{$\Delta$-restricted interpretation}. Then $\Dc{\RS}(n) = \bO(n^2)$. 
Moreover there exists a TRS $\RS$ such that $\Dc{\RS}(n) = \Omega(n^2)$.
\end{itemize}
Moreover, subsequent research revealed the existence of a tight correspondence
between a subclass of context-dependent interpretations and
restricted triangular matrix interpretations, cf.~\cite{MSW:2008}. 
On the one hand this correspondence allows for a
much simpler and more powerful method to automatically 
deduce polynomial derivational complexity. On the other hand this 
result reveals a connection between seemingly very different termination
techniques: matrix interpretations and context-dependent interpretations. Moreover
this result would not have been observed if we had investigated these techniques
directly and not the induced complexity. (Observe that no indication of this correspondence
result could be found in the literature.)

Buchholz's result (described in Section~\ref{ProofTheory}) suggests another approach.
Conceptually~\cite{Buchholz:1995} provides a new well-foundedness proof of MPO
and LPO (by induction on the accessible parts of these orders) and \emph{miniaturises}
this proof in the context of termination analysis. This entails the idea to directly
study miniaturisations of well-known reduction orders in such a way that 
infeasible growth rates are prohibited. 
Of course these miniaturisations have to be done carefully to prevent us from robbing the
order from any real termination power.

To this avail we introduce in~\cite{fsttcs:2005} (see Chapter~\ref{fsttcs05}) the 
\emph{path order for $\fptime$} (\emph{POP} for short). We could
show that POP characterises the functions computable in polytime, i.e., the
complexity class $\fptime$. In particular any function in 
$\fptime$ is representable as a TRS compatible with POP.
Moreover, we established the following result:
\begin{itemize}
\item[\eyecatch] A termination proof for a TRS $\RS$ via POP implies that for any $f \in \FS$ of arity $m$
  $\dheight(f(\Succ^{n_1}(0),\dots,\Succ^{n_m}(0)),\rsrew{\RS})$ is polynomially bounded 
 in the sum of the (binary) length of the input $\Succ^{n_1}(0),\dots,\Succ^{n_m}(0)$.
\end{itemize}

Still, in practice, the applicability of POP is limited. 
Many natural term-rewriting representations of polytime computable functions cannot be handled 
by POP as the imposed restrictions are sometimes not general enough.
To remedy this situation I studied generalisations of POP that are more broadly applicable.
These investigations resulted in the definition of a syntactic restriction of MPO, 
called~\POPSTAR, and the following result, cf.~\cite{AM:2008} (see Chapter~\ref{flops08}). 
\begin{itemize}
\item[\eyecatch] A termination proof for a TRS $\RS$ via \POPSTAR\ implies that the
innermost runtime complexity function $\Rci{\RS}$ is polynomially bounded.
\end{itemize}
Moreover \POPSTAR\ is complete for $\fptime$. 
It should be stressed that as characterisations of complexity classes the orders POP and \POPSTAR\ are
closely related. However, with respect to direct applicability and in particular
automatisation the latter result is a lot stronger.

It is worth noting that our result in~\cite{AM:2008} depends on the careful
combination of the miniaturisation of the multiset path order together with a 
specific strategy. 
Hence~\cite{AM:2008} provides an important indication of the need to consider
rewrite strategies in complexity analysis, see Section~\ref{Broad}.

\section{Related Work} \label{Related}

I mention here only work that is not already cited in Sections~\ref{Modern}--\ref{Broad}.
Concerning low-complexity bounding functions, I want to mention the connection between
the complexity analysis of a TRS $\RS$ and the \emph{computability} of $\RS$.
Roughly speaking a function $f$ is computable by a terminating TRS $\RS$ if there are function symbols
$\mF,\Succ,\Succo,0,0'$ such that
\begin{equation*}
  \mF(\Succ^{n_1}(0),\dots,\Succ^{n_m}(0)) \rssrew{\RS} \Succo^{f(n_1,\dots,n_m)}(0') \tkom
\end{equation*}
holds for all $n_1 , \dots, n_m$, cf.~\cite{HuetOppen:1980}.
The distinction between the input successor $\Succ$ and the output successor 
$\Succo$, as well as between $0$ and $0'$ is sometimes necessary to allow finer distinctions.

We say a function $f$ is computable with respect to a termination method $\mathsf{M}$, 
if $f$ is computable by a TRS that is $\mathsf{M}$-terminating.
For large complexity classes, as for example the primitive
recursive functions, the derivational complexity induced by a termination method 
implies its computability, cf.~\cite{CichonWeiermann:1997}. 
For example the class of functions computable with respect to MPO equals the primitive recursive functions, 
cf.~\cite{Hofbauer:1991}. 
For small complexity classes this equivalence is lost. 
Consider the class of polytime computable function $\fptime$. The class $\fptime$ is 
representable as the set of functions computable by TRSs $\RS$ that are compatible with
restricted polynomial interpretations $\A$, cf.~\cite{BCMT:2001}. 
On the other hand, the derivational complexity induced by $\A$ is double-exponentially, cf.~Section~\ref{Complexity}. 

This seems to strengthen the argument made above that the derivational complexity function
$\Dc{\RS}$ is not always a suitable measure of the complexity of a TRS. Kindly observe that the
runtime complexity $\Rc{\RS}$ with respect to $\RS$ induced by the interpretations $\A$ is polynomial. 
Still, we cannot equate (runtime) complexity and computability in general. The
fact that a given polytime computable function $f$ is computable by a TRS $\RS$ need not
imply that $\Rc{\RS}$ is indeed polynomial (see~\cite{fsttcs:2005} but also~\cite{BeckmannWeiermann:1996}).

The study of the computability of a given function $f$ with respect to a termination
method as outlined above is clearly connected to the investigations in
\emph{implicit computational complexity theory}. 
In the analysis of the implicit computational complexity of programs, 
one is interested in the analysis of the complexity of a given program rather than the 
study of the complexity of the function computed, or
of the problem solved. Much attention is direction towards the characterisation of 
``nice'' classes of programs that define complexity classes in the polynomial hierarchy,
most prominently the class of polytime computable functions $\fptime$. 

In particular I want to mention related work employing term rewriting as
abstract model of computation and consequently use existing techniques
from rewriting to characterise several computational complexity
classes. Interesting techniques in this context comprise the miniaturisation of simplification
orders like MPO and LPO, by Cichon and Marion, cf.~\cite{CichonMarion:1999,Marion:2003}, 
as well as the use of quasi-interpretations or sup-interpretations to characterise complexity classes
by Bonfante, Marion, Moyen, P\'echoux and others, cf.~\cite{MarionMoyen:2000,BCMT:2001,
Amadio:2005,Marion:2006,BMP:2007,BMM:2009:tcs}.

On a more general level I want to mention additional work on tiering or ramification concepts 
by Leivant, Marion, and Pfenning, cf.~\cite{LeivantMarion:1993,Leivant:1994,LeivantMarion:1997,Pfenning:2001}.
Moreover I cite Hofmann's seminal work~\cite{Hofmann:1999,Hofmann:2002} 
as well as related results by Aehlig, Schwichtenberg, and others,
cf.~\cite{AS:2002,ABHS:2004,Schwichtenberg:2006,BT:2009}. In addition there is highly interesting
work on recouce bounds of imperative programs by Niggl, Jones, Kristansen, and others,
see~\cite{Niggl:2005,NW:2006,BJK:2008,NK:2009}.

\section{Future Research} \label{Future}

In Section~\ref{Broad} I discussed the general aim to extend existing work on
complexity analysis for first order rewriting to the higher-order case. As already
mentioned one way to represent higher-order programs like the \texttt{map} function
defined in Section~\ref{Broad} are $S$-expression rewrite systems. For clarity,
we recall the definition from \cite{Toyama:2004,Toyama:2008}.
Let $\CC$ be a set of constants, $\VS$ be a set of
variables such that $\VS \cap \CC = \varnothing$, and
$\circ \notin \CC \cup \VS$ a variadic function symbol. 
We define the set $\mathcal{S}(\CC,\VS)$ of \emph{S-expressions} built from $\CC$ and $\VS$ 
as $\TA(\CC \cup \{ \circ \},\VS)$.
We write $(\sexpr{s})$ instead of $\circ(\seq{s})$.
An \emph{S-expression rewrite system} (\emph{SRS} for short) 
is a TRS with the property that the left- and right-hand
sides of all rewrite rules are S-expressions. 

Applying transformation steps, like \emph{case analysis} and \emph{rewriting of right-hand sides},
the function $\texttt{map}$, as defined in Section~\ref{Broad}, becomes representable as the following SRS:
\begin{align*}
\SL{\CCC{map} \V{f} \CCC{nil}} &\rew \CCC{nil} \\
\SL{\CCC{map} \V{f} \SL{\CCC{cons} \V{x} \V{xs}}} &\rew
\SL{\CCC{cons} \SL{\V{f} \V{x}} \SL{\CCC{map} \V{f} \V{xs}}}  
\end{align*}

In recent work together with Avanzini, Hirokawa and Middeldorp (see~\cite{AHMM:2007})
we study the runtime complexity of (a subset of)
Scheme programs by a translation into SRSs. Scheme is a statically scoped and
properly tail-recursive dialect of the Lisp programming
language invented by Guy Lewis Steele Jr. and Gerald Jay Sussman, cf.~\cite{R6RS}.
Due to its clear and simple semantics Scheme appears as an ideal candidate
to apply our results on the complexity analysis of TRSs in the context of
functional programming. 

By designing the translation to be complexity preserving (or at least closed under
polynomial functions) the complexity of the initial Scheme program can be estimated by analysing 
the complexity of the resulting SRS. Here we indicate how the above result on \POPSTAR\ is applicable to 
(a subset of) S-expression rewrite systems.

Let $\RSS$ be an SRS over $\mathcal{S}(\CC,\VS)$ and let $\CC = \DD \cup \KK$ such
that $\DD \cap \KK = \varnothing$. We call the elements of $\KK$ \emph{constructor}
constants and the elements of $\DD$ \emph{defined} constants.
We define the notion of \emph{value} in the context of SRSs.
The set of \emph{values} $\Val(\RSS)$ of $\RSS$ with respect to $\KK$ is
inductively defined as follows: (i) if $v \in \CC$ then $v \in \Val(\RSS)$,
(ii) if $\seq{v} \in \Val(\RSS)$ and $\CCC{c} \in \KK$ then
$(\CCC{c} \ v_1 \dots  v_n) \in \Val(\RSS)$.

Observe that (defined) constants are values, this reflects that in Scheme
procedures are values, cf.~\cite{R6RS} and allows for a representation
of higher-order programs. Scheme programs are conceivable as SRSs, allowing conditional
\CCC{if} expressions in conjunction with an eager, i.e., innermost
rewrite strategy. 
Thus we can delineate a class of SRSs that easily 
accommodates a relative large subset of Scheme programs, called 
\emph{constructor} SRSs in~\cite{AHMM:2007}. Based on Toyama's observation
that recursive path orders can be successfully employed to prove termination
of SRSs, we invented an automatic complexity analyser for Scheme programs, cf.~\cite{AHMM:2007}. 
The main theoretical contribution of this work can be paraphrased as follows:
\begin{itemize}
\item Let $\RSS$ be a constructor SRS compatible with \POPSTAR. 
Then the  innermost runtime complexity function $\Rci{\RS}$ (suitably adapted to
constructor SRSs) is polynomially bounded.
\end{itemize}
In conjunction with the fact that the transformation of Scheme programs 
into SRS is complexity preserving this result provides us with a
complexity analysis of Scheme programs that is fully automatable. Still,
this is only a partial result as the considered subset of Scheme programs is
only of limited practical interest. In particular
we cannot yet handle integer values. This will be subject to future research.

\section{Conclusion} \label{Conclusion}

In order to assess the complexity of a TRS it is natural to look at the maximal
length of derivation sequences, a program that has been suggested by Hofbauer
and Lautemann in~\cite{HofbauerLautemann:1989}. 
This concept has given rise to the area of \emph{derivational complexity analysis}
that produced a number of deep insights into the strength of direct termination methods,
described in Section~\ref{Complexity} and~\ref{ProofTheory}.

The goal of my subsequent research was and still is to make 
the (derivational) complexity analysis of rewrite systems modern, useful, and
broad.
For that purpose I have analysed 
the established results in order to assess their applicability in the context of 
modern termination provers. These investigations
(notably in~\cite{MoserWeiermann:2003,Moser:2006lpar}) resulted in an improved
understanding and clarification of the used concepts that often allowed the deduction
of more general results.

During this research it became apparent that the ``standard'' notion
of derivational complexity with respect to a given TRS was not the right tool
to modernise complexity analysis. Instead its generalisation to the above
introduced runtime complexity function with respect to a TRS and
a given rewrite strategy proved (up-to now) as the most useful.

Based on this conceptional advance I was able (together with various
co-authors) to modernise (derivational) complexity analysis to accommodate
modern termination techniques like context-dependent interpretations,
match-bounds, matrix interpretations, semantic labeling and dependency pairs,
as documented in~\cite{MS:2008,HM:2008,HM:2008b}.

Moreover, through the research published in~\cite{fsttcs:2005,AM:2008} the viewpoint
of (derivational) complexity analysis, is today much more focused on \emph{feasible} bounding
functions than in earlier research. This has important consequences for the
applicability of this research. Earlier investigations were mainly conducted to reveal
the strength of termination methods, while my research pushed the interest
towards the strength or complexity of rewrite systems, proper. This opens the
door to exciting applications in (automated) program analysis.

Lastly my research in this direction aims at the automation of the introduced
techniques. To this avail I am building (together with Avanzini and Schnabl)
the software tool \tct\ to analyse the complexity of 
rewrite systems automatically. In this context a specialised category for complexity
analysers has been integrated into the termination competition, see
\begin{center}
\url{http://termcomp.uibk.ac.at} \tpkt  
\end{center}
The goal of this competition is twofold. On one hand the most advanced techniques
become comparable in a direct contest. Hence different tools compete to
provide for each system the best possible complexity certificate. For example,
if we consider estimation of upper bounds, then the tool that provides the tightest
bound, gets the highest score.
On the other hand this competition provides a forum that allows to publicise the
gained results and insights. A necessity if we want to apply these results
outside rewriting.

I anticipate that the research described here will considerably advance 
the field of term rewriting. Moreover, I anticipate impact on the fields 
of \emph{implicit computational complexity theory} and 
\emph{proof theory}.

In the context of implicit computational complexity theory (see Section~\ref{Related}) 
my main interest lies in studies that employ term rewriting as
abstract model of computation and consequently use existing techniques
from rewriting to characterise several computational complexity
classes, as described in Section~\ref{Related}. 
Here I highlight the latter approach to implicit computational
complexity. In Section~\ref{Modern} we considered a functional program $\Program$ that
implements insertion sort. Interestingly $\Program$ is a challenge for 
implicit computational complexity theory as its obvious polynomial runtime
complexity cannot be easily verified. 
This was first observed by Caseiro~\cite{Caseiro:1997}, see also~\cite{Hofmann:1999,AS:2002}.
Observe that program $\Program$ can be easily transformed into the following TRS
$\RSc$:
\begin{align*}
  \mif(\top,x,y) & \rew x & 
  \ins(x,\nil) & \rew x \cons \nil\\
  \mif(\bottom,x,y) & \rew y & 
  \ins(x,y\cons z) &\rew \mif(x \leqslant y,x \cons y \cons z,y \cons \ins(x,z))\\
  \m{0} \leqslant \m{s}(y) &\rew \top & 
  \sort(\nil) & \rew \nil\\
  x \leqslant 0 &\rew \bottom &
  \sort(x \cons z) &\rew \ins(x,\sort(z))\\
  \m{s}(x) \leqslant \m{s}(y) &\rew x \leqslant y
\end{align*}
It is easy to see that $\RSc$ is MPO-terminating. Moreover, there
exists a weakly monotone max-polynomial interpretation $\AA$ such that
the interpretation of constructor symbols is restricted to weight functions.
The induced order $\geqord{\AA}$ weakly orients all rules, 
cf.~Bonfante et al.~\cite{BMM:2009:tcs}. Hence $\RSc$ belongs to a specific subclass of
rewrite systems studied in~\cite{BMM:2009:tcs} such that each function
computed by such a TRS is polytime computable.%
\footnote{Note that this does not imply that the
runtime complexity function $\Rc{\RSc}$ is polynomial, but that the function computed
is polytime computable (in the usual sense).}

In my research I am genuinely interested in ``applicable'' upper bounds on the
complexities of rewrite systems and therefore I am less concerned with the 
classification of computational complexity classes. Moreover, it seems a
not too important statement that insertion sort is a polytime computable 
function. Instead the exciting question is whether a given \emph{implementation} 
$\Program$ of insertion sort admits (at most) polynomial runtime complexity.
We thus have to clarify what exactly we accept as an implementation or program.
I would argue that in this context term rewriting systems would be a good choice
and the complexity of $\Program$ ought to be measured in the natural way
for computation model. Unfortunately, we cannot conclude
polynomial runtime complexity of $\RSc$ from the results by Bonfante et al.\
(see~\cite{BMM:2009:tcs} but also Chapter~\ref{fsttcs05}).

Still, there are many connections between complexity analysis of term rewrite systems
as discussed here and implicit computational complexity theory. 
For example the use of rewriting techniques opens the way for automatisation.
Recently, Avanzini, Schnabl and myself implemented a fully automated system that incorporates the 
majority of these techniques. See~\cite{AMS:2008} for the findings of
this experimental comparison.

Furthermore derivational complexity studies have stirred some attention in proof theory, 
cf.~\cite{Weiermann:1995,Buchholz:1995,Arai:1998:trs,Friedman:2001}. 
Clearly my research has implications for proof theory, see Section~\ref{ProofTheory}.
Here I want to emphasise that we are implicitly dealing
with the connection of partial orders and the growth-rate of functions defined by induction
on these orders: We say that a TRS $\RS$ is $\alpha$-terminating if $\RS$ is compatible with an
$\FS$-algebra $(\alpha,>)$, where $>$ denotes ordinal comparison. Any function computable by
an $\alpha$-terminating $\RS$ gives rise to a function 
defined by transfinite induction up-to $\alpha$.

A related connection was first observed by Cichon, who conjectured that the slow-growing
hierarchies would connect the order type of a termination order compatible with $\RS$ with the derivational
complexity of $\RS$, cf.~\cite{Cichon:1992}. Unfortunately, this claim is incorrect, as shown
by Touzet~\cite{Touzet:1998}. On the other hand, the principal connection refers to deep proof theoretic
questions as for example the ``naturalness'' of a given ordinal notation system, cf.~\cite{Feferman:1996,BeklemGPS:2004}, 
see also Section~\ref{ProofTheory}.

}

\chapter{Relating Derivation Lengths with the
Slow-Growing Hierarchy Directly} 
\label{rta03}

\subsection*{Publication Details}

G.~Moser and A.~Weiermann.
\newblock Relating derivation lengths with the slow-growing hierarchy directly.
\newblock In \emph{Proceedings of the 14th International Conference on
Rewriting Techniques and Applications}, number 2706 in LNCS, pages 296--310.
Springer Verlag, 2003.

\subsection*{Ranking}
The International Conference on Rewriting Techniques and Applications has been ranked
\textbf{A} by the \emph{Computing Research and Education Association of Australasia}
(\emph{CORE} for short) in 2007.%
\footnote{http://www.core.edu.au/}

\subsection*{Abstract}

In this article we introduce the notion
of a generalized system of fundamental
sequences and we define its associated
slow-growing hierarchy. 
We claim that these concepts are genuinely 
related to the classification of the
complexity---the derivation length---
of rewrite systems
for which termination is provable by a
standard termination ordering.

To substantiate this claim, we re-obtain
multiple recursive bounds on the 
the derivation length for
rewrite systems terminating under 
lexicographic path ordering, originally
established by the second author.

{
\input{rta03.sty}

\begin{section}{Introduction} \label{rta03:Introduction}

To show termination of a rewrite system $R$ one usually
shows that the induced reduction relation $\rew$ is
contained in some abstract ordering known to be well-founded.
One way to assess the strength of such a termination ordering is
to calculate its \emph{order type}, cf.~\cite{DershowitzOkada:1988}.
There appears to be a subtle relationship between these order
types and the \emph{complexity} of the rewrite system $R$ considered.
Cichon~\cite{Cichon:1992} discussed (and investigated) whether 
the complexity of a rewrite system for which termination is provable using a
termination ordering of order type $\alpha$ is eventually dominated 
by a function from the \emph{slow-growing hierarchy} along $\alpha$. 
It turned out that this principle---henceforth referred to as \CP---is
valid for the (i) \emph{multiset path ordering} ($\gmpo$)
and the (ii) \emph{lexicographic path ordering} ($\glpo$). 

More precisely, Hofbauer~\cite{Hofbauer:1992} proved 
that $\gmpo$ as termination
ordering implies primitive recursive derivation length, 
while the second author showed
that $\glpo$ as termination
ordering implies multiply-recursive derivation length~\cite{Weiermann:1995}.
If one regards the order types of $\gmpo$ and $\glpo$, respectively, 
then these results imply the correctness of \CP\ for (i) and (ii).
Buchholz~\cite{Buchholz:1995} has given an alternative proof of \CP\ for (i) and (ii).
His proof avoids the (sometimes lengthy) calculations with
functions from subrecursive hierarchies
in~\cite{Hofbauer:1992,Weiermann:1995}. 
Instead a clever application of proof-theoretic
results is used. 
Although this proof is of striking beauty, one might
miss the link to term rewriting theory that is
provided in~\cite{Hofbauer:1992,Weiermann:1995}. 

The mentioned proofs \cite{Hofbauer:1992,Weiermann:1995,Buchholz:1995}
of \CP---with respect to (i) and (ii)---%
are \emph{indirect}. I.e.\ without direct reference to
the slow-growing hierarchy.
By now, we know from the work of Touzet~\cite{Touzet:1998} and 
Lepper~\cite{Lepper:2001a,Lepper:2004} that \CP\ fails to hold in general. 
However, our interest in \CP\ is motivated by our strong belief that
there exist reliable ties between \emph{proof theory} and
\emph{term rewriting theory}. Ties which become particularly apparent if one
studies those termination orderings for which \CP\ holds.

To articulate this belief we give yet another \emph{direct}
proof of \CP\ (with respect to (i) and (ii)). To this avail
we introduce the notion of a \emph{generalized system of fundamental sequences} and we define 
its associated \emph{slow-growing hierarchy}.
These concepts are genuinely related 
to classifying derivation
lengths for rewrite systems
for which termination is proved by a
standard termination ordering. 
To emphasize this let us present the 
general outline of the proof method.

Let terms $s=t_0,t_1,\ldots,t_n$ be given, such that
$s \rew t_1 \rew \cdots \rew t_n$ holds, where $t_n$ is in normal form
and term-depth of $s$ ($\depth{s}$) is $\mbox{}\leq m$. Assume $\rew$ is contained
in a termination ordering $\succ$. Hence
$s \succ t_1 \succ \cdots \succ t_n$
holds. Assume further the sequence $(s,t_1,\ldots,t_n)$
is chosen so that $n$ is maximal. Then in the realm of
classifications of derivation lengths one usually defines
an \emph{interpretation} 
$\intN \colon \TA{\Sigma}{\VS} \to \N$ such that
$\intN(s) > \intN(t_1) > \cdots > \intN(t_n)$
holds. ($\TA{\Sigma}{\VS}$ denotes
the term algebra over the signature $\Sigma$ and the set of
variables $\VS$.)
The existence of such an interpretation then directly yields
a bound on the derivation length. 

The problem with this approach is
to guess the right interpretation from
the beginning. More often than not this is not at all obvious.
Therefore we want to generate the interpretation function
directly from the termination ordering in an intrinsic
way. To this avail we proceed as follows.
We separate $\intN$ into an \emph{ordinal interpretation} 
$\int \colon \GTA{\Sigma} \to T$ and an ordinal theoretic
function $g \colon T \to \N$. ($T$ denotes a suitable chosen
set of terms representing an initial segment of the ordinals, 
cf.~Definition~\ref{d:T}.) 
This works smoothly.
Firstly, we can employ the connection between the 
termination ordering $\succ$ and the ordering
on the notation system $T$. This connection was already
observed by Dershowitz and Okada, cf.~\cite{DershowitzOkada:1988}.
Secondly, it turns out that $g$ can be defined in terms of the
slow-growing function $G_{x} \colon T \to \N$; $x \in \N$.
(Note that we have swapped the usual
denotation of arguments, see Definition~\ref{d:G} and
Definition~\ref{l:G_schlange}.)

To simplify the presentation we restrict our attention
to a rewrite system $R$ 
whose termination can be shown by a \emph{lexicographic path ordering} $\glpo$. 
It will become apparent later that the proof presented below is (relative) 
easily adaptable to the case where the rewrite relation $\rew$ is contained
in a \emph{multiset path ordering} $\gmpo$.
We assume the signature $\Sigma$ contains at least one constant $c$.

Let $R$ be a rewrite system over $\TA{\Sigma}{\VS}$ 
such that $\rew$ is contained in a lexicographic path ordering.
Let terms $s=t_0,t_1,\ldots,t_n$ be given, such that
$s \rew t_1 \rew \cdots \rew t_n$ holds, where $t_n$ is in normal form
and $\depth{s} \leq m$.
By our choice of $R$ this implies
\begin{equation}
\label{eq:rta03:2}
s \glpo t_1 \glpo \cdots \glpo t_n \quad .
\end{equation}
We define a ground substitution $\rho$: $\rho(x) = c$, for all $x \in \VS$.
Let $>$ denote a suitable defined (well-founded) ordering relation  
on the ordinal notation system $T$. 
Let $l,r \in \TA{\Sigma}{\VS}$. 
Depending on $m$ and properties of $R$, 
we show the existence of a natural number $h$ such that
$l \glpo r$ implies $\int(l\rho) > \int(r\rho)$ and 
$G_h(\int(l\rho)) >  G_h(\int(r\rho))$, respectively.
Employing this form of an \emph{Interpretation Theorem}
we conclude from (\ref{eq:rta03:2}) for some $\alpha \in T$
\begin{equation*}
\alpha > \int(s\rho) > \int(t_1\rho) > \cdots > \int(t_n\rho) \quad.
\end{equation*}
and consequently 
\begin{equation*}
G_h(\alpha) > G_h(\int(s\rho)) > G_h(\int(t_1\rho)) > \cdots > G_h(\int(t_n\rho)) \quad.
\end{equation*}
Thus $G_h(\alpha)$ calculates an upper bound for $n$. 
Therefore the \emph{complexity} of $R$ can be measured
in terms of the \emph{slow-growing hierarchy} along the \emph{order type} of $T$.

To see that this method calculates an optimal bound, 
it remains to relate the function $G_x \colon T \to \N$ to
the multiply-recursive functions. 
We employ Girard's Hierarchy Comparison Theorem~\cite{Girard:1981}. 
Due to (a variant) of this theorem any multiple-recursive function
can be majorized by functions from the slow-growing hierarchy and
vice versa.%
\footnote{A $k$-ary function $g$ is said to be \emph{majorized} by a unary function $f$ if there
exists a number $n < \omega$ such that $\term{g}{x}{1}{k} < f(\max\{x_1,\ldots,x_k\})$,
whenever $\max\{x_1,\ldots,x_k\} \geq n$.}
(For further details see Section~\ref{Fundamental_Sequences}.)

Contrary to the original proof in\cite{Weiermann:1995}, we can thus 
circumvent technical calculations with the $F$-hierarchy (the fast-growing
hierarchy) and can shed light on the way the slow-growing hierarchy
relates the order type of the termination ordering $\succ$ to
the bound on the length of reduction sequences along $\rew$. 

\end{section}

\begin{section}{The Lexicographic Path Ordering}
\label{LPO}

We assume familiarity with the basic concepts of term rewriting. 
However, we fix some notations. 
Let $\Sigma = \{f_1,\ldots,f_K\}$ denote a finite signature such that 
any function symbol $f \in \Sigma$ has a unique \emph{arity}, denoted as 
$\ar{f}$. The cardinality $K$ is assumed to be fixed in the sequel. 
To avoid trivialities we demand
that $\Sigma$ is non-empty and contains at least one constant, i.e.\
a function symbol of arity $0$.
We set $N \defsym \max \{ \ar{f} \colon f \in \Sigma \}$.

The set of terms over $\Sigma$ and the countably infinite
set of variables \VS\ is denoted as \TA{\Sigma}{\VS}. 
We will use the meta-symbols $l,r,s,t,u,\ldots$ to denote
terms. The set of variables occurring in a term $t$ is denoted
as $\var{t}$. A term $t$ is called \emph{ground} or \emph{closed} if
$\var{t}=\emptyset$. 
The set of ground terms over $\Sigma$ is denoted as \GTA{\Sigma}.
If no confusion can arise, the reference to the signature $\Sigma$
and the set of variables $\VS$ is dropped.
With $\depth{s}$ we denote the \emph{term depth} of $s$,
defined as $\depth{s} \defsym 0$, if $s\in \VS$ or $s \in \Sigma$ and
otherwise 
$\depth{\term{f}{s}{1}{m}} \defsym 
\max \{ \depth{s_i} \colon 1 \leq i \leq m \} + 1$.
A \emph{substitution} $\sigma \colon \VS \to \T$ is a mapping from the
set of variables to the set of terms. 
The application of a substitution $\sigma$ to a term $t$ is
(usually) written as $t\sigma$ instead of $\sigma(t)$.

A \emph{term rewriting system} (or \emph{rewrite system}) $R$ over 
\T\ is a finite set of rewrite rules $(l,r)$. 
The {\em rewrite relation} $\rew$ on $\T$ is the least binary 
relation on $\T$ containing $R$ such that 
(i) if $s \rew t$ and $\sigma$ a substitution, then 
$s\sigma \rew t\sigma$ holds, and
(ii) if $s \rew t$, then $f(\ldots, s,\ldots) \rew f(\ldots ,t,\ldots)$.
A rewrite system $R$ is {\em terminating}
if there is no infinite sequence $\langle t_i\colon i \in \N \rangle$ 
of terms such that 
$t_1\rew t_2\rew\cdots\rew t_m\rew \cdots$.
Let $\succ$ denote a total order on $\Sigma$ such that $f_j \succ f_i \gdw j > i$ for
$i,j \in \{1,\ldots,K\}$. The \emph{lexicographic path ordering} 
$\glpo$ on $\T$ (induced by $\succ$) is defined
as follows, cf.~\cite{BaaderNipkow:1998}.

\begin{definition}
\label{d:LPO}
$s \glpo t$ iff
\begin{enumerate}
\item $t \in \var{s}$ and $s \not= t$, or
\item \label{en:LPO}
  $s= \term{f_j}{s}{1}{m}$, $t = \term{f_i}{t}{1}{n}$, and
  \begin{itemize}
  \item there exists $k$ ($1 \leq k \leq m$) with 
    $s_k\geqlpo t$, or
  \item $j > i$ and $s \glpo t_l$ for all $l = 1,\ldots,n$, or
  \item $i=j$ and $s \glpo t_l$ for all $l = 1,\ldots,n$, and 
    there exists an $i_0$ ($1 \leq i_0 \leq m$) such that 
    $s_1=t_1, \ldots s_{i_0-1} = t_{i_0-1}$ and $s_{i_0} \glpo t_{i_0}$.
  \end{itemize}
\end{enumerate}
\end{definition}

\begin{proposition} (Kamin-Levy).
\begin{enumerate}
\item If $s \glpo t$, then $\var{t} \subseteq \var{s}$.
\item For any total order $\prec$ on $\Sigma$, the induced
lexicographic order $\glpo$ is a simplification order on
\T.
\item If $R$ is a rewrite system such that $\rew$ is contained
in a lexicographic path ordering, then $R$ is terminating.
\end{enumerate}
\end{proposition}

\begin{proof}
Folklore.
\end{proof}

\end{section}

\begin{section}{Ordinal Terms and the Lexicographic Path Ordering}
\label{OrdinalTerms}

Let $N$ be defined as in the previous section. 
In this section we define a set of terms $T$ (and a subset $P \subset T$)
together with a well-ordering $<$ on $T$. The elements of $T$ are
built from $0$, $+$ and the $(N+1)$-ary function symbol $\psi$. It
is important to note that the elements of $T$ are \emph{terms} not
ordinals. Although these terms can serve as representations of
an initial segment of the set of ordinals \On, we will not make
any use of this \emph{interpretation}. 
In particular the reader 
not familiar with proof theory should have no difficulties to understand
the definitions and propositions of this section. However
some basic amount of understanding in proof theory may
be useful to grasp the origin and meaning of the presented concepts,
cf.~\cite{DershowitzOkada:1988,Lepper:2004,Schuette:1977}.
For the reader familiar with proof theory:
Note that $P$ corresponds to the set of additive principal numbers in $T$,
while $\psi$ represents the (set-theoretical) fixed-point free 
Veblen function, cf.~\cite{Schuette:1977,Lepper:2004}.
%
%
\begin{definition} 
\label{d:T}
Recursive definition of a set $T$ of 
ordinal terms, a subset $P \subset T$, and a
binary relation $>$ on $T$.
\begin{enumerate}
\item $0 \in T$.
\item If $\alpha_1, \ldots, \alpha_m \in P$ and 
  $\alpha_1 \geq \cdots \geq \alpha_m$, then 
  $\alpha_1 + \cdots + \alpha_m \in T$.
\item If $\alpha_1,\ldots,\alpha_{N+1} \in T$, then
  $\psi(\alpha_1,\ldots,\alpha_{N+1}) \in P$ and
  $\psi(\alpha_1,\ldots,\alpha_{N+1}) \in T$.
\item $\alpha \not= 0$ implies $\alpha > 0$.
\item $ \alpha > \beta_1,\ldots,\beta_m$ and $\alpha \in P$
  implies $\alpha > \beta_1 + \cdots + \beta_m$.
\item Let $\alpha = \alpha_1 + \cdots + \alpha_m$, 
  $\beta = \beta_1 + \cdots + \beta_n$. Then $\alpha > \beta$ iff
  \begin{itemize}
    \item $m > n$, and for all $i$ ($i \in \{1,\ldots,n\}$) $\alpha_i = \beta_i$, or
    \item there exists $i$ ($i \in \{1,\ldots,m\}$) such that
      $\alpha_1 = \beta_1, \dots, \alpha_{i-1} = \beta_{i-1}$, and $\alpha_i > \beta_i$.
  \end{itemize}  
\item \label{en:T}
  Let $\alpha = \term{\psi}{\alpha}{1}{N+1}$,
  $\beta =  \term{\psi}{\beta}{1}{N+1}$. Then $\alpha > \beta$ iff
  \begin{itemize}
    \item there exists $k$ ($1 \leq k \leq N+1$) with 
    $\alpha_k \geq \beta$, or
    \item $\alpha > \beta_l$ for all $l=1,\ldots,N+1$ and 
    there exists an $i_0$ ($1 \leq i_0 \leq N+1$) such that 
    $\alpha_1=\beta_1, \ldots \alpha_{i_0-1} = \beta_{i_0-1}$ and 
    $\alpha_{i_0} > \beta_{i_0}$.
    \end{itemize}
\end{enumerate}
\end{definition}

We use lower-case Greek letters to denote the elements of $T$. 
Furthermore we formally define $\alpha + 0 = 0 + \alpha = \alpha$
for all $\alpha \in T$.

We sometimes abbreviate sequences of (ordinal) terms like 
$\alpha_1,\ldots,\alpha_n$ by $\overline{\alpha}$. Hence, instead of 
$\psi(\alpha_1,\ldots,\alpha_{N+1})$ we may write $\psi(\overline{\alpha})$.
To relate the elements of $T$ to more expressive ordinal notations, we
define $1 \defsym \psi(\overline{0})$, 
$\omega \defsym \psi(\overline{0},1)$, and
$\epsilon_0 \defsym \psi(\overline{0},1,0)$. 
Let \Lim\ be the set of elements in $T$ which are neither $0$ nor of the form
$\alpha + 1$. Elements of \Lim\ are called \emph{limit} ordinal terms.

\begin{proposition} 
Let $(T,<)$ be defined as above.
Then $(T,<)$ is a well-ordering.
\end{proposition}

\begin{proof}
Let $\length{\alpha}$ denote the number of symbols in the ordinal
term $\alpha$. Exploiting induction on $\length{\alpha}$ one
easily verifies that the ordering $(T,<)$ is well-defined. To show
well-foundedness one uses induction on the lexicographic path ordering
$\llpo$, exploiting the 
close connection between Definition~\ref{d:LPO}.\ref{en:LPO}
in Section~\ref{LPO} and
Definition~\ref{d:T}.\ref{en:T} above. 
\end{proof}

In the following proposition we want to relate the \emph{order type}
of the well-ordering $(T,<)$ and the well-partial ordering
$\llpo$. Concerning the latter it is best to momentarily restrict
our attention to the well-ordering $(\GTA{\Sigma},\llpo)$. 
We indicate the arity of the
function symbol $\psi$ employed in Definition~\ref{d:T}. 
We write $(T(N+1),<)$ instead of $(T,<)$.
Similarly we write $(\GTA{\Sigma(N)},\llpo)$ to indicate
the maximal arity of function symbols in the finite signature $\Sigma$.
Let $\veblen$ denote the small Veblen ordinal~\cite{Schuette:1977}
and let $\otyp(M)$ denote the order type of a well-odering $M$.

\begin{proposition}
\label{p:LPO_T}
\begin{enumerate}
\item 
For any number $k$, there exists an order isomorphic embedding from
$(\GTA{\Sigma(k)},\llpo)$ into $(T(k+1),<)$.
\item For any number $k > 2$, there exists an order isomorphic embedding from
$(T(k),<)$ into $(\GTA{\Sigma(k)},\llpo)$. 
\item 
$\sup_{k < \omega} (\otyp((T(k),<))) = 
\sup_{k < \omega} (\otyp((\GTA{\Sigma(k)},\llpo))) = \veblen$.
\end{enumerate}
\end{proposition}

\begin{proof}
The first two assertions are a consequence of 
the well-ordering proof of $(T,<)$. We only comment
on the stated lower bound in the second one.
The statement fails for $(T(2),<)$ and $(\GTA{\Sigma(2)},\llpo)$.
The presence of the binary function symbol $+$ in $T(2)$
can make the ordering $<$ more expressive than
$\llpo$. This difference vanishes for $k \geq 3$.
The third assertion follows from~\cite{Schmidt:1979}.
\end{proof}
\end{section}

\begin{section}{Fundamental Sequences and Sub-recursive Hierarchies}
\label{Fundamental_Sequences}

To each ordinal term $\alpha \in T$ we assign a canonical sequence of ordinal
terms $\langle \alpha[x] \colon x \in \N \rangle$, the \emph{fundamental
sequence}. The concept of fundamental sequences is a crucial one in
(ordinal) proof theory. 
The main idea of utilizing fundamental sequences in term rewriting, is that the
descent along the branches of such a sequence can, informally speaking,
code rewriting steps. 
We have to wade through some technical definitions.

We define the set $\is{\overline{\alpha}}{\gamma}$, the set of
\emph{interesting subterms} of $\gamma$ (relative to $\overline{\alpha}$)
by induction on $\gamma$. 
We set $\is{\overline{\alpha}}{0} \defsym \emptyset$,
$\is{\overline{\alpha}}{\gamma_1 + \cdots + \gamma_m} \defsym 
  \bigcup_{i=1}^m \is{\overline{\alpha}}{\gamma_i}$, and
finally 
%
\vspace{-\topsep}
\begin{equation*}
\is{\overline{\alpha}}{\term{\psi}{\gamma}{1}{N+1}} \defsym 
  \left \{ \begin{array}{ll}
      \{ \psi(\overline{\gamma}) \} \quad & \mbox{if}\ (\gamma_1,\ldots,\gamma_N) \geqlex 
      (\alpha_1,\ldots,\alpha_N)\\[3pt]
      \bigcup_{i=1}^{N+1} \is{\overline{\alpha}}{\gamma_i} & \mbox{otherwise}.
      \end{array}
      \right . 
\end{equation*}

The (relative to $\overline{\alpha}$) \emph{maximal interesting
subterm} $\ms{\overline{\alpha}}{\gamma_1,\ldots,\gamma_n}$ of a non-empty
sequence $(\gamma_1,\ldots,\gamma_n)$ is defined as the maximum of the
terms occurring in $\is{\overline{\alpha}}{\gamma_i}$.
Let $\glex$ denote the lexicographic ordering on sequences of ordinal terms 
induced by $>$. 
%
Let $\overline{\alpha} = \alpha_1,\ldots,\alpha_{N} \in T$ and $\beta \in T$.
Then set
\begin{equation*}
\fix{\overline{\alpha}} \defsym \{ \psi(\overline{\gamma}, \delta) \colon 
\overline{\gamma} \glex \overline{\alpha} \mbox{ and } 
\psi(\overline{\gamma},\delta) > \alpha_i \mbox{ for all } i=1,\ldots,N
\} \quad .
\end{equation*}

For a unary function symbol $f$ we define the $n^{\rm th}$ iteration
$f^n$ inductively as (i) $f^0(x) \defsym x$, and (ii) $f^{n+1}(x) \defsym f(f^n(x))$.
We will make use of this notation for functions of higher arity by assuming
that all but one argument remain fixed. We use $\cdot$ to indicate the free
position. 
In the sequel $\lambda$ (possibly extended by a subscript) 
will always denote a limit ordinal term. 

\begin{definition}
\label{d:fundamental}
Recursive definition of $\alpha[x]$ for $x < \omega$.
\begin{eqnarray*}
0[x] & \defsym & 0\\
(\alpha_1 + \cdots + \alpha_m)[x] & \defsym & \alpha_1 + \cdots + \alpha_m[x] \qquad
  m > 1, \alpha_1 \geq \cdots \geq \alpha_m\\
\psi(\overline{0})[x] & \defsym & 0\\
\psi(\overline{0},\beta+1)[x] & \defsym & \psi(\overline{0},\beta) \cdot (x+1)\\
\psi(\overline{0},\lambda)[x] & \defsym & \psi(\overline{0},\lambda[x]) \qquad
  \lambda \not\in \fix{\overline{0}}\\
\psi(\overline{0},\lambda)[x] & \defsym & \lambda \cdot (x+1) \qquad
  \lambda \in \fix{\overline{0}}\\
\psi(\alpha_1,\ldots,\alpha_i + 1, \overline{0}, 0)[x] & \defsym &
   \psi(\alpha_1,\ldots,\alpha_i,\cdot,\overline{0})^{x+1}(0)\\
\psi(\alpha_1,\ldots,\alpha_i + 1, \overline{0}, \beta+1)[x] & \defsym &
   \psi(\alpha_1,\ldots,\alpha_i,\cdot,\overline{0})^{x+1}
   (\psi(\alpha_1,\ldots,\alpha_i + 1, \overline{0}, \beta))\\
\psi(\alpha_1,\ldots,\alpha_i + 1, \overline{0}, \lambda)[x] & \defsym &
   \psi(\alpha_1,\ldots,\alpha_i+1,\overline{0},\lambda[x]) \qquad
   \lambda \not\in \fix{\overline{\alpha},\overline{0}}\\
\psi(\alpha_1,\ldots,\alpha_i + 1, \overline{0}, \lambda)[x] & \defsym &
   \psi(\alpha_1,\ldots,\alpha_i,\cdot,\overline{0})^{x+1}(\lambda) \qquad
   \lambda \in \fix{\overline{\alpha},\overline{0}}\\
\psi(\alpha_1,\ldots,\lambda_i, \overline{0}, 0)[x] & \defsym &
   \psi(\alpha_1,\ldots,\lambda_i[x],\overline{0}, 
   \ms{\overline{\alpha},\lambda_i,\overline{0}}{\overline{\alpha},\lambda_i})\\
\psi(\alpha_1,\ldots,\lambda_i, \overline{0}, \beta+1)[x] & \defsym &
   \psi(\alpha_1,\ldots,\lambda_i[x],\overline{0}, 
        \psi(\alpha_1,\ldots,\lambda_i,\overline{0},\beta))\\
\psi(\alpha_1,\ldots,\lambda_i, \overline{0}, \lambda)[x] & \defsym &
   \psi(\alpha_1,\ldots,\lambda_i,\overline{0}, \lambda[x]) \qquad 
   \lambda \not\in \fix{\overline{\alpha},\overline{0}}\\
\psi(\alpha_1,\ldots,\lambda_i, \overline{0}, \lambda)[x] & \defsym &
   \psi(\alpha_1,\ldots,\lambda_i[x],\overline{0}, \lambda) \qquad 
   \lambda \in \fix{\overline{\alpha},\overline{0}}
\end{eqnarray*}
\end{definition}

The above definition is given
in such a way as to simplify the comparison between
the fundamental sequences for $T$ and the 
fundamental sequences for the set of ordinal terms $T(2)$ 
(built from $0$, $+$, and a 2-ary function symbol $\psi$)
as presented in~\cite{Weiermann:2001}.
Note that our definition is equivalent to the more compact one 
presented in~\cite{Lepper:2004}.
The following proposition is stated without proof.
A proof (for a slightly different 
assignment of fundamental sequences) can be found in~\cite{Buchholz:2003}.

\begin{proposition}
Let $\alpha \in T$ be given; assume $x < \omega$. If $\alpha > 0$,
then $\alpha > \alpha[x]$. For $\alpha > 1$ we get $\alpha[x] > 0$, and
if $\alpha \in \Lim$, then $\alpha[x+1] > \alpha[x]$. Finally, if
$\beta < \alpha \in \Lim$, then there exists $x < \omega$, such that $\beta < \alpha[x]$
holds.
\end{proposition}

In the definition of $\psi(\alpha_1,\ldots,\lambda_i, \overline{0}, 0)[x]$
we introduce at the last position of $\psi$ the term 
$\ms{\overline{a},\overline{0}}{\overline{\alpha}}$.
We cannot simply dispense of this term.
To see this, we alter the definition of the crucial case.
We momentarily consider only $3$-ary $\psi$-functions;
we set $\Gamma_0 \defsym \psi(1,0,0)$ and calculate $\psi(0,\Gamma_0,0)[x]$:
\begin{eqnarray*}
\psi(0,\Gamma_0,0)[x] & = &\psi(0,\psi(1,0,0)[x],0)\\
& = & \psi(0,\psi(0,\cdot,0)^{x+1}(0),0)\\
& = & \psi(0,\cdot,0)^{x+2}(0)\\
& < & \psi(1,0,0) \quad .
\end{eqnarray*}

Hence for every $x < \omega$; $\psi(0,\Gamma_0,0)[x] < \Gamma_0$ holds. This
contradicts the last assertion of the proposition as
$\Gamma_0 < \psi(0,\Gamma_0,0)$.
As a side-remark we want to mention that the given assignment of
fundamental sequences even fulfills the \emph{Bachmann} property, 
see~\cite{Bachmann:1955}.
Utilizing Definition~\ref{d:fundamental} we are now in the position
to define sub-recursive hierarchies of ordinal functions.

%
%
\begin{definition} (The slow-growing hierarchy).
\label{d:G}
Recursive definition of the function $G_{\alpha} \colon \omega \to \omega$
for $\alpha \in T$. 
\vspace{-\topsep}
\begin{eqnarray*}
G_0(x) & \defsym & 0\\
G_{\alpha +1}(x) & \defsym & G_{\alpha}(x) + 1\\
G_{\lambda}(x) & \defsym & G_{\lambda[x]}(x) \quad .
\end{eqnarray*}
\end{definition}

%
%
\begin{definition} (The fast-growing hierarchy.)
Recursive definition of the function $F_{\alpha} \colon \omega \to \omega$
for $\alpha \in T$. 
\vspace{-\topsep}
\begin{eqnarray*}
F_0(x) & \defsym & x+1\\
F_{\alpha +1}(x) & \defsym & F_{\alpha}^{x+1}(x)\\
F_{\lambda}(x) & \defsym & F_{\lambda[x]}(x) \quad .
\end{eqnarray*}
\end{definition}

It is easy to see that $G_{\alpha}(x) < F_{\alpha}(x)$ for all $\alpha > 0$.
To see that the name of the hierarchy $\{G_{\alpha} \colon \alpha \in T\}$ 
is appropriate, it suffices to calculate  some examples. 
Take e.g.\ $G_{\omega}$:
$G_{\omega}(x) =  G_{\psi(\overline{0})\cdot (x+1)}(x) = G_{x+1}(x) = G_{x}(x) + 1 = x+1$.

Recall that a function $f$ is \emph{elementary} (in a function $g$) if $f$ is
definable explicitely from $0$, $1$, $+$, $\modminus$ (and $g$), using bounded
sum and product. $E(g)$ denotes the class of all such functions $f$.
Then $G_{\epsilon_0}$ majorizes the elementary functions $E$.
In contrast the function $F_{\omega}$ already majorizes the primitive recursive functions,
i.e.\ its growth rate is comparable to the (binary) Ackermann function. 
Furthermore the class of multiple recursive functions can be characterized
by the hierarchy $\{E(F_{\gamma}) \colon \gamma < \omega^{\omega}\}$, 
cf.~\cite{Peter:1967,Robbin:1965}.

However, the following theorem states a (surprising) connection 
between the slow- and fast-growing hierarchy. 
See e.g.~\cite{Girard:1981,CichonWainer:1983,Weiermann:2001} 
for further reading on the Hierarchy Comparison Theorem.

\begin{theorem} 
\label{t:hierarchy_comparission}
(The Hierarchy Comparison Theorem.)
%
\begin{equation*}
\bigcup_{\alpha \in T} E(G_{\alpha})
   = \bigcup_{\gamma < \omega^{N+1}} E(F_{\gamma}) \quad.
\end{equation*}
\end{theorem}

\begin{proof}
We do not give a detailed proof, but only state the main idea.
In~\cite{Weiermann:2001} the hierarchy comparison
theorem has been established for the set of ordinal terms $T(2)$ 
(built from $0$, $+$, and the function symbol $\psi$, where $\ar{\psi} = 2$).
To extend the result to $T$ it suffices
to follow the pattern of the proof in~\cite{Weiermann:2001}.

The difficult direction is to show that 
every function in the hierarchy 
$\{ F_{\gamma} \colon \gamma < \omega^{N+1} \}$
is majorized by some $G_{\alpha}$. To show this one in particular needs to extend
the proofs of Lemma~5 and Theorem~1 in~\cite{Weiermann:2001}
adequately.
The reversed direction follows by standard techniques, cf.~\cite{CichonWainer:1983}.
\end{proof}
\end{section}

\begin{section}{The Interpretation Theorem}

For all $\alpha \in T$ there are uniquely determined ordinal terms 
$\alpha_1 \geq \cdots \geq \alpha_m \in P$ such that 
$\alpha = \alpha_1 + \cdots + \alpha_m$ holds. 
In addition, for every $\alpha \in P$ there exist unique 
$\alpha_1,\ldots,\alpha_{N+1}$ such that $\alpha = \psi(\alpha_1,\ldots,\alpha_{N+1})$.
(This normal form property is trivial by definition.) 
Now assume $\alpha, \beta \in T$ with 
$\alpha = \gamma_1 + \cdots + \gamma_{m_0}$, $\beta = \gamma_{m_0+1} + \cdots + \gamma_m$.
Then the \emph{natural sum} $\alpha \# \beta$ is defined as
$\gamma_{\rho(1)} + \cdots + \gamma_{\rho(m)}$,
where $\rho$ denotes a permutation on $\{1,\ldots,m\}$ such that 
$\gamma_{\rho(1)} \geq \cdots \geq \gamma_{\rho(m)}$ holds.

Let $R$ denote a finite rewrite system whose induced rewrite
relation is contained in $\glpo$. 

\begin{definition} 
\label{d:interpretation}
Recursive definition of the interpretation function $\int \colon \GTA{\Sigma} \to T$.
Let $N$ denote the maximal arity of a function symbol in $\Sigma$. 
If $s = f_j \in \Sigma$, then set $\int(s) \defsym \psi(j,\overline{0})$. 
Otherwise, let $s=\term{f_j}{s}{1}{m}$ and set
\begin{equation*}
\int(s) \defsym \psi(j,\int(s_1),\ldots,\int(s_m)+1,\overline{0}) \quad.
\end{equation*}
\end{definition}

In the sequel of this section we show that $\int$ defines an \emph{interpretation}
for $R$ on $(T,<)$; i.e.\ we establish the following theorem.

\begin{theorem}
\label{t:interpretation1}
For all $s,t \in \GTA{\Sigma}$ we
have 
$s \rew t$ implies $\int(s) > \int(t)$. 
\end{theorem}

Unfortunately this is not strong enough. 
The problem being that $\alpha > \beta$ implies that $G_{\alpha}$ majorizes
$G_{\beta}$, only. 
Whereas to proceed with our general program---see Section~\ref{rta03:Introduction}---%
we need an interpretation theorem for a binary relation $\succ$ on $T$, such that
$\alpha \succ \beta \Rightarrow G_{\alpha}(x) > G_{\beta}(x)$ holds for all $x$.
We introduce a notion of a generalized system of fundamental
sequences. Based on this generalized notion, it is then possible
to define a suitable ordering $\succ$. 

\begin{definition} (Generalized system of fundamental sequences for $(T,<)$.)
\label{d:alpha^x}
Recursive definition of $(\alpha)^x$ for $x < \omega$.
\begin{enumerate}
\item $(0)^x \defsym \emptyset$
\item Assume $\alpha = \alpha_1 + \cdots + \alpha_m$; $m > 1$. Then
  $\beta \in (\alpha)^x$ if either
  \begin{itemize}
  \item $\beta = \alpha_1 \# \cdots \alpha_i^{\ast} \cdots \# \alpha_m$ and
    $\alpha_i^{\ast} \in (\alpha_i)^x$ holds, or
  \item $\beta = \alpha_i$.
  \end{itemize}
\item Assume $\alpha = \psi(\overline{\alpha})$. Then
  $\beta \in (\alpha)^x$ if
  \begin{itemize}
   \item $\beta = \psi(\alpha_1,\ldots,\alpha_i^{\ast},\ldots,\alpha_{N+1})$, and
      $\alpha_i^{\ast} \in (\alpha_i)^x$, or
   \item $\beta = \alpha_i + x$, 
     where $\alpha_i > 0$, or
   \item $\beta = \psi(\overline{\alpha})[x]$.
   \end{itemize}
\end{enumerate}
\end{definition}

By recursion we define the \emph{transitive closure} of the ownership
$(\alpha)^x \owns \beta$:
$(\alpha \gx{x} \beta) \gdw (\exists \gamma \in (\alpha)^{x} (
                               \gamma \gx{x} \beta \lor \gamma = \beta)
                               )$.
Let $\alpha, \beta \in T$. It is easy to verify that 
$\alpha \gx{x} \beta$ (for some $x < \omega$) implies
$\alpha > \beta$. If no confusion can arise we write
$\alpha^x$ instead of $(\alpha)^x$.

%
%
\begin{lemma} (Subterm Property)
\label{l:subterm}
Let $x < \omega$ be arbitrary.
\begin{enumerate}
\item \label{en:subterm:i}
$\alpha \lx{x} \gamma_1 \# \cdots \alpha \cdots \# \gamma_m$.
\item \label{en:subterm:ii}
$\alpha \lx{x} \psi(\gamma_1,\ldots,\alpha,\ldots,\gamma_{N+1})$.
\end{enumerate}
\end{lemma}

\begin{proof}
The first assertion is trivial. 
The second assertion follows by the definition of $\lx{x}$ and assertion~\ref{en:subterm:i}.
\end{proof}

%
%
\begin{lemma} (Monotonicity Property)
\label{l:monotonie}
Let $x < \omega$ be arbitrary.
\begin{enumerate}
\item \label{en:mono:i} 
  If $\alpha \gx{x} \beta$, then 
  $\gamma_1 \# \cdots \alpha \cdots \# \gamma_m \gx{x}
  \gamma_1 \# \cdots \beta \cdots \# \gamma_m$.
\item \label{en:mono:ii}
  If $\alpha \gx{x} \beta$, then
  $\psi(\gamma_1,\ldots,\alpha,\ldots,\gamma_{N+1}) \gx{x}
  \psi(\gamma_1,\ldots,\beta,\ldots,\gamma_{N+1})$.
\end{enumerate}
\end{lemma}

\begin{proof}
We employ induction on $\alpha$ to prove~\ref{en:mono:i}).
We write (ih) for induction hypothesis.
We may assume that $\alpha > 0$. By definition of $\alpha \gx{x} \beta$
we either have 
(i) that there exist $\delta \in \alpha^x$ and $\delta \gx{x} \beta$
or (ii) $\beta \in \alpha^x$. Firstly, one considers the latter case. Then
$(\gamma_1 \# \cdots \beta \cdots \# \gamma_m) \in 
     (\gamma_1 \# \cdots \alpha \cdots \# \gamma_m)^{x}$
holds by Definition~\ref{d:alpha^x}. Therefore 
$(\gamma_1 \# \cdots \beta \cdots \# \gamma_m) \lx{x}
     (\gamma_1 \# \cdots \alpha \cdots \# \gamma_m)$
follows. 
Now, we consider the first case. 
By assumption $\delta \gx{x} \beta$ holds, by (ih) this
implies 
$(\gamma_1 \# \cdots \delta \cdots \# \gamma_m) \gx{x}
     (\gamma_1 \# \cdots \beta \cdots \# \gamma_m)^{x}$. 
Now
$(\gamma_1 \# \cdots \alpha \cdots \# \gamma_m) \gx{x} 
      (\gamma_1 \# \cdots \delta \cdots \# \gamma_m)$
follows by definition of $\gx{x}$, 
if we replace $\beta$ by $\delta$ in the proof of
the second case. This completely proves~\ref{en:mono:i}).

To prove~\ref{en:mono:ii}) we proceed by 
induction on $\alpha$. By definition of $\alpha \gx{x} \beta$
we have either 
(i) $\delta \in \alpha^x$ and $\delta \gx{x} \beta$
or (ii) $\beta \in \alpha^x$. It is sufficient to consider the latter case,
the first case follows from the second as above. 
By Definition~\ref{d:alpha^x}, $\beta \in \alpha^x$ implies 
$\psi(\gamma_1,\ldots,\beta,\ldots,\gamma_{N+1}) \in 
     \psi(\gamma_1,\ldots,\alpha,\ldots,\gamma_{N+1})^{x}$.
\end{proof}

In the sequel we show the existence of a natural number $e$, such that 
for all $s,t \in \T$, and any ground substitution $\rho$,
$s \rew t$ implies $\int(s\rho) \gx{e} \int(t\rho)$. 
Theorem~\ref{t:interpretation1} follows then as a corollary.
The proof is involved, and makes use of a sequence of
lemmas.

\begin{lemma}
\label{l:lemma21}
Assume $\alpha,\beta \in \Lim$; $x \geq 1$. If $\alpha \gx{x} \beta$,
then $\alpha \gx{x+1} \beta + 1$ holds.
\end{lemma}

To prove the lemma we exploit the following auxiliary lemma. 

\begin{lemma}
\label{l:auxilliary}
We assume the assumptions and notation of Lemma~\ref{l:lemma21};
assume Lemma~\ref{l:lemma21} holds for all 
$\gamma,\delta \in \Lim$ with $\gamma,\delta < \alpha$. 
Then $\alpha \gx{x+1} \alpha[x+1] \geqx{x+1} \alpha[x] +1$.
\end{lemma}

\begin{proof}
The lemma follows by induction on the form of $\alpha$
by analyzing all cases of Definition~\ref{d:fundamental}.
\end{proof}

\begin{proof} (of Lemma~\ref{l:lemma21})
The proof proceeds by induction on the form of $\alpha$.
We consider only the case where $\alpha = \term{\psi}{\alpha}{1}{N+1}$.
The case where $\alpha = \alpha_1 + \cdots + \alpha_m$ is similar but simpler. 

By definition of $\alpha \gx{x} \beta$ we have either 
(i) $\gamma \in \alpha^x$ and $\gamma \gx{x} \beta$
or (ii) $\beta \in \alpha^x$.  
Assume for $\gamma \in \alpha^x$
we have already shown that $\gamma +1 \lx{x+1} \alpha$. Then for
$\beta \lx{x} \gamma$, we conclude by (ih) and the Subterm Property
$\beta + 1 \lx{x+1} \gamma \lx{x+1} \gamma + 1 \lx{x+1} \alpha$. 
Hence, it suffices to consider the second case. 
We proceed by case distinction on the form of $\beta$.

\textsc{Case} $\beta=\psi(\alpha_1,\ldots,\alpha_i^{\ast},\ldots,\alpha_{N+1})$ 
where $\alpha_i^{\ast} \in (\alpha_i)^x$ for some $i$ ($1 \leq i \leq N+1$). 
Note that $\alpha_i < \alpha$, hence (ih) is applicable
to establish $\alpha_i^{\ast} +1 \lx{x+1} \alpha_i$. 

Furthermore by the Subterm Property follows 
$\alpha_i^{\ast} \lx{x+1} \alpha_i^{\ast} +1$ and therefore
\begin{equation*}
  \psi(\alpha_1,\ldots,\alpha_i^{\ast},\ldots,\alpha_{N+1}) \lx{x+1}
\psi(\alpha_1,\ldots,\alpha_i^{\ast}+1,\ldots,\alpha_{N+1})
\end{equation*}
holds with Monotonicity. 
Applying (ih) with respect to
$\psi(\alpha_1,\ldots,\alpha_i^{\ast}+1,\ldots,\alpha_{N+1})$ we obtain
\begin{eqnarray*}
\psi(\alpha_1,\ldots,\alpha_i^{\ast},\ldots,\alpha_{N+1}) + 1 & \lx{x+1} &
    \psi(\alpha_1,\ldots,\alpha_i^{\ast}+1,\ldots,\alpha_{N+1})\\
& \lx{x+1} &
    \psi(\alpha_1,\ldots,\alpha_i,\ldots,\alpha_{N+1}) = \alpha \quad .
\end{eqnarray*}
The last inequality follows again by an application of the Monotonicity Property.

\textsc{Case} $\beta = \alpha_i + x$: Then 
$(\alpha_i + x) +1 = \alpha_i + (x+1) \lx{x+1} \alpha$.

\textsc{Case} $\beta = \psi(\overline{\alpha})[x]$. Clearly $\beta \in \Lim$. 
Then the auxiliary lemma becomes applicable. Thus 
$\psi(\overline{\alpha})[x] + 1 \leqx{x+1} \alpha[x+1] \lx{x+1} \alpha$.
\end{proof}

\begin{lemma}
\label{l:lemma17}
Let $t \in \GTA{\Sigma}$ be given. Assume $\depth{t} \leq d$, and $f_j \in \Sigma$.
If $f_j \glpo t$, then $\int(f_j) \gx{2d} \int(t)$.
\end{lemma}

\begin{proof}
We proceed by induction on \depth{t}. In the presentation of
the argument, we will frequently employ the Subterm and the
Monotonicity Property without further notice.
Set $\alpha \defsym \int(f_j)$, and $\beta \defsym \int(t)$.
Furthermore it is a crucial observation that $0 \lx{x} \alpha$
holds for any $x < \omega$, $\alpha \in T$. (This follows by a
simple induction on $\alpha$.)

\textsc{Case} $\depth{t} = 0$:
Then by assumption $t=f_i \in \Sigma$, $i < j$. Hence $i \lx{2d} j$ holds
and we conclude
$\int(t) = \psi(i,\overline{0}) \lx{2d} \psi(j,\overline{0}) = \int(f_j)$.

\textsc{Case} $\depth{t} > 0$: Let $t=\term{f_i}{t}{1}{n}$. 
Set $\beta_l \defsym \int(t_l)$ for all $l=1,\ldots,n$.
By (ih) one
obtains $\beta_l \lx{2(d-1)} \alpha$ for all $l$. 
For all $l$, we need only consider the case where
$\beta_l \in \alpha^{2(d-1)}$.
We consider $\psi(j,\overline{0})[2d]$ and apply the following sequence
of descents via $\gx{2d}$:
\begin{eqnarray*}
\psi(j,\overline{0})[2d] & = & \psi(j-1,\cdot,\overline{0})^{2d+1}(0) \\
& = & \psi(j-1,\psi(j-1,\cdot,\overline{0})^{2d}(0),\overline{0})\\
& \gx{2d} & \psi(j-1,\psi(j-1,\cdot,\overline{0})^{2d-1}(0)+1,\overline{0})\\
& \gx{2d} & \psi(j-1,\underbrace{\psi(j-1,\cdot,\overline{0})^{2d-1}(0)}_{\psi(j,\overline{0})[2(d-1)]},
                 \cdot,\overline{0})^{2d+1}(0) \quad .
\end{eqnarray*}

We define $\gamma_1 \defsym \psi(j,\overline{0})[2(d-1)]$ and
$\gamma_{k+1} \defsym \psi(j-1,\gamma_1,\ldots,\gamma_k+1,\overline{0})[2(d-1)]$.
By iteration of the above descent, we see 
\begin{eqnarray*}
\alpha[2d] & = & \psi(j,\overline{0})[2d]\\
& \gx{2d} & \psi(j-1,\gamma_1,\ldots,\gamma_n+1,\overline{0})\\ 
& \gx{2d} & \psi(j-1,\alpha[2(d-1)],\ldots,\alpha[2(d-1)]+1,\overline{0})\qquad (\delta) \quad. 
\end{eqnarray*}

Let $l$ ($1 \leq l \leq n$) be fixed. 
By assumption we have $\beta_l \in (\alpha)^{2(d-1)}$. We proceed by
case distinction on the definition of $\beta_l$. 

Assume $\beta_l = \psi(j,\overline{0})[2(d-1)]$. Then 
$\delta = \psi(j-1,\alpha[2(d-1)],\ldots,\beta_l,\ldots,\alpha[2(d-1)]+1,\overline{0})$.
Assume $\beta_l = \psi(j^{\ast},\overline{0})$, where
$j^{\ast} \in (j)^{2(d-1)}$, i.e.\ $j^{\ast} \leqx{2d} j-1 \lx{2d} j$. 
Therefore $\alpha[2(d-1)] \gx{2d} \psi(j-1,\overline{0})$. Hence
$\delta \gx{2d} \psi(j-1,\alpha[2(d-1)],\ldots,\beta_l,\ldots,\alpha[2(d-1)]+1,\overline{0})$. 
Finally assume $\beta_l = j + 2(d-1)$. Then 
$\beta_l \lx{2(d-1)+1} \psi(j-1,\overline{0}) \lx{2(d-1)+1} 
                       \psi(j-1,\cdot,\overline{0})^{2d-1}(0) = \alpha[2(d-1)]$.
Hence $\beta_l \lx{2d} \alpha[2(d-1)]$ by Lemma~\ref{l:lemma21} and therefore
$\delta \gx{2d} \psi(j-1,\alpha[2(d-1)],\ldots,\beta_l,\ldots,\alpha[2(d-1)]+1,\overline{0})$.

As $l$ was fixed but arbitrary, the above construction is valid for
all $l $. And the lemma follows.                    
\end{proof}

\begin{lemma}
\label{l:lemma18}
Let $\term{f_i}{t}{1}{n}, \term{f_j}{s}{1}{m} \in \GTA{\Sigma}$ be given;
let $d > 0$. Then
\begin{enumerate}
\item \label{en:lemma18:i}
If $i < j$, $\int(f_j(\overline{s})) \gx{2(d-1)} \int(t_l)$
  for all $l=1,\ldots,n$. Then $\int(f_j(\overline{s})) \gx{2d} \int(f_i(\overline{t}))$ holds.
\item \label{en:lemma18:ii}
If $s_1 = t_1, \ldots, s_{{i_0}-1} = t_{{i_0}-1}$, 
  $\int(s_{i_0}) \gx{2(d-1)} \int(t_{i_0})$, and
    $\int(f_j(\overline{s})) \gx{2(d-1)} \int(t_l)$, for all $l=i_0+1,\ldots,n$, then
    $\int(f_j(\overline{s})) \gx{2d} \int(f_i(\overline{t}))$ holds.
\end{enumerate}
\end{lemma}

\begin{proof}
The proof of assertion~\ref{en:lemma18:i}) is similar to
the proof of assertion~\ref{en:lemma18:ii}) but simpler. Hence, we concentrate
on~\ref{en:lemma18:ii}). 
Set $\alpha \defsym \int(f_j(\overline{s}))$; $\beta \defsym \int(f_i(\overline{t}))$;
finally set $\alpha_i \defsym \int(s_i)$ for all $i=1,\ldots,m$, and
$\beta_i \defsym \int(t:i)$ for all $i=1,\ldots,n$.
As above, we consider only the case 
where $\beta_l \in (\alpha)^{2(d-1)}$. The other case follows easily. 
\begin{eqnarray*}
\alpha[2d] & = & \psi(j,\alpha_1,\ldots,\alpha_m+1,\overline{0})[2d]\\
& = & \psi(j,\alpha_1,\ldots,\alpha_m,
            \psi(j,\alpha_1,\ldots,\alpha_m,\cdot,\overline{0})^{2d}(0),\overline{0})\\
& \gx{2d} &  \psi(j,\alpha_1,\ldots,\alpha_m,
            \psi(j,\alpha_1,\ldots,\alpha_m,\cdot,\overline{0})^{2d-1}(0)+1,\overline{0})\\
& = & \psi(j,\alpha_1,\ldots,\alpha_m,
            \underbrace{\psi(j,\alpha_1,\ldots,\alpha_m+1,\overline{0})[2(d-1)]}_{\alpha[2(d-1)]}+1,
            \overline{0}) \quad .
\end{eqnarray*}

Similar to above, we define
$\gamma_1\defsym \alpha[2(d-1)] = \psi(j,\alpha_1,\ldots,\alpha_m+1,\overline{0})[2(d-1)]$
and $\gamma_{k+1} \defsym 
\psi(j,\alpha_1,\ldots,\alpha_m,\gamma_1,\ldots,\gamma_{k+1}+1,\overline{0})[2(d-1)]$
and obtain
\begin{eqnarray*}
\alpha[2d] & \gx{2d} & \psi(j,\alpha_1,\ldots,\alpha_m,\gamma_1,\ldots,\gamma_{N-m}+1)\\
& \gx{2d} &  \psi(j,\alpha_1,\ldots,\alpha_m,\alpha[2(d-1)],\ldots,\alpha[2(d-1)]+1)\\
& \gx{2d} & \psi(j,\alpha_1,\ldots,\alpha_{i_0},\overline{0},\alpha[2(d-1)]+1) \quad .
\end{eqnarray*}

By assumption $\beta_{i_0} \lx{2(d-1)} \alpha_{i_0}$ and
by Lemma~\ref{l:lemma21} this implies $\beta_{i_0}+1 \lx{2d} \alpha_{i_0}$.
We set $\overline{\alpha} \defsym \alpha_1,\ldots,\alpha_{i_0-1}$, then
we obtain
\begin{eqnarray*}
\psi(j,\overline{\alpha},\alpha_{i_0},\overline{0},\alpha[2(d-1)]+1) 
& \gx{2d} & \psi(j,\overline{\alpha},\beta_{i_0}+1,\overline{0},\alpha[2(d-1)]+1)\\
& \gx{2d} & \psi(j,\overline{\alpha},\beta_{i_0}+1,\overline{0},\alpha[2(d-1)]+1)[2d]\\
%
%
& \gx{2d} & \psi(j,\overline{\alpha},\beta_{i_0},
                 \psi(j,\overline{\alpha},\beta_{i_0}+1,\overline{0},\alpha[2(d-1)]),
                 \overline{0})\\
& \gx{2d} & \psi(j,\overline{\alpha},\beta_{i_0},\alpha[2(d-1)]+1), \overline{0})\\
& = & \psi(j,\beta_{1},\ldots,\beta_{i_0},\alpha[2(d-1)]+1), \overline{0}) \quad .
\end{eqnarray*}

As in the first part of the proof, we obtain
$\alpha[2d] \gx{2d} 
\psi(j,\overline{\alpha},\alpha_{i_0},\overline{0},\alpha[2(d-1)]+1) \gx{2d} \mbox{}$
\begin{equation*}
\mbox{} \gx{2d}
   \psi(j,\beta_{1},\ldots,\beta_{i_0},\alpha[2(d-1)],\ldots,\alpha[2(d-1)]+1,\overline{0}) \quad.
\end{equation*}

By assumption we have $\beta_l \lx{2(d-1)} \alpha$ for all $l=1,\ldots,n$. 
It remains to prove that
this implies $\beta_l \leqx{2d} \gamma$. For this it is sufficient
to consider the case where $\beta_l \in (\alpha)^{2(d-1)}$. 
The proof proceeds by case-distinction on the construction of $\beta_l$.
The proof is similar to the respective part in the proof of 
Lemma~\ref{l:lemma17}, and hence omitted.
\end{proof}

\begin{lemma}
\label{l:lemma19}
Let $s,t \in \T$ be given. Assume $s = \term{f_j}{s}{1}{m}$,
$\rho$ is a ground substitution, $\depth{t} \leq d$. 
Assume further $s_k \glpo u$ and $\depth{u} \leq d$ implies $\int(s_k\rho) \gx{2d} \int(u\rho)$
for all $u \in \T$. Then $s \glpo t$ implies $\int(s\rho) \gx{2d} \int(t\rho)$.
\end{lemma}

\begin{proof}
The proof is by induction on $d$. 

\textsc{Case} $d=0$: 
Hence $\depth{t}=0$; therefore $t \in \VS$ or
$t = f_i \in \Sigma$. Consider $t \in \VS$.
Then $t$ is a subterm of $s$. Hence there exists 
$k$ ($1 \leq k \leq m$) s.t.\ $t$ is subterm of $s_k$. Hence $s_k \geqlpo t$,
and by assumption this implies $\int(s_k\rho) \gx{2d} \int(t\rho)$,
and therefore $\int(s\rho) \gx{2d} \int(t\rho)$ by 
the Subterm Property.

Now assume $t = f_i \in \Sigma$. As $s \glpo t$
by assumption either $i < j$ or $s_k \geqlpo t$ holds. In the latter
case, the assumptions render $\int(s_k\rho) \geqx{2d} \int(t\rho)$; hence
$\int(s\rho) \gx{2d} \int(t\rho)$. 
Otherwise, 
$\int(s\rho) = \psi(j,\int(s_1\rho),\ldots,\int(s_m\rho)+1,\overline{0})$,
while $\int(t\rho) = \int(t) = \psi(i,\overline{0})$.
As $\int(s_k\rho) \gx{x} 0$ holds for arbitrary $x < \omega$, we
conclude $\int(s\rho) \gx{2d} \int(t\rho)$.

\textsc{Case} $d>0$: 
Assume $\depth{t} > 0$. (Otherwise, the proof follows the pattern of
the case $d=0$.) 
Let $t=\term{f_i}{t}{1}{n}$, and clearly 
$\depth{t_l} \leq (d-1)$ for all $l=1,\ldots,n$. 
We start with the following observation:
Assume there exists $i_0$ s.t.\ $s \glpo t_l$ holds for all 
$l=i_0+1,\ldots,n$. Then by (ih) we have 
$\int(s\rho) \gx{2(d-1)} \int(t_l\rho)$. 

We proceed by case-distinction on $s \glpo t$. 
Assume firstly there exists $k$ ($1\leq k \leq m$) s.t.\
$s_k \geqlpo t$. Utilizing the assumptions of the lemma,
we conclude $\int(s\rho) \gx{2d} \int(t\rho)$.
Now assume $i < j$ and $s \glpo t_l$ for all $l=1,\ldots,n$. 
Clearly $s\rho,t\rho \in \GTA{\Sigma}$. By the 
observation $\int(s\rho) \gx{2(d-1)} \int(t_l\rho)$ holds. 
Hence Lemma~\ref{l:lemma18}.\ref{en:lemma18:i} becomes applicable
and therefore $\int(s\rho) \gx{2d} \int(t\rho)$ holds true.
Finally assume $i=j$; $s_1=t_1,\ldots,s_{i_0-1}=t_{i_0-1}$;
$s_{i_0} \glpo t_{i_0}$; $s \glpo t_l$, for all $l=i_0+1,\ldots,m$.
Utilizing the observation, we see that 
Lemma~\ref{l:lemma18}.\ref{en:lemma18:ii} becomes applicable
and therefore $\int(s\rho) \gx{2d} \int(t\rho)$.
\end{proof}

\begin{lemma}
\label{l:lemma25}
Let $t \in \GTA{\Sigma}$ be given, assume $\depth{t} \leq d$.
Then $ \psi(K+1,\overline{0})\gx{2d} \int(t)$.
\end{lemma}

\begin{proof}
The inductive proof follows
the pattern of the proof of Lemma~\ref{l:lemma17}.
\end{proof}

\begin{theorem}
\label{t:lemma20}
Let $l,r \in \T$ be given. Assume $\rho$ is a ground substitution,
$\depth{t} \leq d$. 
Then $l \glpo r$ implies $\int(l\rho) \gx{2d} \int(r\rho)$.
\end{theorem}

\begin{proof}
We proceed by induction on $\depth{s}$.

\textsc{Case} $\depth{s}=0$: 
Then $s$ can either be a constant or a variable. As $s \glpo t$ holds,
we can exclude the latter case. Hence assume $s=f_j$. As $f_j \glpo t$,
$t$ is closed. Hence the assumptions of the theorem imply
the assumptions of Lemma~\ref{l:lemma17} and we conclude 
$\int(s\rho) = \int(s) \gx{2d} \int(t) = \int(t\rho)$.

\textsc{Case} $\depth{s}>0$: 
Then $s$ can be written as $\term{f_j}{s}{1}{m}$. By (ih) 
$s_k \glpo u$ and $\depth{u} \leq d$ imply
$\int(s_k\rho) \gx{2d} \int(t\rho)$. 
Therefore the present assumptions contain 
the assumptions of Lemma~\ref{l:lemma19} and hence 
$\int(s\rho) \gx{2d} \int(t\rho)$ follows.
\end{proof}

\begin{theorem} (The Interpretation Theorem.)
\label{t:interpretation2}
Let $R$ denote a finite rewrite system whose induced rewrite
relation is contained in $\glpo$.
Then there exists $k < \omega$, such that 
for all $l,r \in \T$, and any ground substitution $\rho$
$l \rew r$ implies $\int(l\rho) \gx{k} \int(r\rho)$.
\end{theorem}

\begin{proof}
Set $d$ equal to $\max \{ \depth{r} \colon \exists l \ (l,r) \in R  \}$.
Then the theorem follows as a corollary to Theorem~\ref{t:lemma20} if
$k$ is set to $2d$.
\end{proof}
\end{section}

\begin{section}{Collapsing Theorem}

We define a variant of the slow-growing hierarchy, 
cf.~Definition~\ref{d:G}, suitable for our purposes.

%
%
\begin{definition} 
Recursive definition of the function $\G_{\alpha} \colon \omega \to \omega$
for $\alpha \in T$. 
\begin{eqnarray*}
\G_0(x) & \defsym & 0\\
\G_{\alpha}(x) & \defsym & \max \{\G_{\beta}(x) \colon \beta \in (\alpha)^x \} + 1 \quad.
\end{eqnarray*}
\end{definition}

\begin{lemma}
\label{l:G_schlange}
Let $\alpha \in T$, $\alpha > 0$ be given. Assume
$x < \omega$ is arbitrary.
\begin{enumerate}
\item \label{en:G_schlange:i}
$\G_{\alpha}$ is increasing. (Even strictly if $\alpha > \omega$.)
\item \label{en:G_schlange:ii}
If $\alpha \gx{x} \beta$, then $\G_{\alpha}(x) > \G_{\beta}(x)$.
\end{enumerate}
\end{lemma}

\begin{proof}
Both assertions follow by induction over $<$ on $\alpha$. 
\end{proof}

We need to know that this variant of the slow-growing hierarchy
is indeed slow-growing. We show this by verifying
that the hierarchies $\{ \G_{\alpha} \colon \alpha \in T \}$ and
$\{ G_{\alpha} \colon \alpha \in T \}$ coincide with respect
to growth-rate. 
It is a triviality to verify that there exists $\beta \in T$ such that
$\G_{\beta}$ majorizes $G_{\alpha}$. (Simply set $\beta = \alpha$.)
The other direction is less trivial. One first
proves that for any $\alpha \in T$ there exists $\gamma < \omega^{N+1}$
such that $\G_{\alpha}(x) \leq F_{\gamma}(x)$ for almost all $x$. 
Secondly one employs the Hierarchy Comparison Theorem once
more to establish the existence of $\beta \in T$ such that 
$\G_{\alpha}(x) \leq G_{\beta}(x)$ holds for almost all $x$.

\begin{theorem}
\label{t:satz4}
\begin{equation*}
\bigcup_{\alpha \in T} E(G_{\alpha}) = \bigcup_{\alpha \in T} E(\G_{\alpha})
   = \bigcup_{\gamma < \omega^{N+1}} E(F_{\gamma}) \quad .
\end{equation*}
\end{theorem}
\end{section}

\begin{section}{Complexity Bounds}
The \emph{complexity} of a terminating finite rewrite system $R$ is
measured by the \emph{derivation length} function.

\begin{definition}
The derivation length function 
$\dl{R} \colon \omega \to \omega$. 
Let $m < \omega$ be given.
$\dl{R}(m) \defsym \max\{n \colon \exists t_1,\ldots,t_n \in \T\
   \left( (t_1\rew\cdots\rew t_n) \land (\depth{t_1} \leq m) \right) \}$.
\end{definition}

Let $R$ be a rewrite system over $\T$ such that $\rew$ is contained
in a lexicographic path ordering.
Now assume that there exist $s=t_0,t_1,\ldots, t_n \in \T$
with $\depth{s} \leq m$ such that
\begin{equation*}
s \rew t_1 \rew \cdots \rew t_n
\end{equation*}
holds. By our choice of $R$ this implies 
$s \glpo t_1 \glpo \cdots \glpo t_n$.
By assumption on $\Sigma$ there exists $c \in \Sigma$, with $\ar{c} = 0$. 
We define a ground substitution $\rho$: $\rho(x) = c$, for all $x \in \VS$.
Let $k < \omega$ be defined as in Theorem~\ref{t:interpretation2}. 
Recall that $K$ denotes the cardinality of $\Sigma$.
We conclude from the Interpretation Theorem and Lemma~\ref{l:lemma25},
$\int(s\rho) \gx{k} \int(t_1\rho) \gx{k} \cdots \gx{k} \int(t_n\rho)$
and $\psi(K+1,\overline{0}) \gx{2m} \int(s\rho)$.
Setting $h \defsym \max \{2m,k\}$ and utilizing Lemma~\ref{l:lemma21}, we obtain
$\psi(K+1,\overline{0}) \gx{h} \int(s\rho) \gx{h} \cdots \gx{h} \int(t_n\rho)$.
An application of Lemma~\ref{l:G_schlange}.\ref{en:G_schlange:ii} yields
\begin{equation*}
\G_{\psi(K+1,\overline{0})}(h) > \G_{\int(s\rho)}(h) > \cdots > \G_{\int(t_n\rho)}(h) \quad.
\end{equation*}
Employing Theorem~\ref{t:satz4} we conclude the existence
of $\gamma < \omega^{\omega}$, such that
\begin{equation*}
F_{\gamma}(\max\{2m,k\}) \geq \G_{\psi(K+1,\overline{0})}(\max\{2m,k\}) \geq \dl{R}(m) \quad.
\end{equation*}

The class of multiply-recursive functions is captured by
$\bigcup_{\gamma < \omega^{\omega}} E(F_{\gamma})$, see~\cite{Robbin:1965}).
Thus we have established a multiply-recursive upper bound for the
derivation length of $R$ if $\rew$ is contained
in a lexicographic path ordering. Furthermore, this bound is essentially 
optimal, cf.~\cite{Weiermann:1995}.
\end{section}

\begin{section}{Conclusion}
The presented proof method is generally applicable. 
Let $R$ denote a rewrite system whose termination can
be shown via $\gmpo$. To yield a primitive recursive upper bound for 
the complexity of $R$ the above proof can be employed. Firstly the 
definition of the interpretation function $\int$ has
to be changed as follows.
If $s=\term{f_j}{s}{1}{m}$, then we set
\begin{equation*}
\int(s) \defsym \psi(j,\int(s_1) \# \cdots \#\int(s_m) \# 1) \quad .
\end{equation*}
Then the presented proof needs only partial changes. It suffices
to reformulate (and reprove) Lemma~\ref{l:lemma17}, \ref{l:lemma18}, \ref{l:lemma19},
and~\ref{l:lemma25}, respectively. 

Future work will be concerned with the \emph{Knuth-Bendix ordering}. 
Due to the more complicated nature of
this ordering the statement of the interpretation is not so simple.
Still we believe that only mild alterations of the given proof
are necessary.
\end{section}
}

\chapter{Proofs of Termination of Rewrite Systems for
Polytime Functions}
\label{fsttcs05}

\subsection*{Publication Details}
 T.~Arai and G.~Moser.
\newblock Proofs of termination of rewrite systems for polytime functions.
\newblock In \emph{Proceedings of the Annual Conference on Foundations of
  Software Technology and Theoretical Computer Science}, number 3821 in LNCS,
  pages 529--540. Springer Verlag, 2005.

\subsection*{Ranking}
The Conference on Foundations of Software Technology and Theoretical Computer Science has been ranked
\textbf{A} by CORE in 2007.

{
\input{fsttcs05.sty}

\subsection*{Abstract}
We define a new path order $\lpop$ so that 
for a finite rewrite system $\RS$ compatible with $\lpop$,
the \emph{complexity} or \emph{derivation length function} $\Drf{\RS}{f}$ for
each function symbol $f$ is 
guaranteed to be bounded by a polynomial in the length of the inputs.
Our results yield a simplification and clarification of the results obtained 
by Beckmann and Weiermann (Archive for Mathematical Logic, 36:11--30, 1996).

\section{Introduction}

Suppose $\C$ denotes an inductively defined class of recursive 
number-theoretic functions and suppose each $f \in \C$ is
defined via an equation (or more generally a system of equations) 
of the form
\begin{equation}
  \label{eq:defining_equation}
  f(\vecx) = 
  t(\lambda \vecy.f(\vecy),\vecx) \tkom
\end{equation}
where $t$ may involve previously defined functions. In a term-rewriting
context these defining equations are oriented from left to right and
the canonical \emph{term-rewriting characterisation} $\RC$ of $\C$ 
can be defined as follows: The signature $\Sigma$ of $\RC$ includes 
for each function $f$ in $\C$ a corresponding function symbol
$f$.
In order to represent natural numbers $\Sigma$ includes 
a constant $0$ and a unary function symbol $S$. 
I.e.\ numbers are represented by their numerals. (Later we 
represent natural numbers in the form of binary strings.) 
For each function $f \in \C - \{0,S\}$, 
defined by~\eqref{eq:defining_equation}, the rule
\begin{equation*}
  f(\vecx) \rew t(\lambda \vecy.f(\vecy), \vecx) \tkom
\end{equation*}
is added to $\RC$. In all non-pathological cases the term rewrite system (TRS for short) 
$\RC$ is terminating and confluent. $\RC$ is best understood
as a constructor TRS, where the constructors are $0$ and $S$. Hence $\RC$ may be
conceived as a \emph{functional program} implementing the functions in $\C$. 

Term-rewriting characterisations have
been studied e.g.\ in~\cite{BeckmannWeiermann:1996,CichonWeiermann:1997,Oitavem:2002,BMM:2005:rta}. The analysis of $\RC$ provides insight into the 
structure of $\C$ or renders us with a delineation of a class of rewrite systems whose complexity (measured by the length of derivations) is guaranteed to belong to the class $\C$.
Term-rewriting characterisations turn the
emphasis form the \emph{definition} of a function $f$ to its
\emph{computation}. 
An essential property of term-rewriting characterisations $\RC$ is its 
\emph{feasibility}: $\RC$ is called \emph{feasible}, if 
for each $n$-ary function $f \in \C$, there exists a function symbol $g$ in 
the signature of $\RC$ such that $g(\overline{m}_1,\dots,\overline{m}_n)$ 
computes the value of $f(m_1,\dots,m_n)$ and the derivation length
of this computation is bounded by a function from $\C$. 

We study \emph{term-rewriting characterisations} of the 
complexity class $\PTIME$. In particular, our starting point is a clever 
characterisation $\RBp$ of $\PTIME$ introduced by Beckmann and 
Weiermann. In~\cite{BeckmannWeiermann:1996}
the feasibility of $\RBp$ is established and conclusively shown that
any reduction strategy for $\RBp$ yields an algorithm for 
$f \in \PTIME$ that runs in polytime.
We provide a slight generalisation of the fact that $\RBp$ is
feasible. Moreover, we flesh out the 
crucial ingredients of the TRS $\RBp$ by defining a
\emph{path order for $\PTIME$}, denoted as $\lpop$. 
We show that for a finite TRS $\RS$, compatible with $\lpop$, 
the \emph{derivation length function} $\Drf{\RS}{f}$
is bounded by a polynomial in the length of the inputs
for any defined function symbol $f$.
Furthermore $\lpop$ is \emph{complete} 
in the sense that for any function $f \in \PTIME$, 
there exists a TRS $\RS$ computing $f$ such that termination of 
$\RS$ can be shown by $\lpop$. 

\section{A Rewrite System for $\PTIME$} \label{PREC}

In the following we need some notions from term rewriting and assume (at least nodding) 
acquaintance with term rewriting. (For background information, please see~\cite{BaaderNipkow:1998}.)
Let $\VS$ denote a countably infinite set of variables and $\Sigma$ a 
signature. The set of terms over $\Sigma$ and $\VS$ is denoted as 
$\T{\Sigma}{\VS}$, while the set of ground terms is
written as $\GTA{\Sigma}$. 
The rewrite relation induced by a rewrite system $\RS$ is
denoted as $\rsrew{\RS}$, and its transitive closure by $\rssrew{\RS}$.
We write $\norm{t}$ to denote the \emph{size} of a term $t$, i.e.\ the 
number of symbols in $t$.

\emph{Conventions}:
Terms are denoted by $r,s,t$, possibly extended by subscripts. 
We write $\vect$, to denote sequences of terms 
$t_1,\ldots,t_k \in \T{\Sigma}{\VS}$ and
$\vecg$ to denote sequences of function symbols $g_1,\ldots,g_k$, respectively.
The letters $i,j,k,l,m,n$, possible extended by subscripts will always
refer to natural numbers. 
The set of natural numbers is denoted as usual by $\N$.

We consider the class $\PTIME$ of \emph{polytime computable functions},
i.e.\ those functions computable by a deterministic Turing machine
$M$, such that $M$ runs in time $\leq p(n)$ for all inputs of length $n$, 
where $p$ denotes a polynomial. 
We consider equivalent formulations of
the class of polytime computable functions in terms of recursion schemes. 

Recursion schemes such as \emph{bounded recursion} due to Cobham~\cite{Cobham:1965} generate
exactly the functions computable in polytime.
In contrast to this, Bellantoni-Cook~\cite{BellantoniCook:1992} introduce certain \emph{unbounded} recursion schemes that distinguish between arguments as to their position in a function.
This separation of
variables gives rise to the following definition of the 
\emph{predicative recursive functions} $\mathcal{B}$;
for further details see~\cite{BellantoniCook:1992}. 
We fix a suitable signature of \emph{predicative recursive function symbols} 
$B$.
\begin{definition}
  For $k,l \in \N$ we define $B^{k,l}$ inductively.
  \begin{itemize}
  \item $S^{0,1}_i \in B^{0,1}$, where $i \in [0,1]$.
  \item $O^{k,l} \in B^{k,l}$.
  \item $U^{k,l}_r \in B^{k,l}$, for all $r \in [1,k+l]$.
  \item $P^{0,1} \in B^{0,1}$.
  \item $C^{0,3} \in B^{0,3}$.
  \item If $f \in B^{k',l'}$, 
    $g_1,\dots,g_{k'} \in B^{k,0}$, and 
    $h_1,\dots,h_{l'} \in B^{k,l}$, \\
    then $\SUB^{k,l}_{k',l'}[f,\vecg,\vech] \in B^{k,l}$.
  \item If $g \in B^{k,l}$, $h_0,h_1 \in B^{k+1,l+1}$,
    then $\PREC^{k+1,l}[g,h_{1},h_{2}] \in B^{k+1,l}$.
  \end{itemize}  
Set $B \defsym \bigcup_{k,l \in \N} B^{k,l}$.
\end{definition}

To simplify notation we usually drop the superscripts, 
when denoting predicative
recursive function symbols.
Occasionally, we even write $\SUB$ (, $\PREC$),  
instead of $\SUB^{k,l}[f,\vecg]$ (,$\PREC^{n+1}[g,h]$).
No confusion will arise from this.

The binary successor function $m \mapsto 2m+i$, $i \in \{0,1\}$ is
denoted as $\mathcal{S}_i$. Every natural number can be buildt
up from $0$ with repeated applications of $\mathcal{S}_i$. 
The binary length of a number $m$ is defined as follows:
$\len{0} \defsym 0$ and $\len{\mathcal{S}_i(m)} \defsym \len{m} + 1$.

We write $\N^{k,l}$ for $\N^k \times \N^l$ and for
$f\colon \N^{k,l} \to \N$, write 
$f(m_1,\dots,m_k;n_1,\dots,n_l)$ instead of
$f(\langle m_1,\dots,m_k \rangle, \langle n_1,\dots,n_l \rangle)$. 
The arguments occurring to the left of the semi-colon are called 
\emph{normal}, while the arguments to the right are called \emph{safe}. 
We define the following functions:
$\mathcal{S}^{0,1}_i$, $i \in \{0,1\}$ denotes the function 
$\langle ; m \rangle \mapsto 2m + i$.
$\mathcal{O}^{k,l}$ denotes the function
$\langle \vecm;\vecn \rangle \mapsto 0$.
$\mathcal{U}^{k,l}_r$ denotes the function
$\langle m_1,\dots,m_k;m_{k+1},\dots,m_{k+l} \rangle \mapsto m_r$.
$\mathcal{P}^{0,1}$ denotes the unique number-theoretic function
satisfying the following equations:
$f(;0) = 0$, $f(;\mathcal{S}_i(m)) = m$.
$\mathcal{C}^{0,3}$ denotes the unique function
satisfying:
$f(;0,m_0,m_1) = m_0$, $f(;\mathcal{S}_i(m),m_0,m_1) = m_i$.
If $f \colon \N^{k',l'} \mathrel{\to} \N$, $g_i \colon \N^{k,0} \mathrel{\to} \N$ for 
$i \in [1,k']$, $h_j \colon \N^{k,l}  \mathrel{\to}\N$ for $j \in [1,l']$, then
$\mathcal{SUB}^{k,l}_{k',l'}[f,\vecg,\vech]$ denotes the function
\begin{equation*}
  \langle \vecm;\vecn \rangle \mathrel{\mapsto} f(g_1(\vecm;), \dots, g_{k'}(\vecm;);
  h_1(\vecm;\vecn), \\
\dots, h_{l'}(\vecm;\vecn)) \tpkt
\end{equation*}
%
If $g \colon \N^{k,l} \to \N$, $h_i \colon \N^{k+1,l+1} \to \N$
for $i \in [0,1]$ then $\mathcal{PREC}^{k+1,l}[g,h_{1},h_{2}]$ denotes
the number-theoretic function $f$ satisfying:
$f(0,\vecm;\vecn) = g(\vecm;\vecn)$ and
$f(\mathcal{S}_i(m),\vecm;\vecn) = 
h_{i}(m,\vecm;\vecn,f(m,\vecm;\vecn))$.

\begin{definition}
  For $k,l \in \N$ we define $\B^{k,l}$ inductively.
  \begin{itemize}
  \item $\mathcal{S}^{0,1}_i \in \B^{0,1}$, where $i \in [0,1]$.
  \item $\mathcal{O}^{k,l} \in \B^{k,l}$.
  \item $\mathcal{U}^{k,l}_r \in \B^{k,l}$, for all $r \in [1,k+l]$.
  \item $\mathcal{P}^{0,1} \in \B^{0,1}$.
  \item $\mathcal{C}^{0,3} \in \B^{0,3}$.
  \item If $f \in \B^{k',l'}$, 
    $g_1,\dots,g_{k'} \in \B^{k,0}$, and 
    $h_1,\dots,h_{l'} \in \B^{k,l}$, then 
    $\mathcal{SUB}^{k,l}_{k',l'}[f,\vecg,\vech] \in \B^{k,l}$.
  \item If $g \in \B^{k,l}$, $h_0,h_1 \in \B^{k+1,l+1}$,
    then $\mathcal{PREC}^{k+1,l}[g,h_{1},h_{2}] \in \B^{k+1,l}$.
  \end{itemize}
The set of \emph{predicative recursive functions} is defined as $\B = \bigcup_{k,l} \B^{k,l}$.
\end{definition}

It follows from the definitions that for each $f \in B$, there exists a
unique predicative recursive function $f^{\B}$; the latter is called
the \emph{interpretation} of $f$ in $\B$.
For every number $m$
we define its \emph{numeral} $\overline{m} \in \T{B}{\VS}$ as follows:
$\overline{0} \defsym 0$, 
$\overline{\mathcal{S}_i(;m)} \defsym S_i(;m)$ for $i \in [0,1]$.
We write $\overline{\vecm}$ to denote a sequence of
numerals $\overline{m}_1,\dots,\overline{m}_k$.
Now the polytime computable functions $\PTIME$ can be defined as follows,
see~\cite{BellantoniCook:1992}:
\begin{equation*}
\PTIME = \bigcup_{k} \, \B^{k,0} \tpkt  
\end{equation*}
In~\cite{BeckmannWeiermann:1996} a clever \emph{feasible} term-rewriting characterisation $\RBp$ of 
the predicative recursive functions $\B$ is given. By Bellantoni's result this yields
a feasible term-rewriting characterisation of the class of polytime computable
functions $\PTIME$. The (infinite) TRS is given in Table~1.

\begin{table}[t]  \label{tab:predrec-trs}
\begin{center}
\noindent\fbox{\begin{minipage}{72ex}
\caption{A Feasible Term-Rewriting Characterisation of the Predicative Recursive Functions}
\vspace{-5ex}
\begin{align*}
  & O^{k,l}(\bfx;\bfa) \rew 0 \tkom && \textbf{[zero]} \\
  & U^{k,l}(x_{1},\ldots,x_{k};x_{k+1},\ldots,x_{k+l}) \rew x_{r} \tkom && \textbf{[projection]}\\
  & P^{0,1}(;0) \rew 0\tkom && \textbf{[predecessor]}\\
  & P^{0,1}(;S_i(;a)) \rew a \tkom \\
  & C^{0,3}(;0,a_{0},a_{1}) \rew a_{0} \tkom && \textbf{[conditional]}\\
  & C^{0,3}(;S_{i}(;a),a_{1},a_{0}) \rew a_{2-i} \tkom \\
  & \SUB^{k,l}[f,\bfg,\bfh](\bfx;\bfn)  \rew f(\bfg(\bfx;);\bfh(\bfx;\bfn)) \tkom && \textbf{[safe composition]}\\
  & \PREC^{k+1,l}[g,h_{1},h_{2}](0,\bfx;\bfn) \rew g(\bfx;\bfn) \tkom && \textbf{[predicative recursion}\\
  & \PREC^{k+1,l}[g,h_{1},h_{2}](S_{i}(;b),\bfx;\bfn) \rew  && \textbf{ on notation]} \\
  & \qquad \rew h_{i}(b,\bfx;\bfn,\PREC^{k+1,l}[g,h_{1},h_{2}](b,\bfx;\bfn)) \tpkt
\end{align*}
We use the following notation: $i \in [0,1]$ and $r \in [1,k+l]$.
\end{minipage}
}
\end{center}
\end{table}

The TRS $\RBp$ is terminating and confluent. 
Termination follows by the multiset path order.
Confluence is a consequence of the fact that $\RBp$ is orthogonal. 
Note the restriction in the rewrite rules for \emph{safe composition} and \emph{predicative recursion}. These rules only apply if all \emph{safe} arguments are 
numerals, i.e.\ in normal-form.
This peculiar restriction is necessary as the canonical term-rewriting characterisation $\RB$ of $\B$, admits exponential lower-bounds, hence $\RB$ is
\emph{non-feasible}, compare.~\cite{BeckmannWeiermann:1996}.

Let $\RS$ denote a TRS.
A \emph{derivation} is a sequence of terms $t_i$, $i \in \N$, such that for all
$i$, $t_i \rsrew{\RS} t_{i+1}$. The $(i+1)^{th}$ element of a 
sequence $a$ is denoted as $(a)_i$.
We write ${}\cons{}$ for the concatenation of sequences and define the
length $\len{a}$ of a sequence $a$ as usually.
We define a partial order $\subseteq$ on 
pairs of sequences. $a \subseteq b$, if $b$ is an \emph{extension} of $a$, 
i.e.\ $\len{a} \leq \len{b}$ and for all $i < \len{a}$ we have $(a)_i = (b)_i$.
A derivation $d$ with $(d)_0 = t$ is called 
\emph{derivation starting with $t$}.
The \emph{derivation tree} $\mathcal{T}_{\RS}(t)$ of $t$ is defined as the
structure $(T(t),\subseteq)$, where 
$T(t) \defsym \{ d | \text{$d$ is a derivation starting with $t$}\}$.
The root of $\mathcal{T}_{\RS}(t)$ is denoted by $t$ (instead of $(t)$).

We measure the \emph{complexity} or \emph{derivation length} 
of the computation of $f(\overline{\vecm})$ 
by the \emph{height} of $\mathcal{T}_{\RS}(f(\overline{\vecm}))$,
i.e., we define the \emph{derivation length function} 
$\Drf{\RS}{f} \colon \GTA{\Sigma} \to \N$:
 \begin{equation*}
\Drf{\RS}{f}(\overline{\vecm}) \defsym 
\max\{n \;|\; \exists \; t_0,\dots,t_n \in \GTA{\Sigma} 
\left( t_n \rsrewr{\RS} \dots \rsrewr{\RS} t_0 = f(\overline{\vecm}) \right)\} \tpkt
\end{equation*}

Based on these definitions we make
the notion of \emph{feasible} term-rewriting characterisation precise. 
  A term-rewriting characterisation $\RC$ of a function class $\C$ is called
\emph{feasible}, if for each $n$-ary function $f \in \C$, there exists a 
function symbol $g$ in the signature of $\RC$ such that 
$g(\overline{m}_1,\dots,\overline{m}_n)$ computes the value of 
$f(m_1,\dots,m_n)$ and $\Drf{\RC}{f}$ is bounded by a function from $\C$.
For the rewrite system $\RBp$ we have the following proposition.

\begin{proposition}\label{p:am2004}
  For every $f\in \B$, $\Drf{\RBp}{f}$ is bounded by a monotone polynomial 
in the length of the normal inputs. Specifically for each $f$ we can find a number $\ell(f)$ so that
$\Drf{\RBp}{f}(\overline{\vecm};\overline{\vecn}) \leq(2+\len{\vecm})^{\ell(f)}$,
where $\len{\bfm}$ denotes the sum of the length normal inputs $m_i$.
\end{proposition}
\begin{proof}
  See~\cite{AraiMoser:2004} for a proof, essentially we employ the observation that the 
derivation trees $\mathcal{T}_{\RBp}(f(\bfm;\bfn))$ are \emph{isomorphic} no matter 
how the safe input numerals $\bfn$ vary, to drop the dependency on the length of the
normal inputs.
\end{proof}

\section{A Path Ordering for $\PTIME$ } \label{POP}

To extend the above results and to facilitate the study of the 
polytime computable functions in a term-rewriting framework, 
we introduce in this section a new
\emph{path order for $\PTIME$}, which is a \emph{miniaturisation} of the 
recursive path order, cf.~\cite{BaaderNipkow:1998}, see also~\cite{Buchholz:1995}.

In the definition we make use of an auxiliary varyadic 
function symbol `$\seq$' of arbitrary, but finite arity, 
to denote sequences $s_0,\ldots,s_n$ of terms. 
Instead of $\seq(s_0,\ldots,s_n)$ we write $(s_0,\ldots,s_n)$. 
We write $a \cons b$ for sequences $a=(s_0,\ldots,s_n)$, 
$b=(s_{n+1},\ldots,s_{n+m})$ to denote the concatenation 
$(s_0,\dots,s_{n+m})$ of $a$ and $b$.

Let $\Sigma$ be a signature. 
We write $\Tseq{\Sigma}{\VS}$ to denote the set of all finite sequences 
of terms in $\T{\Sigma}{\VS}$.
To ensure that $\T{\Sigma}{\VS} \subset \Tseq{\Sigma}{\VS}$, any term 
is identified with the sequence $\seq(t) = (t)$. 
We denote sequences by $a,b,c$, both possible extended with subscripts.
Sometimes we write $fa$ as abbreviations of $f(t_0,\dots,t_n)$, if 
$a=(t_0,\dots,t_n)$.

We suppose a partial well-founded relation on $S$, 
the \emph{precedence}, denoted as $<$. We write $f \sim g$ if 
$(f \lesssim g) \land (g \lesssim f)$ and
we write $f > g$ and $g < f$ interchangeably.
Further, we suppose that the signature $\Sigma$ contains two unary
symbols $S_0,S_1$ of lowest
rank in the precedence. I.e.\ $\Sigma = \{S_0,S_1\} \cup \Sigma'$ and
$S_0 \sim S_1$ and for all $f \in \Sigma'$, $S_0,S_1 < f$. 
Moreover, we define $0 \defsym ()$.
For every number $m$ we define its \emph{numeral}
$\overline{m} \in \T{\Sigma}{\VS}$ as follows: $\overline{0} \defsym ()$;
$\overline{\mathcal{S}_i(m)} \defsym S_i(\overline{m})$ for 
$i \in [0,1]$.

The definition of the \emph{path order for $\PTIME$} (\emph{POP} for short) 
$\lpop$ (induced by $<$) is based on an auxiliary
order $\sqsubset$. The separation in two orders is necessary to break
the strength of the recursive path order that induces primitive recursive
derivation length, cf.~\cite{Hofbauer:1992}.

\begin{definition} \label{d:pop0}
Inductive definition of $\sqsubseteq$ induced by $<$.
\begin{enumerate}
\item \label{en:pop:sbs} 
  $\exi j\in [1,n] \, 
  (s\sqsubseteq t_{j}) \folgt s\sqsubset f(t_{1},\ldots,t_{n}) \tkom$
  \smallskip
\item $t = f(t_{1},\ldots,t_{n}) \spand 
s = g(s_{1},\ldots,s_{m}) \ \text{with} \ g < f
\spand \fal i\in [1,m] \, (s_{i}\sqsubset t) \\
\folgt s\sqsubset t \tpkt$
\end{enumerate}
\end{definition}

\begin{definition} \label{d:pop1}
Inductive definition of $\lpop$ induced by $<$; $\lpop$ is based on 
$\sqsubset$.
\begin{enumerate}
\item 
  $s\sqsubset t \folgt s \lpop t \tkom$
  \smallskip
\item 
  $\exi j \in [1,n] \; (s \leqpop t_{j}) \folgt s \lpop f(t_{1},\ldots,t_{n}) 
  \spand s \lpop (t_1,\dots,t_n) \tkom$
  \smallskip
\item \label{en:pop:inaccess} 
  $t = f(t_{1},\ldots,t_{n}) 
  \spand \left(  m=0 \; \text{or} \; ( \exists i_0 \; (
  \fal i\neq i_0 \; (s_{i}\sqsubset t) \spand 
  s_{i} \lpop t) \right) \\
  \folgt (s_{1},\ldots,s_{m})\lpop t \tkom$
  \smallskip
\item \label{en:pop:monoton} 
  $t = f(t_{0},\ldots,t_{n}) \spand s = g(s_{0},\ldots,s_{m}) \ \text{with} \ f \sim g 
  \spand (s_{0},\ldots,s_{m}) \lpop (t_{0},\ldots,t_{n})\\
  \folgt s \lpop t \tkom$
  \smallskip
\item \label{en:pop:mult} 
  $a \approx a_{0} \cons \cdots \cons a_{n}\spand \fal i\leq n \; (a_{i}\leqpop b_{i}) \spand 
  \exi i\leq n \, (a_{i} \lpop b_{i}) \\
  \folgt a \lpop (b_{0},\ldots,b_{n}) \; \text{if} \ n\geq 1 \tkom$

 \medskip
    $a \approx a_{0} \cons \cdots \cons a_{n}$ denotes the fact that the sequence $a$ of 
    terms is obtained from the concatenated 
    $a_{0} \cons \cdots \cons a_{n}$ by permutation.
\end{enumerate}
\end{definition}
Note that due to rule~\ref{en:pop:inaccess} $() \lpop a$ for any sequence $a \in \Tseq{\Sigma}{\VS}$.
Further, we write $s \gpop t$ for $t \lpop s$.
It is not difficult to argue that $\lpop$ is a reduction order. 
A number of relations are missing; we mention only the following: 
\begin{itemize}
\item $t = f(t_{1},\ldots,t_{n}) \spand 
  s = g(s_{1},\ldots,s_{m}) \ \text{with} \ g < f
  \spand \fal i\in [1,m] \, (s_{i} \lpop t) \folgt s\lpop t$.
\end{itemize}
We indicate the reasons for the omission of this clause.

\begin{example}
Consider the following TRS, where $\Sigma$ contains 
additionally the symbols $a,g,h,f$ with
precedence $a,h < f$, $g < h$. 
\begin{equation*}
  f(0) \rew a \qquad f(S_i(x)) \rew h(f(x)) \qquad
  h(x) \rew g(x,x) \tpkt
\end{equation*}
\end{example}

It is easy to see that $\lpop$ cannot handle the TRS in the
example, but would if rule above is included. However,
note that the TRS admits an \emph{exponential lower-bound} on 
the derivation length function. 

We introduce suitable \emph{approximations} $\prec_{k}$ of $\lpop$. 
\begin{definition} \label{d:approx:sq}
Inductive definition of $\sqsubset_{k}^{l}$ induced by $<$;
we write $\sqsubset_{k}$ to abbreviate $\sqsubset_{k}^{k}$.
\begin{enumerate}
\item 
  $\exi j\in[1,n] \, (s\sqsubseteq^{l}_{k} t_{j}) \folgt s\sqsubset^{l}_{k} f(t_{0},\ldots,t_{n}) \tkom$
  \smallskip
\item \label{en:sq:ii}
  $t = f(t_{0},\ldots,t_{n}) \spand 
  s= g(s_{0},\ldots,s_{m}) \ \text{with} \ g < f \spand m<k
  \spand \fal i \, (s_{i}\sqsubset^{l}_{k} t) \\
  \folgt s\sqsubset^{l+1}_{k} t \tpkt$
\end{enumerate}
\end{definition}

\begin{definition} \label{d:approx}
Inductive definition of $\prec_k$ induced by $<$; $\prec_k$ is based on $\sqsubset_k$.
\begin{enumerate}
\item \label{en:a:i}
  $s\sqsubset_{k} t \folgt s\prec_{k} t \tkom$
  \smallskip
\item \label{en:a:ii}
  $\exi j \in [1,n] \; ( s\preceq_{k} t_{j}) \folgt s\prec_{k} f(t_{1},\ldots,t_{n}) \tkom$
  \smallskip
\item \label{en:a:iii}
  $t = f(t_{1},\ldots,t_{n}) 
  \spand \left( 
    m=0 \; \text{or} \; \exi i_{0} \in [1,m] \, (\fal i\neq i_{0} \,
  (s_{i}\sqsubset_{k} t) \spand s_{i_{0}}\prec_{k} t) \right) \\
  \spand m<k 
  \folgt (s_{1},\ldots,s_{m})\prec_{k} t \tkom$
  \smallskip
\item \label{en:a:iv}
  $t = f(t_{0},\ldots,t_{n}) \spand s = g(s_{0},\ldots,s_{m}) \ \text{with} \ f \sim g 
  \spand 
  (s_{0},\ldots,s_{m}) \prec_{k} \\ 
  (t_{0},\ldots,t_{n}) \spand m<\max\{k,n\} \folgt s\prec_{k} t \tkom$
  \smallskip
\item \label{en:a:v}
  $a \approx a_{0} \cons \dots \cons a_{n} \spand \fal i\leq n \; (a_{i}\preceq_{k} b_{i}) \spand 
  \exi i\leq n \, (a_{i}\prec_{k} b_{i}) \folgt 
  a \prec_{k} (b_{0},\ldots,b_{n}) \ \text{if} \ n\geq 1 \tpkt$
\end{enumerate}
\end{definition}

In the following we prove that if for a finite rewrite system $\RS$,
$\RS \subseteq \lpop$, then it even holds that $\rsrew{\RS} \subseteq \prec_k$,
where $k$ depends on $\RS$ only. 

\begin{lemma} \label{l:approx:1}
  If $s \prec_{k} t$ and $k<l$, then $s \prec_{l} t$.
\end{lemma}

We introduce the auxiliary measure $\bN{.} \colon \Tseq{\Sigma}{\VS} \to \N$:
(i) $\bN{x} \defsym 1$, $x \in \VS$,
(ii) $\bN{(s_1,\dots,s_n)} \defsym \max \{n,\bN{s_1},\dots,\bN{s_n}\}$,
(iii) $\bN{fa} \defsym \bN{a} + 1$.
\begin{lemma} \label{l:approx:2}
  If $s \lpop t$, then for any substitution $\sigma$, 
  $s\sigma \prec_{\bN{s}} t\sigma$.
\end{lemma}

\begin{lemma} \label{l:approx:3}
  If $t= f(t_1,\dots,v,\dots,t_n)$, $s=f(t_1,\dots,u,\dots,t_n)$ with 
  $u \prec_{k} v$, where
  $k \geq \max \{\ar(f) \colon f\in \Sigma\}$, then $s \prec_{k} t$. 
\end{lemma}

Recall that $\lpop$ is
a reduction order. Hence the assumption $\RS \subseteq \lpop$ implies
$\rsrew{\RS} \subseteq \lpop$.
\begin{lemma} \label{l:approx:4}
  Let $k = \max \{ \max\{ \norm{r} | (l \rew r) \in \RS \}, 
  \max \{\ar(f) | f\in S\} \}$. Then $t \rsrew{\RS} s$ implies $s \prec_{k} t$.
\end{lemma}

We set
\begin{eqnarray*}
  \Slow{k}(\sigma) &  \defsym  & \max\{n\in \N \mid \exists (a_0,\ldots,a_{n}) \; \left(a_{n}\prec_{k} \cdots \prec_{k} a_{0} = a \right)\} \tkom \\
F_{k,p}(n) & \defsym & \max \{ \Slow{k}(fa) \colon \rk(f) = p \spand \Slow{k}(a) \leq n \} \tkom
\end{eqnarray*}
where $\rk(f)\colon \Sigma \to \N$ is defined inductively:
$\rk(f) \defsym \max \{\rk(g) + 1 \colon g \in \Sigma \land g \prec f\}$.
We collect some properties of the function $\Slow{k}$ in the next lemma.
\begin{lemma} \label{l:pointwise} 
  \begin{enumerate}
  \item \label{en:pw:i}  $\Slow{k}((s_0,\ldots,s_n)) = \sum_{i=0}^{n} \Slow{k}(a_i)$.
  \item \label{en:pw:iii} $\Slow{k}(\overline{m}) = \len{m}$ for any natural number $m$.
  \end{enumerate}
\end{lemma}

\begin{lemma}\label{l:upper}
  Inductively we define $d_{k,0} \defsym 2$ and 
$d_{k,p-1} \defsym (d_{k,p})^k+1$. Then there exists a constant $c$ (depending only
on $k$ and $p$) such that
$  F_{k,p}(n) \leq c \cdot n^{d_{k,p}} + c$.
\end{lemma}
\begin{proof}
The lemma is proven by main induction on $p$ and
side induction on $\sigma$. 

Set $a \defsym (t_0,\ldots,t_n)$ and 
 let $w \prec_k f(t_0,\ldots,t_n) =: t$, $\rk(f) = p$ and $w$ maximal. 
 By assumption $\Slow{k}(a) \leq n$.
 We prove
   \begin{equation*}
     \Slow{k}(w) < c n^{d_{k,p}} \quad \text{for almost all $n$} \tkom
   \end{equation*}
 by case-distinction on the definition of $\prec_k$. Without loss of generality, 
we only consider the case $w = (r_0,\ldots,r_m)$.

 \medskip
 \noindent
 \textsc{Case}. $p=0$ and $\forall i\leq m \; (r_i \sqsubset_k t)$. 
 By definition of $\lpop$ we have
 $\forall i\leq m \; \exists j \leq n \; (r_i \preceq_{k} t_j)$. 
 Then $\Slow{k}(w) \leq \Slow{k}(a) =  n$. 
 Hence
 \begin{equation*}
   \Slow{k}(w) \leq kn <  c n^{2} \tkom
 \end{equation*}
 where we set $c \defsym k$.

 \medskip
 \noindent
 \textsc{Case}. $p=0$, $\forall i\neq i_0 \; (r_i \sqsubset_k t)$, and $r_{i_0} \prec_{k} t$. 
 By definition of $\lpop$ we have
 $\forall i\leq m \; \exists j \leq n \; (r_i \preceq_{k} t_j)$ and 
 $r_{i_0} = f(s_0,\ldots,s_l)$, $\rk(f)=0$, 
 with $(s_0,\ldots,s_l) \prec_{k} a$. 
 Hence by induction hypothesis on $a$, there exists a constant $c$, such that 
 $G_k(r_{i_0}) \leq c (n-1)^2$ a.e. Employing Lemma~\ref{l:pointwise}(\ref{en:pw:i}) we obtain:
 \begin{equation*}
   \Slow{k}(w) = \Slow{k}((r_0,\ldots,r_m)) = 
   \sum_{i=0}^{m} \Slow{k}(r_i) \leq c (n-1)^2 + (k-1)n <  c n^{2} \tkom
 \end{equation*}
 as we can assume $c > k$.

 \medskip
 \noindent
 \textsc{Case}. $p > 0$ and $\forall i \leq m \; (r_i \sqsubset_k t) $. 
 Let $i$ be arbitrary. We can assume $r_i = g(s_0,\ldots,s_l)$, $g \prec f$, and $\forall i \leq l \; (s_i \sqsubset_k^{k-1} t)$.
Otherwise, if $r_i = g(s_0,\ldots,s_l)$ with $g \succ f$ s.t.\ there 
 $\exists j \leq n \; (r_i \sqsubseteq t_j)$ we proceed as in the first case.
 By induction hypothesis there exists $c$ and $d =d_{k,p}$ s.t.\ 
 $F_{k,p}(n) \leq c n^{d}$ a.e.

 We show the existence of a constant $c'$ s.t.\ $F_{k,p+1}(n) \leq c' n^{d'}$, where 
 $d' = d_{k,p+1}$.
 We define $f(a) \defsym c a^d$ and $g^{(0)}(a) \defsym a$, $g^{(l+1)}(a) = f(g^{(l)}(a)\cdot k)$;
 we obtain:
 \begin{equation}
   s \sqsubset^l_k t \folgt \Slow{k}(s) \leq g^{(l)}(n) \ \text{a.e.} \tag{$\star$}
 \end{equation}
 
 To see $(\star)$ we show by induction on $l$, that 
 $s \sqsubset^l_k t$ implies $\Slow{k}(s) \leq g^{(l)}(n)$, where
 $g^{(l)}(n) = c_0 a^{d^{(l)}}$ with 
 $c_0 = {c}^{\sum_{i=0}^{l-1} d^i} k^{\sum_{i=1}^{l} d^i}$. Suppose $l > 0$, then 
 we obtain by induction hypothesis on the claim and $F_{k,p}(n) \leq c n^{d}$ we obtain:
 \begin{equation*}
   \Slow{k}(s) \leq c [(c_0 n^{d^{l}})\cdot k]^{d} = c_1 n^{d^{l+1}} \text{a.e.} \tkom 
 \end{equation*}
 where $c_1 = {c}^{\sum_{i=0}^{l} d^i} k^{\sum_{i=1}^{l+1} d^i}$. This accomplishes
 the claim.

 Now the upper-bound for $\Slow{k}(w)$ follows:
 \begin{equation*}
 \Slow{k}(w) \leq k g^{(k)}(n)  < c' n^{d'} \ \text{a.e.} \tkom 
 \end{equation*}
 where $c' = {c}^{\sum_{i=0}^{k-1} d^i} k^{\sum_{i=0}^{k} d^i}$ and $d' = d^{k+1}+1 = d_{k,p+1}$. 

 \medskip
 \noindent
 \textsc{Case}. $p > 0$, $\forall i\neq i_0 \; (r_i \sqsubset_k t)$, 
and $r_{i_0} \prec_{k} t$. 
By definition $\forall i\leq m \; \exists j \leq n \; (r_i \preceq_{k} t_j)$,
and $r_{i_0} = f_(s_0,\ldots,s_l)$
so that $(s_0,\ldots,s_l) \prec_{k} a$. 
 Let $c, c', d'$ be defined as above.
 By induction hypothesis on $\sigma$ we obtain $\Slow{k}(r_{i_0}) \leq c' (n-1)^{d'}$ and thus
 \begin{equation*}
   \Slow{k}(w) \leq c' (n-1)^{d'} + (k-1) \cdot c \cdot n^{d^k} < c' n^{d'} \tpkt
 \end{equation*} 
\end{proof}

Recall the definition of the derivation length function:
\begin{equation*}
\Drf{\RS}{f}(\overline{\vecm}) = \max\{l \;|\; \exists \; t_0,\dots,t_n \in \GTA{\Sigma} 
\left( t_n \rsrewr{\RS} \dots \rsrewr{\RS} t_0 = f(\overline{\vecm}) \right)\}
\end{equation*}
We have established the following theorem.

\begin{theorem}\label{th:A}
  If for a finite TRS $\RS$ defined over $\T{\Sigma}{\VS}$, $\RS \subseteq \lpop$ then
 for each $f \in \Sigma$, $\Drf{\RS}{f}$ is bounded by a \emph{monotone polynomial} 
 in the sum of the binary length of the inputs.
\end{theorem}
\begin{proof}
  Let $\RS$ be a finite TRS defined over $\T{\Sigma}{\VS}$, such that
for every rule $(l \to r) \in \RS$, $r \lpop l$ holds. This implies
that for any two terms $t,s$, $t \rsrew{\RS} s$ implies $s \lpop t$. 
Hence by Lemma~\ref{l:approx:4} there exists $k \in \N$, s.t.\ 
$\rsrewr{\RS} \subseteq \prec_k$.
Suppose $f$ is an $n$-ary function symbol and set 
$t \defsym f(\overline{m}_1,\dots,\overline{m}_n)$.
By definition it follows that
\begin{equation*}
  \Drf{\RS}{f}(\overline{m}_1,\dots,\overline{m}_n) \leq 
  \Slow{k}(f(\overline{m}_1,\dots,\overline{m}_n)) \tpkt 
\end{equation*}
By Lemma~\ref{l:upper} there exists a polynomial $p$, depending only on $k$ and
the rank of $f$, s.t.\
\begin{equation*}
  \Slow{k}(f(\overline{m}_1,\dots,\overline{m}_n)) \leq
  p(\Slow{k}((\overline{m}_1,\dots,\overline{m}_n)) \tpkt
\end{equation*}
Employing with Lemma~\ref{l:pointwise}, we obtain 
$\Drf{\RS}{f}(\overline{m}_1,\dots,\overline{m}_{n}) \leq p(\sum_{i=1}^{n} \len{m_i})$. 
\end{proof}

\section{Predicative Recursion and POP}

In the previous section we have shown that if for a finite TRS $\RS$, 
defined over $\Tseq{\Sigma}{\VS}$, $\RS \subseteq \lpop$, 
then the derivation length function $\Drf{\RS}{f}$ is 
bounded by a monotone polynomial in the binary length of the inputs. 
As an application of Theorem~\ref{th:A}, we prove in this section 
that $\Drf{\RBp}{f}$ is bounded by a monotone polynomial in the binary length of the normal inputs. I.e.\ we give an alternative
proof of Prop.~\ref{p:am2004}.
As $\RBp$ exactly characterises the functions in $\PTIME$ this yields that
$\lpop$---via the mapping $\ints$ defined below---exactly 
characterises the class of polytime computable functions $\PTIME$.

It suffices to define a mapping 
$\ints \colon \GT{B}\to \GTseq{\Sigma}$, such that $\ints$ is a monotone
interpretation such that $\ints(l\sigma) \gpop \ints(r\sigma)$ holds
for all $(l \to r) \in \RBp$. 
We suppose the signature $\Sigma$ is defined such that for 
any function symbol $f\in B^{k,l}$ there is 
a function symbol $f' \in \Sigma$ of arity $k$. Moreover,
$\Sigma$ includes two constants $S_0,S_1$ and a
varyadic function symbol $\li$ of lowest rank.
We need a few auxiliary notions:
$\sn(\overline{n}) \defsym n$ for numerals $\overline{n}$;
$\sn(f(\bft;\bfs))=\sum_{j}(\sn(s_{j}))$, otherwise.
For every number $m$ we define its representation 
$\widehat{m} \in \T{\Sigma}{\VS}$ as follows: $\widehat{0} \defsym \li$;
$\widehat{\mathcal{S}_i(m)} \defsym \li(S_i) * \widehat{m}$ for 
$i \in [0,1]$, where 
$\li(s_0,\dots,s_i) * \li(s_{i+1},\dots,s_n) \defsym \li(s_0,\dots,s_n)$.
We define $\ints \colon \GT{B}\to \GTseq{\Sigma}$ by mutual induction together
with the interpretation $\intn \colon \GT{B}\to \GTseq{\Sigma}$. 
\begin{definition} \mbox{}
\begin{itemize}
\item $\ints(\overline{n}) \defsym ()$ and
$\ints(S_{i}(;t)) \defsym (S_{i}) \cons \ints(t)$ for $t\not\equiv\overline{n}$ (i.e.\ $t$ is not a numeral).
\smallskip
\item For $f\not= S_{i}$, define
$\ints(f(\bft;\bfs)) \defsym (f(\intn(t_{0}),\ldots,\intn(t_{n})), \ints(s_{0}),\ldots,\ints(s_{m}))$.
\smallskip
\item $\intn(t) \defsym \li\ints(t) * \widehat{sn(t)}$.
\end{itemize}
\end{definition}

First we show that for $\Q \in \{\ints,\intn\}$, 
$\Q(l\sigma) \gpop \Q(r\sigma)$. More precisely we show the following lemma.
\begin{lemma} \label{l:RBP:i}
Let $(l \to r) \in \RBp$, $\sigma$ a
ground substitution, such that $l\sigma,r\sigma \in \GT{B}$. Then there exists
$k$, depending on the rule $(l \to r)$, such that $\Q(r\sigma) \prec_k \Q(l\sigma)$.
\end{lemma}
\begin{proof}
Let $(l \to r)$ and $\sigma$ as in the assumptions of the lemma. We 
sketch the proof by considering the rule:
\begin{equation*}
\PREC^{p+1,q}[g,h_{1},h_{2}](S_{i}(;t),\bft;\bfn) \rew
  h_{i}(t,\bft;\bfn,\PREC[g,h_{1},h_{2}](t,\bft;\bfn)) \tpkt
\end{equation*}
We abbreviate $F \defsym \PREC^{p+1,q}[g,h_{1},h_{2}]$ and set $k \defsym 1+\max\{3,p+1,q+1\}$. 
 Let $\lh(f)$, $f \in B$ be defined as follows:
 $\lh(f) \defsym 1$,  for $f \in \{S_i, O, U, P \}$.
 $\lh(\SUB[f,\bfg,\bfh]) \defsym 
1 + \lh(f) + \lh(g_1) + \cdots + \lh(g_{k'}) + \lh(h_1) + \cdots + \lh(h_{l'})$.
 $\lh(\PREC[g,h_1,h_2]) \defsym 1 + \lh(g) + \lh(h_1) + \lh(h_2)$. 
Then we define the precedence $<$ over $\Sigma$ 
compatible with $\lh$, i.e. $f' < g'$ if $\lh(f) < \lh(g)$. 
For $\Q = \ints$, we employ the following sequence of comparisons:
\begin{align*}
&  \ints(F(S_{i}(;t),\bft;\bfn)) \\
& = (F'(\intn(S_i(;t)),\intn(t_1),\ldots,\intn(t_p)),\ints(\overline{n}_1),\ldots,\ints(\overline{n}_q)) \\
& = F'(\intn(S_i(;t)),\intn(t_1),\ldots,\intn(t_p)) \\
& = F'(\li(S_i)*\intn(t)),\intn(t_1),\ldots,\intn(t_p)) \tpkt
\intertext{By definition $\ints(\overline{n}_i)=()$ and for each $t \in \T{\Sigma}{\VS}$,
$t=(t)$. Moreover it is a direct consequence of the definitions
that $N(S_i(;t)) = \li(S_i)*N(t)$. Further:}
& F'(\li(S_i)*\intn(t),\intn(t_1),\ldots,\intn(t_p)) \\
& \succ_k (h'_i(\intn(t),\intn(t_1),\ldots,\intn(t_p)),F'(\intn(t),\intn(t_1),\ldots,\intn(t_p))) \tkom \\
\intertext{By Definition~\ref{d:approx}(\ref{en:a:iv}) 
we obtain 
$\li(S_i)*\intn(t) \succ_k \intn(t)$. This yields by 
rule~\ref{d:approx}(\ref{en:a:iv}) and rule~\ref{d:approx}(\ref{en:a:v}), using $k > p+1$: 
\begin{equation*}
  F'(\li(S_i)*\intn(t),\intn(t_1),\ldots,\intn(t_p)) \succ_k 
F'(\intn(t),\intn(t_1),\ldots,\intn(t_p)) \tpkt
\end{equation*}
%
Finally applying Definition~\ref{d:approx}(\ref{en:a:iii}) together with 
rule~\ref{d:approx}(\ref{en:a:ii}) and~\ref{d:approx:sq}(\ref{en:sq:ii}) yields the inequality.
In these rule applications we employ $k > q+1$ and $F' > h'_i$.}
& (h'_i(\intn(t),\intn(t_1),\ldots,\intn(t_p)),F'(\intn(t),\intn(t_1),\ldots,\intn(t_p)))  \\ 
& = (h'_i(\intn(t),\intn(t_1),\ldots,\intn(t_p))),\ints(n_1),\ldots,\ints(n_l),F'(\intn(t),\intn(t_1),\ldots,\intn(t_p))) \\
& = \ints(h_{i}(t,\bft;\bfn,F(t,\bft;\bfn)))\tpkt
\end{align*}

Finally, it is easy to see that 
$\intn(F(S_{i}(;t),\bft;\bfn)) \succ_k \intn(h_{i}(t,\bft;\bfn,F(t,\bft;\bfn))$. 
We established the lemma for the rule
$F(S_{i}(;t),\bft;\bfn) \rew h_{i}(t,\bft;\bfn,F(t,\bft;\bfn))$.
The other rules follow similar. 

Note that the definition of $k$ in
all cases depends on the arity-information encoded in the head function symbol
on the left-hand side. Moreover at most $3$ iterated applications of $\sqsubset_k$ are necessary. 
\end{proof}

The next lemma establish monotonicity for the interpretations $\ints,\intn$. 
\begin{lemma} \label{l:RBP:ii}
For $k\in\N$ and for $u,v \in \GT{\Sigma}$, $\Q(u) \prec_k \Q(v)$
for $\Q \in \{\ints,\intn\}$. Suppose $f \in B^{p,q}$ and $\overline{t},\overline{s} \in \GT{\Sigma}$.
Then 
\begin{itemize}
\item $\Q(f(t_1,\dots,u,\dots,t_p;\overline{s}) \prec_k 
  \Q(f(t_1,\dots,v,\dots,t_p;\overline{s})$ for $\Q \in \{\ints,\intn\}$, and
\item $\Q(f(\overline{t};s_1,\dots,u,\dots,s_q) \prec_k 
  \Q(f(\overline{t};s_1,\dots,v,\dots,s_q))$ for $\Q \in \{\ints,\intn\}$.
\end{itemize}
\end{lemma}

We define the derivation length function $\Drf{\RB'}{f}$ over the ground term-set
$\GT{\Sigma}$:
\begin{equation*}
  \Drf{\RBp}{f}(\overline{\vecm};\overline{\vecn}) \defsym 
  \max\{n \;|\; \exists \; t_0,\dots,t_n \in \GT{B} 
  \left( t_n \rsrewr{\RBp} \dots \rsrewr{\RBp} t_0 = f(\overline{\vecm};\overline{\vecn}) \right)\} \tpkt
\end{equation*}
Recall the definition of the derivation tree $\mathcal{T}_{\RBp}$. Note that for each
$t \in \T{B}{\VS}$, $\mathcal{T}_{\RBp}(t)$ is finite. This follows from the fact that
$\RBp$ is terminating and $\mathcal{T}_{\RBp}(t)$ is finitely branching. The latter is shown
by well-founded induction on $\rsrew{\RBp}$. 
Let $f \in B$ be a fixed predicative recursive function symbol. As the derivation
tree $\mathcal{T}_{\RBp}(f(\overline{\vecm};\overline{\vecn}))$ is finite only finitely
many function symbols occur in $\mathcal{T}_{\RBp}(f(\overline{\vecm};\overline{\vecn}))$. 
This allows to define a finite subset $F \subset B$, such that all terms occurring in
$\mathcal{T}_{\RBp}(f(\overline{\vecm};\overline{\vecn}))$ belong to $\GT{F}$. 
We define
\begin{equation*}
k \defsym 1+\max (\{3\} \cup \{p,q+1 | \text{$f^{p,q} \in B$ occurs in 
    $\mathcal{T}_{\RBp}(f(\overline{\vecm};\overline{\vecn}))$}\}) \tpkt
\end{equation*}
Let $R'$ denote the restriction of $\RBp$ to $\GT{F}$. Then, we have
$\Drf{\RBp}{f}(\overline{\vecm};\overline{\vecn}) = \Drf{R'}{f}(\overline{\vecm};\overline{\vecn})$.
From these observations together with Lemma~\ref{l:RBP:i} and~\ref{l:RBP:ii} we conclude
\begin{lemma}
  Let $s,t \in \GT{F}$ such that $t \rsrew{R} s$. Then $\ints(s) \prec_k \ints(t)$.
\end{lemma}

In summary we obtain, by following the pattern of the proof of Thm.~\ref{th:A}:
\begin{theorem} \label{t:B}
  For every $f\in B$, 
$\Drf{\RBp}{f}(\overline{m}_1,\dots,\overline{m}_p;\overline{n}_1,\dots,\overline{n_q})$ is bounded by a 
monotone polynomial in the sum of the length of the normal inputs $m_1,\dots,m_p$. 
\end{theorem}

\section{Conclusion}

The main contribution of this paper is the definition of 
a \emph{path order for $\PTIME$}, denoted as $\lpop$. 
This path order has the property that for a finite TRS $\RS$ 
compatible with $\lpop$, the \emph{derivation length function} $\Drf{\RS}{f}$
is bounded by a polynomial in the length of the inputs
for any defined function symbol $f$ in the signature of $\RS$.
Moreover $\lpop$ is \emph{complete} in the sense that for a function
$f \in \PTIME$, there exists a TRS $\RS$ computing $f$
such that such that termination of 
$\RS$ follows by $\lpop$. 
Another feature of $\lpop$ is, that its definition is devoid
of the separation of normal and safe arguments, present in the
definition of the predicative recursive functions and therefore in
the definition of the term-rewriting characterisation $\RBp$. 

We briefly relate our findings to the notion of the 
\emph{light multiset path order}, denoted
as $\llmpo$, introduced by Marion in~\cite{Marion:2003}. 
It is possible to define a variant of $\lpop$---denoted as $\lpopp$---%
such that Theorem~\ref{th:A} remains true for $\lpopp$ when suitably 
reformulated. While Definition~\ref{d:pop0} and~\ref{d:pop1} are based on
an arbitrary signature, the definition of $\lpopp$ assumes that 
normal and safe arguments are separated as in 
Section~\ref{PREC}. 
It is easy to see that $\lpopp \subset \llmpo$ and
this inclusion is strict as $\llmpo$ proves termination of the
non-feasible rewrite system $\RB$, while $\lpopp$ clearly does not. 
On the other hand let $\RS$ be a functional
program (i.e.\ a constructor TRS) computing a 
number-theoretic function $f$. 
A termination proof of $\RS$ via $\llmpo$ guarantees the
existence of a polytime algorithm for $f$. However, a termination proof of 
$\RS$ via or the introduced path order $\lpopp$ (or $\lpop$) guarantees that
$\RS$ itself is already a polytime algorithm for $f$. It seems clear
to us that the latter property is of more practical value.

}

\chapter{Derivational Complexity of Knuth-Bendix Orders revisited}
\label{lpar06}

\subsection*{Publication Details}
G.~Moser.
\newblock Derivational complexity of {K}nuth {B}endix orders revisited.
\newblock In \emph{Proceedings of the 13th International Conference on Logic
  for Programming Artificial Intelligence and Reasoning}, number 4246 in LNCS,
  pages 75--89. Springer Verlag, 2006.

\subsection*{Ranking}
The Conference on Logic Programming and Automated Reasoning has been ranked
\textbf{A} by CORE in 2007.

{
\input{lpar06.sty}

\subsection*{Abstract}

We study the derivational complexity of rewrite systems $\RS$ compatible
with Knuth-Bendix orders (KBOs for short), if the signature of $\RS$ is infinite. We show that the known bounds on the
derivation height are preserved,  if $\RS$ fulfils some mild conditions. 
This allows us to obtain bounds on the derivational height of non simply terminating TRSs.
Furthermore, we re-establish the $2$-recursive upper-bound on the derivational complexity of 
finite rewrite systems $\RS$ compatible with KBO.

\section{Introduction} \label{lpar06:Introduction}

One of the main themes in rewriting is \emph{termination}. Over the years powerful methods
have been introduced to establish termination of a given term rewrite system (TRS) $\RS$.
Earlier research mainly concentrated on inventing suitable reduction orders---%
for example simplification orders, see Chapter 6, authored by Zantema in~\cite{Terese}---capable of
proving termination directly.
In recent years the emphasis shifted towards transformation techniques like 
the \emph{dependency pair method} or \emph{semantic labelling}, see~\cite{Terese}. 
The dependency pair method is easily automatable and lies at the heart of many 
successful termination provers like \ttt~\cite{TTT:2005} or \aprove~\cite{Aprove:2004}.
Semantic labelling with infinitely labels was conceived to be unsuitable for automation. Hence, only the variant 
with finitely many elements was incorporated (for example in \aprove ~\cite{Aprove:2004} or \torpa~\cite{Torpa:2005}). 
Very recently this belief was proven wrong. 
\tpa~\cite{TPA:2006} implements semantic labelling with natural numbers, in combination with multiset path
orders (MPOs)  efficiently. 
As remarked in~\cite{KoprowskiZantema:2006} a sensible extension of this implementation is the combination of 
semantic labelling with Knuth-Bendix orders (KBOs for short).

In order to assess the power and weaknesses of different termination techniques it is natural to look at the
length of derivation sequences, induced by different techniques. This program has been suggested 
in~\cite{HofbauerLautemann:1989}. The best known result is that for finite rewrite systems, 
MPO induces primitive recursive derivational complexity. This bound is essentially optimal, 
see~\cite{Hofbauer:1991,Hofbauer:1992}. Similar optimal results 
have been obtained for lexciographic path orders (LPOs) and KBOs.
Weiermann~\cite{Weiermann:1995:tcs} showed that LPO induce multiply recursive derivational complexity. 
In~\cite{Lepper:2001a} Lepper showed that 
for \TRSs\ compatible with KBO, the derivational complexity is bounded by the Ackermann function.

These results not only assess different proof techniques for termination, but constitute an a priori
complexity analysis for term rewrite systems (TRSs for short) provably terminating by MPO, LPO or KBO.
The application of termination provers as basis for the termination analysis of logic or functional programs is currently a very hot topic. Applicability of an a priori complexity analysis for TRSs in this direction seems likely.

While the aforementioned program has spawned a number of impressive
results, not much is known about the derivational complexity induced by
the dependency pair method or semantic labelling (for fixed base orders, obviously). 
We indicate the situation with an example. 
\begin{example} \label{ex:1}
  Consider the \TRS\ $(\FS,\RS)$~\cite{Bachmair:1987} consisting of the following rewrite rules:
  \begin{align*}
    &f(h(x)) \rew f(i(x)) && h(a) \rew b\\
    &g(i(x)) \rew g(h(x)) && i(a) \rew b \tpkt
  \end{align*}
\end{example}

It is not difficult to see that termination of $\RS$ cannot be established directly with path orders
or \KBOs. On the other hand, termination is easily shown via the dependency pair method or
via semantic labelling. For the sake of the argument we show termination via semantic labelling
with \KBOs.

We use natural numbers as semantics and as labels. 
As interpretation for the function symbols we use $a_{\N} = b_{\N} = g_{\N}(n) = f_{\N}(n) = 1$, $i_{\N}(n) = n$, and
$h_{\N}(n) = n+1$. The resulting algebra $(\N,>)$ is
a quasi-model for $\RS$. It suffices to label the symbol $f$. 
We define the labelling function $\ell_f \colon \N \to \N$ as $\ell_f(n)=n$.
Replacing
\begin{equation*}
  f(h(x)) \rew f(i(x)) \tkom
\end{equation*}
by the infinitely many rules
\begin{equation*}
  f_{n+1}(h(x)) \rew f_{n}(i(x)) \tkom
\end{equation*}
we obtain the labelled \TRS, $(\FSlab,\Rlab)$. Further the \TRS\ $(\FSlab,\Dec)$ consists of all
rules
\begin{equation*}
  f_{n+1}(x) \rew f_{n}(x) \tpkt
\end{equation*}
Now we can show termination of $\RS' \defsym \Rlab \cup \Dec$ by an instance $\gkbo$
of \KBO. We set the weight for all occurring function symbols to $1$. Further, the precedence is defined
as 
\begin{equation*}
  f_{n+1} \succ f_n \succ \dots \succ f_0 \succ i \succ h \succ g \succ a \succ b \tpkt
\end{equation*}
It is easy to see that $\RS' \subseteq \gkbo$. Thus termination of
$\RS$ is guaranteed.

As the rewrite system $\RS'$ is infinite we cannot directly apply the aforementioned result
on the derivational complexity induced by Knuth-Bendix orders. A careful study of~\cite{Lepper:2001a} reveals
that the crucial problem is not that $\RS'$ is infinite, but that the signature $\FSlab$ is infinite, as
Lepper's proof makes explicit use of the finiteness of the signature: To establish an upper-bound on the
derivational complexity of a \TRS\ $\RS$, compatible with \KBO, 
an interpretation function $\mathcal{I}$ is defined, where the cardinality of the underlying
signature is hard-coded into $\mathcal{I}$, cf.~\cite{Lepper:2001a}.

We study the situation by giving an alternative proof of Lepper's result 
compare~\cite{Lepper:2001a}.
The outcome of this study is that the assumption of finiteness of the rewrite system can be weakened. 
By enforcing conditions that are still weak enough to treat interesting rewrite systems, 
we show that for (possibly infinite) \TRSs\ $\RS$ over infinite signatures, compatible with \KBO, the 
derivation height of $\RS$ can be bounded by the Ackermann function. 
Using an example that stems from~\cite{Hofbauer:1991} we show that this upper-bound is essentially optimal. 

Specialised to Example~\ref{ex:1}, our results provide an upper bound on the derivation height function
with respect to $\RS$: 
For every $t \in \GTA{\FS}$ there exists a constant $c$ (depending only on $t$, $\RS'$, and $\gkbo$) such that 
the derivation height $\Dh{\RS}(t)$ with respect to $\RS$ is ${} \leq \Ack(c^n,0)$. 
As the constant $c$ can be made precise, the method is capable of automation. 

\medskip
This paper is organised as follows: In Section~\ref{Preliminaries} and~\ref{KBO} some basic facts
on rewriting, set theory and \KBOs\ are recalled. In Section~\ref{OrderType} we define 
an embedding from $\gkbo$ into $\glex$, the lexicographic comparison of sequences of natural numbers.
This embedding renders an alternative description of the derivation height of a term, based on
the partial order $\glex$. This description is discussed in Section~\ref{DerivationLength} and
linked to the Ackermann function in Section~\ref{Ackermann}. The above mentioned central result is
contained in Section~\ref{Sum}. Moreover in Section~\ref{Sum} we apply our result to a non simply terminating
\TRS, whose derivational complexity cannot be primitive recursively bounded.

\section{Preliminaries} \label{Preliminaries}

We assume familiarity with term rewriting. For further details see~\cite{Terese}.
Let $\VS$ denote a countably infinite set of variables and $\FS$ a signature. 
We assume that $\FS$ contains at least one constant.
The set of terms over $\F$ and $\VS$ is denoted as 
$\TA{\F}{\VS}$, while the set of ground terms is
written as $\GTA{\F}$. The set of variables occurring in a term $t$ is denoted
as $\var{t}$. The set of function symbols occurring in $t$ is denoted as $\fs{t}$.
The \emph{size} of a term $t$, written as $\size{t}$, is the number of variables and
functions symbols in it. The number of occurrences of a symbol $a \in \FS \cup \VS$ in $t$ is denoted
as $\Anz{a}{t}$.
A \emph{\TRS} $(\FS,\RS)$ over $\TA{\FS}{\VS}$ is a set of rewrite rules. 
The smallest rewrite relation that contains $\RS$ is denoted as $\rsrew{\RS}$. 
The transitive closure of $\rsrew{\RS}$ is denoted by $\rstrew{\RS}$, and
its transitive and reflexive closure by $\rssrew{\RS}$. 
A \TRS\ $(\FS,\RS)$ is called {\em terminating} if there is no infinite rewrite sequence.
As usual, we frequently drop the reference to the signature $\FS$.

A \emph{partial order} $\succ$ is an irreflexive and transitive relation. The converse of
$\succ$ is written as $\prec$. A partial order
$\succ$ on a set $A$ is \emph{well-founded} if there exists no infinite descending sequence 
$a_1 \succ a_2 \succ \cdots$ of elements of $A$. 
%
%
A rewrite relation that is also a partial order is called \emph{rewrite order}. A well-founded 
rewrite order is called \emph{reduction order}. A \TRS\ $\RS$ and
a partial order $\succ$ are \emph{compatible} if $\RS\subseteq \succ$. We also say
that $\RS$ is compatible with $\succ$ or vice versa. A \TRS\ $\RS$ is terminating
iff it is compatible with a reduction order $\succ$.

Let $(\A,>)$ denote a well-founded weakly monotone $\FS$-algebra. 
$(\A,>)$ consists of a carrier $A$, interpretations $f_\A$ for each function symbol in $\FS$, 
and a well-founded partial order $>$ on $A$ such that every $f_\A$ is weakly monotone in all arguments.
We define a quasi-order $\geqord{\A}$: $s \geqord{\A} t$ if for all assignments $\alpha \colon \VS \to A$
$\eval{\alpha}{\A}{s} \Geqord \eval{\alpha}{\A}{t}$. Here $\Geqord$ denotes the reflexive closure of $>$.
The algebra $(\A,>)$ is a \emph{quasi-model} of a \TRS\ $\RS$, if $\RS \subseteq \geqord{\A}$. 

A \emph{labelling} $\ell$ for $\A$ consists of a set of labels
$L_f$ together with mappings $\ell_f \colon A^n \to L_f$ for every $f \in \FS$, $f$ $n$-ary. 
A labelling is called \emph{weakly monotone} if all labelling functions $\ell_f$ are weakly monotone
in all arguments. The labelled signature $\FSlab$ consists of $n$-ary functions symbols $f_a$ 
for every $f \in \FS$, $a \in L_f$, together with all $f \in \FS$, such that $L_f = \emptyset$. 
The \TRS\ $\Dec$ consists of all rules
\begin{equation*}
  f_{a+1}(x_1,\dots,x_n) \rew f_a(x_1,\dots,x_n) \tkom
\end{equation*}
for all $f \in \FS$. The $x_i$ denote pairwise different variables. 
Our definition of $\Dec$ is motivated by a similar definition in~\cite{KoprowskiZantema:2006}.
Note that the rewrite relation $\rssrew{\Dec}$ is not changed by this modification of $\Dec$.
For every assignment $\alpha$, we inductively
define a mapping $\lab{\alpha} \colon \TA{\FS}{\VS} \to \TA{\FSlab}{\VS}$:
\begin{equation*}
  \lab{\alpha}(t) \defsym 
  \begin{cases}
    t & \text{if $t \in \VS$} \tkom \\
    f(\lab{\alpha}(t_1),\dots,\lab{\alpha}(t_n)) & \text{if $t=f(t_1,\dots,t_n)$ and $L_f = \emptyset$} \tkom \\
    f_a(\lab{\alpha}(t_1),\dots,\lab{\alpha}(t_n)) & \text{otherwise} \tpkt
  \end{cases}
\end{equation*}
The label $a$ in the last case is defined as $l_f(\eval{\alpha}{\A}{t_1},\dots,\eval{\alpha}{\A}{t_n})$.
The \emph{labelled} \TRS\ $\Rlab$ over $\FSlab$ is defined as
\begin{equation*}
  \{\lab{\alpha}(l) \rew \lab{\alpha}(r) \mid \text{$l \rew r \in \RS$ and $\alpha$ an assignment}\} \tpkt
\end{equation*}

\begin{theorem}[Zantema~\cite{Zantema:1995}]
Let $\RS$ be a \TRS, $(\A,>)$ a well-founded weakly monotone quasi-model for $\RS$, and
$\ell$ a weakly monotone labelling for $(\A,>)$. Then $\RS$ is terminating iff
$\Rlab \cup \Dec$ is terminating.
\end{theorem}

The proof of the theorem uses the following lemma.
\begin{lemma} \label{l:lpar06:zantema}
  Let $\RS$ be a \TRS, $(\A,>)$ a quasi-model of $\RS$, and $\ell$ a weakly monotone labelling for
$(\A,>)$. If $s \rsrew{\RS} t$, then $\lab{\alpha}(s) \rssrew{\Dec} \cdot \rsrew{\Rlab} \lab{\alpha}(t)$
for all assignments $\alpha$.
\end{lemma}

We briefly review a few basic concepts from set-theory in particular ordinals, 
see~\cite{Jech}. We write $>$ to denote the well-ordering of ordinals.
Any ordinal $\alpha \not= 0$, smaller than $\epsilon_0$, can uniquely be represented by its 
\emph{Cantor Normal Form} (\emph{CNF} for short)
\begin{equation*}
  \omega^{\alpha_1}n_1 + \dots \omega^{\alpha_k}n_k \qquad \text{with} \ \alpha_1 > \dots > \alpha_k \tpkt
\end{equation*}
To each well-founded partial order $\succ$ on a set $A$ we can associate a (set-theoretic) ordinal, 
its \emph{order type}. First we associate an ordinal to each element $a$ of $A$ by setting
$\otp_{\succ}(a) \defsym \sup \{ \otp_{\succ}(b) + 1 \colon b \in A
\ \text{and} \ b \succ a \}$.
The \emph{order type} of $\succ$, denoted by $\otp(\succ)$, is the supremum of 
$\otp_{\succ}(a) + 1$ with $a \in A$.
For two partial orders $\succ$ and $\succ'$ on $A$ and $A'$, respectively, a mapping $\opm \colon A \to A'$
\emph{embeds} $\succ$ into $\succ'$ if for all $p,q \in A$, $p \succ q$ implies $\opm(p) \succ' \opm(q)$.
Such a mapping is an \emph{order-isomorphism} if it is bijective and the partial orders 
$\succ$ and $\succ'$ are linear .

\section{The Knuth Bendix Orders} \label{KBO}

%
%
A \emph{weight function} for $\FS$ is a pair $(\Weight,w_0)$ consisting of 
a function $\Weight \colon \FS \to \N$ and a minimal weight $w_0 \in \N$, $w_0 >0$ such that
$\weight{c} \geq w_0$ if $c$ is a constant.
A weight function $(\Weight,w_0)$ is called \emph{admissible} for a precedence $\succ$
if $\spsym \succ g$ for all $g \in\FS$ different from $\spsym$, when $\spsym$ is unary with 
$\weight{\spsym} = 0$. The function symbol $\spsym$ (if present) is called \emph{special}.
The \emph{weight} of a term $t$, denoted as $\wei{t}$
is defined inductively. Assume $t$ is a
variable, then set $\wei{t} \defsym w_0$,
otherwise if $t=\term{g}{t}{1}{n}$, we define
$\wei{t} \defsym \weight{g} + \wei{t_1} + \dots + \wei{t_n}$.

%
%
The following definition of \KBO{} is tailored to our purposes. 
It is taken from~\cite{Lepper:2001a}. We write $s = \spsym^a s'$ if
$s = \spsym^a(s')$ and the root symbol of $s'$ is distinct from the special symbol $\spsym$. 
Let $\succ$ be a precedence. The \emph{rank} of a function symbol is defined as: 
$\rk(f) \defsym \max \{\rk(g)+1 \mid f \succ g\}$. (To assert well-definedness we
stipulate $\max(\emptyset) = 0$.)

\begin{definition}\label{d:kbo}
Let $(\Weight,w_0)$ denote an admissible weight function for $\FS$ and let 
$\succ$ denote a precedence on $\FS$.
We write $\spsym$ for the special symbol. 
The \emph{Knuth Bendix order} $\gkboVar$ on $\TA{\FS}{\VS}$ is inductively
defined as follows:
$s \gkboVar t$ if $\Anz{x}{s} \geq \Anz{x}{t}$ for all $x \in \VS$ and
\begin{enumerate}
\item $\wei{s} > \wei{t}$, or
\item $\wei{s} = \wei{t}$, $s = \spsym^a s'$, $t = \spsym^b t'$, where
  $s' = \term{g}{s}{1}{n}$, $t'=\term{h}{t}{1}{m}$, and one of the 
  following cases holds.
  \begin{enumerate}
  \item $a > b$, or
  \item $a=b$ and $g \succ h$, or 
  \item $a=b$, $g=h$, and $(s_1,\dots,s_n) \gkboVarlex (t_1,\dots,t_n)$.
  \end{enumerate}
\end{enumerate}
\end{definition}

Let $\gkbo$ denote the \KBO\ on terms in its usual definition, see~\cite{Terese}.
The following lemma, taken from~\cite{Lepper:2001a}, states that both orders are
interchangeable.
\begin{lemma}[Lepper~\cite{Lepper:2001a}]
  The orders $\gkbo$ and $\gkboVar$ coincide.
\end{lemma}

In the literature \emph{real-valued} \KBOs{} and other generalisations of \KBOs{} are studied as well, 
cf.~\cite{Martin:1987,Dershowitz:1987}. 
However, as established in~\cite{KorovinVoronkov:2003}
any \TRS{} shown to be terminating by a real-valued \KBO{} 
can be shown to be terminating by a integer-valued \KBO{}. 

\section{Exploiting the Order-Type of \KBOs} \label{OrderType}

We write $\N^\ast$ to denote the set of finite sequences of natural numbers. Let $p \in \N^\ast$,
we write $\length{p}$ for the \emph{length} of $p$, i.e.\ the number of positions in the sequence $p$. 
The $i^{\mathrm{th}}$ element of the sequence $a$ is denoted as $(p)_{i-1}$. We write $p \cons q$ 
to denote the concatenation of the sequences $p$ and $q$.
The next definition is standard but included here, for sake of completeness.
\begin{definition}
We define the \emph{lexicographic order} on $\N^\ast$. If $p,q \in \N^\ast$, then $p \glex q$ if,
\begin{itemize}
\item $\length{p} > \length{q}$, or
\item $\length{p} = \length{q}=n$ and there exists $i \in [0,n-1]$, such that
  for all $j \in [0,i-1]$ $(p)_j = (q)_j$ and $(p)_i > (q)_i$.
\end{itemize}
\end{definition}
It is not difficult to see that $\otp(\glex) = \omega^\omega$, moreover in~\cite{Lepper:2001a} it
is shown that $\otp(\gkbo) = \omega^{\omega}$. Hence $\otp(\glex)=\otp(\gkbo)$, a fact we exploit below.
However, to make this work, we have to restrict our attention to signatures $\FS$ with bounded
arities. The maximal arity of $\FS$ is denoted as $\Ar{\FS}$.

\begin{definition}
  Let the signature $\FS$ and a weight function $(\Weight,w_0)$ for $\FS$ be fixed. 
  We define an embedding $\tw \colon \TA{\FS}{\VS} \to \N^\ast$. Set $b \defsym \max\{\Ar{\FS},3\}+1$.
\begin{equation*}
  \tw(t) \defsym 
  \begin{cases}
    (w_0,a,0)\cons 0^m & \text{if $t = \spsym^a x$, $x \in \VS$}  \tkom \\
    (\wei{t},a,\rk(g)) \cons \tw(t_1) \cons \dots \cons \tw(t_n) \cons 0^m & \text{if $t=\spsym^a g(t_1,\dots,t_n)$} \tpkt
  \end{cases}
\end{equation*}
The number $m$ is set suitably, so that $\length{\tw(t)} = b^{\wei{t}+1}$.
\end{definition}

The mapping $\tw$ flattens a term $t$ by transforming it
into a concatenation of triples. Each triple holds the weight of the considered subterm $r$, the
number of leading special symbols and the rank of the first non-special function symbol of $r$. In this
way all the information necessary to compare two terms via $\gkbo$ is expressed as a very simple
data structure: a list of natural numbers. 

\begin{lemma} \label{l:lpar06:lemma1}
 $\tw$ embeds $\gkbo$ into $\glex$: If $s \gkbo t$, then $\tw(s) \glex \tw(t)$.
\end{lemma}
\begin{proof}
The proof follows the pattern of the proof of Lemma 9 in~\cite{Lepper:2001a}.

Firstly, we make sure that the mapping $\tw$ is \emph{well-defined}, i.e., we show 
that the length restriction can be met. 
We proceed by induction on $t$; let $t=\spsym^a t'$. 
We consider two cases (i) $t' \in \VS$ or (ii) $t' = g(t_1,\dots,t_n)$. Suppose the former:
\begin{equation*}
  \length{(w_0,a,0)} = 3 \leq b^{\wei{t}+1} \tpkt
\end{equation*}
Now suppose case (ii): Let $j = \rk(g)$, we obtain
\begin{align*}
  \length{(\wei{t},a,j) \cons \tw(t_1) \cons \dots \cons \tw(t_n)} & = 
  3 + b^{\wei{t_1}+1} + \dots + b^{\wei{t_n}+1} \\
  & \leq 3+ n \cdot b^{\wei{t}} \leq b^{\wei{t}+1} \tpkt
\end{align*}

Secondly, we show the following, slight generalisation of the lemma:
\begin{equation}
  \label{eq:lpar06:embedding}
  s \gkbo t \land \length{\tw(s) \cons r} = \length{\tw(t) \cons r'} \folgt \tw(s) \cons r \glex \tw(t) \cons r' \tpkt
\end{equation}
To prove~\eqref{eq:lpar06:embedding} we proceed by induction on $s \gkbo t$. Set $p = \tw(s) \cons r$, $q = \tw(t) \cons r'$.

\medskip
\noindent
\textsc{Case} $\wei{s} > \wei{t}$: By definition of the mapping $\tw$, we have: If $\wei{s} > \wei{t}$, then
$(\tw(s))_0 > (\tw(t))_0$. Thus  $p \glex q$ follows.

\medskip
\noindent
\textsc{Case} $\wei{s} = \wei{t}$: We only consider the sub-case where $s= \spsym^a g(s_1,\dots,s_n)$ and
$t=\spsym^a g(t_1,\dots,t_n)$ and there exists $i \in [1,n]$ such that
$s_1=t_1,\dots,s_{i-1}=t_{i-1}$, and $s_i \gkbo t_i$. (The other cases are treated as in the case above.)
The induction hypothesis expresses that if $\length{\tw(s_i) \cons v} = \length{\tw(t_i) \cons v'}$, then
$\tw(s_i) \cons v \glex \tw(t_i) \cons v'$. For $j = \rk(g)$, we obtain
\begin{align*}
  & p = \overbrace{(\wei{s},a,j) \cons \tw(s_1) \cons \dots \cons \tw(s_{i-1})}^{w} \cons \tw(s_i) \cons \dots \cons \tw(s_n) \cons r \tkom \\
  & q = \underbrace{(\wei{s},a,j) \cons \tw(s_1) \cons \dots \cons \tw(s_{i-1})}_{w} \cons \tw(t_i) \cons \dots \cons \tw(t_n) \cons r' \tpkt
\end{align*}
Due to $\length{p} = \length{q}$, we conclude
\begin{equation*}
  \length{\tw(s_i) \cons \dots \cons \tw(s_n) \cons r} = \length{\tw(t_i) \cons \dots \cons \tw(t_n) \cons r'} \tpkt
\end{equation*}
Hence induction hypothesis is applicable and we obtain
\begin{equation*}
  \tw(s_i) \cons \dots \cons \tw(s_n) \cons r \glex \tw(t_i) \cons \dots \cons \tw(t_n) \cons r' \tkom
\end{equation*}
which yields $p \glex q$. This completes the proof of~\eqref{eq:lpar06:embedding}. 

Finally, to establish the lemma, we assume
$s \gkbo t$. By definition either $\wei{s} > \wei{t}$ or $\wei{s} = \wei{t}$. In the latter case $\tw(s) \glex \tw(t)$ follows by~\eqref{eq:lpar06:embedding}. While in the former $\tw(s) \glex \tw(t)$ follows as $\wei{s} > \wei{t}$ implies $\length{\tw(s)} > \length{\tw(t)}$.
\end{proof}

\section{Derivation Height of Knuth-Bendix Orders} \label{DerivationLength}

Let $\RS$ be a \TRS\ and $\gkbo$ a \KBO\ such that
$\gkbo$ is compatible with $\RS$. The \TRS\ $\RS$ and the \KBO\ $\gkbo$ are fixed
for the remainder of the paper.
We want to extract an upper-bound on the length of derivations in $\RS$. We recall the central definitions. Note that we can restrict the definition to the set ground terms.
The \emph{derivation height} function $\Dh{\RS}$ (with respect to
$\RS$ on $\GTA{\FS}$) is defined as follows.
\begin{align*}
  \Dh{\RS}(t) & \defsym \max (\{n \mid 
  \exists (t_0,\dots,t_n) \ t=t_0 \rsrew{\RS} t_1 \rsrew{\RS} \dots \rsrew{\RS} t_n \}) \tpkt 
\end{align*}

We introduce a couple of \emph{measure functions} for term and sequence complexities,
respectively. The first measure $\spec \colon \TA{\F}{\VS} \to \N$ bounds the maximal nesting of special symbols in
the term:
\begin{align*}
    \spec(t) &\defsym 
    \begin{cases}
      a & \text{if $t = \spsym^a x$, $x \in \VS$} \tkom\\
      \max(\{a\}\cup\{\spec(t_j) \mid j \in [1,n]\}) & \text{if $t = \spsym^a g(t_1,\dots,t_n)$} \tpkt
    \end{cases}\\[2ex]
\intertext{The second and third measure $\rk \colon \TA{\F}{\VS} \to \N$ and $\mrk \colon \TA{\FS}{\VS} \to \N$ 
collect information on the ranks of non special function symbols occurring:
}
    \rk(t) &\defsym 
    \begin{cases}
      0 & \text{if $t = \spsym^a x$, $x \in \VS$} \tkom \\
      j & \text{if $t = \spsym^a g(t_1,\dots,t_n)$ and $\rk(g) = j$} \tkom
    \end{cases}\\[2ex]
     \mrk(t) &\defsym 
    \begin{cases}
      0 & \text{if $t = \spsym^a x$, $x \in \VS$} \tkom \\
      \max(\{j\} \cup \{\mrk(t_i) \mid i \in [1,n]\}) & \text{if $t = \spsym^a g(t_1,\dots,t_n)$, $\rk(g) = j$} \tpkt
    \end{cases}
  \end{align*}

The fourth measure $\mc \colon \N^\ast \to \N$ considers sequences $p$ and bounds the maximal number 
occurring in $p$:
  \begin{equation*}
    \mc(p) \defsym \max (\{ (p)_i \mid i \in [0,\length{p}-1]\}) \tpkt
  \end{equation*}
It is immediate from the definitions that for any term $t$: $\spec(t),\rk(t), \mrk(t) \leq \mc(\tw(t))$. We
write $r \subterm t$ to denote the fact that $r$ is a subterm of $t$.

\begin{lemma} \label{l:lpar06:help1}
  If $r \subterm t$, then $\mc(\tw(t)) \geq \mc(\tw(r))$.
\end{lemma}
We informally argue for the correctness of the lemma.
Suppose $r$ is a subterm of $t$. Then clearly $\wei{r} \leq \wei{t}$. The
maximal occurring nesting of special symbols in $r$ is smaller (or equal) than in $t$. 
And the maximal rank of a symbol in $r$ is smaller (or equal) than in $t$. 
The mapping $\tw$ transforms $r$ to a sequence $p$ whose coefficients
are less than $\wei{t}$, less than the maximal nesting of special symbols and less than the maximal rank of
non-special function symbol in $r$ . Hence $\mc(\tw(t)) \geq \mc(\tw(r))$ holds.

\begin{lemma}  \label{l:lpar06:help2}
  If $p=\tw(t)$ and $q=\tw(\spsym^a t)$, then $\mc(p) + a \geq \mc(q)$.
\end{lemma}
\begin{proof}
  The proof of the lemma proceeds by a case distinction on $t$. 
\end{proof}

\begin{lemma} \label{l:lpar06:lemma2}
We write $m \modminus n$ to denote $\max(\{m-n,0\})$.
Assume $s \gkbo t$ with $\spec(t) \leq K$ and $(\mrk(t) \modminus \rk(s)) \leq K$.
Let $\sigma$ be a substitution and set $p = \tw(s\sigma)$, 
$q = \tw(t\sigma)$. Then $p \glex q$ and $\mc(p)+K \geq \mc(q)$.
\end{lemma}

\begin{proof}
It suffices to show $\mc(p)+K \geq \mc(q)$ as $p \glex q$ follows from
Lemma~\ref{l:lpar06:lemma1}. We proceed by induction on $t$; let $t = \spsym^a t'$.

\medskip
\noindent
\textsc{Case} $t' \in \VS$: Set $t' = x$. We consider two sub-cases: Either
(i) $x\sigma = \spsym^b y$, $y \in \VS$ or (ii) $x\sigma = \spsym^b g(u_1,\dots,u_m)$. It
suffices to consider sub-case (ii), as sub-case (i) is treated in a similar way. 
From $s \gkbo t$, we know that for all $y \in \VS$, $\Anz{y}{s} \geq \Anz{y}{t}$, hence $x \in \var{s}$ and 
$x\sigma \subterm s\sigma$. 
Let $l \defsym \rk(g)$; by Lemma~\ref{l:lpar06:help1} we conclude $\mc(\tw(x\sigma)) \leq \mc(p)$.
I.e.\ $b,l,\mc(\tw(u_1)),\dots,\mc(\tw(u_m)) \leq \mc(p)$. We obtain
\begin{eqnarray*}
  \mc(q) & = & \max (\{w_0,a+b,l\} \cup \{\mc(\tw(u_j)) \mid i \in [1,m]\}) \\
  & \leq & \max (\{\wei{s\sigma},\spec(t)+\mc(p),\mc(p)\} \cup \{\mc(p)\}) \\
  & \leq & \max (\{\wei{s\sigma},\mc(p)+K\} \cup \{\mc(p)\}) = \mc(p) + K \tpkt
\end{eqnarray*}

\noindent
\textsc{Case} $t' = g(t_1,\dots,t_n)$: Let $j = \rk(g)$.
By Definition~\ref{d:kbo} we obtain $s \gkbo t_i$. Moreover $\spec(t_i) \leq \spec(t) \leq K$ and
$\mrk(t_i) \leq \mrk(t)$. Hence for all $i$: $\spec(t_i) \leq K$ and $(\mrk(t_i) \modminus \rk(s)) \leq K$ holds.
Thus induction hypothesis is applicable: For all $i$: $\mc(\tw(t_i\sigma)) \leq \mc(p)+K$.
By using the assumption $(\mrk(t) \modminus \rk(s)) \leq K$ we obtain:
\begin{eqnarray*}
   \mc(q) & = & \max (\{\wei{t\sigma},a,j\} \cup \{\mc(\tw(t_i\sigma)) \mid i \in [1,n]\}) \\
     & \leq & \max (\{\wei{t\sigma},\spec(t),\rk(s)+K\} \cup \{\mc(p)+K\}) \\
     & \leq & \max (\{\wei{s\sigma},\spec(t),\rk(s\sigma)+K\} \cup \{\mc(p)+K\}) \\
     & \leq & \max (\{\wei{s\sigma},K,\mc(p)+K\} \cup \{\mc(p)+K\}) = \mc(p)+K \tpkt
 \end{eqnarray*}
\end{proof}

In the following, we assume that the set
\begin{equation} \label{eq:lpar06:M}
  M \defsym \{\spec(r) \mid l \to r \in \RS\} \cup \{(\mrk(r) \modminus \rk(l)) \mid l \to r \in \RS\}
\end{equation}
is finite. We set $K \defsym \max(M)$ and let $K$ be fixed for the remainder.

\begin{example} \label{ex:2}
  With respect to the \TRS\ $\RS' \defsym \Rlab \cup \Dec$ from Example~\ref{ex:1}, we have
  $M = \{(\mrk(r) \modminus \rk(l)) \mid l \to r \in \RS'\}$.
  Note that the signature of $\RS'$ doesn't contain a special symbol.

Clearly $M$ is finite and it is easy to see that $\max(M) = 1$. 
Exemplary, we consider the rule schemata $f_{n+1}(h(x)) \rew f_{n}(i(x))$. Note that the rank of
$i$ equals $4$, the rank of $h$ is $3$, and the rank of $f_n$ is given by $n+5$. Hence
$\mrk(f_{n}(i(x))) = n+5$ and $\rk(f_{n+1}(h(x))) = n+6$. Clearly $(n+5 \modminus n+6) \leq 1$.
\end{example}

\begin{lemma} \label{l:lpar06:lemma3}
  If $s \rsrew{\RS} t$, $p = \tw(s)$, $q = \tw(t)$, then $p \glex q$ and $u(\mc(p),K) \geq \mc(q)$, where $u$ denotes a monotone polynomial such that $u(n,m) \geq 2n+m$.
\end{lemma}
\begin{proof}
By definition of the rewrite relation there exists a context $C$, a substitution $\sigma$ and 
a rule $l \rew r \in R$ such that $s = C[l\sigma]$ and $t = C[r\sigma]$. We prove
$\mc(q) \leq u(\mc(p),K)$ by induction on $C$. Note that $C$ can only have the form
(i) $C=\spsym^a[\Box]$ or (ii) $C=\spsym^a g(u_1,\dots,C'[\Box],\dots,u_n)$. 

\medskip
\noindent
\textsc{Case} $C = \spsym^a[\Box]$:
By Lemma~\ref{l:lpar06:lemma2} we see $\mc(\tw(r\sigma)) \leq \mc(\tw(l\sigma)) + K$. Employing
in addition Lemma~\ref{l:lpar06:help2} and Lemma~\ref{l:lpar06:help1}, we obtain:
\begin{eqnarray*}
  \mc(q) &=& \mc(\tw(\spsym^a r\sigma)) \leq \mc(\tw(r\sigma)) + a\\
  & \leq & \mc(\tw(l\sigma)) + K + a \\
  & \leq & \mc(p) + K + \mc(p) \leq u(\mc(p),K) \tpkt
\end{eqnarray*}

\medskip
\noindent
\textsc{Case} $C = \spsym^a g(u_1,\dots,C'[\Box],\dots,u_n)$: 
As $C'[l\sigma] \rsrew{\RS} C'[r\sigma]$, induction hypothesis is applicable: Let $p' = \tw(C'[l\sigma])$, $q'=\tw(C'[r\sigma])$. Then
$\mc(q') \leq u(\mc(p'),K)$. For $\rk(g) = l$, we obtain by application of induction hypothesis and Lemma~\ref{l:lpar06:help1}:
\begin{eqnarray*}
  \mc(q) & = & \max(\{\wei{t},a,l\} \cup \{\mc(\tw(u_1)),\dots,\mc(q'),\dots,\mc(\tw(u_n))\}) \\
  & \leq & \max(\{\wei{s},a,l\} \cup \\
  & & \cup \{\mc(\tw(u_1)),\dots,u(\mc(p'),K),\dots,\mc(\tw(u_n))\}) \\
  & \leq & \max(\{\wei{s},a,l\} \cup \{\mc(p),u(\mc(p),K)\}) = u(\mc(p),K) \tpkt
\end{eqnarray*}
\end{proof}

We define \emph{approximations} of the partial order $\glex$. 
\begin{equation*}
  p \glexapprox{n} q \quad \gdw \quad p \glex q \ \text{and} \ u(\mc(p),n) \geq \mc(q) \tkom
\end{equation*}
where $u$ is defined as in Lemma~\ref{l:lpar06:lemma3}.
Now Lemma~\ref{l:lpar06:lemma2} can be concisely expressed as 
follows, for $K$ as above.
\begin{proposition} \label{p:prop1}
  If $s \rsrew{\RS} t$, then $\tw(s) \glexapprox{K} \tw(t)$.
\end{proposition}

In the spirit of the definition of derivation height, we define a family of functions $\Ah{n} \colon \N \to \N$:
\begin{equation*}
  \Ah{n}(p) \defsym \max(\{m \mid 
  \exists (p_0,\dots,p_m) \ p = p_0 \glexapprox{n} p_1 \glexapprox{n} \dots \glexapprox{n} p_m \}) \tpkt
\end{equation*}

The following proposition is an easy consequence of the definitions and Proposition~\ref{p:prop1}.
\begin{theorem} \label{t:thm1}
 Let $(\FS,\RS)$ be a \TRS, compatible with \KBO. Assume the set
$M \defsym \{\spec(r) \mid l \to r \in \RS\} \cup \{(\mrk(r) \modminus \rk(l)) \mid l \to r \in \RS\}$ 
is finite and the arities in of the symbols in $\FS$ are bounded; set $K \defsym \max(M)$.
Then $\Dh{\RS}(t) \leq \Ah{K}(\tw(t))$.
\end{theorem}

In the next section we show that $\Ah{n}$ is bounded by the Ackermann function $\Ack$. Thus
providing the sought upper-bound on the derivation height of $\RS$.

\section{Bounding the Growth of $\Ah{n}$} \label{Ackermann}

Instead of directly relating the functions $\Ah{n}$ to the Ackermann function, we make
use of the fast-growing \emph{Hardy} functions, cf.~\cite{RoseSubrecursion}. The Hardy functions form a hierarchy of 
unary functions $\Hardy{\alpha} \colon \N \to \N$ indexed by ordinals. We will only be interested
in a small part of this hierarchy, namely in the set of functions $\{\Hardy{\alpha} \mid \alpha < \omega^{\omega}\}$.

\begin{definition}
  We define the embedding $\wo \colon \N^\ast \to \omega^{\omega}$ as follows:
  \begin{equation*}
    \wo(p) \defsym \omega^{\ell-1}(p)_0 + \dots \omega (p)_{\ell-2} + (p)_{\ell-1} \tkom
  \end{equation*}
where $\ell=\length{p}$.
\end{definition}

The next lemma follows directly from the definitions.
\begin{lemma}
  If $p \glex q$, then $\wo(p) > \wo(q)$.
\end{lemma}

We associate with  every $\alpha < \omega^{\omega}$ in \emph{\CNF}
an ordinal $\alpha_n$, where $n \in \N$. The sequence $(\alpha_n)_{n}$ is called \emph{fundamental sequence} of
$\alpha$. (For the connection between 
rewriting and fundamental sequences see e.g.~\cite{MoserWeiermann:2003} or Chapter~\ref{rta03}.)
\begin{equation*}
  \alpha_n \defsym
  \begin{cases}
    0 & \text{if $\alpha = 0$} \tkom \\
    \beta & \text{if $\alpha = \beta + 1$} \tkom \\
    \beta + \omega^{\gamma+1} \cdot (k-1) + \omega^{\gamma} \cdot (n+1) & \text{if $\alpha = \beta + \omega^{\gamma + 1} \cdot k$} \tpkt \\
  \end{cases}
\end{equation*}

Based on the definition of 
$\alpha_n$, we define $\Hardy{\alpha} \colon \N \to \N$, for $\alpha < \omega^{\omega}$ 
  by transfinite induction on $\alpha$:
\begin{equation*}
    \Hardy{0}(n) \defsym n \qquad \Hardy{\alpha}(n) \defsym \Hardy{\alpha_n}(n+1) \tpkt
\end{equation*}
Let $\gapprox{n}$ denote the transitive closure of $(.)_{n}$, i.e.\ $\alpha \gapprox{n} \beta$ iff
$\alpha_n \gapprox{n} \beta$ or $\alpha_n = \beta$.
Suppose $\alpha,\beta < \omega^{\omega}$. Let
$\alpha = \omega^{\alpha_1}n_1 + \dots \omega^{\alpha_k}n_k$ and
$\beta = \omega^{\beta_1}m_1 + \dots \omega^{\beta_l}m_l$.
Recall that any ordinal $\alpha \not= 0$ can be uniquely written in \CNF, hence we can assume that $\alpha_1 > \dots > \alpha_k$ and
$\beta_1 > \dots > \beta_l$. Furthermore by our assumption that $\alpha,\beta < \omega^{\omega}$, we have $\alpha_i,\beta_j \in \N$.
We write $\NF(\alpha,\beta)$ if $\alpha_k \geq \beta_1$. 

Before we proceed in our estimation of the functions $\Ah{n}$, we state some simple facts that help us to calculate with 
the function $\Hardy{\alpha}$. 
\begin{lemma} \label{l:lpar06:lemma3a}
  \begin{enumerate}
  \item \label{en:lemma3a:1} If $\alpha \gapprox{n} \beta$, then $\alpha \gapprox{n+1} \beta + 1$ or $\alpha = \beta+1$.
  \item \label{en:lemma3a:2} If $\alpha \gapprox{n} \beta$ and $n \geq m$, then $\Hardy{\alpha}(n) > \Hardy{\beta}(m)$.
  \item \label{en:lemma3a:3} If $n > m$, then $\Hardy{\alpha}(n) > \Hardy{\alpha}(m)$.
  \item \label{en:lemma3a:4} If $\NF(\alpha,\beta)$, then $\Hardy{\alpha+\beta}(n) = \Hardy{\alpha} \circ \Hardy{\beta}(n)$; $\circ$ denotes function composition.
  \end{enumerate}
\end{lemma}

We relate the Hardy functions with the Ackermann function. The stated upper-bound is a gross one,
but a more careful estimation is not necessary here. 
\begin{lemma} \label{l:lpar06:lemma4}
  For $n \geq 1$: $\Hardy{\omega^n}(m) \leq \Ack(2n,m)$.
\end{lemma}
\begin{proof}
  We recall the definition of the Ackermann function:
  \begin{eqnarray*}
    \Ack(0,m) & = & m+1\\
    \Ack(n+1,0) & = & \Ack(n,1) \\
    \Ack(n+1,m+1) & = & \Ack(n,\Ack(n+1,m))
  \end{eqnarray*}
In the following we sometimes denote the Ackermann function as a unary function, indexed by its first argument:
$\Ack(n,m) = \AckVar{n}(m)$. To prove the lemma, we proceed by induction on the lexicographic
comparison of $n$ and $m$. 
We only present the case, where $n$ and $m$ are greater than $0$. As preparation
note that $m+1 \leq \Hardy{\omega^n}(m)$ holds for any $n$ and $\AckVar{n}^2(m+1) \leq \AckVar{n+1}(m+1)$ holds
for any $n,m$.
\begin{align*}
  \Hardy{\omega^{n+1}}(m+1) & = \Hardy{\omega^n (m+2)}(m+2) \\
      & \leq \Hardy{\omega^n (m+2)+\omega^n}(m+1) && \text{Lemma~\ref{l:lpar06:lemma3a}(\ref{en:lemma3a:3},\ref{en:lemma3a:4})}\\
      & = \Hardy{\omega^n}^2\Hardy{\omega^n (m+1)}(m+1) && \text{Lemma~\ref{l:lpar06:lemma3a}(\ref{en:lemma3a:4})}\\
      & = \Hardy{\omega^n}^2\Hardy{\omega^{n+1}}(m) \\
      & \leq \AckVar{2n}^2 \AckVar{2(n+1)}(m) && \text{induction hypothesis} \\
      & \leq \AckVar{2n+1} \AckVar{2(n+1)}(m)\\
      & = \Ack(2(n+1),m+1) \tpkt
\end{align*}
\end{proof}

\begin{lemma} \label{l:lpar06:lemma5}
Assume $u(m,n) \leq 2m +n $ and set $\ell = \length{p}$. For all $n \in \N$:
\begin{equation}\label{eq:lpar06:eq1}
  \Ah{n}(p) \leq \Hardy{\omega^2 \cdot \wo(p)}(u(\mc(p),n)+1) < \Hardy{\omega^{4+\ell}}(\mc(p)+n) \tpkt
\end{equation}
\end{lemma}
\begin{proof}
  To prove the first half of~\eqref{eq:lpar06:eq1} , we make use of the following fact:
  \begin{equation} \label{eq:lpar06:claim}
    p \glex q \land n \geq \mc(q) \folgt \wo(p) \gapprox{n} \wo(q) \tpkt
  \end{equation}
To prove \eqref{eq:lpar06:claim}, one proceeds by induction on $\glex$ and uses that 
the embedding $\wo \colon \N^\ast \to \omega^{\omega}$
is essentially an order-isomorphism. We omit the details.

\medskip
By definition, we have $\Ah{n}(p) = \max(\{ \Ah{n}(q) + 1 \mid p \glexapprox{n} q \})$.
Hence it suffices to prove
\begin{equation} \label{eq:lpar06:star}
  p \glex q \land  u(\mc(p),n) \geq \mc(q) \folgt \Ah{n}(q) < \Hardy{\omega^2 \cdot \wo(p)}(u(\mc(p),n)+1) 
\end{equation}
We fix $p$ fulfilling the assumptions in~\eqref{eq:lpar06:star}; let $\alpha = \wo(p)$, $\beta = \wo(q)$, 
$v = u(\mc(q),n)$. We use~\eqref{eq:lpar06:claim} to obtain $\alpha \gapprox{v} \beta$. We proceed
by induction on $p$.

Consider the case $\alpha_v = \beta$. As $p \glex q$, we can employ induction hypothesis to conclude
$\Ah{n}(q) \leq \Hardy{\omega^2 \cdot \wo(q)}(u(\mc(q),n)+1)$. It is not difficult to see that
for any $p \in \N^\ast$ and $n \in \N$, $4\mc(p)+2n+1 \leq \Hardy{\omega^2}(u(\mc(p),n))$. In sum, we obtain:
\begin{align*}
  \Ah{n}(q) &\leq \Hardy{\omega^2 \cdot \wo(q)}(u(\mc(q),n)+1) \\
  & \leq \Hardy{\omega^2 \cdot \alpha_v}(u(u(\mc(p),n),n)+1) && \mc(q) \leq u(\mc(p),n)\\
  & \leq \Hardy{\omega^2 \cdot \alpha_v}(4\mc(p)+2n+1) && \text{Definition of $u$}\\
  & \leq \Hardy{\omega^2 \cdot \alpha_v} \Hardy{\omega^2}(u(\mc(p),n)) \\
  & = \Hardy{\omega^2 \cdot (\alpha_v+1)}(u(\mc(p),n)) &&\text{Lemma~\ref{l:lpar06:lemma3a}(\ref{en:lemma3a:4})}\\
  & < \Hardy{\omega^2 \cdot (\alpha_v+1)}(u(\mc(p),n)+1) &&\text{Lemma~\ref{l:lpar06:lemma3a}(\ref{en:lemma3a:3})}\\
  & \leq \Hardy{\omega^2 \cdot \alpha}(u(\mc(p),n)+1) &&\text{Lemma~\ref{l:lpar06:lemma3a}(\ref{en:lemma3a:2})}
\end{align*}
The application of Lemma~\ref{l:lpar06:lemma3a}(\ref{en:lemma3a:2}) in the last step is feasible as by definition
$\alpha \gapprox{v} \alpha_v$. An application of Lemma~\ref{l:lpar06:lemma3a}(\ref{en:lemma3a:1}) yields
$\alpha_v + 1 \leqapprox{v+1} \alpha$. From which we deduce 
$\omega^2 \cdot (\alpha_v+1) \leqapprox{v+1} \omega^2 \cdot \alpha$.

Secondly, consider the case $\alpha_v \gapprox{v} \beta$. In this case the proof follows the pattern of the above proof,
but an additional application of Lemma~\ref{l:lpar06:lemma3a}(\ref{en:lemma3a:4}) is required. This completes
the proof of\eqref{eq:lpar06:star}. 

\medskip
To prove the second part of~\eqref{eq:lpar06:eq1}, we proceed as follows: 
The fact that $\omega^{\ell} > \wo(p)$ is immediate from the definitions. Induction on $p$ reveals
that even $\omega^{\ell} \gapprox{\mc(p)} \wo(p)$ holds. Thus in conjunction with the first part of \eqref{eq:lpar06:eq1}, 
we obtain:
\begin{eqnarray*}
  \Ah{n}(p) & \leq & \Hardy{\omega^2 \cdot \wo(p)}(u(\mc(p),n)+1) \leq \Hardy{\omega^{2+\ell}}(u(\mc(p),n)+1) \\
  & \leq & \Hardy{\omega^{4+\ell}}(\mc(p)+n)  \tpkt
\end{eqnarray*}
The last step follows as $2\mc(p)+n+1 \leq \Hardy{\omega^2}(\mc(p)+n)$.
\end{proof}

As a consequence of Lemma~\ref{l:lpar06:lemma4} and~\ref{l:lpar06:lemma5}, we obtain the following proposition.
\begin{theorem} \label{t:thm2}
For all $n \geq 1$: If $\ell = \length{p}$, then $\Ah{n}(p) \leq \Ack(2\ell+8,\mc(p)+n)$.
\end{theorem}

\section{Derivation Height of TRSs over Infinite Signatures Compatible with KBOs} \label{Sum}

Based on Theorem~\ref{t:thm1} and~\ref{t:thm2} we obtain that the derivation height of $t \in \GTA{\FS}$ 
is bounded in the Ackermann function.

\begin{theorem} \label{t:main}
  Let $(\FS,\RS)$ be a \TRS, compatible with \KBO. Assume the set
  $M \defsym\{\spec(r) \mid l \to r \in \RS\} \cup \{(\mrk(r) \modminus \rk(l)) \mid l \to r \in \RS\}$
  is finite and the arities of the symbols in $\FS$ are bounded; set $K \defsym \max(M)$.
  Then $\Dh{\RS}(t) \leq \Ack(\bO(\length{\tw(t)})+\mc(\tw(t))+K,0)$. 
\end{theorem}

\begin{proof}
We set $u(n,m) = 2n+m$ and keep the polynomial $u$ fixed for the remainder. 
Let $p = \tw(t)$ and $\ell = \length{p}$. 
Due to Theorem~\ref{t:thm1} we conclude that $\Dh{\RS}(t) \leq \Ah{K}(p)$.
It is easy to see that $\Ack(n,m) \leq \Ack(n+m,0)$. 
Using this fact and Theorem~\ref{t:thm2} we obtain:
$\Ah{K}(p) \leq \Ack(\bO(\ell),\mc(p)+K) \leq \Ack(\bO(\ell)+\mc(p)+K,0)$.
Thus the theorem follows.
\end{proof}

For fixed $t \in \GTA{\FS}$ we can bound the argument of the Ackermann function in
the above theorem in terms of the size of $t$. 
We define
\begin{align*}
  r_{\max} \defsym \mrk(t) && 
  w_{\max} \defsym \max(\{\weight{u} \mid u \in \fs{t} \cup \var{t}\} \tpkt
\end{align*}

\begin{lemma} \label{l:lpar06:sum}
  For $t \in \GTA{\FS}$, let $r_{\max}$, $w_{\max}$ be as above. Let $b \defsym \max\{\Ar{\FS},3\}+1$,
  and set $n \defsym \size{t}$. Then $\weight{t} \leq w_{\max} \cdot n$, $\spec(t) \leq n$, 
  $\mrk(t) \leq r_{\max}$. Hence $\length{\tw(t)} \leq b^{w_{\max}(n) \cdot n + 1}$ and
  $\mc(\tw(t)) \leq w_{\max}(n) \cdot n + r_{\max}$.
\end{lemma}
\begin{proof}
  The proof proceeds by induction on $t$. 
\end{proof}

\begin{corollary} \label{c:main}
  Let $(\FS,\RS)$ be a \TRS, compatible with a \KBO\ $\gkbo$. Assume 
  the set $\{\spec(r) \mid l \to r \in \RS\} \cup \{(\mrk(r) \modminus \rk(l)) \mid l \to r \in \RS\}$
  is finite and the arites of the symbols in $\FS$ are bounded. 
  Then for $t \in \GTA{\FS}$, there exists a constant $c$---depending on $t$, $(\FS,\RS)$, and $\gkbo$---such that
  $\Dh{\RS}(t) \leq \Ack(c^n,0)$.
\end{corollary}
\begin{proof}
  The corollary is a direct consequence of Theorem~\ref{t:main} and Lemma~\ref{l:lpar06:sum}. 
\end{proof}

\begin{remark}
  Note that it is not straight-forward to apply Theorem~\ref{t:main} to classify the derivational complexity of $\RS$,
over infinite signature, compatible with \KBO. This is only possible in the (unlikely)
case that for every term $t$ the maximal rank $\mrk(t)$ and the weight $\weight{t}$ of $t$ 
can be bounded uniformly, i.e.\ independent of the size of $t$. 
\end{remark}

We apply Corollary~\ref{c:main} to the motivating example introduced in Section~\ref{lpar06:Introduction}.
\begin{example} \label{ex:3}
  Recall the definition of $\RS$ and $\RS' \defsym \Rlab \cup \Dec$ from Example~\ref{ex:1} and~\ref{ex:2} respectively. 
Let $s \in \GTA{\FSlab}$ be fixed and set $n \defsym \size{s}$. 

Clearly the arities of the symbols in $\FSlab$
are bounded. In Example~\ref{ex:2} we indicated that the set $M = \{(\mrk(r) \modminus \rk(l)) \mid l \to r \in \RS'\}$
is finite. Hence, Corollary~\ref{c:main} is applicable to conclude the existence of $c \in \N$ with
$\Dh{\RS'}(s) \leq \Ack(c^n,0)$.
In order to bound the derivation height of $\RS$, we employ Lemma~\ref{l:lpar06:zantema} to observe
that for all $t \in \GTA{\FS}$: $\Dh{\RS}(t) \leq \Dh{\RS'}(\lab{\alpha}(t))$, for arbitrary $\alpha$.
As $\size{t} = \size{\lab{\alpha}(t)}$ the above calculation yields
\begin{equation*}
  \Dh{\RS}(t) \leq \Dh{\RS'}(\lab{\alpha}(t)) \leq \Ack(c^n,0) \tpkt
\end{equation*}
Note that $c$ depends only on $t$, $\RS'$ and the \KBO\ $\gkbo$ employed.
\end{example}

The main motivation of this work was to provide an alternative proof of Lepper's result
that the derivational complexity of any \emph{finite} \TRS, compatible with \KBO, is bounded
by the Ackermann function, see \cite{Lepper:2001a}. We recall the definition of the
\emph{derivational complexity}:
\begin{equation*}
  \Dc{\RS}(n) \defsym \max (\{\Dh{\RS}(t) \mid \size{t} \leq n\}) \tpkt
\end{equation*}

\begin{corollary}
  Let $(\FS,\RS)$ be a \TRS, compatible with \KBO, such that $\FS$ is finite.
  Then $\Dh{\RS}(n) \leq \Ack(2^{\bO(n)},0)$.
\end{corollary}

\begin{proof}
As $\FS$ is finite, the $K = \max(\{(\mrk(r) \modminus \rk(l)) \mid l \to r \in \RS'\})$
and $\Ar{\FS}$ are obviously well-defined. Theorem~\ref{t:main} yields that 
$\Dh{\RS}(t) \leq \Ack(\bO(\length{\tw(t)})+\mc(\tw(t))+K,0)$. 
Again due to the finiteness of $\FS$, for any $t \in \GTA{\FS}$, $\mrk(t)$ and 
$\weight{t}$ can be estimated independent of $t$. A similar argument calculation as in Lemma~\ref{l:lpar06:sum}
thus yields 
$\Dh{\RS}(t) \leq \Ack(2^{\bO(\size{t})},0)$.
Hence the result follows.
\end{proof}

\begin{remark}
Note that if we compare the above corollary to Corollary 19 in~\cite{Lepper:2001a}, we see that Lepper could even
show that $\Dc{\RS}(n) \leq \Ack(\bO(n),0)$. On the other hand, as already remarked above, Lepper's
result is not admissible if the signature is infinite.  
\end{remark}
 
In concluding, we want to stress that the method is also applicable to obtain bounds
on the derivational height of non simply terminating \TRSs, a
feature only shared by Hofbauer's approach to utilise context-dependent interpretations, cf.~\cite{Hofbauer:2001}.
\begin{example}
  Consider the \TRS\ consisting of the following rules:
  \begin{align*}
    & f(x) \circ (y \circ z) \to x \circ (f^2(y) \circ z) && a(a(x)) \to a(b(a(x))) \\
    & f(x) \circ (y \circ (z \circ w)) \to x \circ (z \circ (y \circ w)) \\
    & f(x) \to x
  \end{align*}
Let us call this \TRS\ $\RS$ in the following.
Due to the rule $a(a(x)) \to a(b(a(x)))$, $\RS$ is not simply terminating. And due to the
three rules, presented on the left, the derivational complexity of $\RS$ cannot be bounded by
a primitive recursive function, compare~\cite{Hofbauer:1991}.
\end{example}

Termination can be shown by semantic labelling, where the natural numbers are used as semantics and as labels. 
The interpretations $a_\N(n) = n+1$, $b_\N(n) = \max(\{0, n-1\})$, $f_\N(n) = n$, and $m \circ_\N n = m+n$
give rise to a quasi-model. Using the labelling function $\ell_{a}(n) = n$,
termination of $\RS' \defsym \Rlab \cup \Dec$ can be shown by an instance $\gkbo$ of \KBO\ with 
weight function $(\Weight,1)$: 
$\Weight(\circ) = \Weight(f) = 0$, $\Weight(b) = 1$, and $\Weight(a_n) = n$
and precedence: $f \succ \circ \succ \dots a_{n+1} \succ a_n \succ \dots \succ a_0 \succ b$. The
symbol $f$ is special. Clearly the arities of the symbols in $\FSlab$
are bounded. Further, it is not difficult
to see that the set 
$M = \{\spec(r) \mid l \to r \in \RS'\} \cup \{(\mrk(r) \modminus \rk(l)) \mid l \to r \in \RS'\}$
is finite and $K \defsym \max(M) = 2$. 

Proceeding as in Example~\ref{ex:3}, we see that for each $t \in \GTA{\FS}$, there exists a 
constant $c$ (depending on $t$, $\RS'$ and $\gkbo$) such that $\Dh{\RS}(t) \leq \Ack(c^n,0)$.

}

\chapter{Complexity Analysis by Rewriting}
\label{flops08}

\subsection*{Publication Details}
M.~Avanzini and G.~Moser.
\newblock Complexity analysis by rewriting.
\newblock In \emph{Proceedings of the 9th International Symposium on Functional
  and Logic Programming}, number 4989 in LNCS, pages 130--146. Springer Verlag,
  2008.%
\footnote{This research was partially supported by FWF (Austrian Science Fund) project P20133.}

\subsection*{Ranking}
The International Symposium on Functional and Logic Programming has been ranked
\textbf{A} by CORE in 2007.

{
\input{flops08.sty}

\subsection*{Abstract}
In this paper we introduce a restrictive version of the multiset path order, 
called \emph{polynomial path order}. This recursive path order 
induces polynomial bounds on the maximal number of innermost rewrite steps.
This result opens the way to automatically verify for a given program, 
written  in an eager functional programming language, 
that the maximal number of evaluation steps starting from any function call is polynomial 
in the input size.
To test the feasibility of our approach we have implemented 
this technique and compare its applicability to existing methods.

\section{Introduction}
Term rewriting is a conceptually simple but powerful abstract model of
computation that underlies much of declarative programming. 
In rewriting, proving termination is an important research field. 
Powerful methods have been introduced to establish termination of a given term
rewrite system.
One of the most natural ways to proof termination is the 
use of \emph{interpretations}. Consequentially this technique has been introduced 
quite early. Moreover, if one is interested in automatically proving termination,
\emph{polynomial interpretations} provide a natural starting point, cf.~\cite{contejean:2005}. 
However, termination proofs via polynomial interpretations are 
limited as the longest possible rewrite sequences admitted by
rewrite systems compatible with a polynomial interpretation are 
double-exponential (in the size of the initial term), see~\cite{HofbauerLautemann:1989}. Another well-studied
(and direct) termination technique is the use of reduction orders---for example 
simplification orders. Still this technique is limited, which can again
be shown by the analysis of the induced derivation length, cf.~\cite{Hofbauer:1992,Weiermann:1995,Lepper:2001a}. 
In recent years the emphasis shifted towards 
transformation techniques like the dependency pair method or 
semantic labeling.
Transformation techniques have significantly 
increased the possibility to automatically prove termination. 

Once we have established termination of a given rewrite system $\RS$, 
it seems natural to direct the attention to the analysis of the \emph{complexity} of $\RS$. 
In rewriting the complexity of a rewrite system $\RS$ is measured as the \emph{maximal derivation length}
with respect to $\RS$. As mentioned above for \emph{direct} termination methods a significant
amount of investigations has been conducted, providing a suitable foundation
for further research. Unfortunately, almost nothing is known 
about the length of derivations induced by 
state-of-the-art termination techniques like 
the dependency pair method or semantic labeling.
For the dependency pair method no results on the induced derivation
length are known. Partial result with respect to
semantic labeling are reported in~\cite{Moser:2006lpar}.

In this paper we introduce a restriction of the multiset path order, called
\emph{polynomial path order} (denoted as $\gpop$). 
Our main result states that this recursive path order induces polynomial bounds on the 
maximal length of \emph{innermost} rewrite steps. As we have successfully 
implemented this technique, we thus can \emph{automatically} verify 
for a given term rewrite system $\RS$ that $\RS$ admits at most 
polynomial innermost derivation length (on the set of \emph{constructor-based} terms).
This opens the way to automatically verify for a given program---%
written  in an eager functional programming language---%
that its runtime complexity is polynomial (in the input size). 
The only restrictions in the applicability of the result are that (i) the
functional program $\mathsf{P}$ is transformable into a term rewrite system $\RS$ and (ii) 
a feasible (i.e., polynomial)  derivation length with respect to $\RS$ 
gives rise to a feasible runtime complexity of $\mathsf{P}$. 
In short the transformation has to be 
\emph{non-termination} and \emph{complexity preserving}.

The definition of polynomial path orders employs the idea of 
\emph{tiered recursion}~\cite{BellantoniCook:1992}. Syntactically this amount to a 
separation of arguments into \emph{normal} and \emph{safe} argument. 
(Below this will be governed by the presences of mappings $\safe$ and $\normal$
associating with each function symbol a list of argument positions.)
We explain our approach by an example rewrite system
that clearly admits at most polynomial derivation length.
\begin{example}
Consider the following rewrite system $\RSmult$. 
\begin{align*}
    \add(x,\Null) &\rew x & \mult(\Null,y) &\rew \Null \\
    \add(\mS(x),y) &\rew \mS(\add(x,y)) & \mult(\mS(x),y) &\rew \add(y,\mult(x,y))
  \end{align*}  
We suppose that all arguments of the successor ($\mS$) are
safe ($\safe(\mS) = \{1\}$), that the second argument of addition ($\add$) is
safe ($\safe(\add) = \{2\}$) and that all arguments of multiplication $(\mult)$
are normal ($\safe(\mult) = \varnothing$). Furthermore let the
(strict) precedence $>$ be defined as $\mult > \add > \mS$. 
Then $\RSmult$ is compatible with $\gpop$ (see Definition~\ref{d:pop3}) and as a 
consequence of our main theorem (see Section~\ref{Main}) 
we conclude that the number of rewrite steps
starting from $\mult(\mS^n(\Null),\mS^m(\Null))$ is polynomially bounded in
$n$ and $m$. (Here we write $\mS^n(\Null)$ as abbreviation of 
$\mS(\dots(\mS(\Null)\dots))$ with $n$ occurrences of the successor symbol $\mS$.)
\end{example}

The polynomial path order is an extension of the 
\emph{path order for $\mathbf{FP}$} introduced by Arai and the 
second author in~\cite{fsttcs:2005} (see also Chapter~\ref{fsttcs05}). 
A central motivation of this research
is the observation that the direct application of the latter order 
is only successful on a handful  of (very simple) rewrite systems. 
The path order for $\mathbf{FP}$  gains only power if additional 
transformations are performed. Unfortunately, such 
powerful transformations are difficult to find automatically.

Further note that the polynomial path order is to some extent related to the 
\emph{light multiset path order} introduced
by Marion~\cite{Marion:2003}. Roughly speaking the light multiset path order 
is a tamed version of the multiset path order, characterising the functions computable
in polytime. It seems important to stress that the below stated main theorem fails for the 
light multiset path order. This can be easily seen from the next example.
\begin{example}
Consider the following rewrite system  $\RSbin$. 
(This is Example 2.21 about binomial coefficients from \cite{SteinbachKuehler:1990}.)
  \begin{align*}
    \bin(x,\Null) &\rew \mS(\Null) & \bin(\mS(x),\mS(y)) &\rew 
    \mP(\bin(x,\mS(y)),\bin(x,y)) \\
    \bin(\Null,\mS(y)) &\rew \Null
  \end{align*}
For a precedence that fulfills $\bin > \mS$, $\bin > \mP$ 
and separations of arguments $\safe(\bin) = \varnothing$, $\safe(\mP) = \{1,2\}$,
we obtain that $\RSbin$ is compatible with the light multiset path order, cf.~\cite{Marion:2003}.
However it is straightforward to verify that the (innermost) derivation
height of $\bin(\mS^n(\Null),\mS^m(\Null))$ is exponential in $n$. 
\end{example}

To test the feasibility of our approach we have implemented a small complexity analyser
based on the polynomial path order and compare its applicability to existing techniques.
To do so, we also have implemented the light multiset path order and a restricted form
of polynomial interpretations, so-called \emph{additive polynomial interpretations}, cf.~\cite{BCMT:2001}.
Note that compatibility with addivite polynomial interpretations 
induces polynomial derivation length for constructor-based terms, cf.~\cite{BCMT:2001}.

The research in~\cite{BCMT:2001,Marion:2003} falls into the realm of
\emph{implicit complexity theory}. In this context related work to our
research is due to Bonfante \textsl{et al.}~\cite{BMM:2009:tcs} but see
also seminal work by Hofmann~\cite{Hofmann:1999} and Schwichtenberg~\cite{Schwichtenberg:2006}. 
While~\cite{Hofmann:1999,Schwichtenberg:2006} are incomparable to our techniques, a comparison 
to~\cite{BMM:2009:tcs} is also not straightforward. Our principal concern is that the
\emph{termination} techniques employed allow for an \emph{complexity analysis} of the
subjected program. On the other hand the crucial feature 
of \emph{quasi-interpretations} (the central contribution of~\cite{BMM:2009:tcs}) 
is their weak monotonicity, hence termination can only be shown in 
conjunction with other termination techniques. 
For example the class of polytime computable functions can
be characterised as the class of functions computable by 
confluent constructor rewrite systems compatible with the multiset path order \emph{and}
that admit only additive quasi-interpretations, cf.~\cite{BMM:2009:tcs}. This
interesting result renders an insightful implicit characterisation of
the polytime computable function, but it is of little help, if one wants
to obtain a complexity analysis of a term rewrite system subjected to a modern termination
prover. 
Recently an interesting application of quasi-interpretations
has been reported by Lucas and Pe\~na~\cite{LP:2007}. Here the dependency pair
method is used in conjunction with quasi-interpretations to obtain bounds on
the \emph{memory consumption} of Safe programs. This method is easily
automatable, but new ideas are necessary 
to yield bounds on the \emph{runtime behaviour} of functional programs.

\medskip
The remainder of this paper is organised as follows.
In the next section we recall basic notions and starting points of
this paper. In Section~\ref{Main} we have collected our main results. In order
to prove these results we extend results originally presented in~\cite{fsttcs:2005}.
Our findings in this direction are presented in Section~\ref{Preparations}. The
central argument to prove the main theorem is then given in Section~\ref{Predicative}.
In Section~\ref{flops08:Experiments} we give the experimental evidence mentioned above.
In Section~\ref{Application} we touch upon an
application of our main theorem in recent work (together with Hirokawa and Middeldorp) 
where we study the termination behaviour of Scheme programs. Finally
in Section~\ref{flops08:Conclusion} we conclude and mention possible future work.

\section{Preliminaries} \label{Basics}

We assume familiarity with term rewriting~\cite{BaaderNipkow:1998,Terese}.
Let $\VS$ denote a countably infinite set of variables and $\FS$ a 
signature. The set of terms over $\FS$ and $\VS$ is denoted by 
$\TERMS$. We always assume that $\FS$ contains at least one constant.
The \emph{arity} of a function symbol $f$ is denoted as $\ar(f)$.
Let $>$ be a precedence on the signature $\FS$. 
The \emph{rank} of a function symbol is defined inductively as follows:
$\rk(f) = 1 + \max\{\rk(g) \mid g \in \FS \land f > g\}$.
(Here we employ the convention that the maximum of an empty set equals $0$.)
We write $\subterm$ to denote the subterm relation
and $\superterm$ for its converse.
The strict part of $\superterm$ is denoted by $\rhd$.
$\Var(t)$ denotes the set of variables occurring in a term $t$. 
The \emph{size} (\emph{depth}) of a term $t$ is denoted as $\size{t}$ ($\depth(t)$). 
The width of a term $t$ is defined inductively as follows:
$\width(t) = 1$, if $t$ is a variable or a constant, otherwise if
$t = f(t_1,\dots,t_n)$ with $n > 0$, we set $\width(t) = \max\{n,\width(t_1),\dots,\width(t_n)\}$.
The \emph{Buchholz norm} of a term $t$ is defined inductively as follows:
$\bN{t} = 1$, if $t$ is a variable and for $t=f(\seq{t})$ we set 
$\bN{t} = 1+ \max\{n,\bN{t_1},\dots,\bN{t_n}\}$.
We write $[t_1,\dots,t_n]$ to denote multisets
and $\uplus$ for the summation of multisets.

A \emph{term rewrite system} (\emph{TRS} for short) $\RS$ over
$\TERMS$ is a set of rewrite
rules $l \to r$, such that $l \notin \VS$ and $\Var(l) \supseteq \Var(r)$.
(If not mentioned otherwise, we assume $\RS$ is \emph{finite}.)
The root symbols of left-hand sides of rewrite rules
are called \emph{defined}, while all other function symbols are called 
\emph{constructors}. For a given signature $\FS$ the defined symbols are denoted
as $\DS$, while the constructor symbol are collected in $\CS$.
The smallest rewrite relation that contains $\RS$ is denoted by
$\to_{\RS}$. We simply write $\to$ for $\to_{\RS}$ if $\RS$ is clear from context.
Let $s$ and $t$ be terms. 
If exactly $n$ steps are preformed to contract $s$
to $t$ we write $s \to^n t$.
A term $s \in \TERMS$ is called a \emph{normal form} if there is no
$t \in \TERMS$ such that $s \to t$. 
The \emph{innermost rewrite relation} $\irew_\RS$
of a TRS $\RS$ is defined on terms as follows: $s \irew_\RS t$ if 
there exist a rewrite rule $l \to r \in \RS$, a context $C$, and
a substitution $\sigma$ such that $s = C[l\sigma]$, $t = C[r\sigma]$,
and all proper subterms of $l\sigma$ are normal forms of $\RS$.
A TRS is called \emph{confluent} if for all $s, t_1, t_2 \in \TERMS$
with $s \to^* t_1$ and $s \to^* t_2$ there exists a term $t_3$ such that
$t_1 \to^* t_3$ and $t_2 \to^* t_3$.
A TRS is \emph{non-overlapping} if
it has no critical pairs, cf.~\cite{BaaderNipkow:1998}. 
A TRS $\RS$ is \emph{left-linear} if for all rules $l \to r \in \RS$, 
all variables in $l$ occur at most once. If $\RS$ is additionally 
non-overlapping, then $\RS$ is called \emph{orthogonal}. Note that
every orthogonal TRS is confluent.
A \emph{constructor} TRS is a TRS whose signature $\FS$
can be partitioned into the defined symbols $\DS$ and constructor symbols $\CS$ 
in such a way that the left-hand side of each rule has the form $f(\seq{s})$ with
$f \in \DS$ and for all $i$: $s_i \in \TA(\CS,\VS)$. 
A defined function symbol is \emph{completely defined} if it does not occur
in any ground term in normal form.
A TRS is \emph{completely defined} if each defined symbol is completely defined.
An element of $\TA(\CS,\VS)$ is called a \emph{value};
we set $\Val(\RS) = \TA(\CS,\VS)$. 
We call a TRS \emph{terminating} if no infinite rewrite sequence
exists.
The \emph{derivation length} of a term $t$ with respect to a terminating
TRS $\RS$ and rewrite relation $\rsrew{\RS}$ is defined as usual: 
$\Dh{\RS}{\rew}(s) = \max\{ n \mid \exists t \; s \to^n t \}$. 
We call a term $t=f(\seq{t})$ \emph{constructor-based} if all its arguments $t_i$ 
are values, i.e., $t_i \in \TA(\CS,\VS)$ for all $1 \leqslant i \leqslant n$. The
set $\TB$ collects all constructor-based terms. 

A proper order is a transitive and irreflexive relation. 
The \emph{reflexive closure} of a proper order $\succ$ is denoted
as $\succ^=$.
A proper order $\succ$ is \emph{well-founded} if there is 
no infinite decreasing sequence $t_1 \succ t_2 \succ t_3 \cdots$.
A well-founded proper order that is also a rewrite relation is
called a \emph{reduction order}. 
We say a reduction order $\succ$ and a TRS $\RS$ are \emph{compatible}
if $\RS \subseteq {\succ}$. It is well-known that a TRS is terminating
if and only if there exists a compatible reduction order.

\section{Main Result} \label{Main}

In the sequel $\RS$ denotes a constructor TRS over 
a (possible variadic) signature $\FS$. Let $>$ denote a precedence on 
$\FS$ such that for all $f \in \DS$ 
we have for all $c \in \CS$: $f > c$. (Recall that $\FS$ contains at least one constant.) 
We assume that $\RS$ is completely defined, i.e.,
ground normal forms and ground values coincide.%
\footnote{The assumption that $\RS$
is completely defined arises naturally in the context of implicit characterisation of 
complexity classes. We follow this convention to some extent, but show 
that this restriction is not necessary.}

For each $n$-ary function symbol $f \in \DS$ of fixed arity,
we suppose the existence of a mapping $\safe$ that associates with $f$ a
(possibly empty) list $\{i_1,\dots,i_m\}$ with $1 \leqslant i_1 < \dots < i_m \leqslant n$. 
For a mapping $\safe$ and a term $t=f(\seq{t})$, 
$\safe(f)$ denotes the \emph{safe argument positions} of $t$. 
The argument positions of $t$ \emph{not} included in $\safe(f)$ are 
called \emph{normal} and are denoted by $\normal(f)$.
The mapping $\safe$ ($\normal$) is referred to as \emph{safe}
(\emph{normal}) mapping.
We generalise safe (normal) mappings to constructor symbols and variadic
function symbols as follows: For each function symbol $f \in \CS$, 
we fix $\safe(f) = \{1,\dots,\ar(f)\}$ and for each variadic function symbol $f \in \DS$ 
we assert $\safe(f) = \varnothing$.
The \emph{normalised signature} $\FSn$ contains a function symbol $f^{\Normal}$
for each $f \in \FS$. If $f$ is of fixed-arity and $\normal(f) = \{i_1,\dots,i_p\}$, 
then $\ar(f) = p$. The normalised signature $\CS^{\Normal}$ is defined accordingly.

\begin{definition} \label{d:pop2}
Let $>$ be a precedence and $\safe$ a safe mapping. 
We define $\gsq$ inductively as follows: $s = f(\seq{s}) \gsq t$ if one of the following
alternatives holds:
\begin{enumerate}
\item $f$ is a constructor and $s_i \geqsq t$
  for some $i \in \{1,\dots,n\}$,
\item $s_i \geqsq t$ for some $i \in\normal(f)$, or
\item $t = g(\seq[m]{t})$ with $f\in\DS$ and $f > g$ and $s \gsq t_i$
  for all $1 \leqslant i \leqslant m$.
\end{enumerate}
\end{definition}

We write $s \gsq t~\cl{i}$ if $s \gsq t$ follows by application of clause (i) in Definition~\ref{d:pop2}. 
A similar notation will be used for the orders defined below.

\begin{definition} \label{d:pop3}
Let $>$ be a precedence and $\safe$ a safe mapping. 
We define the \emph{polynomial path order} $\gpop$ (\emph{\POPSTAR} for short) 
inductively as follows: $s = f(\seq{s}) \gpop t$ if one of the following
alternatives holds:
\begin{enumerate}
\item $s \gsq t$, 
\item $s_i \geqpop t$ for some $i \in \{1,\dots,n\}$,
\item $t = g(\seq[m]{t})$, with $f\in\DS$, $f > g$, and the following properties hold:
\begin{itemize}
\item $s \gpop t_{i_0}$ for some $i_0 \in \safe(g)$ and
\item either $s \gsq t_i$ or $s \rhd t_i$ and $i\in\safe(g)$ for all $i \neq i_0$,
\end{itemize}
\item \label{en:pop3:iv} $t = f(\seq[m]{t})$ and for 
  $\normal(f) = \{i_1,\dots,i_p\}$, $\safe(f) = \{j_1,\dots,j_q\}$ the following properties hold:
  \begin{itemize}
  \item $[s_{i_1},\dots,s_{i_p}]~(\gpop)_{\mul}~[t_{i_1},\dots,t_{i_p}]$,
  \item $[s_{j_1},\dots,s_{j_q}]~(\geqpop)_{\mul}~[t_{j_1},\dots,t_{j_q}]$.
  \end{itemize}    
\end{enumerate}
Here $(\gpop)_{\mul}$ denotes the multiset extension of $\gpop$ and
recall that for variadic function symbols, the set of safe arguments is empty.
\end{definition}

\begin{example}
Consider the following TRS $\RSinsert$ (This is a simplification of an example from~\cite{Marion:2003}.)
\begin{equation*}
\begin{array}{rcl@{\hspace{10ex}}rcl}
\mif(\mtrue,x,y) &\rew & x  & x \mgeq \Null &\rew &\mtrue \\[1ex]
\mif(\mfalse,x,y) &\rew & y & \Null \mgeq \mS(x) &\rew &\mfalse \\[1ex]
\mins(x,\mnil) &\rew & \mcs(x,\mnil) & \mS(x) \mgeq \mS(y) &\rew & x \mgeq y \\[1ex]
\mins(x,\mcs(y,ys)) &\rew & 
\multicolumn{4}{l}{\mif(y \mgeq x,\mcs(x,\mcs(y,ys)),\mcs(y,\mins(x,ys)))}
\end{array}
\end{equation*}
We represent lists with the help of the constructors $\mnil$ and $\mcs$. 
To show compatibility with \POPSTAR, we assume a precedence $\succ$ that fulfills 
$\mins \succ \mif$, $\mins \succ \mgeq$,
$\mins \succ \mcs$,  $\Null \succ \mtrue$, and $\Null \succ \mfalse$. 
%
Further we define a safe mapping $\safe$ as follows:
\begin{align*}
  \safe(\mS) & = \{1\} &  \safe(\mif) & = \{1,2,3\} & \safe(\mins) &= \varnothing \\
  \safe(\mcs) &= \{1,2\} & \safe(\mgeq) & = \{2\} & 
\end{align*}
It is straightforward to verify that the induced polynomial path order $\gpop[\succ]$ 
is compatible with $\RSinsert$.
\end{example}

An easy inductive argument shows that if $s\in\Val(\RS)$ and $s \gpop t$, 
then $t \in \Val(\RS)$.
Note that $\gpop$ is \emph{not} a reduction order. Although $\gpop$ is a 
well-founded proper order that is closed under substitutions, 
the order is not closed under contexts due to
the restrictive definition of clause~\ref{en:pop3:iv} in the above definition. However
we still have the following theorem, which follows as the multiset path order extends $\gpop$. 

\begin{theorem}
  Every TRS $\RS$ that is compatible with $\gpop$ for some well-founded precedence $>$ is
terminating. 
\end{theorem}

As \emph{normal} and \emph{safe} arguments are distinguisable, we strengthen the 
notion of runtime complexity as follows:
\begin{equation*}
  \Rcn{\RS}(m) = \max\{\Dh{\RS}{\irew}(t) \mid t=f(\seq{t})\in\TB \ \text{and} \
\sum_{i \in \normal(f)}\size{t_i} \leqslant m \} \tpkt
\end{equation*}
This function is called the \emph{normal} runtime complexity.

\begin{mainth}
  Let $\RS$ be a finite, completely defined constructor TRS. Assume further $\RS$ is compatible with
$\gpop$, i.e., $\RS\subseteq~\gpop$.
Then the induced (normal) runtime complexity is polynomial.
\end{mainth}

Assume $\RS$ is a finite, constructor TRS that is not completely defined;
i.e., at least one defined function symbol occurs in a ground normal form.
To obtain a completely defined TRS it suffices to add suitable rules, 
thus we arrive at the following corollary, see~\cite{TR:2008} for the proof.

\begin{corollary} \label{c:1}
    Let $\RS$ be a finite, constructor TRS. Assume further $\RS$ is compatible with
$\gpop$, i.e., $\RS\subseteq~\gpop$.
Then the induced (normal) runtime complexity is polynomial.
\end{corollary}

\begin{definition} \label{d:predicative}
The \emph{predicative} rewrite relation $s \prew t$ is defined as follows: $s \prew t$ if
$s \rew t$ by contracting safe argument positions first, i.e., if
there exist a rewrite rule $l \to r \in \RS$, a context $C$, and 
a substitution $\sigma$ such that $s = C[l\sigma]$, $t = C[r\sigma]$ 
and all safe argument position of $l\sigma$ are in normal form.
\end{definition}

Clearly predicative rewriting is a generalisation of innermost rewriting.
%
Essentially following the pattern of the proof of the theorem,
we arrive at the following corollary.

\begin{corollary}
  Let $\RS$ be a finite constructor TRS. Assume further $\RS$ is compatible with
$\gpop$, i.e., $\RS\subseteq~\gpop$. 
Then for all $f \in \FS$ of arity $n$, with $\normal(f) = \{i_1,\dots,i_p\}$ and 
for all values $s_1,\dots,s_n$: $\Dh{\RS}{\prew}(f(s_1,\dots,s_n))$ is 
bounded by a polynomial
 in the sum of the sizes of the normal argument terms $s_{i_1},\dots,s_{i_p}$.
\end{corollary}

\begin{remark}
  Beckmann and Weiermann observed in~\cite{BeckmannWeiermann:1996} that
general rewriting is too powerful to serve as a suitable computation model
to characterise the class of polytime computable functions as a TRS. Their
notion of a \emph{feasible} rewrite system is reflected adequately in the notion
of predicative rewriting.
\end{remark}

\section{Polynomial Path Order on Sequences} \label{Preparations}

In this section we extend definitions and results originally presented in~\cite{fsttcs:2005}
(see also Chapter~\ref{fsttcs05}).
The main aim is to define a \emph{polynomial path order} $\gpopv$ \emph{on sequences of terms} such that
$\gpopv$ induces polynomial derivation length with respect to a compatible TRS $\RS$.

Let $\odot \not \in \FSn$ be a variadic function symbol. 
We extend the normalised signature $\FSn$ by $\odot$ and define
$\SE(\FSn,\VS) = \TA(\FSn\cup\{\odot\},\VS)$. Elements of $\SE(\FSn,\VS)$
are sometimes referred to as \emph{sequences}.
Instead of $\odot(\seq{s})$, we usually write $(\sexpr{s})$ and
denote the empty sequence $()$ as $\varnothing$.
Let $a=(a_1 \dots a_n)$ and $b=(b_1 \dots b_m)$ be elements of 
$\SE(\FSn,\VS)$. For $a\not=\varnothing$ and $b\not=\varnothing$ define 
$a \cons b = (a_1 \dots a_n~b_1 \dots b_m)$. If $a=\varnothing$ ($b = \varnothing$) 
we set $a\cons b = b$ ($a \cons b = a$).

Let $>$ denote the precedence on $\FSn$ induced by 
the total precedence $>$ on $\FS$.
Buchholz~\cite{Buchholz:1995} was the first to observe 
that finite term rewrite systems compatible with recursive path 
orders $\succ$ are even compatible to \emph{finite} approximations of $\succ$. 
This observation carries over to polynomial path orders. 
The following definitions generalise the \emph{path order on $\mathbf{FP}$} (\emph{POP} for short) 
as defined in~\cite{fsttcs:2005}. To keep this exposition short, we 
only state the definition of  \emph{approximations} of 
the \emph{polynomial path order} $\gpopv$ \emph{on sequences}. 
The general definitions for $\gppv$ and $\gpopv$ is obtained 
by dropping the restrictions on depth and width, cf.~\cite{TR:2008}. 
Note that $\gpopv$ can be conceived as the \emph{limit} of
the finite approximations $\gpopv_k$.
We use the convention that $f \in \FSn$, i.e., $s = f(\seq{s})$
implicitly indicates that $f \not= \odot$.

\begin{definition} 
Let $k,l \geqslant 1$ and let $>$ be a precedence. 
We define $\gppv_{k}^{l}$ inductively as follows: 
$s \gppv_{k}^{l} t$ for $s=f(\seq{s})$ or $s=(\sexpr{s})$ if one of the following alternatives holds:
\begin{enumerate}
\item $s_i~(\geqppv)_{k}^{l}~t$ for some $i \in \{1,\dots,n\}$,
\item $s=f(\seq{s})$ such that of the following two possibilities holds:
  \begin{itemize}
  \item $t = g(\seq[m]{t})$ with $f > g$ or
  \item $t=(\sexpr[m]{t})$,
  \end{itemize}
and $s \gppv_{k}^{l-1} t_i$ for all $1 \leqslant i \leqslant m$, 
and $m < k+\width(s)$, or
\item $s=(\sexpr{s})$, $t = (\sexpr[m]{t})$ and the following properties hold:
  \begin{itemize}
  \item $[\seq[m]{t}] = N_1 \uplus \cdots \uplus N_n$,
  \item there exists $i \in \{1,\ldots,n\}$ such that $[s_i] \not= N_i$,
  \item for all $1 \leqslant i \leqslant n$ such that $[s_i] \not= N_i$ 
    we have $s_i \gppv_k^l r$ for all $r \in N_i$
  \item $m < k + \width(s)$.
  \end{itemize}
\end{enumerate}
We write $\gppv_k$ to abbreviate $\gppv_k^k$.
\end{definition}

\begin{definition} 
Let $k,l \geqslant 1$ and let $>$ be a precedence.
We define the \emph{approximation of the polynomial path order}
$\gpopv^l_k$ \emph{on sequences} inductively as follows: $s \gpopv_k^l t$ 
for $s=f(\seq{s})$ or $s=(\sexpr{s})$ if one of the following
alternatives holds:
\begin{enumerate}
\item $s \gppv_k^l t$, 
\item $s_i~(\geqpopv)_k^l~t$ for some $i \in \{1,\dots,n\}$,
\item $s = f(\seq{s})$, $t = (\sexpr[m]{t})$, and the following
  properties hold:
  \begin{itemize}
  \item $s \gpopv_k^{l-1} t_{i_0}$ for some $i_0 \in\{1,\dots,n\}$, 
  \item $s \gppv_k^{l-1} t_i$ for all $i \neq i_0$, and
  \item $m < k + \width(s)$,
  \end{itemize}

\item $s = f(\seq{s})$, $t = f(\seq[m]{t})$ with
$(\sexpr{s}) \gpopv_k^l (\sexpr[m]{t})$, or

\item $s=(\sexpr{s})$, $t = (\sexpr[m]{t})$ and the following properties hold:
  \begin{itemize}
  \item $[\seq[m]{t}] = N_1 \uplus \cdots \uplus N_n$,
  \item there exists $i \in \{1,\ldots,n\}$ such that $[s_i] \not= N_i$,
  \item for all $1 \leqslant i \leqslant n$ such that $[s_i]\not=N_i$: 
    $s_i \gppv_k^l r$ for all $r \in N_i$, and
  \item $m < k + \width(s)$.
  \end{itemize}
\end{enumerate}
We write $\gpopv_k$ to abbreviate $\gpopv_k^k$.
\end{definition}

Note that $\varnothing$ is the minimal element of $\gppv_k$ and $\gpopv_k$
and that $\gpopv$ is a reduction order. 
The following lemmas are direct consequences of the definitions.

\begin{lemma} \label{l:approx:popalt}
\mbox{}
\begin{enumerate}
\item If $s \gpopv_{k} t$ and $k<l$, then $s \gpopv_{l} t$.
\item If $s \gpopv_k t$, then $C[s] \gpopv_k C[t]$, where $C[\ctx]$ denotes
a context over $\SE(\FSn,\VS)$.
\end{enumerate}
\end{lemma}

\begin{lemma} \label{l:approx}
  If $s \gpopv_k^l t$, then $\depth(t) \leqslant \depth(s)+l$ and
  $\width(t) \leqslant k + \width(s)$. Moreover, if $s \gpopv_k^l t$, 
  then $\bN{t} \leqslant \bN{s} +k+l$. 
\end{lemma}
  
By Lemma~\ref{l:approx}, there exists a (uniform) constant $c$ such that 
$\bN{t} \leqslant \bN{s} + c$, whenever $s \gpopv_k t$. And thus if
we have a $\gpopv_k$-descending sequence $s = t_0 \gpopv_k t_1 \gpopv_k \dots \gpopv_k t_\ell$
we conclude that $\bN{t_i} \leqslant ci + \bN{s}$ for all $i \geqslant 1$.

\begin{definition} \label{d:slow}
We define
\begin{align*}
  \Slow(s) & \defsym \max\{\ell\in \N \mid \exists (t_0,\dots,t_\ell) \colon 
  s = t_0 \gpopv_k t_1 \gpopv_k \dots \gpopv_k t_\ell \} \\
  \Fpop_{k,p}(m) &\defsym \max\{ \Slow(f(\seq{t})) \colon \rk(f) = p \land \sum_{i}\Slow(t_i) \leqslant m \}
 \end{align*}
In the definition of $\Fpop_{k,p}$, we assume $f \in \FSn$.
\end{definition}

A direct consequence of Definition~\ref{d:slow} is that 
$\Slow((\sexpr{t})) = n+\sum_{i=1}^{n} \Slow(t_i)$ holds.
The following lemma is generalisation of a similar lemma in~\cite{fsttcs:2005} and
the proof given in~\cite{fsttcs:2005} can be easily adapted.
\begin{lemma}\label{l:flops08:upper}
We define $d_{k,0} \defsym k+1$ and $d_{k,p+1} \defsym (d_{k,p})^{k}+1$. 
Then for all $k,p$ there exists a constant $c$ (depending only on $k$ and $p$)
such that for all $m$: $\Fpop_{k,p}(m) \leqslant c (m+2)^{d_{k,p}}$.
\end{lemma}

As a consequence of Lemma~\ref{l:flops08:upper} we obtain that $\Fpop_{k,p}(m)$ is asymptotically
bounded by $m^{d_{k,p}}$ for large enough $m$. 
The following lemma follows by a standard inductive argument.
\begin{lemma} \label{l:value}
  For all $k$, there exists a constant $c$ such that for $s \in \TA(\CS^{\Normal} \cup \{\odot\},\VS)$:
  $\Slow(s) \leqslant c \cdot \size{s}^2$.
\end{lemma}

We arrive at the main theorem of this section.

\begin{theorem} \label{t:2}
For all $f \in \FSn$ of arity $n$, for all $s_1,\dots,s_n \in \TA(\CSn\cup\{\odot\})$, and for all $k$:
$\Slow(f(s_1,\dots,s_n))$ is bounded by a polynomial in the sum of the sizes of $s_1,\dots,s_n$.   
The polynomial depends only on $k$ and the rank of $f$.
 \end{theorem}
\begin{proof}
Let $f \in \FSn$ and let $s_1,\dots,s_n \in \TA(\CS^{\Normal} \cup \{\odot\})$. 
By Lemma~\ref{l:flops08:upper} there exists $c_1 \in \N$ depending on $k$ and $\rk(f)$ such that 
  \begin{equation} \label{eq:t:2}
    \Slow(f(s_1,\dots,s_n)) \leqslant m^{c_1} 
  \end{equation}
if $\sum_i \Slow(s_i) \leqslant m$ and $m$ is large enough. 
By Lemma~\ref{l:value}, there exists a constant $c_2$ 
(depending on the rank of the function symbols in $s_i$) 
such that $\Slow(s_i) \leqslant c_2 \cdot \size{s_i}^2$.
Replacing $m$ in~\eqref{eq:t:2} by $c_2 \cdot (\sum_i \size{s_i})^2$ and setting 
$c = c_2^{c_1}$ yields:
\begin{equation*}
 \Slow(f(s_1,\dots,s_n)) \leqslant [c_2 \cdot (\sum_i \size{s_i})^2]^{c_1} = 
c \cdot (\sum_i \size{s_i})^{2c_1} 
\end{equation*}
\end{proof}

\section{Predicative Interpretation} \label{Predicative}

The purpose of this section is to prove our main theorem. Let $\RS$ denote
a completely defined, constructor TRS. We embed the order $\gpop$ into 
$\gpopv_k$ such that $k$ depends only on $\RS$. 
This becomes possible if we represent the information on normal and safe arguments 
underlying the definition of $\gpop$ explicitly by interpreting the signature $\FS$
in the normalised signature $\FSn$.

Let $t$ be a term and recall that $\bN{t}$ denotes its (Buchholz) norm. 
We represent the norm unary. Let $\mS$ denote a fresh nullary function symbol that
is minimal in the precedence $>$ on $\FSn$. 
We define $\SV(t) = \U(\bN{t})$, where $\U \colon \N \to \TA(\{\mS,\circ\})$
denotes the representation of $n$ as a sequence $(\mS~\cdots~\mS)$ 
with $n$ occurrences of the constant $\mS$.
As a direct consequence of the definition, we have: 
$s \rhd t$ implies $\SV(s) \gpopv_k \SV(t)$ for any $k$.

\begin{definition} \label{d:pi}
Let $\safe$ denote a safe mapping. A \emph{predicative interpretation} (with respect to  $\TA(\FS,\VS)$) 
is a pair $(\ints,\intn)$ of mappings $\ints \colon \TA(\FS,\VS) \to \TA(\FSn,\VS)$ and 
$\intn \colon \TA(\FS,\VS) \to \TA(\FSn,\VS)$, defined as follows:
  \begin{align*}
    \ints(t) & =
    \begin{cases}
       \varnothing & \text{if $t \in \Val(\RS)$} \\
       (f^{\Normal}(\intn(s_{j_1}),\dots,\intn(s_{j_p}))~\ints(s_{i_1})~\dots~\ints(s_{i_q})) 
      & \text{if $t \not\in \Val(\RS)$} 
    \end{cases}\\
    \intn(t) & = (\ints(t))\cons\SV(t)
  \end{align*}
In the definition of $\ints$, we assume 
$t = f(\seq{s})$, $\normal(f) = \{j_1,\dots,j_p\}$ and $\safe(f) = \{i_1,\dots,i_q\}$.
(Recall that $\safe(f) \cup \normal(f) = \{1,\dots,n\}$.)
\end{definition}

Note that $\intn(s) \gppv_k \ints(s)$ (and thus
$\intn(s) \gpopv_k \ints(s)$) holds for any $k$. 
Moreover, observe that for any term $t$, we have
$\width(\intn(t)) = 1 + \bN{t}$ which follows by a simple
inductive argument.
We arrive at the two main lemmas  of this section.

\begin{lemma} \label{l:3}
Let $f(l_1,\dots,l_n) \rew r \in \RS$, let $\sigma \colon \VS \to \Val(\RS)$ be a substitution
and let $k = 2 \cdot \max\{\size{r} \mid l \rew r \in \RS\}$.
If $f(l_1,\dots,l_n) \gsq r$ then 
$f^{\Normal}(\intn(l_{i_1}\sigma),\dots,\intn(l_{i_p}\sigma)) \gppv_{k} \Q(r\sigma)$ 
for $\Q \in \{\ints,\intn\}$, where $\normal(f) = \{i_1,\dots,i_p\}$.
\end{lemma}
\begin{proof}
We sketch the proof plan: Instead of showing the lemma directly, one shows the 
following stronger property for terms $s,t \in \TA(\FS,\VS)$ 
where $s$ is either a value or of form $f(s_1,\dots,s_n)$ such that $s_i\sigma \in \Val(\RS)$ 
for all $1 \leqslant i \leqslant n$.
\begin{center}
\begin{tabular}{l@{\hspace{5ex}}c}
$(\dag)$ &
\parbox{60ex}{%
Let $\ell = \bN{t}$, if $f \in \DS$, then $s \gsq  t$ implies
$\Q(s\sigma) \gppv_{2\ell} f^{\Normal}(\intn(s_1\sigma),\dots,\intn(s_p\sigma)) \gppv_{2\ell} \Q(t\sigma)$;
otherwise $\intn(s\sigma) \gppv_{2\ell} \intn(t\sigma)$ holds.%
}
\end{tabular}
\end{center}
Here we suppose $\safe(f) = \{p+1,\dots,n\}$. 
To show $(\dag)$ one proceeds by induction on $\gsq$. 
See~\cite{TR:2008} for the complete proof.
\end{proof}

\begin{lemma} \label{l:4}
Let $l \rew r \in \RS$, let $\sigma \colon \VS \to \Val(\RS)$ be a substitution,
and let $k = 2\cdot\max\{\size{r} \mid l \rew r \in \RS\}$.
If $l \gpop r$ then $\Q(l\sigma) \gpopv_{k} \Q(r\sigma)$ for $\Q\in\{\ints,\intn\}$.
\end{lemma}
\begin{proof}
  Similar to the proof of Lemma~\ref{l:3} one shows the following property
for terms $s,t \in \TA(\FS,\VS)$ where $s$ is either a value or of form $f(s_1,\dots,s_n)$
such that $s_i\sigma \in \Val(\RS)$ for all $1 \leqslant i \leqslant n$.

\begin{center}
\begin{tabular}{l@{\hspace{5ex}}c}
$(\ddag)$
&
\parbox{60ex}{%
Let $\ell = \bN{t}$. If $f\in\DS$, then $s = f(s_1,\dots,s_n) \gpop  t$ implies 
(i) $f^{\Normal}(\intn(s_1\sigma),\dots,\intn(s_p\sigma)) \gpopv_{2\ell} \ints(t\sigma)$
and 
(ii) $(f^{\Normal}(\intn(s_1\sigma),\dots,\intn(s_p\sigma)))\cons\SV(s\sigma) \gpopv_{2\ell} \intn(t\sigma)$. 
Otherwise if $f\in\CS$ then $\intn(s\sigma) \gpopv_{2\ell} \intn(t\sigma)$ holds.
}
\end{tabular}
\end{center}
Here we suppose $\safe(f) = \{p+1,\dots,n\}$. 
To show $(\ddag)$ one proceeds by induction on $\gpop$. 
See~\cite{TR:2008} for the complete proof.
\end{proof}

From Lemmata~\ref{l:3} and~\ref{l:4} the main lemma of this section follows.

\begin{lemma} \label{l:main}
Let $s$ and $t$ be terms such that $s \irew t$ and let 
$k = 2\cdot\max\{\size{r} \mid l \rew r \in \RS\}$.
Then $\Q(s) \gpopv_k \Q(t)$ for $\Q\in\{\ints,\intn\}$.
\end{lemma}

\begin{mainth} 
Let $\RS$ be a finite, completely defined constructor TRS. 
Assume further $\RS$ is compatible with $\gpop$.
Then the induced (normal) runtime complexity is polynomial.
\end{mainth}
\begin{proof}
Let $t=f(\seq{t})$ be term in $\TB$ and without loss of generality 
let $\safe(f) = \{p+1,\dots,n\}$. We set $k = 2\cdot\max\{\size{r} \mid l \rew r \in \RS\}$. 
By Lemma~\ref{l:main} any innermost rewrite steps $t \irew u$ induces
$\ints(t) \gpopv_{k} \ints(u)$. Thus we obtain:
\begin{align*}
    \Dh{\RS}{\irew}(f(t_1,\dots,t_n)) 
    & = \max\{\ell \mid 
    \exists u \; t \irew^{\ell} u \} \\
    & \leqslant \max\{\ell \mid 
    \exists \; (s'_1,\dots,s'_{\ell}) \colon \ints(t)
    \gpopv_k s'_1 \gpopv_k \dots \gpopv_k s'_{\ell} \}\\
    & \leqslant \Slow(\ints(f(t_1,\dots,t_n)))
\end{align*}
Next note that $\ints(f(t_1,\dots,t_n)) = (f^{\Normal}(\intn(t_1),\dots,\intn(t_p))~\varnothing \dots \varnothing)$. 
By Theorem~\ref{t:2} and the observation following Definition~\ref{d:slow} we see that 
\begin{equation*}
  \Slow((f^{\Normal}(\intn(t_1),\dots,\intn(t_p))~\varnothing \dots \varnothing)) 
  \leqslant n + 1 + \Slow(f^{\Normal}(\intn(t_1),\dots,\intn(t_p)))
\end{equation*}
Employing Lemma~\ref{l:flops08:upper}, we see (for a fixed $f$) 
that $n + 1 + \Slow(f^{\Normal}(\intn(t_1),\dots,\intn(t_p)))$ is asymptotically
bounded by a polynomial in the sum of the sizes of the arguments $\intn(t_1)$,\dots,$\intn(t_p)$.  
By definition $\size{\intn(t_i)} = \bN{t_i} \leqslant \size{t_i}$ for all $1 \leqslant i \leqslant p$. 

Hence for each term $t \in \TB$, $\Dh{\RS}{\irew}(t)$ is bounded by a polynomial in the sum of the sizes of the
normal argument terms of $t$. In particular, as the signature $\FS$ is finite, the 
normal runtime complexity function is polynomial.
\end{proof}

\begin{remark} \label{r:constructor}
In the above theorem we assume a constructor TRS. It is not difficult to see that this
restriction is not necessary. (Essentially one replaces the application of Lemmata~\ref{l:3} and~\ref{l:4}
by the application of the properties $(\dag)$ and $(\ddag)$ respectively.)
However, the restriction that the arguments of $f$ are in normal form is necessary. 
Hence we prefer the given formulation of the theorem.
\end{remark}

\section{Experimental Data} \label{flops08:Experiments}

To prove compatibility of a given TRS $\RS$ with recursive path orders
we have to find a \emph{precedence} $>$ such that the induced order is compatible with $\RS$.
When we want to orient $\RS$ by a polynomial path order $\gpop$ we additionally
require a suitable \emph{safe mapping}.
To automate this search we encode the constraint $s \gpop t$ into a propositional formula:
\begin{equation*}
    \enc(s \gpop t) = \enc_1(s \gpop t) \vee \enc_2(s \gpop t) \vee \enc_3(s \gpop t) \vee \enc_4(s \gpop t)
\end{equation*}
Here $\enc_i(\cdot)$ is designed to encode clause $(i)$ from Definition \ref{d:pop3}. 
Based on such an encoding, compatibility of a TRS with $\gpop$ 
becomes expressible as the satisfiability of the formula
$\bigl(\bigwedge_{l \rew r \in \RS} \enc(l \gpop r)\bigr) \wedge P \wedge S$. 
Here the subformula $P$ is satisfiable 
if and only if all the variables $>_{f,g}$ (defined below) encode a strict precedence,
see~\cite{ZM:2007} for a suitable definition of $P$. The subformula $S$ is used to cover 
the additional conditions imposed on safe mappings defined in the beginning of 
Section \ref{Main}.

We only describe cases $(2)$--$(4)$, 
the encoding for case $(1)$---the comparison using the weaker order $\gsq$---%
can be easily derived in a similar fashion. 
If $s = f(\seq{s})$ we set $\enc_2(s \gpop t) =  \bigvee_{i} {s_i \geqpop t}$, 
otherwise $\enc_2(s \gpop t) = \bot$.
For case $(3)$ we introduce for every function symbol $f$ and argument 
position $i$ of $f$ the (propositional) variables $\beta_{f,i}$, such that 
$\beta_{f,i} = \true$ represents the assertion $i \in \safe(f)$.
Moreover, for all function symbols $f,g$ we introduce variables $>_{f,g}$ 
such that truth of $>_{f,g}$ expresses that $f > g$ holds.
If $s=f(\seq{s})$ and $t=g(\seq[m]{t})$ for $f \in \DS$ with $f \not= g$, we define 
$\enc_3(s \gpop t)$ as:
\begin{equation*}
  {>}_{f,g}~\land~\bigvee_{i_0=1}^m (\enc(s\gpop t_i)~\land~\beta_{g,i_0}~\land 
  \bigwedge_{i = 1,i \not=i_0}^m (\enc(s \gsq t_i) \lor (\beta_{g,i} \land (s \rhd t_i)))
\end{equation*}
(For $s$, $t$ of different shape, we set $\enc_3(s \gpop t)=\bot$.)
To deal with case $(4)$ we follow~\cite{SK:2007}. The main idea is to describe 
a multiset comparison in terms of \emph{multiset covers}. 
Formally, a multiset  cover is a pair of mappings 
$\gamma \colon \set{1,\dots,m} \to \set{1,\dots,n}$ 
and $\varepsilon \colon \set{1,\dots,n} \to \set{\true,\false}$
such that for all $i$, $j$ ($1 \leqslant i \leqslant n$, $1 \leqslant j \leqslant m$):
if $\varepsilon(i) = \true$ then the set $\set{j \mid \gamma(j) = i}$ is a singleton. 
It is easy to see that  $[\seq{s}]~(\succ^{=})_{\mul}~[\seq[m]{t}]$ if there 
exists a multiset cover $(\gamma,\varepsilon)$ such that for each $j$ there exists
an $i$ with $\gamma(j) = i$ and $\varepsilon(i) = \true$ implies $s_i = t_j$,
while $\varepsilon(i) = \false$ implies $s_i \succ t_j$.
Similarly we obtain $[\seq{s}] \succ_{\mul} [\seq[m]{t}]$ if  
$[\seq{s}]~(\succ^{=})_{\mul}~[\seq[m]{t}]$ and $\varepsilon(i) = \false$ for some 
$i\in\{1,\dots,n\}$.

This definition allows an easy encoding of multiset comparisons and
based on it, clause $(4)$ of Definition~\ref{d:pop3} becomes representable 
(for terms $s = f(\seq{s})$ and $t = f(\seq[m]{t})$) as the conjunction of the 
following two conditions together with the assumption that there exists a suitable
multiset cover $(\gamma,\varepsilon)$:
\begin{itemize}
\item whenever $\gamma(j) = i$ then the indicated argument positions $i$ and $j$, 
  are either both normal or both safe,
\item 
  at least one cover is strict ($\varepsilon(i) = \false$) for some
  normal argument position $i$ of $f$.
\end{itemize}
We introduce variables $\gamma_{i,j}$ and 
$\varepsilon_i$, where $\gamma_{i,j} = \true$ represents $\gamma(j) = i$ and
$\varepsilon_i = \true$ denotes $\varepsilon(i) = \true$
($1 \leqslant i \leqslant n$, $1 \leqslant j \leqslant m$).
Summing up, we set $\enc_4(s~(\gpop)_{\mul}~t)$ ($s = f(\seq{s})$ and $t = f(\seq[m]{t})$)
equal to:
\begin{multline*}
  \bigwedge_{i=1}^n \bigwedge_{j=1}^m \Bigl( \gamma_{i,j} \to \bigl( \varepsilon_i \to \enc(s_i = t_j) \bigr)
  \wedge \bigl( \neg \varepsilon_i \to \enc(s_i \succ t_j) \bigr)
  \wedge \bigl( \beta_{f,i} \leftrightarrow \beta_{f,j} \bigr) \Bigr) \\
  \wedge \bigwedge_{j=1}^m \one(\gamma_{1,j},\dots,\gamma_{n,j}) 
  \wedge \bigwedge_{i=1}^n \bigl(\varepsilon_i \to \one(\gamma_{i,1},\dots,\gamma_{i,m})\bigr) 
  \wedge  \bigvee_{i=1}^n \bigl( \neg \beta_{f,i} \wedge \neg \varepsilon_i \bigr) 
\end{multline*}
Here $\one(\alpha_1,\dots,\alpha_n)$ is satisfiable if and only if 
exactly one of the variables $\alpha_1,\dots,\alpha_n$ is $\true$. And if
$s$, $t$ do not have the assumed form, we set $\enc_4(s~(\gpop)_{\mul}~t)=\bot$.

We compare the polynomial path order \POPSTAR\
to a restricted class of polynomial interpretations (SMC for short)~\cite{BCMT:2001} 
and to LMPO~\cite{Marion:2003}.
SMC refers to simple-mixed polynomial interpretations where 
constructor symbols are interpreted by a strongly linear
(also called \emph{additive}) polynomial~\cite{BCMT:2001}. 
Defined symbols on the other hand are interpreted by simple-mixed polynomials~\cite{contejean:2005}. 
Since \POPSTAR\ and LMPO~are in essence syntactic 
restrictions of MPO we also provide a 
comparison to MPO. 
\POPSTAR\ is implemented using the previously described propositional encoding; 
while the implementation of SMC rests on a propositional encoding of the
techniques described in~\cite{contejean:2005}.
To check satisfiability we employ MiniSat.%
\footnote{Available online at \url{http://minisat.se}.}
LMPO and MPO are implemented using an extension of the constraint 
solving technique described in~\cite{HirokawaMiddeldorp:2003}, which 
allows us to compare different implementation techniques at the same
time. 

As testbed we use those TRSs from the termination problem data base version 4.0
that can be shown terminating with at least one of the 
tools that participated in the termination competition 2007.%
\footnote{These 957 systems can be found online: \url{http://www.lri.fr/~marche/termination-competition/2007/webform.cgi?command=trs&file=trs-standard.db&timelimit=120}}
We use three different testbeds: \textbf{T} collects the 957 terminating
TRSs from TPDB, \textbf{TC} collects the 449 TRSs from the TPDB that are also
constructor systems, and \textbf{TCO} collects the 236 TRSs that are terminating, constructor based
and orthogonal.%
\footnote{The main reason for this delineation is that in related
work~\cite{BCMT:2001,Marion:2003} \emph{confluent} constructor TRS are considered.}
The results of our comparisons are given in Table~\ref{tab:exp}.
The tests presented below were conducted on a small complexity analyser
running single-threaded on a 2.1 GHz Intel Core 2 Duo with 1 GB of memory.
For each system we used a timeout of 30 seconds.%

\begin{table}[h]
\caption{Experimental results}
\label{tab:exp}
\smallskip
  \centering
  \begin{tabular}{l|@{\hspace{7ex}}lccccc}
    \hline
    \multicolumn{2}{c}{\TOP \BOT}  & \POPSTAR & LMPO    & SMC     & MPO \\
    \hline
    \textbf{T} & Yes \TOP    & 65       & 74       & 156      & 106   \\
    & Maybe                  & 892      & 812      & 395      & 847  \\
    & Timeout (30 sec.) \BOT & 0        & 71       & 406      & 4    \\
    \hline\hline
    \textbf{TC} & Yes \TOP   & 41       & 54       & 83       & 65  \\
    & Maybe                  & 408      & 372      & 271      & 381  \\
    & Timeout (30 sec.) \BOT & 0        & 23       & 95       & 3    \\
    \hline\hline
    \textbf{TCO} & Yes \TOP  & 19       & 25       & 38       & 29   \\
    & Maybe                  & 217      & 201      & 147      & 207  \\
    & Timeout (30 sec.) \BOT & 0        & 10       & 51       & 0    \\

    \hline
    \multicolumn{2}{l}{\TOP \BOT Average yes time (milliseconds)} & 15 & 14 & 1353 & 10 \\
    \hline
  \end{tabular}
 \end{table}

Some comments: 
What is noteworthy is the good performance of \POPSTAR
as a direct \emph{termination} method in comparison to MPO. It is well-known
that MPO implies primitive recursive derivation length, cf.~\cite{Hofbauer:1992}. In contrast to
this \POPSTAR\ implies polynomial runtime complexity and is thus a much weaker
order. Still more than half of the TRSs compatible with MPO are also compatible with \POPSTAR.
On the other hand the comparison between \POPSTAR\ and LMPO is quite favourable 
for our approach. Compatibility with LMPO tells us that the given TRS is 
(in principle)  polytime computable, while compatibility with \POPSTAR\ tells 
additionally that the runtime of a straightforward implementation (using an
innermost strategy) is polytime computable. Hence compatibility with \POPSTAR\ provides
us with a theoretical stronger result, while the difference on the experimental
data appears negligible.

The good performance of SMC in strength is a clear indication that currently
(restrictions of) semantic termination techniques (like polynomial interpretations) are of some interest
in automatically estimating the runtime complexity of TRSs. This may be surprising, 
as for additive polynomial interpretations it is (almost) trivial
to check that the induced upper bound on the derivation height is polynomial. 
However, the significant increase in the time necessary to find an additive polynomial interpretation, as indicated in Table~\ref{tab:exp},
clearly shows the limits of semantic methods for large examples.

\section{An Application: Complexity of Scheme Programs}
\label{Application}

In recent work together with Hirokawa and Middeldorp (see~\cite{AHMM:2007})
we study the runtime complexity of (a subset of)
Scheme programs by a translation into so-called \emph{S-expression rewrite systems}
(\emph{SRS} for short).
By designing the translation to be complexity preserving, 
the complexity of the initial Scheme program can be estimated by analysing 
the complexity of the resulting SRS.
 Here we indicate how our main theorem is applicable to 
(a subset of) S-expression rewrite systems, cf.~\cite{Toyama:2004}.

\begin{definition}
  Let $\CC$ be a set of constants, $\VV$ be a set of
variables such that $\VV \cap \CC = \varnothing$, and
$\circ \notin \CC \cup \VV$ a variadic function symbol. 
We define the set $\SEXPRS$ of \emph{S-expressions} built from $\CC$ and $\VV$ 
as $\TA(\CC \cup \{ \circ \},\VS)$.
We write $(\sexpr{s})$ instead of $\circ(\seq{s})$.
An \emph{S-expression rewrite system} (\emph{SRS} for short) 
is a TRS with the property that the left- and right-hand
sides of all rewrite rules are S-expressions.
\end{definition}

Let $\RSS$ be an SRS over $\SEXPRS$ and let $\CC = \DD \cup \KK$ such
that $\DD \cap \KK = \varnothing$. We call the elements of $\KK$ \emph{constructor}
constants and the elements of $\DD$ \emph{defined} constants.
We momentarily redefine the notion of \emph{value} in the context of SRSs.
%
The set of \emph{values} $\Val(\RSS)$ of $\RSS$ with respect to $\KK$ is
inductively defined as follows:
\begin{enumerate}
\item
if $v \in \CC$ then $v \in \Val(\RSS)$,
\item
if $\seq{v} \in \Val(\RSS)$ and $\C{c} \in \KK$ then
$(\C{c} \ \V{v_1} \dots  \V{v_n}) \in \Val(\RSS)$.
\end{enumerate}

Observe that (defined) constants are values, this reflects that in Scheme
procedures are values, cf.~\cite{R6RS} and allows for a representation
of higher-order programs.
Scheme programs are conceivable as SRSs allowing conditional
\C{if} expressions in conjunction with an eager, i.e., innermost
rewrite strategy. 
Thus we can delineate a class of SRSs that easily 
accommodate a suitably large subset of Scheme programs.

\begin{definition}
$\RSS$ is called a \emph{constructor} if, for every $l \to r \in \RSS$,
$l = (l_0 \cdots\ l_n)$ with
$l_0 \in \DD$ and
$l_i \in \Val(\RSS)$ for all $i \in \{ 1, \dots, n \}$.
(Here the set of values $\Val(\RSS)$ is defined with respect to  $\KK$.)
\end{definition}

\begin{corollary} \label{t:srs}
Let $>$ denote a precedence on $\CC$ such that for all $f \in \DS$ 
we have for all $c \in \KK$: $f > c$ and let $\gpop$ denote the
induced \POPSTAR. Let $\RSS$ be a constructor SRS compatible with $\gpop$. 
Then for all $f \in \DS$ of arity $n$ and for all values $s_1,\dots,s_n$: 
$\Dh{\RSS}{\irew}((\C{f} \ \V{s_1} \ \dots \ \V{s_n}))$ is bounded by a polynomial
in the sum of the sizes of the arguments $s_1,\dots,s_n$.
\end{corollary}
\begin{proof}
It is important to note that the set of S-expressions $\SEXPRS$ equals
$\TA(\CC \cup \{\circ\},\VS)$, i.e., SRSs are \emph{first-order} rewrite
systems, whose single defined symbol is the variadic function symbol $\circ$.

Hence Theorem~\ref{t:srs} follows almost immediately from Corollary~\ref{c:1}.
However the fact that according to the above definition
values may contain defined symbol need to be taken into account. For that is
suffices to redefine Definition~\ref{d:pi} in the natural
way. It is not difficult to argue that suitable adaption of Lemmata~\ref{l:3} 
and~\ref{l:4} to SRSs are provable.
\end{proof}

\section{Conclusion} \label{flops08:Conclusion}

In this paper we have introduced a restriction of the multiset path order, called
\emph{polynomial path order} (\emph{\POPSTAR} for short). 
Our main result states that \POPSTAR\ induces 
polynomial runtime complexity. 
In Section~\ref{flops08:Experiments} we have provided evidence 
that our approach performs well in comparison to related methods.
In Section~\ref{Application} the necessary theory to
apply our main theorem in the context of (higher-order) 
functional languages with eager evaluations has been developed. 
In related work (together with Hirokawa and Middeldorp),
studying the termination behaviour and the runtime complexity 
of (a subclass of higher-order) Scheme programs, this basis
has proven quite useful, cf.~\cite{AHMM:2007}.

In concluding we also want to mention that 
as an easy corollary to our main theorem we obtain that 
\POPSTAR\ also characterises the polytime computable
functions. To be precise the polytime computable functions 
are exactly the functions computable by an orthogonal constructor 
TRS (based on a simple signature) compatible with \POPSTAR.
(Here \emph{simple} signature means that the size
of any constructor term depends linearly on its depth, 
an equivalent restriction is necessary in~\cite{Marion:2003}.) See~\cite{TR:2008}
for details.

In future work we will strengthen the applicability of our method. 
The experimental evidence presented in Section~\ref{flops08:Experiments} shows that 
compatibility of rewrite systems with \POPSTAR\ can be easily and quickly
tested. However, the strength of the method seems to be improvable. One 
possible field of future work is to extend \POPSTAR\ to quasi-precedences. The
theoretical changes necessary to accomodate quasi-precedences seem to be manageable.
Another natural extension is to combine \POPSTAR\ with the transformation
technique of semantic labeling, cf.~\cite{Zantema:1995}. 
It is easy to see that semantic labeling (in the basic form) does 
not affect the derivation length. Furthermore for \emph{finite} models 
the main theorem remains directly applicable. 

}

\chapter{Proving Quadratic Derivational Complexities using Context Dependent Interpretations}
\label{rta08}

\subsection*{Publication Details}
G.~Moser and A.~Schnabl.
\newblock Proving quadratic derivational complexities using context dependent
  interpretations.
\newblock In \emph{Proceedings of the 19th International Conference on Rewrite
  Technques and Applications}, number 5117 in LNCS, pages 276--290. Springer
  Verlag, 2008.%
\footnote{This research was partially supported by FWF (Austrian Science Fund) project P20133.}

\subsection*{Ranking}
The International Conference on Rewriting Techniques and Applications has been ranked
\textbf{A} by CORE in 2007.

\subsection*{Abstract}
In this paper we study \emph{context dependent interpretations}, a
semantic termination method extending interpretations over the natural numbers, 
introduced by Hofbauer. 
We present two subclasses of context dependent interpretations and
establish tight upper bounds on the induced derivational complexities. In
particular we delineate a class of interpretations that induces 
\emph{quadratic derivational complexity}. 
Furthermore, we present an algorithm for mechanically proving termination 
of rewrite systems with context dependent interpretations. 
This algorithm has been implemented and we present ample numerical data
for the assessment of the viability of the method.

{
\input{rta08.sty}

\section{Introduction}

In order to assess the complexity of a (terminating) term rewrite system (TRS for short) 
it is natural to look at the maximal length of derivation sequences, as
suggested by Hofbauer and Lautemann in~\cite{HofbauerLautemann:1989}.
To be precise, let $\RS$ denote a finitely branching and terminating TRS over
a finite signature.   
The \emph{derivational complexity function} 
with respect to $\RS$ (denoted as $\dc{\RS}$) 
relates the length of the longest derivation sequence to the size of the initial term. 
For \emph{direct} termination techniques it is often possible to infer an upper bound on
$\dc{\RS}(n)$ from the termination proof of $\RS$, cf.~\cite{HofbauerLautemann:1989,Hofbauer:1992,Weiermann:1995,Moser:2006lpar,GHWZ07}.
(Currently it is unknown how to estimate the \emph{derivational complexity} of 
a TRS $\RS$, if termination of $\RS$ has been shown via 
transformation methods like the dependency pair method or semantic labeling, but see~\cite{Moser:2006lpar,HM:2008}
for partial results in this direction.)
For example \emph{linear} derivational complexity can be
verified by the use of automata techniques: linear match-bounded TRSs 
induce linear derivational complexity, see~\cite{GHWZ07}. 
Unfortunately such a feasible growth rate is not typical. Already 
termination proofs by polynomial interpretations imply a double-exponential
upper bound on the derivational complexity, cf.~\cite{HofbauerLautemann:1989}.
In both cases the upper bounds are tight. 

However, the tightness of the mentioned bounds does not imply 
that the upper bounds are \emph{always} optimal.
In particular polynomial interpretations typically overestimate the derivational complexity.
In~\cite{Hofbauer:2001} Hofbauer introduced so-called \emph{context dependent interpretations} as a remedy. 
These interpretations extend traditional interpretations by introducing an
additional parameter. The parameter changes in the course of evaluating a term, 
which makes the interpretation dependent on the context. 
The crucial advantage is that context dependent interpretations typically improve the
induced bounds on the derivational complexity of TRSs. Furthermore this technique
allows the handling of non-simple terminating systems. 
(See~\cite{Hofbauer:2001} and Section~\ref{rta08:Preliminaries} for further details.)

In this paper, we establish theoretical and practical extensions of Hofbauer's approach. 
As theoretic contributions, we present two subclasses of context dependent interpretations,
i.e., we introduce \emph{$\Delta$-linear} and \emph{$\Delta$-restricted interpretations}. 
We show that $\Delta$-linear interpretations induce \emph{exponential} derivational complexity,
while $\Delta$-restricted interpretations induce \emph{quadratic} derivational complexity.
Furthermore, we provide examples showing that these bounds are tight. In~\cite{Hofbauer:2001} it
is shown that context dependent interpretations are expressive enough to show termination of TRSs
that are not simply terminating. 
We improve upon this and show that $\Delta$-restricted interpretations suffice here.
On the practical side, we design an algorithm that automatically searches for 
$\Delta$-linear interpretations and $\Delta$-restricted interpretations, which
shows that the technique can be mechanised.
This answers a question posed by Hofbauer in~\cite{Hofbauer:2001}.
The procedure has been implemented and we provide ample numerical data to 
assess its viability. TRSs with \emph{polynomial} derivational complexity appear to be of special
interest. Thus, we finally compare the applicability of our method 
to other termination techniques that also induce polynomial derivational complexity. 

\medskip
The remainder of this paper is organised as follows. In the next
section we recall basic notions and starting points of this paper. In
Section~\ref{Automation} we introduce the class of $\Delta$-linear interpretations and
describe the algorithm that mechanises the search for 
$\Delta$-linear and $\Delta$-restricted interpretations. 
In Section~\ref{rta08:Complexity}, we obtain the mentioned results on the 
derivational complexities induced by either of these interpretations. 
Furthermore, we show in this section that already $\Delta$-restricted interpretations allow the treatment of
non-simple terminating TRSs. 
Section~\ref{rta08:Experiments} provides experimental data and finally in Section~\ref{rta08:Conclusion}
we conclude and mention future work.

\section{Context Dependent Interpretations} \label{rta08:Preliminaries}

We assume familiarity with the basics of term rewriting, see \cite{BaaderNipkow:1998,Terese}.
Knowledge of context dependent interpretations~\cite{Hofbauer:2001} will be helpful.
Below we recall the basic results from the latter paper in a slightly different,
but equivalent way, compare~\cite{Hofbauer:2001,Schnabl:2007}. See~\cite{Hofbauer:2001} for the motivation
and intuition underlying the introduced concepts. 

Let $\FS$ be a finite signature, let $\VS$ be a set of variables 
and let $\RS$ denote a terminating TRS over $\FS$. The induced relation $\rsrew{\RS}$ is assumed to be 
finitely branching. We simply write $\to$ for $\to_{\RS}$ if $\RS$ is clear from context.
The \emph{derivation length} of a term $t$ with respect to $\RS$ 
is defined as follows: 
$\dl{\RS}(t) = \max\{ n \mid \exists u \; t \to^n u \}$. 
The \emph{derivational complexity} (with respect to $\RS$) 
is defined as: $\dc{\RS}(n) = \max\{\dl{\RS}(t) \mid \size{t} \leqslant n\}$,
where $\size{t}$ denotes the size of $t$, i.e., the number of symbols of $t$ as usual. 
(For example the size of the term $\mf(\ma,x)$ is $3$.)
We say the derivational complexity of $\RS$ is linear, quadratic, double-exponential, if
$\dc{\RS}(n)$ is bounded by a linear, quadratic, double-exponential function in $n$,
respectively. 
A \emph{context dependent $\FS$-algebra} (\emph{CDA} for short) 
$\C$ is a family of $\FS$-algebras over the reals parametrised by a set 
$D \subseteq \RPOS$ of positive reals. 
A CDA $\C$ associates to each function symbol $f\in\FS$ of arity $n$,
a collection of $n+1$ mappings: 
$f_{\C} \colon D \times (\RNNEG)^n  \to \RNNEG$ and $f^i_{\C} \colon D \to D$ for
all $1 \leqslant i \leqslant n$. As usual $f_{\C}$ is called \emph{interpretation
function}, while the mappings $f^i_{\C}$ are called \emph{parameter functions}.
In addition $\C$ is equipped with a 
set $\{>_{\Delta}\mid\Delta\in D\}$ of proper orders, 
where we define: $z >_{\Delta} z'$ if and only if $z - z' \geqslant \Delta$.

Let $\C$ be a CDA and let a \emph{$\Delta$-assignment} denote 
a mapping: $\alpha:D\times\VS \to \RNNEG$. 
We inductively define a mapping $[\alpha,\Delta]_{\C}$ from the set of terms into 
the set $\RNNEG$ of non-negative reals:
\begin{equation*}
[\alpha,\Delta]_{\C}(t) \defsym 
\begin{cases}
  \alpha(\Delta,t) & 
  \text{if $t \in \VS$} \\
  f_{\C}(\Delta,[\alpha,f_{\C}^1(\Delta)]_{\C}(t_1), \dots, [\alpha,f_{\C}^n(\Delta)]_{\C}(t_n)) &
  \text{if $t = f(t_1,\dots,t_n)$} \tpkt
\end{cases} 
\end{equation*}

We fix some notational conventions: Due to the special role of the additional variable
$\Delta$, we often write $f_{\C}[\Delta](z_1,\dots,z_n)$ 
instead of $f_{\C}(\Delta,z_1,\dots,z_n)$. 
Furthermore, we usually denote the evaluation of $t$ as 
$\eval{t}$, 
if the respective algebra is clear from context.

We say that a CDA $\C$ is \emph{\emph{$\Delta$}-monotone} if 
for all $\Delta \in D$ and for all $a_1,\ldots,a_n,b\in \RNNEG$ with $a_i >_{f_{\C}^i(\Delta)} b$
for some $i\in\{1,\ldots,n\}$, we have
\begin{equation*}
f_{\C}[\Delta](a_1,\ldots,a_i,\ldots,a_n) 
>_{\Delta} f_{\C}[\Delta](a_1,\ldots,b,\ldots,a_n) \tpkt
\end{equation*}
Note that if all interpretation functions $\cdifunc{f}{\Delta}$ are weakly
monotone with respect to the standard ordering on $\RNNEG$, then
validity of the inequalities
\begin{equation*}
f_{\C}[\Delta](z_1,\ldots,z_i+f^i_{\C}(\Delta),\ldots,z_n)-
f_{\C}[\Delta](z_1,\ldots,z_i,\ldots,z_n)\geqslant\Delta \tkom
\end{equation*}
suffices in order to conclude $\Delta$-monotonicity of $\C$, cf.~\cite{Hofbauer:2001}.

A CDA $\C$ is \emph{compatible} with a TRS $\RS$ (or 
$\RS$ is compatible with $\C$) if for every rewrite rule $l\rew r\in\RS$, 
every $\Delta \in D$, and any assignment $\alpha$:
$[\alpha,\Delta](l) >_{\Delta} [\alpha,\Delta](r)$ holds.

\begin{example}[\cite{Hofbauer:2001}] \label{ex:rta08:1}
As running example, we consider the TRS $\RSa$ with the single rewrite rule
$\ma(\mb(x))\rew\mb(\ma(x))$. We assume $D = \RPOS$.
The following interpretation and parameter functions 
\begin{align*}
\cdifunc{\ma}{\Delta}(z)&=(1+\Delta)z & \cdictx{\ma}{1}(\Delta)&=\frac{\Delta}{1+\Delta}\\
\cdifunc{\mb}{\Delta}(z)&=z+1 & \cdictx{\mb}{1}(\Delta)&=\Delta \tkom
\end{align*}
define a CDA $\C$ that  is $\Delta$-monotone and
compatible with $\RSa$, compare~\cite{Hofbauer:2001}.
\end{example}

\begin{theorem}[\cite{Hofbauer:2001}] \label{t:cdi}
  Let $\RS$ be a TRS and suppose that there exists a $\Delta$-monotone and compatible 
CDA $\C$. Then $\RS$ is terminating and 
\begin{equation} \label{eq:cdi}
  \dl{\RS}(t)\leqslant\inf_{\Delta\in D}\frac{\eval{t}}{\Delta}
\end{equation}
holds for all terms $t \in \TA(\FS,\VS)$.
\end{theorem}
The next example clarifies the impact of Theorem~\ref{t:cdi}, compare~\cite{Hofbauer:2001}.
\begin{example} \label{ex:1:dh}
  Consider the TRS $\RSa$ together with the CDA $\C$ in Example~\ref{ex:rta08:1}. 
Suppose $\mc \in \FS$ is a constant and $\cdifunc{\mc}{\Delta} =0$.
We assert $D = \RPOS$. Then we obtain $\eval{\ma^n(\mb^m(\mc))} = (1+\Delta n) m$ 
and hence:
\begin{equation*}
  \inf_{\Delta > 0} \frac{\eval{\ma^n(\mb^m(\mc))}}{\Delta} = 
    \inf_{\Delta > 0} \big( \frac{1}{\Delta} + n \big)m = nm \geqslant \dl{\RSa}(\ma^n(\mb^m(\mc))) \tpkt
\end{equation*}

Furthermore, an easy inductive argument reveals:
$\dl{\RSa}(\ma^n(\mb^m(\mc))) = nm$. 
Hence with respect to the term $\ma^n(\mb^m(\mc))$, compatibility with $\C$
entails an optimal upper bound on the derivation length of $\RSa$. 
This is also true for all ground terms. A proof of 
$\inf_{\Delta > 0} \frac{\eval{t}}{\Delta} = \dl{\RSa}(t)$ for all $t \in \TA(\FS)$
can be found in~\cite{Hofbauer:2001}. 
\end{example} 

\begin{definition} \label{d:delta-quotient}
A \emph{$\Delta$-quotient} is an expression of the form 
\begin{equation*}
 \frac{\Delta}{a+b\Delta} \tkom
\end{equation*}
where $a, b \in \N$ and either $a>0$ or $b>0$. A $\Delta$-quotient $d$
is \emph{nontrivial}, if $d \not= \Delta$.
\end{definition}

\begin{lemma} \label{l:Delta-quotient}
Let $d_1$, $d_2$ be $\Delta$-quotients and let $d = d_1[\Delta \defsym d_2]$ 
denote the result of substituting  $d_2$ for $\Delta$ in $d_1$. Then $d$ is
a $\Delta$-quotient.
\end{lemma}

As usual a \emph{polynomial} $P$ in the variables $z_1,\dots,z_n$ (over the
reals) is a finite sum $\sum_{i=1}^{m} c_i z_1^{i_1} \dots z_n^{i_n}$. 
To accommodate $\Delta$-quotients we slightly generalise polynomials.

\begin{definition}
An \emph{extended monomial} $M$ in the variables $\Delta$ and $z_1,\dots,z_n$ 
is a finite product $c\cdot\prod_{i}v_i$ such that
$c$ is an integer and $v_i$ is $x^n$, $x \in \{\Delta,z_1,\dots,z_n\}$ or
$v_i$ is a $\Delta$-quotient. 
The integer $c$ is called the \emph{coefficient} and the expression
$v_i$ a \emph{literal}.
Finally, an \emph{extended polynomial} $P$ over $\Delta \in D$ and $z_1,\dots,z_n \in \RNNEG$ 
is a finite sum $\sum_{i} M_i$ of extended monomials $M_i$ (in $\Delta$ and $z_1,\dots,z_n$).
\end{definition}
Note that the coefficients of an extended polynomial are integers. 
If the context clarifies what is meant, we will drop the qualifier ``extended''.
Examples~\ref{ex:rta08:1} and~\ref{ex:1:dh} as well as the examples
studied in~\cite{Hofbauer:2001} suggest a restricted notion of context dependent algebras.
This is the subject of the next definition.
\begin{definition} \label{d:pcdi}
A \emph{polynomial context dependent interpretation} of $\FS$ is a CDA
$(\C,\{>_{\Delta}\mid\Delta\in D\})$ satisfying the following
properties:
\begin{itemize}
\item the interpretation function $f_{\C}$ is an extended polynomial,
\item the parameter set $D$ equals $\RPOS$, and
\item for each $f \in\FS$ the parameter functions $f^i_{\C}$ are $\Delta$-quotients.
\end{itemize}
\end{definition}

\begin{lemma} \label{l:eval}
  Let $\C$ denote a polynomial context dependent interpretation, let
$\alpha$ be a $\Delta$-assignment, and let $t$ be a term. Then
$\eval{t}$ is an extended polynomial.
\end{lemma}
\begin{proof}
  The lemma is a direct consequence of the definitions and
Lemma~\ref{l:Delta-quotient}. 
\end{proof}

\begin{remark}
Hofbauer showed in \cite{Hofbauer:2001} that for any monotone polynomial interpretation 
compatible with a TRS $\RS$, there exists a polynomial context dependent interpretation which is
$\Delta$-monotone and compatible with $\RS$ and induces at least the
same upper bound on the derivational complexity as the polynomial interpretation.
\end{remark}

\section{Automated Search for Context Dependent Interpretations}
\label{Automation}

One approach to find context dependent interpretations
(semi-)automatically was already mentioned in Hofbauer's paper
\cite{Hofbauer:2001}. 
A given polynomial interpretation is suitably lifted to
a context dependent interpretation 
such that monotonicity and compatibility are preserved, but
the upper bound on the derivational complexity is often improved. 
Unfortunately, experimental evidence suggests that the applicability of this
heuristics is limited, if one is interested in automatically
finding complexity bounds, see Section~\ref{rta08:Experiments} for further details.
However, the standard approach for automatically
proving termination via polynomial interpretations as stipulated by
Contejean et al.~\cite{contejean:2005} can be adapted.
The description of this adaption is the topic of this section.
We restrict the form of parametric interpretations that we consider.
\begin{definition} \label{d:delta-linear}
A \emph{(parametric) $\Delta$-linear interpretation} 
is a polynomial context dependent interpretation $\C$ whose 
interpretation functions and parameter functions have the following form:
\begin{align*}
f_{\C}(\Delta,\seq{z}) &= \sum_{i=1}^n a_{(f,i)} z_i + \sum_{i=1}^n b_{(f,i)} 
z_i\Delta + c_f\Delta+ d_f \\
\cdictx{f}{i}(\Delta) &= \frac{\Delta}{a_{(f,i)}+b_{(f,i)}\Delta} 
\end{align*}
where the occurring coefficients are supposed to be natural numbers.
For a \emph{parametric} $\Delta$-linear interpretation, 
$a_{(f,i)}$, $b_{(f,i)}$, $c_f$, and $d_f$ ($f\in\FS$, $1 \leqslant i \leqslant n$)
are called \emph{coefficient variables}.
\end{definition}
Note that for any $\Delta$-linear interpretation, we have 
$a_{(f,i)} > 0$ or $b_{(f,i)} > 0$ ($f\in\FS$, $1 \leqslant i \leqslant n$):
Any $\Delta$-linear interpretation is a
polynomial context dependent interpretation by
definition. And hence the parameter functions
have to be $\Delta$-quotients, cf.~Definition~\ref{d:pcdi}.
Moreover the coefficients $a_{(f,i)}$, $b_{(f,i)}$ 
are used in the interpretation function \emph{and} the parameter functions. 
This is necessary for the correctness of Lemma~\ref{l:monotone} below. 
\begin{example} \label{ex:1:para}
Consider the TRS $\RSa$ from Example~\ref{ex:rta08:1}. 
The parametric interpretation and parameter functions have the form:
\begin{align*}
\ma_{\C}[\Delta](z)&=az +bz \Delta + c\Delta + d
&\ma_{\C}^1(\Delta)&=\frac{\Delta}{a+b\Delta} \\
\mb_{\C}[\Delta](z)&= ez + f z \Delta + g\Delta + h
&\mb_{\C}^1(\Delta)&=\frac{\Delta}{e + f\Delta} \tpkt
\end{align*}
\end{example}

The following lemma is a direct consequence of the definitions.
\begin{lemma} \label{l:monotone}
Let $\C$ be an $\Delta$-linear interpretation.
Then $\C$ is $\Delta$-monotone.
\end{lemma}

Due to Lemma~\ref{l:monotone}, in order to prove termination
of a given TRS $\RS$, it suffices to find a $\Delta$-linear interpretation 
compatible with $\RS$. This observation is reflected in the following definition.
\begin{definition}
\label{ruleconstraintsdef}
Let $\RS$ be a TRS and let $\C$ be a parametric $\Delta$-linear interpretation. 
The \emph{compatibility constraints} of $\RS$ with respect to $\C$ are
defined as
\begin{align*}
  \CC(\RS,\C) = & \{\eval{l}- \eval{r}-\Delta \geqslant 0 \mid l\rew r\in\RS\} \cup \\
  & \cup \{a_{(f,i)}+b_{(f,i)}-1\geqslant0 \mid f\in\FS, 1 \leqslant i \leqslant \arity(f)\} \tpkt
\end{align*}
Here $\arity(f)$ denotes the arity of $f$ and $\alpha$ refers to
a \emph{symbolic $\Delta$-assignment}: Expressions of the
form $\eval{x}$ for $x \in \VS$ remain unevaluated.
\end{definition}

While the first half of $\CC(\RS,\C)$ represents compatibility with
$\RS$, the second set of constraints guarantees that the
denominators of the occurring $\Delta$-quotients are different from $0$. 
Thus any solution to $\CC(\RS,\C)$, instantiating coefficients with natural
numbers, represents a polynomial context dependent interpretation compatible with $\RS$.
\begin{example} \label{ex:1:constraint}
Consider the (parametric) CDA $\C$ from Example~\ref{ex:1:para} and 
set $\Delta_1 = \ma_{\C}^1(\Delta)$ and $\Delta_2 = \mb_{\C}^1(\Delta)$.
Let $\alpha_1 = [\alpha,\Delta_2[\Delta \defsym \Delta_1]](x)$ and 
let $\alpha_2 = [\alpha,\Delta_1[\Delta \defsym \Delta_2]](x)$. 
Then the constraint $\eval{\ma(\mb(x))} - \eval{\mb(\ma(x))} - \Delta \geqslant 0$ becomes:
\begin{gather*}
\bigl( ae \alpha_1 + af \alpha_1 \Delta_1 + ag \Delta_1 + 
 be \alpha_1 \Delta + bf \alpha_1 \Delta_1 \Delta + bg \Delta_1 \Delta + (bh + c) \Delta + \\
+ ah + d \bigr) - \bigl( ae \alpha_2 + be \alpha_2 \Delta_2 + ce \Delta_2 + 
af \alpha_2 \Delta + bf \alpha_2 \Delta_2 \Delta + cf \Delta_2 \Delta + \\
+ (df + g) \Delta + de + h \bigr) - \Delta \geqslant 0 \tpkt
\end{gather*} 
\end{example}

For all constraints $(P \geqslant 0) \in \CC(\RS,\C)$, 
$P$ is an extended polynomial, cf.~Lemma~\ref{l:eval}. 
It is easy to see how an extended polynomial
(over $\Delta,z_1,\dots,z_n$) is transferable into a (standard) polynomial 
(over $\Delta,z_1,\dots,z_n$): Multiply (symbolically) with
denominators of (nontrivial) $\Delta$-quotients till all (nontrivial)
$\Delta$-quotients are eliminated. This simple procedure is denoted as $\Alg$.
Correctness and termination of the procedure follow trivially.

\begin{definition} \label{d:pcc}
Let $\RS$ be a TRS and let $\C$ be a parametric $\Delta$-linear interpretation. 
The \emph{polynomial compatibility constraints} of $\RS$ with respect to $\C$ are
defined as follows:
$\PCC(\RS,\C) \defsym \{ P' \geqslant 0 \mid \text{$P \geqslant 0 \in \CC(\RS,\C)$ and
    $P' \defsym \Alg(P)$}\}$.
\end{definition}

\begin{example} \label{ex:1:pcc}
Consider the constraint $P \geqslant 0$ depicted in 
Example~\ref{ex:1:constraint}. 
To apply the algorithm $\Alg$ we first have to symbolically multiply with the
expression $a + b\Delta$ and later with $e + f\Delta$. 
The resulting constraint $P' \geqslant 0$
(with the polynomial $P'$ in the ``variables'' $\Delta$, $\alpha_1$, and $\alpha_2$) has the form:
\begin{gather*}
\bigl((b^2ef + bf^2) \alpha_1 \Delta^3 + (2abef + af^2 + b^2e^2 + bef) \alpha_1 \Delta^2 \\
+ (2abe^2 + a^2ef + aef) \alpha_1 \Delta + (a^2 e^2) \alpha_1 \bigr) \\
-  \bigl( (abf^2 + b^2f) \alpha_2 \Delta^3 + (a^2f^2 + 2abef + abf + b^2e) \alpha_2 \Delta^2 \\
+  (2a^2ef + abe^2 + aeb) \alpha_2 \Delta + (a^2 e^2) \alpha_2 \bigr) 
+ \bigl( (b^2fh - bdf^2 - bf) \Delta^3  \\
+  (2abfh + b^2eh +bdf - adf^2 - 2bdef - bfh -be - af) \Delta^2 \\
+  (a^2fh + 2abeh + adf + bde - 2adef - afh - bde^2 - beh - ae)) \Delta \\
+  (a^2eh + ade - ade^2 - aeh) \bigr) \geqslant 0 \tpkt
\end{gather*}
We obtain $\PCC(\RSa,\C) = \{P' \geqslant 0, a+b-1 \geqslant 0, e + f- 1 \geqslant 0 \}$,
where the last two constraints reflect that all denominators of $\Delta$-quotients 
are non-zero.
\end{example}

Let $P \geqslant 0$ be a constraint in $\PCC(\RS,\C)$ such that
$n$ distinct symbolic assignments $[\alpha,d](x)$ occur in $P$
($x \in \VS$, $d$ a $\Delta$-quotient). (In Example~\ref{ex:1:pcc} two
symbolic assignments occur: $\alpha_1$ and $\alpha_2$.)
Then $P$ is conceivable as a polynomial in $\Z[\Delta,z_1,\dots,z_n]$. 
It remains to verify 
that (a suitable instance of) $P$ is \emph{positive}, i.e., we have to prove that 
$P(\Delta,z_1,\dots,z_n) \geqslant 0$ for any values $\Delta > 0$, $z_i \geqslant 0$. 
This is achieved by testing for absolute positivity instead of positivity, compare~\cite{contejean:2005}.

A polynomial $P$ is \emph{absolutely positive} if $P$ has non-negative
coefficients only. A parametric polynomial $P$ is called \emph{absolutely positive}
if there exists an instance $P'$ of $P$ such that $P'$ is absolutely positive.
Clearly any absolutely positive polynomial is positive. Thus for a given 
constraint $P \geqslant 0 \in \PCC(\RS,\C)$ it suffices to find instantiations
of the coefficient variables such that all coefficients are natural numbers. 
This is achieved through the construction of suitable Diophantine inequalities over the
coefficients. 

\begin{lemma} \label{l:compatible}
  Let $\RS$ be a TRS and let $\C$ denote a parametric $\Delta$-linear interpretation. 
If for all $P \geqslant 0 \in \PCC(\RS,\C)$, $P$ is absolutely positive then
there exists an instantiation of $\C$ compatible with $\RS$.
\end{lemma}
\begin{proof}
If $P$ is absolutely positive, there exist natural numbers that can be substituted to the
coefficient variables in $P$ such that the resulting polynomial $P'$ is
absolutely positive and thus positive. By definition this implies that the
constraints in $\CC(\RS,\VS)$ are fulfilled. 
We define an instantiation $\C'$ of $\C$ by applying the same 
substitution to the coefficient variables in $\C$. 
Then $\C'$ is compatible with $\RS$.
\end{proof}

As an immediate consequence of Lemmata~\ref{l:monotone},~\ref{l:compatible},  
and Theorem~\ref{t:cdi} we obtain the following theorem. 
\begin{theorem} \label{t:correctness}
Let $\RS$ be a TRS and let $\C$ denote a 
parametric $\Delta$-linear interpretation.
Suppose for all $P \geqslant 0 \in \PCC(\RS,\C)$, $P$ is absolutely positive.
Then $\RS$ is terminating and property~\eqref{eq:cdi} holds for $D = \RPOS$.
\end{theorem}

It is easy to see that the Diophantine inequalities 
induced by Example~\ref{ex:1:pcc} cannot be solved, if
the symbolic assignments $\alpha_1$ and $\alpha_2$ are 
treated as different variables.
This motivates the next definition.
\begin{definition} \label{d:VE}
Given a TRS $\RS$ and a $\Delta$-linear interpretation
$\C$, the \emph{equality constraints} of $\RS$
with respect to $\C$ are defined as follows:
\begin{equation*}
\VE(\RS,\C)=\{(a+b\Delta) - (c+d\Delta) =0 \mid \text{Property \emph{($\ast$)} is fulfilled}\}
\end{equation*}
\begin{description}
  \item[($\ast$)] There exists $P \geqslant 0 \in \PCC(\RS,\C)$, $x \in \VS$ such that
$[\alpha,d_1](x)$ and $[\alpha,d_2](x)$ occur in $P$ and
$d_1 = \frac{\Delta}{a + b\Delta} \not= \frac{\Delta}{c + d\Delta} = d_2$.
\end{description}
\end{definition}

\begin{example} \label{ex:1:EC}
  Consider Example~\ref{ex:1:pcc}. Property $(\ast)$ is
applicable to the $\Delta$-quotients $d_1$, $d_2$ 
in the $\Delta$-assignments $\alpha_1 = [\alpha,d_1]$ and 
$\alpha_2 = [\alpha,d_2]$ as
\begin{equation*}
  d_1 = \frac{\Delta}{ae + (be + f) \Delta} \not= 
  \frac{\Delta}{ae + (af + b) \Delta} = d_2 \tpkt
\end{equation*}
Thus the constraint $(ae + (be + f) \Delta) - (ae + (af + b) \Delta) = 0$
occurs in  $\VE(\RSa,\C)$. This is the only constraint in $\VE(\RSa,\CS)$.
\end{example}

Let $P \geqslant 0 \in \PCC(\RS,\C)$, assume the equality constraints in $\VE(\RS,\C)$
are fulfilled and assume we want to test for absolute positivity of $P$.
By assumption distinct symbolic assignments can be treated as equal, which
may change the coefficients we need to consider in $P$.
This is expressed by writing $P \geqslant 0 \in \PCC(\RS,\C) \cup \VE(\RS,\C)$.
Furthermore, we call a parametric polynomial a \emph{zero polynomial}
if there exists an instance $P'$ of $P$ such that $P' = 0$. 
\begin{corollary} \label{c:correctness}
Let $\RS$ be a TRS and let $\C$ denote a parametric $\Delta$-linear interpretation. 
Suppose for all $P \geqslant 0 \ (P = 0) \in \PCC(\RS,\C) \cup \VE(\RS,\C)$, 
$P$ is absolutely positive ($P$ is a zero polynomial). 
Then $\RS$ is terminating and property~\eqref{eq:cdi} holds for $D = \RPOS$.
\end{corollary}

Corollary~\ref{c:correctness} opens the way to efficiently search for 
CDAs: Finding a $\Delta$-monotone and compatible CDA $\C$ amounts to solving the Diophantine constraints in 
$\PCC(\RS,\C) \cup \VE(\RS,\C)$. Recall that solvability of Diophantine
constraints is undecidable~\cite{Matiyasevich:1970}. However, there is an easy remedy for this: 
we restrict the domain of the coefficient variables to a finite one. 

\begin{example} \label{ex:1:CP}
Consider the TRS $\RSa$ from Example~\ref{ex:rta08:1} and the
$\Delta$-linear interpretation $\C$ from Example~\ref{ex:1:para}. Applying the above described
algorithm, the following Diophantine (in)equalities need to be solved.
\begin{align*}
&  \begin{aligned}
&  b^2ef + bf^2 - abf^2 - b^2f \geqslant 0 & &abe^2 + a^2ef + aef -2a^2ef - aeb \geqslant 0 \\
&  b^2fh - bdf^2 -bf \geqslant 0 &  &a^2eh + ade - ade^2 - aeh \geqslant 0\\
&  a+b-1 \geqslant 0 & &e + f- 1 \geqslant 0 \\ 
&  be + f - af - b = 0 & &af^2 + b^2e^2 + bef -a^2f^2 - abf - b^2e  \geqslant 0
  \end{aligned} \\ 
&  2abfh + b^2eh +bdf - adf^2 - 2bdef -bfh -be - af \geqslant 0 \\
&  a^2fh + 2abeh + adf + bde - 2adef - afh - bde^2 - beh -ae \geqslant 0 \tpkt
\end{align*}  
Here the constraints $a+b-1 \geqslant 0$, $e + f- 1 \geqslant 0$ guarantee that
the denominators of occurring $\Delta$-quotients are positive,
and the equality $be + f - af - b = 0$ expresses the 
equality constraint in $\VE(\RSa,\C)$.
Our below discussed implementations of the algorithm 
presented in this section find the following satisfying 
assignments for the coefficient variables fully automatically:\\[-2ex]
\begin{equation*}
  a = b = e = h = 1 \qquad c = d = f = g = 0 \tpkt
\end{equation*}
\end{example}

\section{Derivational Complexities Induced by Polynomial Context Dependent Interpretations}
\label{rta08:Complexity}

In this section we show that the derivational complexity induced by $\Delta$-linear interpretations
is exponential and that this bound is tight. Furthermore, we introduce a 
restricted subclass of $\Delta$-linear interpretations that induces (tight)
quadratic derivational complexity. 

Recall the TRS $\RSa$ considered in Example~\ref{ex:rta08:1}. This TRS belongs to a family of TRSs $\RS_k$ for
$k > 0$: $\ma(\mb(x))\rew\mb^k(\ma(x))$ and it is not difficult to see that for $k \geqslant 2$
the derivational complexity of $\RS_k$ is exponential. In~\cite{Hofbauer:2001}
$\Delta$-linear interpretations $\C_k$ were introduced such that 
\begin{equation*}
 \inf_{\Delta > 0} \frac{[\alpha,\Delta]_{\C_k}(t)}{\Delta} = \dl{\RS_k}(t) \tkom 
\end{equation*}
holds for any ground term. I.e., for all $k > 0$ there exist 
$\Delta$-linear interpretations that optimally bound the derivational complexities of $\RS_k$.
This triggers the question whether we can find such context dependent interpretations automatically.
The next example answers this question affirmatively, for $k=2$.%
\footnote{The answer remains positive for $k=3$. Detailed
experimental evidence and additional information on the considered
constraints are available 
at~\url{http://cl-informatik.uibk.ac.at/~aschnabl/experiments/cdi/}.
}
\begin{example} \label{ex:rta08:2}
Consider the TRSs $\RSb$: $\ma(\mb(x))\rew\mb(\mb(\ma(x)))$.%
\footnote{This is Example 2.50 in Steinbach and Kühler's collection~\cite{SteinbachKuehler:1990}.}    
To find a $\Delta$-linear interpretation, we employ the same parametric
interpretation $\C$, as in Example~\ref{ex:1:para} and build the
set of constraints $\CC(\RSb,\C)$ and consecutively the polynomial
compatibility constraints $\PCC(\RSb,\C)$ together with the equality
constraints $\VE(\RSb,\C)$. We only state the (automatically) obtained interpretation
and parameter functions:%
\begin{align*}
  \cdifunc{\ma}{\Delta}(z) &=(2 + 2\Delta)z & \cdictx{\ma}{1}(\Delta)& =\frac{1}{2 + 2\Delta}\\
  \cdifunc{\mb}{\Delta}(z)& =z + 1 & \cdictx{\mb}{1}(\Delta)&=\Delta \tpkt
\end{align*}
\end{example}
As a consequence of Example~\ref{ex:rta08:2} we see the existence of TRSs, compatible with
$\Delta$-linear interpretations, whose derivational complexity function is exponential.
Moreover, we have the following lemma.

\begin{lemma} \label{l:exponential}
  Let $\C$ denote a $\Delta$-linear interpretation and let $K$ denote the maximal
coefficient occurring in $\C$. Further let $t$ be a ground term, $\alpha$ 
a $\Delta$-assignment and $\Delta > 0$. Then
$\eval{t} \leqslant (K+2)^{\size{t}} (\Delta + 1)$.
\end{lemma}
\begin{proof}
  Straightforward induction on $t$.
\end{proof}

\begin{theorem} \label{t:exponential}
Let $\RS$ be a TRS and let $\C$ denote a $\Delta$-linear interpretation
compatible with $\RS$. Then $\RS$ is terminating and 
$\dc{\RS}(n) = 2^{\bO(n)}$. Moreover there exists
a TRS $\RS$ such that $\dc{\RS}(n) = 2^{\Omega(n)}$.
\end{theorem}
\begin{proof}
The proof of the upper bound follows 
the pattern of the proof of Theorem~\ref{t:quadratic} below. 
To show that this upper bound is tight, we consider the
TRS $\RSb$ from Example~\ref{ex:rta08:2}. 
It is easy to see that $\dc{\RSb}(n) = 2^{\Omega(n)}$ holds.
\end{proof}

In order to establish a termination method that induces 
\emph{polynomial} derivational complexity, 
we restrict the class of $\Delta$-linear interpretations.

\begin{definition} \label{d:delta-restricted}
A \emph{$\Delta$-restricted interpretation} is a 
$\Delta$-linear interpretation. In addition we require that 
for the interpretation functions and parameter functions
\begin{align*}
f_{\C}(\Delta,\seq{z}) &= \sum_{i=1}^n a_{(f,i)} z_i + \sum_{i=1}^n b_{(f,i)} 
z_i\Delta + c_f\Delta+ d_f \\
\cdictx{f}{i}(\Delta) &= \frac{\Delta}{a_{(f,i)}+b_{(f,i)}\Delta} \tkom
\end{align*}
we have $a_{(f,i)} \in \{0,1\}$ for all $1 \leqslant i \leqslant n$.
\end{definition}

\begin{example} \label{ex:1:quadratic}
Consider the TRS $\RSa$ from Example~\ref{ex:rta08:1}. The assignment
of coefficient variables as defined in Example~\ref{ex:1:CP} induces
a $\Delta$-restricted interpretation. 
\end{example}

\begin{lemma} \label{l:quadratic}
 Let $\C$ denote a $\Delta$-restricted interpretation with coefficients
 $a_{(f,i)}$, $b_{(f,i)}$, $c_f$, $d_f$ ($f \in \FS$, $1 \leqslant i \leqslant \arity(f)$) and we set
 \begin{align*}
   M &\defsym \max (\{c_f,d_f \mid f \in \FS\} \cup \{1\}) \\
   N &\defsym \max (\{ b_{(f,i)} \mid f \in \FS, 1 \leqslant i \leqslant \arity(f) \} \cup \{1\}) \tpkt
 \end{align*}
Further let $t$ be a ground term, $\alpha$ a $\Delta$-assignment and let 
$\Delta > 0$. Then $\eval{t} \leqslant M ( \size{t} + N \size{t}^2 \Delta )$.
\end{lemma}
\begin{proof}
We proceed by induction on $t$. As $t \in \TA(\FS)$, 
the evaluation is independent of the assignment. Hence we write
$\groundeval{t}$ instead of $\eval{t}$. If $t = f\in \FS$, then 
\begin{equation*}
 \groundeval{t} = c_f \Delta + d_f \leqslant  M(\Delta + 1) \leqslant 
 M(\size{t} + N \size{t}^2 \Delta) \tpkt
\end{equation*} 
If on the other hand $t = f(\seq{t})$, then 
\begin{align}
\groundeval{t} &= \sum_{i} (a_{f_i} + b_{f_i}\Delta) [f^i_\C(\Delta)](t_i) + c_f \Delta + d_f 
\label{eq:quadratic:1}  \\ 
&\leqslant \sum_{i} (a_{f_i} + b_{f_i}\Delta) \bigl( M(\size{t_i} + N \size{t_i}^2 \frac{\Delta}{a_{f_i} + b_{f_i}\Delta} \bigr)
  + c_f\Delta + d_f 
\label{eq:quadratic:2}  \\
&=  \sum_{i} \bigl( (a_{f_i} + b_{f_i}\Delta) M \size{t_i} + MN \size{t_i}^2 \Delta \bigr) + c_f\Delta + d_f 
\label{eq:quadratic:3}  \\ 
  &\leqslant \sum_{i} \bigl( (1 + N\Delta) M \size{t_i} + MN \size{t_i}^2 \Delta \bigr) + M(\Delta + 1)
\label{eq:quadratic:4}  \\ 
  &\leqslant \sum_i \size{t_i} \bigl( (1 + N\Delta) M +  MN (\size{t}-1) \Delta \bigr) + M(\Delta + 1) 
\label{eq:quadratic:6}  \\ 
  & = (\size{t}-1) \bigl( (1 + N\Delta) M +  MN (\size{t}-1) \Delta \bigr) + M(\Delta + 1) 
\label{eq:quadratic:8}  \\ 
  & = M \bigl( (\size{t}-1)(1 + N\Delta) + N (\size{t}-1)^2 \Delta + (\Delta + 1) \bigr) 
\label{eq:quadratic:9} 
\\
 &\leqslant M (\size{t} + N \size{t}^2 \Delta) \tpkt
\label{eq:quadratic:10} 
\end{align}
In line~\eqref{eq:quadratic:2} we employ the induction hypothesis, in~\eqref{eq:quadratic:6} 
we use $\size{t_i} \leqslant \size{t} - 1$ and for~\eqref{eq:quadratic:10} 
a simple calculation reveals:
$(\size{t}-1)(1 + N\Delta) + N \Delta (\size{t}-1)^2 + (\Delta + 1) = \\
  \size{t} + N \size{t}^2 \Delta + \Delta - N \size{t} \Delta \leqslant
  \size{t} + N \size{t}^2 \Delta$.
\end{proof}

\begin{theorem} \label{t:quadratic}
Let $\RS$ be a TRS and let $\C$ denote a $\Delta$-restricted interpretation
compatible with $\RS$. Then $\RS$ is terminating and 
$\dc{\RS}(n) = \bO(n^2)$. Moreover there exists a TRS $\RS$ such
that $\dc{\RS}(n) = \Omega(n^2)$.
\end{theorem}
\begin{proof}
By Theorem~\ref{t:cdi} $\RS$ is terminating and by Lemma~\ref{l:quadratic}, there
exists $K \in \N$, such that for any ground term $t$:
$\groundeval{t} \leqslant K ( \size{t} + K \size{t}^2 \Delta) \leqslant K^2 \size{t}^2 (\Delta + 1)$
and hence
\begin{equation*}
  \dl{\RS}(t) \leqslant\inf_{\Delta > 0}\frac{\groundeval{t}}{\Delta} \leqslant
  \inf_{\Delta > 0}\frac{K^2 \size{t}^2 (\Delta + 1)}{\Delta} = K^2 \size{t}^2 \tpkt
\end{equation*}
We obtain $\dl{\RS}(t) = \bO(\size{t}^2)$ for any $t \in \TA(\FS,\VS)$
and thus $\dc{\RS}(n) = \bO(n^2)$. 
The tightness of the bound follows by Example~\ref{ex:rta08:1}.
\end{proof}

By definition the constant employed in Theorem~\ref{t:quadratic} depends only
on the employed interpretation functions. Moreover this dependence is linear.
In concluding this section, we want to stress that $\Delta$-restricted interpretation
are even strong enough to handle non-simple terminating TRSs. 
\begin{example}[\cite{Hofbauer:2001}]
Consider the TRS $\RS$ with the one rule $\ma(\ma(x))\rew \ma(\mb(\ma(x)))$.
By applying the algorithm described in Section~\ref{Automation}, we
find the below given $\Delta$-restricted interpretation $\C$ automatically:
\begin{equation*}
\cdifunc{\ma}{\Delta}(z)=2z\Delta +2 \qquad
\cdifunc{\mb}{\Delta}(z)=z\Delta \qquad
\cdictx{\ma}{1}(\Delta)=\frac{1}{2} \qquad
\cdictx{\mb}{1}(\Delta)=1 \tpkt
\end{equation*}
By Theorem~\ref{t:correctness}, $\C$ is compatible with $\RS$.
Hence Theorem~\ref{t:quadratic} implies that 
the derivational complexity of $\RS$ is (at most) quadratic.
\end{example}

\section{Experimental Results} \label{rta08:Experiments}

In this section we describe the programs \cdiA, \cdiB, and \cdiC. These
programs provide search procedures for context dependent interpretations. 
The program~\cdiA\ implements the heuristics of Hofbauer in~\cite{Hofbauer:2001},
mentioned in Section~\ref{Automation} above.
On the other hand, programs~\cdiB\ and~\cdiC\ implement the 
algorithm presented in Section~\ref{Automation} 
and incorporate constraint solvers for Diophantine (in)equalities.
The program \cdiA\ searches for $\Delta$-linear interpretations, while 
\cdiB\ and \cdiC\ can search for $\Delta$-linear and $\Delta$-restricted interpretations. 
We summarise further differences below:
\begin{description}
\item[\cdiA] Firstly, the program searches for a polynomial interpretation
  compatible with a TRS $\RS$.
  This interpretation is then lifted to a polynomial context dependent interpretation 
  $\C$ as follows: Coefficients of the form $k+1$ are replaced by $k + \Delta$. 
  Finally  \textsc{Mathematica}%
\footnote{\url{http://www.wolfram.com/products/mathematica/}.}
is invoked to verify that the resulting CDA $\C$ is 
$\Delta$-monotone and compatible with $\RS$. 
 
\item[\cdiB] This programs employs a \emph{constraint propagation procedure}
to solve the Diophantine constraints in $\PCC(\RS,\C) \cup \VE(\RS,\C)$. 
Essentially the implementation follows the technique suggested in~\cite{contejean:2005}. 

\item[\cdiC] The Diophantine (in)equalities in $\PCC(\RS,\C) \cup \VE(\RS,\C)$
  are translated into \emph{propositional logic} and suitable assignments are found by employing a SAT solver, in our case
Mini\-Sat%
\footnote{\url{http://minisat.se/}.}. 
The implementation follows ideas presented in~\cite{FGMSTZ07} and
employs the \texttt{plogic} library of \TTTT.%
\footnote{\url{http://colo6-c703.uibk.ac.at/ttt2/}.}
\end{description}
The implementation of the transformation steps as described in Section~\ref{Automation}, 
is the same for \cdiB\ and \cdiC.
The programs~\cdiA, \cdiB, and \cdiC\ are written in OCaml%
\footnote{\url{http://www.caml.inria.fr/}.}
 (and parts of \cdiA\ in C). 
All three programs are fairly small: \cdiA\ consists of about 2000
lines of code, while \cdiB\ and \cdiC\ use roughly 3000 lines of code each.%
In Table~\ref{tab:rta08:exp} we summarise the comparison between the different 
programs \cdiA, \cdiB, and \cdiC. 
The numbers in the third line of the table refer to the number of bits 
maximally used in \cdiC\ to encode coefficients. Correspondingly for \cdiB\ we used 
32 as strict bound on the coefficients.
We are interested in automatically verifying the complexity of terminating TRSs. 
Consequentially, as testbed we employ those 957 TRSs from the version 4.0 of the Termination Problem Data Base 
(TPDB for short) that can be shown terminating with at least one of the 
tools that participated in the termination competition 2007.%
\footnote{These 957 systems and full experimental evidence can be found 
at \url{http://cl-informatik.uibk.ac.at/~aschnabl/experiments/cdi/}.}
The presented tests were performed single-threaded on a 
2.40 GHz Intel\textregistered\ Core\texttrademark\ 2 Duo with 2 GB of memory.
For each system we used a timeout of 60 seconds, the times in the tables are given in 
milliseconds. 

\begin{table}
\caption{957 terminating TRSs}
\label{tab:rta08:exp}
\medskip
  \centering
  \begin{tabular}{l|@{\hspace{1ex}}c|cc|@{\hspace{1ex}}ccc@{\hspace{1ex}}|@{\hspace{1ex}}ccc}
    & \multicolumn{1}{c}{\cdiA} & \multicolumn{2}{c}{\cdiB} & \multicolumn{6}{c}{\cdiC} \\
    & $\Delta$-lin.\ & $\Delta$-restr.\ &  $\Delta$-lin.\ & 
    \multicolumn{3}{c}{$\Delta$-restricted} & \multicolumn{3}{c}{$\Delta$-linear} \\
    & & & & 3 & 4 & 5 & 3 & 4 & 5 \\
    \hline\hline
    \# success & 19 & 61 & 62 & 86 & 86 & 86 & 82 & 82 & 83\\
    average time & - & 3132 & 3652 & 4041 & 4008 & 5496 & 4981 & 5010 & 5527 \\
    \# timeout & - & 276 & 782 & 189 & 222 & 238 & 687 & 751 & 797 \\
  \end{tabular}
\end{table}

Observe that the heuristic proposed in~\cite{Hofbauer:2001} is not suitable
as an automatic procedure. (We have not indicated the time spent by \cdiA\ as the timing is
incomparable to the stand-alone approach of \cdiB\ or \cdiC.) With respect to
the comparison between \cdiB\ and \cdiC, the latter outperforms the former, if
at least 2 bits are used. Perhaps surprisingly the performance of
\cdiB\ and \cdiC\ on $\Delta$-restricted and $\Delta$-linear 
is almost identical.  
This can be explained by the strong impact of larger
bounds for the coefficients $a_{(f,i)}$ ($f \in \FS, 1 \leqslant i \leqslant \ar(f)$)
in the complexity of the issuing Diophantine (in)equalities.
However, for both programs \cdiB\ and \cdiC, 
the stronger technique gains one crucial system: Example~\ref{ex:rta08:2}.

Table~\ref{tab:comp} relates existing methods that induce
polynomial derivational complexities of TRSs to \cdiC. \SL\ refers to
\emph{strongly linear} interpretations, i.e., only interpretation
functions of the form $f_{\A}(x_1,\dots,x_n) = \sum_i x_i + c$, $c \in \N$ are allowed.  
Clearly compatibility with strongly linear interpretations induces \emph{linear}
derivational complexity. 
Secondly, \TTTBOX\ refers to the implementation of the match-bound technique as in~\cite{KM:2007}: Linear TRSs 
are tested for match-boundedness, non-linear, but non-duplicating TRSs are
tested for match-raise-boundedness. This technique
again implies \emph{linear} derivational complexity.
(Employing~\cite{HW:2004} (as in~\cite{GHWZ07}) one sees that any match-raise bounded TRS has linear derivational
complexity. Then the claim follows from Lemma~8 in~\cite{KM:2007}.)
Note that the restriction to non-duplicating TRS is harmless, as 
any duplicating TRS induces at least exponential derivational complexity.
No further termination methods that induce at most 
polynomial derivational complexities for TRSs 
have previously been known. In particular related work on implicit complexity 
(for example~\cite{BCMT:2001,Marion:2003,Marion:2006,AM:2008,BMM:2009:tcs}) 
does not provide methods that 
induce polynomial \emph{derivational complexities},
even if sometimes the derivation length can
be bounded polynomially, if the set of start terms
is suitably restricted.
Finally \CDI\ denotes our standard strategy: First, 
we search for a strongly linear interpretation. 
If such an interpretation cannot be found, then a $\Delta$-restricted interpretation
is sought (with $5$ bits as bound). 

\begin{table}
  \caption{Termination Methods as Complexity Analysers}
  \label{tab:comp}
  \medskip
  \centering
  \begin{tabular}{l|@{\hspace{1ex}}c@{\hspace{1ex}}|@{\hspace{1ex}}
      c@{\hspace{1ex}}|@{\hspace{1ex}}c@{\hspace{1ex}}|@{\hspace{1ex}}c}
    & \SL & \TTTBOX & \cdiC---$\Delta$-restricted & \CDI---$\Delta$-restricted\\
    \hline\hline
    \# success & 41 & 125 & 86 & 87 \\
     average time & 20 & 577 & 3986 & 3010 \\
     \# timeout & 0 & 225 & 238 & 237
  \end{tabular}
\end{table}

Some comments on the results reported in Table~\ref{tab:comp}: By definition the set of TRSs compatible with
a strongly linear interpretation is a (strict) subset of those treatable with \CDI. 
On the other hand the comparison between \TTTBOX\ and \CDI\ (or \cdiC) may appear
not very favourable for our approach. However, \CDI\ (and \cdiC) can handle
TRSs that cannot be handled by \TTTBOX.  
More precisely with respect to $\Delta$-restricted interpretations \CDI\ (and \cdiC)
can handle 38 (37) TRSs that cannot be handled with \TTTBOX. 
For instance the following example can only be handled with \CDI\ (and \cdiC).
\begin{example}
Consider the following rewrite system $\RSam$.
(This is Example 2.11 in Steinbach and Kühler's collection~\cite{SteinbachKuehler:1990}.) 
\begin{align*}
  \mN \m{+} y & \rew y & \mN \m{-} y & \rew y
  & \m{s}(x) \m{-} \m{s}(y) &\rew x \m{-} y\\
  \m{s}(x) \m{+} y & \rew \m{s}(x \m{+} y) & x \m{-} 0 & \rew x 
\end{align*}
It is easy to see that $\RSam$ is compatible with the following
(automatically generated) $\Delta$-restricted interpretation $\C$.
\begin{align*}
  \cdifunc{\m{-}}{\Delta}(x,y) &= x + y + 3y\Delta + 2\Delta &
  \cdifunc{\mN}{\Delta} &= 0\\
  \cdifunc{\m{+}}{\Delta}(x,y) &= x + y + x\Delta + \Delta &
  \cdifunc{\m{s}}{\Delta}(x) &= x + 2 \tkom
\end{align*}
with parameter functions: 
$\cdictx{\m{-}}{1}(\Delta) = \cdictx{\m{+}}{2}(\Delta) 
= \cdictx{\m{s}}{1}(\Delta) = \Delta$, 
$\cdictx{\m{-}}{2}(\Delta) = \frac{\Delta}{1 + 3\Delta}$, and
$\cdictx{\m{+}}{1}(\Delta) = \frac{\Delta}{1 + \Delta}$. Due to
Theorem~\ref{t:quadratic} we conclude \emph{quadratic} derivational
complexity, while the standard polynomial
interpretation would only allow to conclude an \emph{exponential} upper bound.
Note that the deduced quadratic derivational complexity provides an
optimal upper bound.
\end{example}
Another issue is the high average yes time (and the higher
number of timeouts) of \cdiC\ and \CDI\ in relation to existing techniques. 
Although a closer look reveals that the total times spent by \TTTBOX\ and
\CDI\ (or \cdiC) is relatively equal, an improvement of the efficiency
of the introduced tools seems worthwhile.

\begin{remark}
Note that \CDI\ in conjunction with 
\TTTBOX\ can automatically verify that 163 TRSs in the testbed 
are of at most \emph{quadratic} derivational complexity. 
Put differently more than 10\% of all 1381 TRSs (and more than a third of the 445 
non-duplicating TRSs) in version 4.0 of the TPDB are of quadratic derivational complexity.
\end{remark}

\section{Conclusion} \label{rta08:Conclusion}

In this paper we have presented 
two subclasses of context dependent interpretations, and
established tight upper bounds on the induced derivational complexities. 
More precisely, we have delineated two subclasses: 
$\Delta$-linear and $\Delta$-restricted context 
dependent interpretations that induce exponential and quadratic derivational
complexity, respectively.
Further, we introduced an algorithm for mechanically proving termination 
of rewrite systems with context dependent interpretations. As a consequence
we established a technique to automatically verify quadratic derivational
complexity of TRSs. 
Finally, we reported on different implementations of this algorithm and 
presented numerical data to compare these implementations with existing
methods that allow to automatically verify polynomial derivational
complexity of TRSs.

We believe the here presented approach can be extended further. 
A starting point for future work would be to decide whether it is possible 
to define additional subclasses of context dependent interpretations 
inducing polynomial derivational complexities that grow faster than quadratic. 
One possible approach is to drop the restriction to \emph{integer} coefficients
and thus generalise the notion of polynomial context dependent interpretations.
By Tarski's quantifier elimination method, such an extension turns the undecidable 
\emph{positivity problem} for Diophantine (in)equalities into a decidable 
problem. Further research will clarify the impact of this extension.
A crucial problem in practical considerations is the known 
ineffectivity of quantfier elimination, see for example~\cite{CJ:2004}.

}

\chapter{Automated Complexity Analysis Based on the Dependency Pair Method} 
\label{ijcar08a}

\subsection*{Publication Details}
N.~Hirokawa and G.~Moser.
\newblock Automated complexity analysis based on the dependency pair method.
\newblock In \emph{Proceedings of the 4th International Joint Conference on
  Automated Reasoning}, number 5195 in LNAI, pages 364--380. Springer Verlag,
  2008{\natexlab{a}}.%
\footnote{This research was partially supported by FWF (Austrian Science Fund) project P20133.}

\medskip
I am indebted to Dieter Hofbauer, who spotted an unfortunate mistake in 
Lemma~22 of the published version. This mistake has been rectified below. 
The mistake influenced the given experimental evidence and I would
like to thank Andreas Schnabl and Martin Avanzini for providing me 
with adjusted experimental data.

\subsection*{Ranking}
The International Joint Conference on Automated Reasoning has been ranked
\textbf{A+} by CORE in 2007.

\subsection*{Abstract}
In this paper, we present a variant of the dependency pair method for
analysing runtime complexities of term rewrite systems automatically.
This method is easy to implement, but significantly extends
the analytic power of existing direct methods.
Our findings extend the class of TRSs
whose linear or quadratic runtime complexity can be detected automatically.
We provide ample numerical data for assessing the viability of the method.

{
\input{ijcar08a.sty}

\section{Introduction} \label{ijcar08a:Intro}

Term rewriting is a conceptually simple but powerful abstract model of
computation that underlies much of declarative programming. 
In order to assess the complexity of a (terminating) term rewrite system (TRS for short) 
it is natural to look at the maximal length of derivation sequences, as
suggested by Hofbauer and Lautemann in~\cite{HofbauerLautemann:1989}. More precisely, 
the \emph{derivational complexity function} 
with respect to a (terminating and finitely-branching) TRS $\RS$ 
relates the length of the longest derivation sequence to the size of the initial term. 
For direct termination techniques it is often possible to 
establish upper-bounds on the
growth rate of the derivational complexity function 
from the termination proof of $\RS$, see for 
example~\cite{HofbauerLautemann:1989,Hofbauer:1992,Weiermann:1995,Hofbauer:2001,Moser:2006lpar,GHWZ07}.

However, if one is interested in methods that induce feasible (i.e., polynomial)
complexity, the existing body of research is not directly applicable. 
On one hand this is due to the fact that for standard techniques 
the derivational complexity cannot be contained by polynomial growth rates. 
(See~\cite{GHWZ07} for the exception to the rule.)
Already termination proofs by polynomial interpretations induce a double-exponential
upper-bound on the derivational complexity, cf.~\cite{HofbauerLautemann:1989}. 
On the other hand this is---to some extent---the consequence of the \emph{definition} 
of derivational complexity as this measure
does not discriminate between different types of initial terms, while in modelling
declarative programs the type of the initial term is usually quite restrictive. 
The following example clarifies the situation.
\begin{example} \label{ex:ijcar08a:1}
%
Consider the TRS $\RR$
\begin{alignat*}{4}
1\colon &\;& \isempty(\nil) &\rew {\top} & \hspace{5ex}
5\colon && \app(x,y) &\rew \ifapp(x,y,x) \\
2\colon && \isempty({x} \cons {y}) &\rew {\bottom} &
6\colon && \ifapp(x,y,\nil) & \rew y\\
3\colon && \hd({x} \cons {y}) &\rew x  &
7\colon && \ifapp(x,y,{u} \cons {v}) & \rew {u} \cons {\app(v,y)}\\
4\colon &\;& \tl({x} \cons {y}) &\rew y
\end{alignat*}
Although the functions \emph{computed} by $\RS$ are obviously feasible this
is not reflected in the derivational complexity of $\RS$. Consider rule 5, which
we abbreviate as $C[x] \to D[x,x]$. 
Since the maximal derivation length starting with $C^n[x]$ equals $2^{n-1}$ for all $n > 0$,
$\RS$ admits (at least) exponential derivational complexity.
\end{example}
After a moment one sees that this behaviour is forced upon us, 
as the TRS $\RS$ may duplicate variables, i.e., $\RS$ is
\emph{duplicating}. Furthermore, in general 
the applicability of the above results is typically limited to simple termination.
(But see~\cite{Hofbauer:2001,Moser:2006lpar,GHWZ07} for exceptions to this rule.)
To overcome the first mentioned restriction we propose to 
study \emph{runtime complexities} of rewrite systems.
The \emph{runtime complexity function} with respect to a TRS $\RS$ 
relates the length of the longest derivation sequence to the size
of the arguments of the initial term, where the arguments are supposed to be in normal form.
In order to overcome the second restriction, we base our study
on a fresh analysis of the \emph{dependency pair method}. 
The dependency pair method~\cite{ArtsGiesl:2000}
is a powerful (and easily automatable) method for proving termination of term rewrite systems.
In contrast to the above cited direct termination methods, 
this technique is a \emph{transformation} technique, allowing for 
applicability beyond simple termination.

Studying (runtime) complexities induced by the dependency pair method is challenging. 
Below we give an (easy) example showing that 
the direct translations of original theorems formulated 
in the context of termination analysis is destined to failure in the context of 
runtime complexity analysis. 
If one recalls that the dependency pair method is based on 
the observation that from an arbitrary non-terminating term 
one can extract a minimal non-terminating subterm, this is not surprising.
Through a very careful investigation of the original formulation of
the dependency pair method (see~\cite{ArtsGiesl:2000,GAO:2002}, but also \cite{HirokawaMiddeldorp:2007}), we 
establish a runtime complexity analysis based on the dependency pair method.
In doing so, we introduce \emph{weak dependency pairs} and \emph{weak innermost dependency pairs}
as a general adaption of dependency pairs to (innermost) runtime complexity analysis. 
Here the \emph{innermost} runtime complexity function with respect to a TRS $\RS$ 
relates the length of the longest innermost derivation sequence to the size
of the arguments of the initial term, where again the arguments are supposed to be in normal form.

Our main result shows how natural improvements of the dependency pair method, 
like 
\emph{usable rules}, 
\emph{reduction pairs}, and 
\emph{argument filterings} become applicable in this context.
Moreover, for innermost rewriting, 
we establish an easy criterion to decide when 
weak innermost dependency pairs can be replaced 
by ``standard'' dependency pairs without introducing fallacies. 
Thus we establish (for the first time) a method to analyse
the derivation length induced by the (standard) dependency pair method for innermost rewriting.
We have implemented the technique and experimental evidence shows that the
use of weak dependency pairs significantly increases the applicability of the body
of existing results on the estimation of derivation length via
termination techniques. In particular, our findings extend the class of TRSs
whose linear or quadratic runtime complexity can be detected automatically.

The remainder of this paper is organised as follows.
In the next section we recall basic notions and starting points of
this paper. Sections~\ref{dependency pairs} and~\ref{usable} 
introduce weak dependency pairs and discuss the employability of the
usable rule criterion. 
In Section~\ref{semantical gap} we show how to estimate runtime 
complexities through relative rewriting and in 
Section~\ref{Refinements} we state our Main Theorem.
The presented technique
has been implemented and we provide ample numerical data for assessing 
the viability of the method. This evidence can be found in Section~\ref{ijcar08a:Experiments}.
Finally in Section~\ref{ijcar08a:Conclusion} we conclude and mention possible 
future work.

\section{Preliminaries}

We assume familiarity with term rewriting~\cite{BaaderNipkow:1998,Terese} but briefly
review basic concepts and notations.
Let $\VS$ denote a countably infinite set of variables and $\FS$ a 
signature. The set of terms over $\FS$ and $\VS$ is denoted by 
$\TERMS$. The \emph{root symbol} of a term $t$ is either $t$ itself, if
$t \in \VS$, or the symbol $f$, if $t = f(\seq{t})$. 
The \emph{set of position} $\Pos(t)$ of a term $t$ is defined as
usual. We write $\Pos_{\GG}(t) \subseteq \Pos(t)$ for the set of
positions of subterms, whose root symbol is contained in $\GG \subseteq \FS$.
The subterm relation is denoted as $\subterm$.
$\Var(t)$ denotes the set of variables occurring in a term $t$ 
and the \emph{size} $\size{t}$ of a term is defined 
as the number of symbols in $t$. 

A \emph{term rewrite system} (\emph{TRS} for short) $\RS$ over
$\TERMS$ is a \emph{finite} set of rewrite
rules $l \to r$, such that $l \notin \VS$ and $\Var(l) \supseteq \Var(r)$.
The smallest rewrite relation that contains $\RS$ is denoted by
$\to_{\RS}$. 
The transitive closure of $\to_{\RS}$ is denoted by $\rstrew{\RS}$, and
its transitive and reflexive closure by $\rssrew{\RS}$.
We simply write $\to$ for $\to_{\RS}$ if $\RS$ is clear from context.
A term $s \in \TERMS$ is called a \emph{normal form} if there is no
$t \in \TERMS$ such that $s \to t$. 
With $\NF(\RR)$ we denote the set of all normal forms of a term rewrite
system $\RR$.
The \emph{innermost rewrite relation} $\irew{\RR}$
of a TRS $\RR$ is defined on terms as follows: $s \irew{\RR} t$ if 
there exist a rewrite rule $l \to r \in \RR$, a context $C$, and
a substitution $\sigma$ such that $s = C[l\sigma]$, $t = C[r\sigma]$,
and all proper subterms of $l\sigma$ are normal forms of $\RR$.  
The set of defined function symbols is denoted as $\DS$, while the constructor
symbols are collected in $\CS$.
We call a term $t = f(\seq{t})$ \emph{basic} if $f \in \DS$ and $t_i \in \TA(\CS,\VS)$
for all $1 \leqslant i \leqslant n$. A TRS $\RS$ is called \emph{duplicating} if 
there exists a rule $l \to r \in \RS$ such that a variable occurs more often in $r$ than in $l$. 
We call a TRS \emph{terminating} if no infinite rewrite sequence
exists. Let $s$ and $t$ be terms. If exactly $n$ steps are performed to rewrite $s$
to $t$ we write $s \to^n t$.
The \emph{derivation length} of a terminating term $t$ with respect to a 
TRS $\RS$ and rewrite relation $\rsrew{\RS}$ is defined as: 
$\dheight(s,\rsrew{\RS}) = \max\{ n \mid \exists t \; s \to^n t \}$. 
Let $\RR$ be a TRS and $T$ be a set of terms. 
The \emph{runtime complexity function with respect to a relation $\to$ on $T$}
is defined as follows:
\begin{equation*}
\Rc{}(n, T, \rew) = \max\{ \dheight(t, \rew) \mid 
\text{$t \in T$ and $\size{t} \leqslant n$}\} \tpkt
\end{equation*}
In particular we are interested in the 
(innermost) runtime complexity with respect to $\rsrew{\RS}$ ($\irew{\RS}$)
on the set $\TB$ of all \emph{basic} terms, as defined above.%
\footnote{We can replace $\TB$ by the set of terms $f(\seq{t})$ with $f\in \DS$,
whose arguments $t_i$ are in normal form, while keeping all results in this paper.
\label{FN}} 
More precisely, the \emph{runtime complexity function} 
(with respect to $\RS$) is defined as 
$\Rc{\RS}(n) \defsym \Rc{}(n, {\TB}, \rsrew{\RS})$ and we 
define the \emph{innermost runtime complexity function} 
as $\Rci{\RS}(n) \defsym \Rc{}(n, {\TB}, \irew{\RS})$.
Finally, the \emph{derivational complexity function} (with respect to $\RS$) 
becomes definable as follows: $\Dc{\RS}(n) = \Rc{}(n, \TA, \rsrew{\RS})$, where
$\TA$ denotes the set of all terms $\TA(\FS,\VS)$. 
We sometimes say the (innermost) runtime complexity of $\RR$ is \emph{linear},
\emph{quadratic}, or \emph{polynomial} if $\rc_\RR^{(\m{i})}(n)$ is bounded 
linearly, quadratically, or polynomially in $n$, respectively.
Note that the derivational complexity and the runtime complexity of a 
TRS $\RS$ may be quite different: In general
it is not possible to bound $\Dc{\RS}$ polynomially in $\Rc{\RS}$, 
as witnessed by Example~\ref{ex:ijcar08a:1} and the observation that
the runtime complexity of $\RS$ is linear 
(see~Example~\ref{ex:ijcar08a:1:final}, below).

A \emph{proper order} is a transitive and irreflexive relation and
a \emph{preorder} is a transitive and reflexive relation. 
A proper order $\succ$ is \emph{well-founded} if there is 
no infinite decreasing sequence $t_1 \succ t_2 \succ t_3 \cdots$.
A well-founded proper order that also is a rewrite relation is
called a \emph{reduction order}. 
We say a reduction order $\succ$ and a TRS $\RS$ are \emph{compatible}
if $\RS \subseteq {\succ}$. It is well-known that a TRS is terminating
if and only if there exists a compatible reduction order.
An $\FS$-\emph{algebra} $\A$ consists of a carrier set $A$ and a collection
of interpretations $f_\A$ for each function symbol in $\FS$. 
A \emph{well-founded} and \emph{monotone} algebra (\emph{WMA} for short) 
is a pair $(\A,>)$, where $\A$ is an algebra and $>$ is a well-founded proper order on $A$ 
such that every $f_\A$ is monotone in all arguments.
An \emph{assignment} $\alpha \colon \VS \to A$ is a function mapping
variables to elements in the carrier.
A WMA naturally induces a proper order $\gord{\A}$ on terms: 
$s \gord{\A} t$ if $\eval{\alpha}{\A}(s) > \eval{\alpha}{\A}(t)$
for all assignments $\alpha \colon \VS \to A$.

\section{The Dependency Pair Method}
\label{dependency pairs}

The purpose of this section is to take a fresh look at the dependency
pair method from the point of complexity analysis.
Familiarity with~\cite{ArtsGiesl:2000,HirokawaMiddeldorp:2007} will be helpful. 
The dependency pair method for termination analysis is based on 
the observation that from an arbitrary non-terminating term
one can extract a minimal non-terminating subterm.
For complexity analysis we employ a similar observation:
From a given term $t$ one can extract a list of subterms 
whose sum of the derivation lengths is equal to the derivational 
length of $t$.

Let $X$ be a set of symbols.
We write $\SC{t_1, \ldots, t_n}_X$ to denote $C[t_1,\ldots,t_n]$,
whenever $\rt(t_i) \in X$ for all $1 \leqslant i \leqslant n$
and $C$ is an $n$-hole context containing no $X$-symbols.  
(Note that the context $C$ may be degenerate and does not contain a hole $\ctx$
or it may be that $C$ is a hole.)
Then, every term $t$ can be uniquely written in
the form $\SC{t_1, \ldots, t_n}_X$.

\begin{lemma}
\label{l:sum:1}
Let $t$ be a terminating term, and let $\sigma$ be a substitution.
Then we have $\dheight(t\sigma, \to_\RS) = \sum_{1 \leqslant i \leqslant n} \dheight(t_i\sigma, \rsrew{\RS})$,
whenever $t = \SC{t_1, \ldots, t_n}_{\DD \cup \VV}$.
\end{lemma}
We define the function $\COM$ as a mapping from tuples of terms to terms 
as follows: $\COM(\seq{t})$ is $t_1$ if $n = 1$, and
$c(t_1,\ldots,t_n)$ otherwise.
Here $c$ is a fresh $n$-ary function symbol called \emph{compound symbol}.
The above lemma motivates the next definition of 
\emph{weak dependency pairs}. 

\begin{definition}
\label{d:WDP}
Let $t$ be a term. We set $t^\sharp \defsym t$ if $t \in \VV$, and 
$t^\sharp \defsym f^\sharp(t_1,\dots,t_n)$ if $t = f(\seq{t})$.
Here $f^\sharp$ is a new $n$-ary function symbol called 
\emph{dependency pair symbol}. For a signature
$\FS$, we define $\FS^\sharp = \FS \cup \{f^\sharp \mid f\in \FS\}$.
Let $\RR$ be a TRS. 
If $l \rew r \in \RS$ and $r = \SC{\seq{u}}_{\DD \cup \VV}$ then 
the rewrite rule $l^\sharp \to \COM(u_1^\sharp,\ldots,u_n^\sharp)$ is called
a \emph{weak dependency pair} of $\RS$.
The set of all weak dependency pairs is denoted by $\WDP(\RS)$.
\end{definition}
%
While dependency pair symbols are defined with respect to $\WDP(\RS)$,
these symbols are not defined with respect to the original system $\RS$. In the
sequel defined symbols, refer to the defined function symbols of $\RS$.

\begin{example}[continued from Example~\ref{ex:ijcar08a:1}] \label{ex:ijcar08a:1:WDP}
The set $\WDP(\RR)$ consists of the next seven weak dependency pairs:
\begin{alignat*}{4}
5\colon & \; & \isempty^\sharp(\nil) & \rew \m{c} & \hspace{5ex}
9\colon & \; & \app^{\sharp}(x,y) &\rew \ifapp^{\sharp}(x,y,x)\\
6\colon && \isempty^\sharp({x} \cons {y}) & \rew \m{d} &
10\colon && \ifapp^\sharp(x,y,\nil) &\rew y\\ 
7\colon && \hd^\sharp({x} \cons {y}) & \rew x &
11\colon && \ifapp^\sharp(x,y,{u} \cons {v}) &\rew \m{e}(u,\app^\sharp(v,y))\\
8\colon && \tl^\sharp({x} \cons {y}) & \rew y \tpkt
\end{alignat*}
\end{example}
                
\begin{lemma}
\label{l:dp:1}
Let $t \in \TA(\FS,\VS)$ be a terminating term with $\rt(t) \in \DD$. 
We have
$\dheight(t, \rsrew{\RS}) = \dheight(t^\sharp, \rsrew{\WDP(\RR) \cup \RR})$.
\end{lemma}
\begin{proof}
We show 
$\dheight(t, \rsrew{\RS}) \leqslant
 \dheight(t^\sharp, \rsrew{\WDP(\RR) \cup \RR})$ by 
induction on $\ell = \dheight(t, \rsrew{\RS})$.
If $\ell = 0$, the inequality is trivial.
Suppose $\ell > 0$.  Then there exists a term $u$ such that
$t \rsrew{\RS} u$ and $\dheight(u, \rsrew{\RS}) = \ell - 1$.
We distinguish two cases depending on the rewrite position $p$.
\begin{itemize}
\item
If $p$ is a position below the root, then clearly $\rt(u) = \rt(t) \in \DD$ and
$t^\sharp \rsrew{\RS} u^\sharp$.  The induction hypothesis yields
$\dheight(u, \rsrew{\RS}) \leqslant
 \dheight(u^\sharp, \rsrew{\WDP(\RS) \cup \RS})$,
and we obtain
$\ell \leqslant \dheight(t^\sharp, \rsrew{\WDP(\RS) \cup \RS})$.
\item
If $p$ is a root position, then there exist a rewrite rule $l \to r \in \RS$ and
a substitution $\sigma$ such that $t = l\sigma$ and $u = r\sigma$.
We have $r = \SC{\seq{u}}_{\DD \cup \VV}$ and thus by definition
$l^\sharp \to \COM(u_1^\sharp,\ldots,u_n^\sharp) \in \WDP(\RR)$ such that
$t^\sharp = l^\sharp\sigma$. Now, either $u_i \in \VV$ or $\rt(u_i) \in \DD$ 
for every $1 \leqslant i \leqslant n$.
Suppose $u_i \in \VV$. Then $u_i^\sharp\sigma = u_i\sigma$ and clearly no
dependency pair symbol can occur and thus,
\begin{equation*}
\dheight(u_i\sigma, \rsrew{\RS})
= \dheight(u_i^\sharp\sigma, \rsrew{\RS})
= \dheight((u_i\sigma)^\sharp, \rsrew{\WDP(\RS) \cup \RS}) \tpkt
\end{equation*}
Otherwise, if $\rt(u_i) \in \DD$ then $u_i^\sharp\sigma = (u_i\sigma)^\sharp$.
Hence $\dheight(u_i\sigma, \rsrew{\RS}) \leqslant \dheight(u, \rsrew{\RS}) < l$,
and we conclude
$\dheight(u_i\sigma, \rsrew{\RS}) \leqslant
 \dheight(u_i^\sharp\sigma, \rsrew{\WDP(\RS) \cup \RS})$
from the induction hypothesis. Therefore,
\begin{align*}
\ell 
& = \dheight(u, \rsrew{\RS}) + 1
  = \sum_{1 \leqslant i \leqslant n} \dheight(u_i\sigma, \rsrew{\RR}) + 1
  \leqslant \sum_{1 \leqslant i \leqslant n} \dheight(u_i^\sharp\sigma, \rsrew{\WDP(\RS) \cup \RR}) + 1\\
& \leqslant \dheight(\COM(u_1^\sharp,\ldots,u_n^\sharp)\sigma, 
                     \rsrew{\WDP(\RS) \cup \RR}) + 1
  = \dheight(t^\sharp, \rsrew{\WDP(\RS) \cup \RR}) \tpkt
\end{align*}
Here we used Lemma~\ref{l:sum:1} for the second equality.
\end{itemize}
Note that $t$ is $\RS$-reducible if and only if $t^\sharp$ is
$\WDP(\RS)\cup\RS$-reducible. Hence as $t$ is terminating, 
$t^\sharp$ is terminating on $\rsrew{\WDP(\RS)\cup\RS}$. Thus, similarly, 
$\dheight(t, \rsrew{\RS}) \geqslant
 \dheight(t^\sharp, \rsrew{\WDP(\RR) \cup \RR})$ 
is shown by induction on $\dheight(t^{\sharp}, \rsrew{\WDP(\RS)\cup\RS})$.
\end{proof}

\begin{lemma}
\label{l:sum:2}
Let $t$ be a terminating term and $\sigma$ a substitution such that $x\sigma$ is a normal form of $\RS$ 
for all $x \in \Var(t)$.
Then $\dheight(t\sigma, \rsrew{\RS}) =  \sum_{1 \leqslant i \leqslant n} \dheight(t_i\sigma, \rsrew{\RS})$,
whenever $t = \SC{t_1, \ldots, t_n}_\DD$.
\end{lemma}
\begin{definition}
Let $\RR$ be a TRS. 
If $l \rew r \in \RS$ and $r = \SC{\seq{u}}_\DD$ then 
the rewrite rule $l^\sharp \to \COM(u_1^\sharp,\ldots,u_n^\sharp)$
is called a \emph{weak innermost dependency pair} of $\RS$.
The set of all weak innermost dependency pairs is denoted by $\WIDP(\RS)$.
\end{definition}

\begin{example}[continued from Example~\ref{ex:ijcar08a:1}] \label{ex:ijcar08a:1:WIDP}
The set $\WIDP(\RR)$ consists of the following seven weak dependency pairs
(with respect to $\ito$):
\begin{alignat*}{4}
&\;& \isempty^\sharp(\nil) & \rew \m{c} & 
&\;& \app^{\sharp}(x,y) &\rew \ifapp^\sharp(x,y,x)\\
&& \isempty^\sharp({x} \cons {y}) & \rew \m{d} & \hspace{5ex}
&& \ifapp^\sharp(x,y,\nil) &\rew \m{g}\\
&& \hd^\sharp({x} \cons {y}) &\rew \m{e} & 
&& \ifapp^\sharp(x,y,{u} \cons {v}) &\rew \app^\sharp(v,y) \\
&& \tl^\sharp({x} \cons {y}) &\rew \m{f} \tpkt
\end{alignat*}
\end{example}

The next lemma adapts Lemma~\ref{l:dp:1} 
to innermost rewriting.
\begin{lemma}
\label{l:dp:2}
Let $t$ be an innermost terminating term in $\TA(\FS,\VS)$ with $\rt(t) \in \DD$.
We have $\dheight(t, \irew{\RS}) = \dheight(t^\sharp, \irew{\WIDP(\RR) \cup \RR})$.
\end{lemma}

We conclude this section by discussing the applicability of 
standard dependency pairs (\cite{ArtsGiesl:2000}) in complexity analysis.
For that we recall the standard definition of dependency pairs.

\begin{definition}[\cite{ArtsGiesl:2000}]
\label{d:DP}
The set $\DP(\RS)$ of (standard) \emph{dependency pairs} of a TRS $\RS$ 
is defined as 
$\{ l^\sharp \to u^{\sharp} \mid l \to r \in \RS, 
u \subterm r, \rt(u) \in \DD \}$.
\end{definition}

The following example shows that Lemma~\ref{l:dp:1} (Lemma~\ref{l:dp:2})
does not hold if we replace weak (innermost) dependency pairs with 
standard dependency pairs.
\begin{example} 
Consider the one-rule TRS $\RR$:
$\m{f}(\m{s}(x)) \to \m{g}(\m{f}(x), \m{f}(x))$.
$\DP(\RR)$ is the singleton of $\m{f}^\sharp(\m{s}(x)) \to \m{f}^\sharp(x)$.
Let $t_n = \m{f}(\m{s}^n(x))$ for each $n \geqslant 0$.
Since $t_{n+1} \rsrew{\RS} \m{g}(t_n,t_n)$ holds for
all $n \geqslant 0$, it is easy to see 
$\dheight(t_{n+1}, \rsrew{\RS}) \geqslant 2^n$, while
$\dheight(t_{n+1}^\sharp, \rsrew{\DP(\RR) \cup \RS}) = n$.
\end{example}

Hence, in general we cannot replace weak dependency pairs with (standard) dependency
pairs. However, if we restrict our attention to innermost rewriting,
we can employ dependency pairs in complexity analysis without introducing fallacies,
when specific conditions are met.
\begin{lemma} \label{l:dp:4}
Let $t$ be an innermost terminating term with $\rt(t) \in \DD$.
If all compound symbols in $\WIDP(\RS)$ are nullary, 
$\dheight(t, \irew{\RS}) \leqslant
 \dheight(t^\sharp,\irew{\DP(\RS) \cup \RS}) + 1$ holds.  
\end{lemma}

\begin{example}[continued from Example~\ref{ex:ijcar08a:1:WIDP}] \label{ex:ijcar08a:1:DP}
The occurring compound symbols are nullary. 
$\DP(\RR)$ consists of the following two dependency pairs:
\begin{alignat*}{4}
\app^\sharp(x,y) &\rew \ifapp^\sharp(x,y,x) & \hspace{5ex}
\ifapp^\sharp(x,y,{u} \cons {v}) & \rew \app^\sharp(v,y) \tpkt
\end{alignat*}
\end{example}

\section{Usable Rules} \label{usable}

In the previous section, we studied the dependency pair method in the light
of complexity analysis. Let $\RS$ be a TRS and $\PP$ a set of 
weak dependency pairs, weak innermost dependency pairs, or
standard dependency pairs of $\RS$. 
Lemmata~\ref{l:dp:1},~\ref{l:dp:2}, and~\ref{l:dp:4} describe
a strong connection between the length of derivations 
in the original TRSs $\RS$ and the transformed
(and extended) system $\PP \cup \RS$.
In this section we show how we can simplify the new TRS $\PP \cup \RS$
by employing \emph{usable rules}.

\begin{definition}
We write $f \depends g$ if there exists a rewrite rule
$l \to r \in \RR$ such that $f = \rt(l)$ and $g$ is a defined function
symbol in $\Fun(r)$. 
For a set $\GG$ of defined function symbols we
denote by $\RR{\restriction}\GG$ the set of
rewrite rules $l \to r \in \RR$ with $\rt(l) \in \GG$. 
The set $\UU(t)$ of usable rules of a term $t$ is defined as
$\RR{\restriction}\{ g \mid \text{$f \depends^* g$ for some $f \in \Fun(t)$} \}$.
Finally, if $\PP$ is a set of (weak) dependency pairs
then $\UU(\PP) = \bigcup_{l \to r \in \PP} \UU(r)$.
\end{definition}

\begin{example}[continued from Examples~\ref{ex:ijcar08a:1:WDP} and~\ref{ex:ijcar08a:1:WIDP}]
The sets of usable rules are empty (and thus equal) for the weak dependency pairs and
for the weak innermost dependency pairs, 
i.e., we have $\UU(\WDP(\RR)) = \UU(\WIDP(\RR)) = \varnothing$.
\end{example}

The usable rule criterion in termination analysis (cf.~\cite{GTSF06,HirokawaMiddeldorp:2007})
asserts that a non-terminating rewrite sequence of $\RR \cup \DP(\RR)$
can be transformed into a non-terminating rewrite sequence of 
$\UU(\DP(\RR)) \cup \DP(\RR) \cup \{ \m{g}(x,y) \to x, \m{g}(x,y) \to y \}$,
where $\m{g}$ is a fresh function symbol.
Because $\UU(\DP(\RR))$ is a (small) subset of $\RR$ and
most termination methods can handle $\m{g}(x,y) \to x$ and
$\m{g}(x,y) \to y$ easily, the termination  analysis often becomes easy
by switching the target of analysis from the former TRS to the latter TRS.
Unfortunately the transformation used in~\cite{GTSF06,HirokawaMiddeldorp:2007} increases
the size of starting terms, therefore we cannot adopt this transformation
approach. Note, however that the usable rule criteria for innermost 
termination~\cite{GAO:2002} can be directly applied in the context of complexity analysis.
Nevertheless, one may show a new type of usable rule criterion
by exploiting the basic property of a starting term.  
Recall that $\TB$ denotes the set of basic terms;
we set $\TBS = \{t^{\sharp} \mid t \in \TB \}$.

\begin{lemma} \label{l:usable}
Let $\PP$ be a set of (weak) dependency pairs and let 
$(t_i)_{i = 0, 1, \ldots}$ be a (finite or infinite) derivation of
$\RR \cup \PP$. If $t_0 \in \TBS$ then $(t_i)_{i = 0, 1, \ldots}$ is 
a derivation of $\UU(\PP) \cup \PP$.
\end{lemma}
\begin{proof}
Let $\GG$ be the set of all non-usable symbols with respect to $\PP$.
We write $P(t)$ if ${\atpos{t}{q}} \in {\NF(\RS)}$ for all $q \in \Pos_\GG(t)$.
Since $t_i \rsrew{\UU(\PP) \cup \PP} t_{i+1}$ holds
whenever $P(t_i)$ and $t_i \rsrew{\RS \cup \PP} t_{i+1}$, 
it is sufficient to show $P(t_i)$ for all $i$.
We perform induction on $i$.
\begin{enumerate}
\item 
Assume $i = 0$.  Since $t_0 \in \TBS$, we have $t_0 \in \NF(\RS)$
and thus ${\atpos{t}{p}} \in {\NF(\RS)}$ for all positions $p$. 
The assertion $P$ follows trivially.
\item
Suppose $i > 0$.  
By induction hypothesis, there exist $l \rew r \in \UU(\PP) \cup \PP$,
$p \in \Pos(t_{i-1})$, and a substitution $\sigma$ such that
${\atpos{t_{i-1}}{p}} = l\sigma$ and $\atpos{t_i}{p} = r\sigma$.
In order to show property $P$ for $t_i$, we fix a position $q \in \GG$. We
have to show $\atpos{t_i}{q} \in \NF(\RS)$. We distinguish three cases:
\begin{itemize}
\item
Suppose that $q$ is above $p$. Then $\atpos{t_{i-1}}{q}$ is 
reducible, but this contradicts the induction hypothesis $P(t_{i-1})$.
\item
Suppose $p$ and $q$ are parallel but distinct.
Since $\atpos{t_{i-1}}{q} = \atpos{t_i}{q} \in \NF(\RS)$ holds,
we obtain $P(t_i)$.
\item
Otherwise, $q$ is below $p$. Then, $\atpos{t_i}{q}$ is a subterm of $r\sigma$.  
Because $r$ contains no $\GG$-symbols by the definition of usable symbols,
$\atpos{t_i}{q}$ is a subterm of $x\sigma$ for some 
$x \in \Var(r) \subseteq \Var(l)$.
Therefore, $\atpos{t_i}{q}$ is also a subterm of $t_{i-1}$, 
from which $\atpos{t_i}{q} \in \NF(\RS)$ follows.  We obtain $P(t_i)$.
\end{itemize}
\end{enumerate}
\end{proof}

The following theorem follows from Lemmata~\ref{l:dp:1},~\ref{l:dp:2}, 
and~\ref{l:dp:4} in conjunction with the above Lemma~\ref{l:usable}.
It adapts the usable rule criteria to complexity analysis.%
\footnote{Note that Theorem~\ref{t:dp:usable} only holds 
for \emph{basic} terms $t \in \TBS$. 
In order to show this, we need some additional technical lemmas, which are the
subject of the next section.}
\begin{theorem} \label{t:dp:usable}
Let $\RR$ be a TRS and let $t \in \TB$. 
If $t$ is terminating with respect to $\rew$ then
$\dheight(t, \rew) \leqslant \dheight(t^{\sharp},\rsrew{\UU(\PP)\,\cup \, \PP})$, 
where $\rew$ denotes $\rsrew{\RS}$ or $\irew{\RS}$
depending on whether $\PP = \WDP(\RS)$ or $\PP = \WIDP(\RS)$.
Moreover, suppose all compound symbols in $\WIDP(\RS)$ are nullary then
$\dheight(t, \irew{\RS}) \leqslant 
 \dheight(t^{\sharp},\rsrew{\UU(\DP(\RS))\,\cup \, \DP(\RS)}) + 1$.
\end{theorem}
It is worth stressing that it is (often) easier to analyse the 
complexity of $\UU(\PP)\, \cup \, \PP$  than the complexity of $\RS$. 
To clarify the applicability  of the theorem in complexity analysis, we consider 
two restrictive classes of polynomial interpretations, 
whose definitions are motivated by~\cite{BCMT:2001}.

A polynomial $P(x_1,\dots,x_n)$ (over the natural numbers) is called
\emph{strongly linear} if $P(x_1,\dots,x_n) = x_1 + \cdots + x_n + c$
where $c \in \NN$. A polynomial interpretation is called \emph{linear restricted} 
if all constructor symbols are interpreted by strongly linear polynomials
and all other function symbols by a linear polynomial (monotone with respect to the
standard order $>$ on $\NN$). If on the other
hand the non-constructor symbols are interpreted by quadratic polynomials,
the polynomial interpretation is called \emph{quadratic restricted}.
Here a polynomial is \emph{quadratic} if it is a sum of monomials of degree at most $2$.
It is easy to see that if a TRS $\RS$ is compatible with a 
linear or quadratic restricted interpretation, 
the runtime complexity of $\RS$ is linear or quadratic, 
respectively (see also~\cite{BCMT:2001}).

\begin{corollary} \label{c:dp:usable}
Let $\RS$ be a TRS and let $\PP = \WDP(\RS)$ or $\PP = \WIDP(\RS)$. 
If $\UU(\PP) \,\cup \, \PP$ is compatible with a linear 
or quadratic restricted interpretation, 
the (innermost) runtime complexity function $\rc^{(\m{i})}_{\RS}$ 
with respect to $\RS$ is linear or
quadratic, respectively. 
Moreover, suppose all compound symbols in $\WIDP(\RS)$ are nullary and
$\UU(\DP(\RS))\,\cup \, \DP(\RS)$ is compatible with a 
linear (quadratic) restricted interpretation, then $\RS$ 
admits at most linear (quadratic) innermost runtime complexity.
\end{corollary}
\begin{proof}
Let $\RS$ be a TRS. For simplicity we suppose $\PP = \WDP(\RS)$ and assume the
existence of a linear restricted interpretation $\A$, compatible with $\UU(\PP) \cup \PP$. 
Clearly this implies the well-foundedness of the relation $\rsrew{\UU(\PP) \cup \PP}$, which
in turn implies the well-foundedness of $\rsrew{\RS}$, cf.~Lemma~\ref{l:usable}.
Hence Theorem~\ref{t:dp:usable} is applicable and we 
conclude $\dheight(t,\rsrew{\RS}) \leqslant \dheight(t^\sharp, \rsrew{\WDP(\RR) \cup \RR})$.
On the other hand, compatibility with $\A$ implies that
$\dheight(t^\sharp, \rsrew{\WDP(\RR) \cup \RR}) = \bO(\size{t^\sharp})$. 
As $\size{t^\sharp} = \size{t}$, we can combine these equalities 
to conclude linear runtime complexity of $\RS$.
\end{proof}

The below given example applies Corollary~\ref{c:dp:usable} to the
motivating Example~\ref{ex:ijcar08a:1} introduced in Section~\ref{ijcar08a:Intro}.
%
\begin{example}[continued from Example~\ref{ex:ijcar08a:1:WDP}.] 
\label{ex:ijcar08a:1:final}
We take a quadratic restricted interpretation $\BB$ into $\NN \setminus \{0\}$ 
with
$\m{c}_\BB = \m{d}_\BB  =  \nil_\BB = 1$, $\m{e}_\BB(x,y) = x + y$, 
${x} \cons_\BB {y} = x + y$, 
$\hd^\sharp_\BB(x) = x + 1$, $\tl^\sharp_\BB(x) = x + 1$,
$\isempty^\sharp_\BB (x) = x + 1$,
$\app^\sharp_\BB (x,y) = x^2 + 3x + y + 1$, and
$\ifapp^\sharp_\BB (x,y,z) = 2x + y + z^2 + z$.
Then $\BB$ interprets $\WDP(\RR)$ as follows:
\begin{equation*}
  \begin{array}{ll@{\hspace{5ex}}ll}
    5\colon & 2 > 1 & 8\colon & x+y + 1 > y\\[1ex]
    6\colon & x + y + 1 > 1 & 9\colon &x^2 + 3x + y + 1 > 3x + y + x^2\\[1ex]
    7\colon & x + y + 1 > x & 10\colon & 2x+y + 2 >y \\[1ex]
    11\colon & \multicolumn{3}{l}{2x + y + u^2 + 2uv +v^2 + u+v > u + v^2 + 3v + y + 1}
  \end{array}
\end{equation*}
Therefore, ${\WDP(\RR)} \subseteq {\gord{\BB}}$
holds. Hence, the runtime complexity of $\RR$ for full rewriting 
is quadratic. (Recall that $\UU(\WDP(\RR)) = \varnothing$.)
\end{example}

\section{The Weight Gap Principle}
\label{semantical gap}

We recall the notion of \emph{relative rewriting} (\cite{Geser:1990,Terese}).
\begin{definition}
Let $\RR$ and $\SS$ be TRSs.
We write $\to_{\RR / \SS}$ for $\to_\SS^* \cdot \to_\RR \cdot \to_\SS^*$ and
we call $\to_{\RR / \SS}$ the \emph{relative rewrite relation} of $\RR$ over $\SS$.%
\footnote{Note that ${\to_{\RR / \SS}} = {\rsrew{\RS}}$, if $\SS = \varnothing$.}
\end{definition}
Since $\dheight(t, {\to_{\RR / \SS}})$ corresponds to
the number of $\to_\RS$-steps in a maximal derivation of 
$\to_{\RR \cup \SS}$ from $t$, we easily see the bound
$\dheight(t, {\to_{\RR / \SS}}) \leqslant \dheight(t, {\to_{\RR \cup \SS}})$.
In this section we study the opposite, i.e., we figure out a way to give an 
upper-bound of $\dheight(t, {\to_{\RR \cup \SS}})$ by a function of 
$\dheight(t, {\to_{\RR / \SS}})$.

First we introduce the key ingredient, \emph{strongly linear interpretations},
a very restrictive form of polynomial interpretations.
Let $\FS$ denote a signature. 
A \emph{strongly linear interpretation} (\emph{SLI} for short) 
is a WMA $(\A,\succ)$ that satisfies the following properties:
(i) the carrier of $\A$ is the set of natural numbers $\N$,
(ii) all interpretation functions $f_{\A}$ are strongly linear,
(iii) the proper order $\succ$ is the standard order $>$ on $\N$.
Note that an SLI $\A$ is conceivable as a weight function. We
define the maximum weight $\MA$ of $\A$ as 
$\max \{ f_\AA(0,\ldots,0) \mid f \in \FF \}$.
Let $\A$ denote an SLI, let $\alpha_0$ denote the assignment mapping 
any variable to $0$, i.e., $\alpha_0(x) = 0$ for all $x \in \VS$, and 
let $t$ be a term. 
We write $[t]$ as an abbreviation for $\eval{\alpha_0}{\A}(t)$.

\begin{lemma} \label{l:SLIa}
Let $\AA$ be an SLI and let $t$ be a term. Then $[t] \leqslant \MA \cdot \size{t}$ holds.
\end{lemma}
\begin{proof}
By induction on $t$. 
If $t \in \VV$ then $[t] = 0 \leqslant \MA \cdot \size{t}$.
Otherwise, suppose $t = f(t_1,\ldots,t_n)$, where
$f_\AA(x_1,\ldots,x_n) = x_1 + \ldots + x_n + c$.
By the induction hypothesis and $c \leqslant \MA$ we obtain the
following inequalities:
\begin{align*}
[t]
& = f_\AA([t_1],\ldots,[t_n]) 
  \leqslant [t_1] + \cdots + [t_n] + c \\
& \leqslant \MA \cdot \size{t_1} + \cdots + \MA \cdot \size{t_n} + \MA 
  = \MA \cdot \size{t} \tpkt
\end{align*}
\end{proof}

The conception of strongly linear interpretations as weight
functions allows us to study (possible) weight increase throughout
a rewrite derivation. This observation is reflected in the
next definition. 
\begin{definition} \label{d:weightgap}
Let $\AA$ be an algebra and let $\RR$ be a TRS.  
The \emph{weight gap}
$\WG(\A,\RS)$ \emph{of $\A$ with respect to $\RR$} is defined on $\NN$ 
as follows: $\WG(\A,\RS) = \max \{[r] \modminus [l] \mid l \to r \in \RR \}$,
where $\modminus$ is defined as usual: $m \modminus n \defsym \max \{m-n,0\}$
\end{definition}

The following \emph{weight gap principle} is a direct consequence of the
definitions.

\begin{lemma} \label{l:WG}
Let $\RR$ be a non-duplicating TRS and $\AA$ an SLI. 
If $s \to_\RR t$ then $[s]+\WG(\A,\RS) \geqslant [t]$.
\end{lemma}
We stress that the lemma does only require the non-duplicating condition. 
Indeed, the implication in the lemma holds even if the TRS $\RR$ is 
\emph {not} compatible with a strongly linear interpretation.
This principle brings us to the next theorem.

\begin{theorem} \label{t:relative}
Let $\RR$ and $\SS$ be TRSs, $\AA$ an SLI compatible with $\SS$ and
$\RR$ non-duplicating.
Then we have
$\dheight(t, \to_{\RR \cup \SS}) \leqslant 
(1 + \WG(\A,\RS)) \cdot \dheight(t,\to_{\RR/\SS}) + \MA \cdot \size{t}$,
whenever $t$ is terminating on $\RS \cup \SS$.
\end{theorem}
\begin{proof}
Let $m = \dheight(t, {\to_{\RR/\SS}})$, let $n = \size{t}$, and set
$\Delta = \WG(\A,\RS)$. 
Any derivation of $\to_{\RR \cup \SS}$ is representable as follows
\begin{equation*}
s_0 \to_\SS^{k_0} 
t_0 \to_\RR 
s_1 \to_\SS^{k_1} 
t_1 \to_\RR \cdots \to_\SS^{k_m}
t_m \tkom
\end{equation*}
and without loss of generality we may assume that the derivation is
maximal. 
We observe the next two facts.
\begin{itemize}
\item[$(\m{a})$]
$k_i \leqslant [s_i] - [t_i]$ holds for all $0 \leqslant i \leqslant m$. 
This is because $[s] \geqslant [t] + 1$ whenever $s \rsrew{\SS} t$
by the assumption $\SS \subseteq {\gord{\AA}}$,
and we have $s_i \rsrew{\SS}^{k_i} t_i$.

\item[$(\m{b})$] 
$[s_{i+1}] - [t_i] \leqslant \Delta$ holds for all $0 \leqslant i < m$ as
due to Lemma~\ref{l:WG} we have $[t_i] + \Delta \geqslant [s_{i+1}]$.

\end{itemize}
We obtain the following inequalities:
\begin{align*}
  \dheight(s_0, \to_{\RR \cup \SS}) & = m + k_0 + \dots + k_m \\
  & \leqslant m + ([s_0] - [t_0]) + \dots + ([s_m] - [t_m]) \\
  & = m + [s_0] + ([s_1] - [t_0]) + \dots + ([s_m] - [t_{m-1}]) - [t_m] \\
  & \leqslant m + [s_0] + m \Delta - [t_m] \\
  & \leqslant m + [s_0] + m \Delta \\
  & \leqslant m + \MA \cdot n + m \Delta 
    = (1 + \Delta) m + \MA \cdot n \tpkt
\end{align*}
Here we used $(\m{a})$ $m$-times in the second line, $(\m{b})$ 
$m-1$-times in the fourth line, and Lemma~\ref{l:SLIa} in the last line.
\end{proof}

The next example clarifies that the conditions expressed in Theorem~\ref{t:relative}
are essentially optimal: We cannot replace the assumption that the algebra 
$\A$ is \emph{strongly} linear with a weaker assumption:
Already if $\A$ is a linear polynomial interpretation, 
the derivation height of $\RS \cup \SS$ cannot be bounded \emph{polynomially} in 
$\dheight(t, \rsrew{\RR/\SS})$ and $\size{t}$ alone.
\begin{example}
Consider the  TRSs $\RS$
\begin{align*}
 \m{exp}(\mN) & \to \ms(\mN) & \m{d}(\mN) & \to \mN \\
\m{exp}(\m{r}(x)) & \to \m{d}(\m{exp}(x)) & \m{d}(\ms(x)) & \to \ms(\ms(\m{d}(x))) 
\end{align*}
This TRS formalises the exponentiation
function. Setting $t_n = \m{exp}(\m{r}^n(\mN))$ we obtain
$\dheight(t_n, \rsrew{\RS}) \geqslant 2^n$ for each $n \geqslant 0$. Thus
the runtime complexity of $\RS$ is (at least) exponential.
In order to show the claim, we split $\RS$ into two TRSs 
$\RS_1 = \{\m{exp}(\mN) \to \ms(0), \m{exp}(\m{r}(x)) \to \m{d}(\m{exp}(x))\}$
and $\RS_2 = \{\m{d}(\mN) \to \mN, \m{d}(\ms(x)) \to \ms(\ms(\m{d}(x))) \}$.
Then it is easy to verify that the next linear polynomial interpretation $\AA$ 
is compatible with $\RS_2$: 
$\mN_{\A} = 0$, $\m{d}_{\A}(x) = 3x$, and $\ms_{\A}(x) = x + 1$.
Moreover an upper-bound of $\dheight(t_n ,\rsrew{\RS_1/\RS_2})$ can be estimated 
by using the following polynomial interpretation $\BB$:
$\mN_{\BB} = 0$, $\m{d}_{\BB}(x) = \ms_{\BB}(x) = x$, and 
$\m{exp}_{\BB}(x) = \m{r}_{\BB}(x) = x + 1$.
Since ${\to_{\RR_1}} \subseteq {\gord{\BB}}$ and 
${\to_{\RS_2}^*} \subseteq {\geqord{\BB}}$ hold, we have 
${\rsrew{{\RS_1}/{\RS_2}}} \subseteq {\gord{\BB}}$.  Hence
$\dheight(t_n, \rsrew{{\RS_1}/{\RS_2}}) \leqslant \eval{\alpha_0}{\BB}(t_n) = n+2$.
But clearly from this
we cannot conclude a polynomial bound on the derivation length of $\RS_1 \cup \RS_2 = \RS$,
as the runtime complexity of $\RS$ is exponential, at least.
\end{example}

To conclude this section, we show that Theorem~\ref{t:dp:usable} 
can only hold for \emph{basic} terms $t \in \TBS$. 
\begin{example} \label{ex.SK}
Consider the one-rule TRS 
$\RS  = \{\m{a}(\m{b}(x)) \to \m{b}(\m{b}(\m{a}(x)))\}$
from \cite[Example 2.50]{SteinbachKuehler:1990}.
It is not difficult to see that
$\dheight(\m{a}^n(\m{b}(x)),\rsrew{\RS}) = 2^n -1$, 
see~\cite{Hofbauer:2001}. 
The set $\WDP(\RS)$ consists of just one dependency pair 
$\m{a}^\sharp(\m{b}(x)) \to \m{a}^\sharp(x)$.
In particular the set of usable rules is empty.
The following SLI $\A$ is compatible with $\WDP(\RS)$:
$\m{a}^\sharp_{\A}(x) = \m{a}_{\A}(x) = x$ and $\m{b}_{\A}(x) = x+1$.
Hence, due to Lemma~\ref{l:SLIa} we can conclude the existence of a constant $K$ such that 
$\dheight(t^\sharp,\rsrew{\WDP(\RS)}) \leqslant K \cdot \size{t}$. 
Due to Theorem~\ref{t:dp:usable} we conclude linear runtime complexity of $\RS$. 
\end{example}

\section{Reduction Pairs and Argument Filterings} \label{Refinements}

In this section we study the consequences of combining
Theorem~\ref{t:dp:usable} and Theorem~\ref{t:relative}. 
In doing so, we adapt reduction pairs and argument filterings
(\cite{ArtsGiesl:2000}) to runtime complexity analysis.
Let $\RS$ be a TRS, and let $\A$ be a strongly linear interpretation and
suppose we consider weak, weak innermost, or (standard) dependency pairs $\PP$,
such that $\PP$ is non-duplicating.
If $\UU(\PP) \subseteq {>_\A}$ then there exist constants $K, L \geqslant 0$
(depending on $\PP$ and $\A$ only) such that
\begin{equation*}
\dheight(t, \rsrew{\RS}) \leqslant 
K \cdot \dheight(t^\sharp, \rsrew{\PP / \UU(\PP)}) +  L \cdot \size{t^\sharp} \tkom
\end{equation*}
for all terminating basic terms $t \in \TB$. 
This follows from the combination of Theorems~\ref{t:dp:usable} and~\ref{t:relative}.
Thus, in order to estimate the derivation length of $t$ with respect to $\RS$ 
it suffices to estimate the maximal $\PP$ steps, i.e., we have to
estimate $\dheight(t^\sharp, \rsrew{\PP / \UU(\PP)})$ suitably.
Consider a maximal derivation $(t_i)_{i=0,\ldots,n}$ of 
$\rsrew{\PP / \UU(\PP)}$ with $t_0 = t^\sharp$. 
For every $0 \leqslant i < n$ 
there exist terms $u_i$ and $v_i$ such that
\begin{equation} \label{eq:refine}
t_i \to_{\UU(\PP)}^* u_i \to_{\PP} v_i \to_{\UU(\PP)}^* t_{i+1} \tpkt
\end{equation}
Let $\gtrsim$ and $\succ$ be a pair of orders with
${\gtrsim} \cdot {\succ} \cdot {\gtrsim} \subseteq {\succ}$.
If $t_i \gtrsim u_i \succ v_i \gtrsim t_{i+1}$ holds for all $0 \leqslant i < n$,
we obtain $t^\sharp = t_0 \succ t_1 \succ \cdots \succ t_n$.  Therefore, 
$\dheight(t^\sharp, \rsrew{\PP / \UU(\PP)})$ can be bounded in the
maximal length of $\succ$-descending steps. 
We formalise these observations through the use of \emph{reduction pairs} and 
\emph{collapsible orders}.

\begin{definition} \label{d:collapsible}
Let $\RS$ be a TRS, let $\PP$ be a set of weak dependency pairs of $\RS$
and let $\Slow$ denote a mapping associating a term (over $\FS^\sharp$ and $\VS$) 
and a proper order $\succ$ with a natural number.
An order $\succ$ on terms is $\Slow$-\emph{collapsible} 
for a TRS $\RR$ 
if $s \rsrew{\PP \cup \UU(\PP)} t$  and $s \succ t$ implies $\Slow(s,\succ) > \Slow(t,\succ)$. 
An order $\succ$ is \emph{collapsible} for a TRS $\RS$, 
if there is a mapping $\Slow$ such that $\succ$ is $\Slow$-collapsible 
for $\RR$.
\end{definition}
Note that most reduction orders are collapsible.  For instance, 
if $\A$ is a polynomial interpretation then $\gord{\A}$ is collapsible, 
as witnessed by the evaluation function $\eval{\alpha_0}{\A}$.
Furthermore, simplification orders like MPO, LPO and KBO are collapsible
(cf.~\cite{Hofbauer:1992,Weiermann:1995,Moser:2006lpar}).%
\footnote{On the other hand it is easy to construct non-collapsible
orders: Suppose we extend the natural numbers $\N$ by a non-standard element $\infty$ such that
for any $n \in \N$ we set $\infty > n$. Clearly we cannot collapse $\infty$ to a
natural number.}

\begin{definition}
A \emph{rewrite preorder} is a preorder on terms which is closed under contexts and
substitutions. A \emph{reduction pair} $(\gtrsim, \succ)$ consists of a
rewrite preorder $\gtrsim$ and a compatible well-founded order $\succ$
which is closed under substitutions.  Here compatibility means 
the inclusion ${\gtrsim \cdot \succ \cdot \gtrsim} \subseteq {\succ}$.
A reduction pair $({\gtrsim}, {\succ})$ is called \emph{collapsible} 
for a TRS $\RR$
if $\succ$ is collapsible for $\RR$.
\end{definition}

Recall the derivation in~\eqref{eq:refine}: Due to compound symbols the
rewrite step $u_i \rsrew{\PP} v_i$ may take place below the root.
Hence $\PP \subseteq {\succ}$ does not ensure $u_i \succ v_i$.  
To address this problem we introduce a notion of \emph{safety} that is based on the next definitions.
\begin{definition}
The set $\TBC$ is inductively defined as follows (i) $\TTs \cup \TT \subseteq \TBC$,
where $\TTs = \{t^{\sharp} \mid t \in \TT \}$ and
(ii) $c(t_1,\ldots,t_n) \in \TBC$, whenever $t_1,\ldots,t_n \in \TBC$ and $c$ is a compound symbol.
\end{definition}

\begin{definition}
A proper order $\succ$ on $\TBC$ 
is called \emph{safe} if
$c(s_1,\ldots,s_i,\ldots,s_n) \succ c(s_1,\ldots,t,\ldots,s_n)$
for all $n$-ary compound symbols $c$ and all terms $s_1,\ldots,s_n,t$ 
with $s_i \succ t$.
A reduction pair $({\gtrsim}, {\succ})$ is called \emph{safe}
if $\succ$ is safe.
\end{definition}

\begin{lemma} \label{l:5}
Let $\PP$ be a set of weak, weak innermost, or standard dependency pairs, and 
let $(\gtrsim, \succ)$ be a safe reduction pair such that
$\UU(\PP) \subseteq {\gtrsim}$ and $\PP \subseteq {\succ}$. 
If $s \in \TBC$ and $s \rsrew{\PP / \UU(\PP)} t$ then
$s \succ t$ and $t \in \TBC$.
\end{lemma}

Employing Theorem~\ref{t:dp:usable}, 
Theorem~\ref{t:relative}, and Lemma~\ref{l:5} we arrive at our Main Theorem.
\begin{theorem} \label{t:ijcar08a:main}
Let $\RR$ be a TRS, let $\AA$ be an SLI, let $\PP$ be the set of weak, weak
innermost, or (standard) dependency pairs, such that $\PP$ is non-duplicating, 
and let $(\gtrsim, \succ)$ be a safe and 
$\Slow$-collapsible reduction pair such that $\UU(\PP) \subseteq {\gtrsim}$ and
$\PP \subseteq {\succ}$. 
If in addition $\UU(\PP) \subseteq {>_\A}$ then 
for any $t \in \TB$, we have $\dheight(t, \rew) \leqslant p(\Slow(t^{\sharp},\succ), \size{t})$,
where $p(m,n) \defsym (1 + \WG(\A,\PP)) \cdot m + \MA \cdot n$
and $\rew$ denotes $\rsrew{\RS}$ or $\irew{\RS}$
depending on whether $\PP = \WDP(\RS)$ or $\PP = \WIDP(\RS)$.
Moreover if all compound symbols in $\WIDP(\RR)$ are nullary  
we have $\dheight(t, \irew{\RS}) \leqslant p(\Slow(t^{\sharp},\succ), \size{t}) + 1$.
\end{theorem}
\begin{proof}
First, observe that the assumptions imply that any basic term $t \in \TB$
is terminating with respect to $\RS$. This is a direct consequence of
Lemma~\ref{l:usable} and Lemma~\ref{l:5} in conjunction with the
assumptions of the theorem.
Without loss of generality, we assume $\PP = \WDP(\PP)$. 
By Theorem~\ref{t:dp:usable} 
and~\ref{t:relative} we obtain:
\begin{align*}
  \dheight(t, \rew)  & \leqslant \dheight(t^\sharp,\rsrew{\UU(\PP) \, \cup \, \PP}) \leqslant p(\dheight(t^\sharp,\rsrew{\PP/\UU(\PP)}),\size{t^{\sharp}})\\ 
  & \leqslant p(\Slow(t^{\sharp},\succ),\size{t^{\sharp}}) = p(\Slow(t^{\sharp},\succ),\size{t}) \tpkt
\end{align*}
In the last line we exploit that $\size{t^{\sharp}} = \size{t}$.
\end{proof}

Note that there exist two subtle disadvantages of Theorem~\ref{t:ijcar08a:main} 
in comparison to Theorem~\ref{t:dp:usable}. 
First the Main Theorem requires that the set of weak, weak innermost, or 
(standard) dependency pairs $\PP$ is non-duplicating. 
Second, the requirement that the usable rules are compatible with some SLI, implies
that all usable rules must be non-duplicating. Hence the set
$\UU(\PP) \cup  \PP$ must not contain duplicating rules.
This is not necessary to meet the requirements of Theorem~\ref{t:dp:usable}.

In order to construct safe reduction pairs 
one may use \emph{safe algebras}, i.e., weakly monotone well-founded 
algebras $(\AA, \succ)$ such that the interpretations of compound symbols
are strictly monotone with respect to $\succ$.
Another way is to apply an argument filtering to a reduction pair.
\begin{definition}
An \emph{argument filtering} for a signature $\FF$ is a
mapping $\pi$ that assigns to every $n$-ary function symbol $f \in \FF$
an argument position $i \in \{ 1, \dots, n \}$ or a (possibly empty)
list $[ \seq[m]{i} ]$ of argument positions with
$1 \leqslant i_1 < \cdots < i_m \leqslant n$.
The signature $\FFpi$ consists of all function symbols $f$ such that
$\pi(f)$ is some list $[ \seq[m]{i} ]$, where in $\FFpi$ the arity of
$f$ is $m$. Every argument filtering $\pi$ induces a mapping from
$\TERM$ to $\TERMpi$, also denoted by $\pi$:
\[
\pi(t) = \begin{cases}
t & \text{if $t$ is a variable} \\
\pi(t_i) & \text{if $t = f(\seq{t})$ and $\pi(f) = i$} \\
f(\pi(t_{i_1}),\dots,\pi(t_{i_m})) &
\text{if $t = f(\seq{t})$ and $\pi(f) = [ \seq[m]{i} ]$}
\end{cases}
\]
An argument filtering $\pi$ is called \emph{safe} if 
$\pi(c) = [1,\ldots, n]$ for all $n$-ary compound symbols.
For a relation $R$ on $\TERM$ we define $R^\pi$ on $\TERMpi$ as follows: 
$s \mathrel{R^\pi} t$ if and only if $\pi(s) \mathrel{R} \pi(t)$.
\end{definition}

\begin{lemma}
If $(\AA, \succ)$ is a safe algebra then
$(\succcurlyeq_{\A},\succ_{\A})$ is a safe reduction pair, where
$\succcurlyeq$ denotes the reflexive closure of $\succ$.
Furthermore, $(\gtrsim^\pi,\succ^\pi)$ is a safe reduction pair
if $(\gtrsim, \succ)$ is a safe reduction pair and $\pi$ is 
a safe argument filtering. 
\end{lemma}

Following the pattern of the proof of Corollary~\ref{c:dp:usable}
it is an easy exercise to extend Theorem~\ref{t:ijcar08a:main} to a method 
for complexity analysis.  
%
\begin{corollary} \label{c:main1}
Let $\RR$ be a TRS, let $\AA$ be an SLI, 
let $\PP$ be the set of weak, weak innermost, or standard dependency pairs,
such that $\PP$ is non-duplicating,
where the compound symbols in $\WIDP(\RR)$ are nullary, if $\PP = \DP(\RS)$.
Moreover let $\BB$ be a linear or quadratic restricted interpretation 
such that $(\geqord{\BB},\gord{\BB})$ forms a safe reduction pair with
$\UU(\PP) \subseteq {\geqord{\BB}}$ and $\PP \subseteq {\gord{\BB}}$.
If $\UU(\PP) \subseteq {>_\A}$ then 
the (innermost) runtime complexity function $\rc^{(\m{i})}_{\RS}$
with respect to $\RS$ is linear or
quadratic, respectively. 
\end{corollary}
Note that if $\UU(\PP) = \varnothing$, the compatibility of $\UU(\PP)$ with an SLI is
trivially satisfiable. In this special case by taking the SLI $\A$ that 
interprets all symbols with the identity function,
we obtain $\dheight(t,\rew) \leqslant \Slow(t^\sharp,\gord{\BB}) + \size{t}$ 
because $\WG(\A,\varnothing) = 0$ and $\MA = 1$. 

As a consequence of Theorem~\ref{t:ijcar08a:main} and Lemma~\ref{l:dp:4} we obtain
the following corollary.

\begin{corollary} \label{c:main2}
Let $\RR$ be a TRS, let $\AA$ be an SLI, let all compound symbols in 
$\WIDP(\RS)$ be nullary, let $\WIDP(\RS)$ be non-duplicating, 
and let $\BB$ be a linear or quadratic
restricted interpretation such that $(\geqord{\BB},\gord{\BB})$ forms 
a reduction pair with $\UU(\DP(\RS)) \subseteq {\geqord{\BB}}$ 
and $\DP(\RS) \subseteq {\gord{\BB}}$.
If in addition $\UU(\DP(\RS)) \subseteq {>_\A}$ then the innermost 
runtime complexity function $\Rci{\RS}$ with respect to $\RS$ is
linear or quadratic, respectively. 
\end{corollary}

Corollary~\ref{c:main2} establishes (for the first time) a method to analyse
the derivation length induced by the standard dependency pair method for 
innermost rewriting. More general, if  all compound symbols in 
$\WIDP(\RS)$ are nullary and there exists a collapsible reduction pair
$(\gtrsim,\succ)$ such that $\UU(\PP) \subseteq {\gtrsim}$ and
$\PP \subseteq {\succ}$, then the innermost runtime complexity of
$\RS$ is linear in the maximal length of $\succ$-descending steps.
Clearly for \emph{string rewriting} (cf. \cite{Terese}) 
the compound symbols in $\WIDP(\RS)$ are always nullary and
all rules in $\WIDP(\RS)$ are non-duplicating. Hence the syntactic
requirements are always met.

\section{Experiments} \label{ijcar08a:Experiments}

In order to test the practical feasibility of the here established methods,
we implementated a complexity analyser based on
syntactical transformations for dependency pairs and usable rules 
together with polynomial orders (based on \cite{contejean:2005}).
To deal efficiently with polynomial interpretations, the issuing
constraints are encoded in \emph{propositional logic} in a similar
spirit as in~\cite{FGMSTZ07}.  
Assignments are found by employing a state-of-the-art SAT solver, in our case
Mini\-Sat%
\footnote{\url{http://minisat.se/}.}. 
Furthermore, strongly linear interpretations are handled by 
a decision procedure for Presburger arithmetic.

In a similar way, the new techniques have also been incorportated into the 
\emph{Tyrolean Complexity Tool} (\tct\ for short) that 
incorporates the most powerful techniques to analyse the complexity of 
rewrite systems that are currently at hand.%
\footnote{\url{http://cl-informatik.uibk.ac.at/software/tct/}.}
For compilation of the here presented experimental data we used the
latter implementation.

As suitable test bed we used the rewrite systems in the Termination Problem Data Base
version 4.0.%
\footnote{\url{http://colo5-c703.uibk.ac.at:8080/termcomp/}.} 
This test bed comprises 1739 TRSs.
The presented tests were performed on a server with 8 Dual-Core 2.6 GHz AMD\textsuperscript\textregistered\ 
Opteron\texttrademark\ Processor 8220 CPUs, for a total of 16 cores. 64 GB of RAM are
available. For each system we used a timeout of 60 seconds, the times 
in the tables are given in seconds. 
Tables~\ref{tab:lc} and~\ref{tab:qc} summarise the results of the conducted experiments.%
\footnote{For full experimental evidence see~\url{http://www.jaist.ac.jp/~hirokawa/08a/}
or \url{http://cl-informatik.uibk.ac.at/software/tct/}.}
Text written in~\textit{italics} below the number of successes or failures
indicates total time of success cases or failure cases, respectively.%
\footnote{Sum of numbers in each column may be less than 1739 because of
stack overflow.}

\begin{table}
\begin{center}
\caption{Results for Linear Runtime Complexities}
\label{tab:lc}
\medskip
\begin{tabular}{@{\hspace{0pt}}lr@{~~}r@{~}r@{~}r@{\qquad}rl@{~}rl@{~}l@{\hspace{0pt}}}
\hline\\[-2ex]
& \multicolumn{4}{c}{\emph{full rewriting}}
& \multicolumn{5}{c}{\emph{innermost rewriting}} \\
& \textsf{LR} 
& Cor.~\ref{c:dp:usable} & Cor.~\ref{c:main1}
& both
& Cor.~\ref{c:dp:usable} & ($\DP$) & Cor.~\ref{c:main1} & ($\DP$) 
& both \\
\hline\\[-2ex]
S
& 139 
& 139 & 93 & 147
& 144 & (136) & 102 & (91) & 166 \\
& \emph{7} 
& \emph{9} & \emph{14} & \emph{31}
& \emph{8} & (\emph{7}) & \emph{16} & (\emph{13}) & \emph{33} \\
F
& 1591 
& 1582 & 1646 & 1587
& 1577 & (1581) & 1637 & (1648) & 1568 \\
& \emph{2474} 
& \emph{4789} & \emph{456} & \emph{3853} 
& \emph{4699} & (\emph{4797}) &  \emph{462} & (\emph{455}) & \emph{3628}\\
T
& 8
& 17 & 0 & 5
& 17 & (22) & 0 & (0) & 5\\
 \hline\\[-3em]
\end{tabular}
\end{center}
\end{table}

We use the following abbreviations: The method~\LR\ (\QR) refers to compatibility with
\emph{linear (quadratic) restricted interpretation}, 
cf.~Section~\ref{dependency pairs}. Moreover ``S'', ``F'', ``T'' denotes
\emph{success}, \emph{failure}, or \emph{timeout} respectively.
In interpreting defined and dependency pair functions, we restrict the search to 
polynomials in the range $\{0,1,\dots,5\}$.
Table~\ref{tab:lc} shows the experimental results for
linear runtime complexities based on \LR. The columns
marked ``Cor.~\ref{c:dp:usable}'' and ``Cor.~\ref{c:main1}'' refer to 
the applicability of the respective corollaries. In the column
marked ``both'' we indicate the results, we obtain when we first
try to apply Corollary~\ref{c:main1} and if this fails Corollary~\ref{c:dp:usable}. 
Table~\ref{tab:qc} summarises experimental results for
quadratic runtime complexities based on \QR.
On the studied test bed there are 1567 TRSs such that one may switch from 
$\WIDP(\RR)$ to $\DP(\RR)$. For the individual tests, we indicated 
the results in parentheses for these versions of Corollary~\ref{c:dp:usable} 
and Corollary~\ref{c:main1}.

\begin{table}
\begin{center}
\caption{Results for Quadratic Runtime Complexities}
\label{tab:qc}
\medskip
\begin{tabular}{@{\hspace{0pt}}lr@{~~}r@{~}r@{~}r@{\qquad}rl@{~}rl@{~}l@{\hspace{0pt}}}
\hline\\[-2ex]
& \multicolumn{4}{c}{\emph{full rewriting}}
& \multicolumn{5}{c}{\emph{innermost rewriting}} \\
& \textsf{QR} 
& Cor.~\ref{c:dp:usable} & Cor.~\ref{c:main1}
& both
& Cor.~\ref{c:dp:usable} & ($\DP$) & Cor.~\ref{c:main1} & ($\DP$) 
& both \\
\hline\\[-2ex]
S
& 182 
& 182 & 93 & 186
& 183 & (166) & 102 & (91) & 193 \\
& \emph{152} 
& \emph{329} & \emph{97} & \emph{614} 
& \emph{324} & (\emph{327}) & \emph{98} & (\emph{81}) & \emph{486} \\
F
& 524
& 473 & 1636 & 564
& 492 & (855) & 1627 & (1636) & 577 \\
& \emph{5469} 
& \emph{5436} & \emph{793} & \emph{5215}
& \emph{5500} & (\emph{4884}) & \emph{825} & (\emph{765}) & \emph{5535}
\\
T
& 864
& 951 & 10 & 853
& 924 & (884) & 10 & (12) & 833 \\
\hline\\[-3em]
\end{tabular}
\end{center}
\end{table}

\section{Conclusion} \label{ijcar08a:Conclusion}

In this paper we studied the runtime complexity of rewrite systems. 
We have established a variant of the dependency pair method 
that is applicable in this context and is easily mechanisable. 
In particular our findings extend the class of TRSs
whose \emph{linear} or \emph{quadratic} runtime complexity can be detected automatically.
We provided ample numerical data for assessing the viability of the method.
To conclude, we mention possible future work.
In the experiments presented, we have restricted our attention to 
interpretation based methods inducing linear or quadratic (innermost) 
runtime complexity. Recently in~\cite{AM:2008} (see Chapter~\ref{flops08}) 
a restriction of the multiset path order, called \emph{polynomial path order} has been introduced that induces polynomial
runtime complexity. In future work we will test to what extent this
is effectively combinable with our Main Theorem.
Furthermore, we strive to extend the approach
presented here to handle dependency graphs~\cite{ArtsGiesl:2000}.

}
\chapter{Complexity, Graphs, and the Dependency Pair Method}
\label{lpar08}

\subsection*{Publication Details}
N.~Hirokawa and G.~Moser.
\newblock Complexity, graphs, and the dependency pair method.
\newblock In \emph{Proceedings of the International Conference on Logic for
  Programming Artificial Intelligence and Reasoning}, number 5330 in LNAI,
  pages 652--666, Springer Verlag,
  2008{\natexlab{b}}.%
\footnote{This research was partially supported by FWF (Austrian Science Fund) project P20133.}

\medskip
The mistake mentioned in Chapter~\ref{ijcar08a} resulted in a
flawed Proposition~9 in the published version. This mistake has been rectified below. 
This influenced the given experimental evidence and I would
like to thank Andreas Schnabl and Martin Avanzini for providing me 
with adjusted experimental data.

\subsection*{Ranking}
The Conference on Logic Programming and Automated Reasoning has been ranked
\textbf{A} by CORE in 2007.

{
\input{lpar08.sty}

\subsection*{Abstract}
This paper builds on recent efforts (see Chapter~\ref{ijcar08a}) to exploit 
the dependency pair method for verifying feasible, i.e., polynomial
\emph{runtime complexities} of term rewrite systems automatically.
We extend our earlier results 
by revisiting dependency graphs in the context
of complexity analysis. 
The obtained new results are easy to implement and 
considerably extend the analytic power of our existing methods.
The gain in power is even more significant when compared to existing methods that
directly, i.e., without the use of transformations, 
induce \emph{feasible} runtime complexities.
We provide ample numerical data for assessing the viability of the method.

\section{Introduction} \label{lpar08:Intro}

Term rewriting is a conceptually simple but powerful abstract model of
computation that underlies much of declarative programming. 
\emph{Runtime complexity} is a notion for capturing time complexities
of functions defined by a term rewriting system (TRS for short) introduced
in~\cite{HM:2008} (see Chapter~\ref{ijcar08a} but also~\cite{Lescanne:1995,BCMT:2001,AM:2008}). 
In recent research we revisited the basic dependency pair method~\cite{ArtsGiesl:2000}
in order to make it applicable for complexity analysis, cf.~\cite{HM:2008}.  
The dependency pair method introduced by Arts and Giesl~\cite{ArtsGiesl:2000} 
is one of the most powerful methods in termination analysis.
The method enables us to use several powerful techniques including, 
usable rules, reduction pairs, argument filterings, and 
dependency graphs. Our main results in~\cite{HM:2008} show how natural 
improvements of the dependency pair method, like usable rules, 
reduction pairs, and argument filterings become applicable in the 
context of complexity analysis. 
In this paper, we will extend these recent results further. 

The dependency pair method for termination analysis is based on 
the observation that from an arbitrary non-terminating term
one can extract a minimal non-terminating subterm. 
For that one considers \emph{dependency pairs} that essentially
encode recursive calls in a TRS. 
Note that with respect to the TRS defined in Example~\ref{ex:lpar08:1} below, 
one finds 5 such pairs (see Section~\ref{DG} for further details).

\begin{example} \label{ex:lpar08:1}
Consider the following TRS $\RS$ which computes a permutation of lists.%
\footnote{This is Example~3.12 in Arts and Giesl's collection of TRSs~\cite{ArtsGiesl:2001}.}
\begin{alignat*}{4}
1\colon && \app(\nil,y) & \rew y &
4\colon && \reverse(n \cons x) & \rew \app(\reverse(x),n \cons \nil)\\
2\colon && \app(n \cons x,y) & \rew n \cons \app(x,y) & \hspace{5ex}
5\colon && \shuffle(\nil) & \rew \nil\\
3\colon && \reverse(\nil) & \rew \nil & 
6\colon && \shuffle(n\cons x) & \rew n \cons \shuffle(\reverse(x))
\end{alignat*}
\end{example}

A very well-studied refinement of the dependency pair method 
are \emph{dependency graphs}.
To show termination of a TRS, it suffices to guarantee that none
of the cycles in $\DG(\RS)$~\cite{ArtsGiesl:2000} can give rise to an infinite rewrite sequence.
(Here a \emph{cycle} $\CC$ is a nonempty set of dependency pairs of 
$\RR$ such that for every two pairs $s \to t$ and $u \to v$ in $\CC$ 
there exists a nonempty path in $\CC$ from $s \to t$ to $u \to v$.)
More precisely it suffices to prove for every cycle
$\CC$ in the dependency graph $\DG(\RS)$, that
there are no $\CC$-minimal rewrite sequences (see~\cite{GAO:2002}, but also~\cite{HirokawaMiddeldorp:2007,GTSF06}). 
To achieve this one may consider each cycle independently, i.e., for each
cycle it suffices to find a reduction pair $(\gtrsim,\succ)$ (cf.~Section~\ref{lpar08:Preliminaries})
such that
${\RS} \subseteq {\gtrsim}$, ${\CC} \subseteq {\gtrsim}$ and 
${{\CC} \cap {\succ}} \not= {\varnothing}$, i.e., at least one dependency pair in $\CC$ 
is strictly decreasing.
\begin{example}[continued from Example~\ref{ex:lpar08:1}] 
The \emph{dependency graph} $\DG(\RS)$,
whose nodes are the mentioned 5 dependency pairs, has the following form
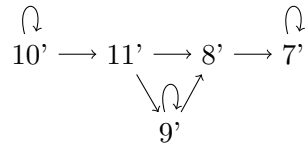
\begin{center}
  \begin{tikzpicture}[node distance=5mm]
    \node(10) {10'} ;
    \node(11) [right=of 10 ] {11'} ;
    \node(9) [below=of 11,xshift=6mm] {9'} ;
    \node(8) [right=of 11] {8'} ;
    \node(7) [right=of 8] {7'} ;

    \path[->] (7) edge [loop above] (7) ;
    \draw[->] (8) to (7);
    \path[->] (9) edge [loop above] (9) ;
    \path[->] (10) edge [loop above] (10) ;
    \draw[->] (9) to (8) ;
    \draw[->] (10) to (11) ;
    \draw[->] (11) to (8) ;
    \draw[->] (11) to (9) ;
  \end{tikzpicture}
\end{center}
This graph contains the 
(maximal) cycles $\{7'\}$, $\{9'\}$, and $\{10'\}$.%
\footnote{Recall that a cycle $\CC$ is maximal, if there is no
longer cycle containing $\CC$.} 
As already mentioned, it suffices to consider each of these three cycles individually. 
\end{example}

The main contribution of this paper is to extend the dependency graph
refinement of the dependency pair method to complexity analysis. 
This is a challenging task, and we face a
couple of difficulties, documented via suitable examples below.
To overcome these obstacles we adapt the standard notion of 
dependency graph suitably and introduce \emph{weak (innermost) dependency graphs},
based on \emph{weak dependency pairs}, which have been studied in~\cite{HM:2008}
(see also Chapter~\ref{ijcar08a}).
Moreover, we observe that in the context of complexity analysis,
it is not enough to focus on the (maximal) cycles of a (weak) dependency graph.
Instead, we show how cycle detection is to be replaced by \emph{path detection},
in order to salvage the (standard) technique of dependency graphs for
runtime complexity considerations.

The remainder of the paper is organised as follows.
After recalling basic notions in Section~\ref{lpar08:Preliminaries}, 
we recall in Section~\ref{DPMethod} main results from~\cite{HM:2008} that will be extended in the sequel. 
In Section~\ref{DG} we establish our dependency graph analysis for 
complexity analysis. Finally, we conclude in Section~\ref{lpar08:Conclusion},
where we assess the applicability of our method.

\section{Preliminaries} \label{lpar08:Preliminaries}

We assume familiarity with term rewriting~\cite{BaaderNipkow:1998,Terese},
but briefly review basic concepts and notations. Moreover, we assume
familiarity with standard notions in graph theory (see for example~\cite[Chapter~1]{Hein:2009}).

Let $\VS$ denote a countably infinite set of variables and $\FS$ a 
signature. The set of terms over $\FS$ and $\VS$ is denoted by 
$\TERMS$ ($\TT$ for short). The \emph{root symbol} of a term $t$ is either $t$ itself, if
$t \in \VS$, or the symbol $f$, if $t = f(\seq{t})$. 
The \emph{set of positions} $\Pos(t)$ of a term $t$ is defined as
usual. We write $\Pos_{\GG}(t) \subseteq \Pos(t)$ for the set of
positions of subterms whose root symbol is contained in $\GG \subseteq \FS$.
The \emph{descendants} of a position with respect to a rewrite
sequence are defined as usual, cf.~\cite{Terese}. 
The subterm relation is denoted as $\subterm$.
$\Var(t)$ ($\Fun(t))$ denotes the set of variables (functions) 
occurring in a term $t$.
The \emph{size} $\size{t}$ of a term is defined 
as the number of symbols in $t$. 
A \emph{term rewrite system} $\RS$ over
$\TERMS$ is a \emph{finite} set of rewrite
rules $l \to r$, such that $l \notin \VS$ and $\Var(l) \supseteq \Var(r)$.
The smallest rewrite relation that contains $\RS$ is denoted by $\to_{\RS}$, 
and its transitive and reflexive closure by $\rssrew{\RS}$.
We simply write $\to$ for $\to_{\RS}$ if $\RS$ is clear from context.
A term $s \in \TERMS$ is called a \emph{normal form} if there is no
$t \in \TERMS$ such that $s \to t$. 
The \emph{innermost rewrite relation} $\irew{\RR}$
of a TRS $\RR$ is defined on terms as follows: $s \irew{\RR} t$ if 
there exists a rewrite rule $l \to r \in \RR$, a context $C$, and
a substitution $\sigma$ such that $s = C[l\sigma]$, $t = C[r\sigma]$,
and all proper subterms of $l\sigma$ are normal forms of $\RR$.  
The set of defined symbols is denoted as $\DS$, while the constructor
symbols are collected in $\CS$.
We call a term $t = f(\seq{t})$ \emph{basic} if $f \in \DS$ and $t_i \in \TA(\CS,\VS)$
for all $1 \leqslant i \leqslant n$. 

We call a TRS \emph{terminating} if no infinite rewrite sequence
exists. The $n$-fold composition of $\to$ is denoted as $\to^n$ and
the \emph{derivation length} of a terminating term $t$ with respect to a 
TRS $\RS$ and rewrite relation $\rsrew{\RS}$ is defined as:
$\dheight(s,\rsrew{\RS}) \defsym \max\{ n \mid \exists t \; s \to^n t \}$. 
Let $\RR$ be a TRS and $T$ be a set of terms. 
The \emph{runtime complexity function with respect to a relation $\to$ on $T$}
is defined as follows:
\begin{equation*}
\Rc{}(n, T, \rew) \defsym \max\{ \dheight(t, \rew) \mid 
\text{$t \in T$ and $\size{t} \leqslant n$}\} \tpkt
\end{equation*}
In particular we are interested in the 
(innermost) runtime complexity with respect to $\rsrew{\RS}$ ($\irew{\RS}$)
on the set $\TB$ of all \emph{basic} terms.%
\footnote{We can replace $\TB$ by the set of terms $f(\seq{t})$ with $f\in \DS$,
whose arguments $t_i$ are in normal form, while keeping all results in this paper.} 
More precisely, the \emph{runtime complexity function} 
(with respect to $\RS$) is defined as 
$\Rc{\RS}(n) \defsym \Rc{}(n, {\TB}, \rsrew{\RS})$ and we 
define the \emph{innermost runtime complexity function} 
as $\Rci{\RS}(n) \defsym \Rc{}(n, {\TB}, \irew{\RS})$.
Note that the \emph{derivational complexity function} (with respect to $\RS$) 
becomes definable as follows: 
$\Dc{\RS}(n) \defsym \Rc{}(n, \TA, \rsrew{\RS})$, where
$\TA$ denotes the set of \emph{all} terms $\TA(\FS,\VS)$, compare~\cite{HofbauerLautemann:1989}.
We sometimes say the (innermost) runtime complexity of $\RR$ is \emph{linear},
\emph{quadratic}, or \emph{polynomial} if $\rc_\RR^{(\mathrm{i})}$ 
is bounded 
by a linear, quadratic, or polynomial function in $n$, respectively.

A \emph{proper order} is a transitive and irreflexive relation and
a \emph{preorder} is a transitive and reflexive relation. 
A proper order $\succ$ is \emph{well-founded} if there is 
no infinite decreasing sequence $t_1 \succ t_2 \succ t_3 \cdots$.
An $\FS$-\emph{algebra} $\A$ consists of a carrier set $A$ and 
an interpretation
$f_\A$ for each function symbol in $\FS$. 
A \emph{well-founded} and \emph{monotone} algebra (\emph{WMA} for short) 
is a pair $(\A,>)$, where $\A$ is an algebra and $>$ is 
a well-founded proper order on $A$ such that every $f_\A$ is monotone (with respect to
$>$) in all arguments.
An \emph{assignment} $\alpha \colon \VS \to A$ is a function mapping
variables to elements in the carrier,
and $\eval{\alpha}{\A}(\cdot)$ denotes the usual evaluation function
associated with $\A$.
A WMA naturally induces a proper order $\gord{\A}$ on terms: 
$s \gord{\A} t$ if $\eval{\alpha}{\A}(s) > \eval{\alpha}{\A}(t)$
for all assignments $\alpha \colon \VS \to A$. 
For the reflexive closure $\geqslant$ of $>$, the preorder $\geqord{\A}$ 
is similarly defined.
Clearly the proper order $\gord{\A}$ is a reduction order, i.e., 
if ${\RS} \subseteq {\gord{\A}}$, for a TRS $\RS$, then
we can conclude termination of $\RS$.
A \emph{rewrite preorder} is a preorder on terms 
which is closed under contexts and substitutions. 
A \emph{reduction pair} $(\gtrsim, \succ)$ consists of a
rewrite preorder $\gtrsim$ and a compatible well-founded order $\succ$
which is closed under substitutions.  Here compatibility means 
the inclusion ${\gtrsim \cdot \succ \cdot \gtrsim} \subseteq {\succ}$. 
Note that for any WMA $\A$ the pair $(\geqord{\A},\gord{\A})$
constitutes a reduction pair. 

We call a WMA $\A$ based on 
the natural numbers $\NN$
a \emph{polynomial interpretation},
if all functions $f_\A$ are polynomials.
A polynomial $P(x_1,\dots,x_n)$ (over the natural numbers) is called
\emph{strongly linear} if $P(x_1,\dots,x_n) = x_1 + \cdots + x_n + c$
where $c \in \NN$. A polynomial interpretation is called \emph{linear restricted} 
if all constructor symbols are interpreted by strongly linear polynomials
and all other function symbols by linear polynomials. If on the other
hand the non-constructor symbols are interpreted by quadratic polynomials,
the polynomial interpretation is called \emph{quadratic restricted}.
Here a polynomial is \emph{quadratic} if it is a sum of monomials of 
degree at most $2$ (see~\cite{contejean:2005}).
It is easy to see that if a TRS $\RS$ is compatible with a 
linear or quadratic restricted interpretation, 
the runtime complexity of $\RS$ is linear or quadratic, 
respectively (see~\cite{HM:2008} but also~\cite{BCMT:2001}).

Finally, we introduce a very restrictive class of polynomial interpretations:
\emph{strongly linear interpretations} (\emph{SLI} for short). A polynomial
interpretation is called \emph{strongly linear} if all functions 
$f$ are interpreted as strongly linear polynomials. 

\section{Complexity Analysis Based on the Dependency Pair Method} 
\label{DPMethod}

In this section, we recall central definitions
and results established in~\cite{HM:2008} (see also Chapter~\ref{ijcar08a}).
We kindly refer the reader to~\cite{HM:2008} for additional
examples and underlying intuitions.

We write $\SC{t_1, \ldots, t_n}_X$ to denote $C[t_1,\ldots,t_n]$,
whenever $\rt(t_i) \in X$ for all $1 \leqslant i \leqslant n$
and $C$ is an $n$-hole context containing no $X$-symbols.
Let $t$ be a term. We set $t^\sharp \defsym t$ if $t \in \VV$, and 
$t^\sharp \defsym f^\sharp(t_1,\dots,t_n)$ if $t = f(\seq{t})$.
Here $f^\sharp$ is a new $n$-ary function symbol called 
\emph{dependency pair symbol}. For a signature
$\FS$, we define $\FS^\sharp = \FS \cup \{f^\sharp \mid f\in \FS\}$.
\begin{definition} \label{d:lpar08:WDP}
Let $\RR$ be a TRS. 
If $l \rew r \in \RS$ and $r = \SC{\seq{u}}_{\DD \cup \VV}$ then 
the rewrite rule $l^\sharp \to \COM(u_1^\sharp,\ldots,u_n^\sharp)$ is called
a \emph{weak dependency pair} of $\RS$. 
Here $\COM$ is defined with a fresh $n$-ary function symbol $c$ 
(corresponding to $l \to r$) as follows: $\COM(\seq{t})$ is $t_1$ if $n = 1$, 
and $c(t_1,\ldots,t_n)$ otherwise.  
The symbol $c$ is called \emph{compound symbol}.
The set of all weak dependency pairs is denoted by $\WDP(\RS)$.
\end{definition}

\begin{example}[continued from Example~\ref{ex:lpar08:1}] \label{ex:pre2}
The set $\WDP(\RR)$ consists of the next 6 weak dependency pairs. 
\begin{alignat*}{4}
7\colon && \app^\sharp(\nil,y) & \rew y & 
10\colon && \reverse^\sharp(n \cons x) & \rew \app^\sharp(\reverse(x),n \cons \nil)\\
8\colon && \app^\sharp(n \cons x,y) &  \rew \m{c}(n,\app^\sharp(x,y)) & \hspace{5ex}
11\colon && \shuffle^\sharp(\nil) & \rew \m{e}\\
9\colon && \reverse^\sharp(\nil) & \rew \m{d} &
12\colon && \shuffle^\sharp(n\cons x) & \rew \m{f}(n,\shuffle^\sharp(\reverse(x)))
\end{alignat*}
\end{example}

\begin{definition}
\label{d:WIDP}
Let $\RR$ be a TRS. 
If $l \rew r \in \RS$ and $r = \SC{\seq{u}}_\DD$ then 
the rewrite rule $l^\sharp \to \COM(u_1^\sharp,\ldots,u_n^\sharp)$
is called a \emph{weak innermost dependency pair} of $\RS$.
The set of all weak innermost dependency pairs is denoted by $\WIDP(\RS)$.
\end{definition}

Definitions~\ref{d:lpar08:WDP} and~\ref{d:WIDP} should be compared to
the definition of ``standard'' dependency pairs.
\begin{definition}[\cite{ArtsGiesl:2000}]
\label{d:lpar08:DP}
The set $\DP(\RS)$ of (standard) \emph{dependency pairs} of a TRS $\RS$ 
is defined as 
$\{ l^\sharp \to u^{\sharp} \mid l \to r \in \RS, 
u \subterm r, \rt(u) \in \DD \}$.
\end{definition}
\begin{example}[continued from Example~\ref{ex:pre2}]
As already mentioned in the introduction, the TRS $\RS$ admits 5
(standard) dependency pairs. Note that
the sets $\DP(\RS)$ and $\WDP(\RS)$ are incomparable. For example
${\app^\sharp(\nil,y) \rew y} \in {\WDP(\RS) \setminus \DP(\RS)}$
while 
${\shuffle^\sharp(x)) \to  \reverse^\sharp(x)} \in {\DP(\RS) \setminus \WDP(\RS)}$.
\end{example}

We write $f \depends g$ if there exists a rewrite rule
$l \to r \in \RR$ such that $f = \rt(l)$ and $g$ is a defined 
symbol in $\Fun(r)$. 
For a set $\GG$ of defined symbols 
we denote by $\RR{\restriction}\GG$ the set of
rewrite rules $l \to r \in \RR$ with $\rt(l) \in \GG$. 
The set $\UU(t)$ of usable rules of a term $t$ is defined as
$\RR{\restriction}\{ g \mid \text{$f \depends^* g$ for some $f \in \Fun(t)$} \}$.
Finally, if $\PP$ is a set of (weak or weak innermost) dependency pairs
then $\UU(\PP) = \bigcup_{l \to r \in \PP} \UU(r)$.

\begin{proposition}[Chapter~\ref{ijcar08a}, Theorem~\ref{t:dp:usable}] \label{p:dp:usable}
Let $\RR$ be a TRS and let $t \in \TB$. 
If $t$ is terminating with respect to $\rew$ then
$\dheight(t, \rew) \leqslant \dheight(t^{\sharp},\rsrew{\UU(\PP)\,\cup \, \PP})$, 
where $\rew$ denotes $\rsrew{\RS}$ or $\irew{\RS}$
depending on whether $\PP = \WDP(\RS)$ or $\PP = \WIDP(\RS)$.
\end{proposition}

We recall the notion of \emph{relative rewriting}~\cite{Terese}.
Let $\RR$ and $\SS$ be TRSs.
We write $\to_{\RR / \SS}$ for $\to_\SS^* \cdot \to_\RR \cdot \to_\SS^*$ and
we call $\to_{\RR / \SS}$ the \emph{relative rewrite relation} of $\RR$ over $\SS$. (Note that ${\to_{\RR / \SS}} = {\rsrew{\RS}}$, if $\SS = \varnothing$.)
Let $\A$ denote a strongly linear interpretation. 
%
\begin{proposition}[Chapter~\ref{ijcar08a}, Theorem~\ref{t:relative}] \label{p:relative}
Let $\RR$ and $\SS$ be TRSs, $\AA$ an SLI compatible with $\SS$, and
$\RR$ non-duplicating.
There exist constants $K$ and $L$, depending only on $\RS$ and $\AA$, such that
$\dheight(t, \to_{\RR \cup \SS}) \leqslant 
K \cdot \dheight(t,\to_{\RR/\SS}) + L \cdot \size{t}$
for all terminating terms $t$ on $\RS \cup \SS$.
\end{proposition}

We need some further definitions.
Let $\RS$ be a TRS, let $\PP$ be a set of 
weak or weak innermost dependency pairs of $\RS$
and let $\Slow$ denote a mapping associating a term (over $\FS^\sharp$ and $\VS$) 
and a proper order $\succ$ with a natural number.
An order $\succ$ on terms is $\Slow$-\emph{collapsible} 
for a TRS $\RR$ 
if $s \rsrew{\PP \cup \UU(\PP)} t$  and $s \succ t$ implies $\Slow(s,\succ) > \Slow(t,\succ)$. 
An order $\succ$ is \emph{collapsible} for a TRS $\RS$, 
if there is a mapping $\Slow$ such that $\succ$ is $\Slow$-collapsible 
for $\RR$.%
\footnote{Note that most reduction orders are collapsible. 
E.g.~if $\A$ is a polynomial interpretation then $\gord{\A}$ is collapsible, 
as one may take any $\alpha$ and set 
$\Slow(t, \gord{\A}) \defsym \eval{\alpha}{\A}(t)$.}

We write $\TBS$ for $\{t^{\sharp} \mid t \in \TB \}$. 
The set $\TBC$ is inductively defined as follows 
(i) $\TTs \cup \TT \subseteq \TBC$, 
where $\TTs = \{t^{\sharp} \mid t \in \TT \}$ and (ii) 
$c(t_1,\ldots,t_n) \in \TBC$, whenever $t_1,\ldots,t_n \in \TBC$ and $c$ a compound symbol.
A proper order $\succ$ on $\TBC$ is called \emph{safe} if
$c(s_1,\ldots,s_i,\ldots,s_n) \succ c(s_1,\ldots,t,\ldots,s_n)$
for all $n$-ary compound symbols $c$ and all terms $s_1,\ldots,s_n,t$ 
with $s_i \succ t$.
A reduction pair $({\gtrsim}, {\succ})$ is called \emph{collapsible} 
for a TRS $\RR$ if $\succ$ is collapsible for $\RR$. It is called \emph{safe}
if the well-founded order $\succ$ is safe. In order to construct safe reduction pairs 
one may use \emph{safe algebras}, i.e., weakly monotone well-founded 
algebras $(\AA, \succ)$ such that the interpretations of compound symbols
are strictly monotone with respect to $\succ$. It is easy to see that
if $(\AA, >)$ is a safe algebra then $(\geqord{\A},\gord{\A})$ is a safe reduction pair.
%
\begin{proposition}[Chapter~\ref{ijcar08a}, Theorem~\ref{t:ijcar08a:main}] \label{p:main}
Let $\RR$ be a TRS, let $\AA$ be an SLI, let $\PP$ be a set of weak or weak innermost dependency pairs, 
such that $\PP$ is non-duplicating, and let $(\gtrsim, \succ)$ be a safe and 
$\Slow$-collapsible reduction pair such that $\UU(\PP) \subseteq {\gtrsim}$ and
$\PP \subseteq {\succ}$. 
If in addition $\UU(\PP) \subseteq {>_\A}$ then 
for any $t \in \TB$, there exist constants $K$ and $L$
(depending only on $\RS$ and $\A$) such that 
$\dheight(t, \rew) \leqslant K \cdot \Slow(t^{\sharp},\succ) + L \cdot \size{t}$.
Here $\rew$ denotes $\rsrew{\RS}$ or $\irew{\RS}$
depending on whether $\PP = \WDP(\RS)$ or $\PP = \WIDP(\RS)$.
\end{proposition}
Suppose the assertions of the proposition are met and there exists a polynomial
$p$ 
such that $\Slow(t^{\sharp},\succ) \leqslant p(\size{t})$ holds. 
Then, as an easy corollary to Proposition~\ref{p:main}, 
we observe that the runtime
complexity induced by $\RS$ is majorised by $p$.

\section{Dependency Graphs} \label{DG}

In this section, we study a natural refinement of the dependency pair
method, namely \emph{dependency graphs} (see~\cite{ArtsGiesl:2000,GAO:2002,Giesl:2005frocos,HirokawaMiddeldorp:2007}) in
the context of complexity analysis. 
We start with a brief motivation.
Let $\RS$ be a TRS, let $\PP$ denote a set of weak or weak innermost dependency
pairs and let $(s_i)_{i=0,\ldots,n}$ denote a maximal 
derivation $D$ with respect to $\RS$ with $s_0 \in \TB$.
In order to estimate the length $\ell$ of this derivation 
it suffices to estimate the length of the 
derivation 
$t_0 \rssrew{\UU(\PP) \cup \PP} t_n$, 
where $t_0 = s_0^\sharp \in \TBS$, cf.~Proposition~\ref{p:dp:usable}.
If we suppose that $\PP$ is non-duplication and 
that there exists an SLI such that ${\UU(\PP)} \subseteq {\AA}$, 
we may estimate the derivation length $\ell$ by finding 
\emph{one} (safe and collapsible) reduction pair $(\gtrsim, \succ)$ 
such that $\UU(\PP) \subseteq {\gtrsim}$ and $\PP \subseteq {\succ}$ holds, 
cf.~Proposition~\ref{p:main}. 
On the other hand in termination analysis---as already mentioned
in the introduction---it suffices to 
guarantee that for any cycle $\CC$ in the dependency graph $\DG(\RS)$, 
there are no $\CC$-minimal rewrite sequences, cf.~\cite{GAO:2002}.
Hence, we strive to extend this idea to complexity analysis. 

\subsection{From Cycle Analysis to Path Detection}
Let us recall the definition of a dependency graph and extend it suitably
to weak and weak innermost dependency pairs. 
\begin{definition} \label{d:DG}
Let $\RS$ be a TRS over a signature $\FS$ and 
let $\PP$ be the set of weak, weak
innermost, or (standard) dependency pairs. 
The nodes of the \emph{weak dependency graph} $\WDG(\RS)$, 
\emph{weak innermost dependency graph} $\WIDG(\RS)$, 
or \emph{dependency graph} $\DG(\RS)$ are the elements of $\PP$ 
and there is an arrow from $s \to t$ to $u \to v$ if and
only if there exist a context $C$ and substitutions 
$\sigma, \tau \colon \VV \to \TT(\FS, \VV)$ such that
$t\sigma \rew^* C[u\tau]$, where $\rew$ denotes 
$\rsrew{\RS}$ or $\irew{\RS}$
depending on whether $\PP = \WDP(\RS)$, $\PP = \DP(\RS)$ or $\PP = \WIDP(\RS)$,
respectively.
\end{definition}

\begin{example}[continued from Example~\ref{ex:pre2}] \label{ex:lpar08:2}
The weak dependency graph $\WDG(\RR)$ has the following form.  
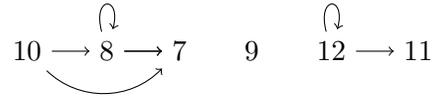
\begin{center}
\begin{tikzpicture}[node distance=5mm]
    \node(10) {10} ;
    \node(8) [right=of 10] {8} ;
    \node(7) [right=of 8] {7} ;
    \node(9) [right=of 7 ] {9} ;
    \node(12) [right=of 9 ] {12} ;
    \node(11) [right=of 12 ] {11} ;

    \path[->] (8) edge [loop above] (8) ;
    \draw[->] (8) to (7);
    \path[->] (12) edge [loop above] (12) ;
    \draw[->] (8) to (7) ;
    \draw[->] (12) to (11) ;
    \path[->] (10) edge [bend right=45] (7) ;
    \draw[->] (10) to (8);
\end{tikzpicture}
\end{center}
\end{example}

We recall a theorem on 
the dependency graph refinement in conjunction with 
usable rules and innermost rewriting (see~\cite{GAO:2002}, but also~\cite{HirokawaMiddeldorp:2005}).
Similar results hold in the
context of full rewriting, see~\cite{GTSF06,HirokawaMiddeldorp:2007}.
\begin{theorem}[\cite{GAO:2002}]
\label{t:GAO02}
A TRS $\RS$ is innermost terminating if for every maximal cycle 
$\CC$ in the dependency graph $\DG(\RS)$
there exists a reduction pair $(\gtrsim,\succ)$ such that
${\UU(\CC)} \subseteq {\gtrsim}$ and ${\CC} \subseteq {\succ}$.
\end{theorem}

The following example shows that we cannot directly employ Theorem~\ref{t:GAO02}
in the realm of complexity analysis. 
Even though in this setting we 
can restrict our attention to a specific strategy: innermost rewriting.
\begin{example} \label{ex:exp}
Consider the TRS $\RSexp$
\begin{align*}
\m{exp}(\mN) & \to \ms(\mN) & \m{d}(\mN) & \to \mN \\
\m{exp}(\m{r}(x)) & \to \m{d}(\m{exp}(x)) 
& \m{d}(\ms(x)) & \to \ms(\ms(\m{d}(x))) 
\end{align*}
$\DP(\RSexp)$ consists of three pairs:
$1\colon \m{exp}^\sharp(\m{r}(x)) \to \m{d}^\sharp(\m{exp}(x))$,
$2\colon \m{exp}^\sharp(\m{r}(x)) \to \m{exp}^\sharp(x)$, and
$3\colon \m{d}^\sharp(\ms(x)) \to \m{d}^\sharp(x)$.
Hence the dependency graph $\DG(\RSexp)$ contains two 
maximal cycles:
$\{2\}$ and $\{3\}$.
It is easy to see how to define two reduction pairs 
$(\geqord{\AA},\gord{\AA})$ and $(\geqord{\BB},\gord{\BB})$ 
such that the conditions of the theorem are fulfilled.
For that it suffices to define interpretations
$\AA$ and $\BB$, respectively. 
Because one can find suitable linear restricted ones for $\AA$ and $\BB$,
compatibility with these interpretations apparently
induces \emph{linear} runtime complexity of $\RSexp$,
cf.~\cite{BCMT:2001,HM:2008} (even for full rewriting). 
However, we must not conclude linear innermost runtime complexity for 
$\RSexp$ 
in this setting, as $\RSexp$ formalises the exponentiation function and
setting $t_n = \m{exp}(\m{r}^n(\mN))$ we obtain 
$\dheight(t_n, \irew{\RS}) \geqslant 2^n$ for each $n \geqslant 0$. 
Thus the innermost runtime complexity of $\RSexp$ is exponential.
\end{example}
Note that the problem exemplified by Example~\ref{ex:exp}
cannot be circumvented by replacing the dependency graph employed in Theorem~\ref{t:GAO02} with 
the weak (innermost) dependency graph. 
Furthermore, observe that while Proposition~\ref{p:dp:usable}
allows us to replace in Example~\ref{ex:exp} the innermost rewrite relation 
$\irew{\RS}$ by the 
(sometimes simpler) rewrite relation 
$\irew{\UU(\DP(\RS)) \cup \DP(\RS)}$, this is of no help:
The exponential length of $t_n^\sharp$ in Example~\ref{ex:exp} with respect to 
$\UU(\DP(\RR)) \cup \DP(\RR)$ is \emph{not} due to the cycles $\{2\}$ or $\{3\}$, but
achieved through the non-cyclic pair $1$ and its usable rules. 
These observations are cast into Definition~\ref{d:1}, below.

A graph is called \emph{strongly connected} if any node is connected with every other
node by a path. A \emph{strongly connected component} (\emph{SCC} for short) is a
maximal strongly connected subgraph.%
\footnote{Note that in the literature SCCs are sometimes defined as \emph{maximal cycles}. This
alternative definition is of limited use in our context as we must not ignore
trivial SCCs.}
\begin{definition} \label{d:1}
Let $\GG$ be a graph, let $\equiv$ denote the equivalence relation induced by SCCs, 
and let $\PP$ be an SCC in $\GG$. 
The set of all source nodes in $\PG{\GG}$ is denoted by $\Src$.
Let $l \to r$ be a dependency pair in $\GG$, let $\KK \in \PG{\GG}$ and
let $\CC$ denote the SCC represented by $\KK$. 
Then we write $l \to r \in \KK$ if $l \to r \in \CC$.
\end{definition}

\begin{example}[Continued from Example \ref{ex:lpar08:2}] \label{ex:lpar08:3}
There are 6 (trivial) SCCs in $\WDG(\RS)$,all being trivial. 
Hence the graph $\PG{\WDG(\RR)}$ has the 
following form:
\begin{center}
\begin{tikzpicture}[node distance=5mm]
    \node(10) {10} ;
    \node(8) [right=of 10] {8} ;
    \node(7) [right=of 8] {7} ;
    \node(9) [right=of 7 ] {9} ;
    \node(12) [right=of 9 ] {12} ;
    \node(11) [right=of 12 ] {11} ;

  \draw[->] (8) to (7) ;
  \draw[->] (12) to (11) ;
  \path[->] (10) edge [bend right=45] (7) ;
  \draw[->] (10) to (8);
\end{tikzpicture}
\end{center}
Here $\Src = \{ \{ 9 \}, \{ 10 \}, \{12\} \}$.
\end{example}

\subsection{Refinement Based on Path Detection}

We re-consider the motivating derivation $D$:
\begin{equation} 
\label{eq:dg:1}
  t_0 \rsrew{\UU(\PP) \cup \PP} t_1 \rsrew{\UU(\PP) \cup \PP} \dots 
  \rsrew{\UU(\PP) \cup \PP} t_n \tkom
\end{equation}
where $t_0 \in \TBS$. To simplify the exposition, we set $\PP = \WDP(\RS)$ 
and $\GG = \WDG(\RS)$. Momentarily we assume that all compound symbol are of arity $0$,
as is for instance the case in Example~\ref{ex:pre2}.
Above we asserted that there exists an SLI $\A$ such that ${\UU(\PP)} \subseteq {\gord{\A}}$. 
Hence Proposition~\ref{p:relative} is applicable. 
Thus, to estimate the length of the derivation~\eqref{eq:dg:1} 
it suffices to consider the following relative rewriting derivation:
\begin{equation} \label{eq:dg:2}
  t_0 \rsrew{\PP/\UU(\PP)} t_1 \rsrew{\PP/\UU(\PP)} \dots 
  \rsrew{\PP/\UU(\PP)} t_n \tpkt
\end{equation}
Exploiting the given assumptions, it is not difficult to see that derivation~\eqref{eq:dg:2} is representable 
as follows:
\begin{equation} \label{eq:dg:3}
t_0 \to_{\PP_1/\UU(\PP_1)}^{\ell_1} 
t_{\ell_1} \to_{\PP_2/\UU(\PP_1) \cup \UU(\PP_2)}^{\ell_2} 
\cdots \to_{\PP_m/\UU(\PP_1) \cup \cdots \cup \UU(\PP_m)}^{\ell_m} t_n \tkom
\end{equation}
where, $(\PP_1,\ldots,\PP_m)$ is a \emph{path} in $\PG{\GG}$ with $\PP_1 \in \Src$
and $\ell_i \geqslant 0$ ($i=1,\dots,n$).
Since the length $\ell$ of the pictured $\rsrew{\PP/\UU(\PP)}$-rewrite
sequence equals $\ell_1 + \cdots + \ell_m$, this suggests that we can 
estimate each $\ell_j$ ($j \in \{1,\dots,m\}$) independently.
We assume the existence of a family of SLIs $\BB_j$ 
($j \in \{1,\dots,m\}$) such that 
${\UU(\PP_1) \cup \dots \cup \UU(\PP_j)} \subseteq {\geqord{\BB_j}}$
and 
${\PP_j} \subseteq {\gord{\BB_j}}$
holds for every $j$. From this we can conclude 
$\ell_j = \bO(\size{t_{\ell_j}})$ for all $j \in \{1,\dots, m\}$.
The next step is to estimate each $\ell_j$ by a function 
(preferable a polynomial) in $\size{t_0}$. As each of the WMAs
$\BB_j$ is assumed to be strongly linear, we can even 
conclude $\eval{\alpha_0}{\BB_j}(t_{\ell_j}) = \Omega(\size{t_{\ell_j}})$.%
%
%
(Here $\alpha_0$ denotes the assignment mapping any variable to $0$.)
In sum, we obtain for each $j \in \{1,\dots,m\}$, the existence of a
constant $c_j$ such that  $\size{t_{\ell_j}}\leqslant c_j \cdot \size{t_0}$
and thus there exists a linear polynomial $p(x)$ such that
$\ell_j \leqslant p(\size{t_0})$. 
However, some care is necessary in assessing this observation:
Note that
the given argument cannot be used to deduce \emph{polynomial} runtime
complexity, if we weaken the assumption that the algebras 
$\BB_j$ are strongly linear only slightly. 
Hence, we replace the direct application of Proposition~\ref{p:relative} as follows.
\begin{lemma}\label{l:shift}
$n \leqslant \dheight(s, \rsrew{\RR_2/(\SS_1 \cup \SS_2)})$
whenever $s \rsrew{\SS_1}^* \cdot \rsrew{\RR_2/\SS_2}^n u$.
\end{lemma}
\begin{proof}
  Straightforward.
\end{proof}

We lift the assumption that all compound symbols are of arity at most $0$. 
Perhaps surprisingly this generalisation complicates the matter considerably.
First a \emph{maximal} derivation need no longer be of the form given in~\eqref{eq:dg:3}
which is exemplified by Example~\ref{ex:dg:1} below.
\begin{example} \label{ex:dg:1}
Consider the TRS 
$\RR = \{\m{f}(\m{0}) \to \m{leaf}, \m{f}(\m{s}(x)) \to \m{branch}(\m{f}(x), \m{f}(x))\}$.
The set $\WDP(\RR)$ consists of the two weak dependency pairs:
$1\colon \m{f}^\sharp(\m{0}) \to \m{c_1}$ and
$2\colon \m{f}^\sharp(\m{s}(x)) \to \m{c_2}(\m{f}^\sharp(x), \m{f}^\sharp(x))$.
Hence the weak dependency graph $\WDG(\RR)$
contains 2 SCCs: $\{2\}$ and $\{1\}$.  
Clearly $\Src = \{ \{ 2 \} \}$.
Let $t_n = \m{f}^\sharp(\m{s}^n(\m{0}))$.
Consider the following sequence:
\begin{align*}
t_2 & \rsrew{\{2\}}^2 \m{c_2}(\m{c_2}(t_0, t_0), t_1) 
\rsrew{\{1\}}   \m{c_2}(\m{c_2}(\m{c_1}, t_0), t_1) \\
    & \rsrew{\{2\}}   \m{c_2}(\m{c_2}(\m{c_1}, t_0), \m{c_2}(t_0, t_0)) 
\rsrew{\{1\}}^3 \m{c_2}(\m{c_2}(\m{c_1}, \m{c_1}), \m{c_2}(\m{c_1},\m{c_1})) \tpkt
\end{align*}
This derivation does not have the form \eqref{eq:dg:3}, 
because it is based on the sequence 
$(\{2\},\{1\},\{2\},\{1\})$, which is not a path in $\PG{\WDG(\RS)}$.
\end{example}

Note that the derivation in Example~\ref{ex:dg:1} can be
reordered (without affecting its length) such that the derivation becomes
based on a path. Still, not every derivation can be abstracted to a path. 
Consider a maximal (with respect to subset inclusion) component of $\PG{\WDG(\RS)}$. 
Clearly this component forms a directed acyclic graph $\GG$, and without loss of
generality we assume in the following that $\GG$ is a tree $T$ with root in $\Src$. 
Otherwise, observe that any directed acyclic graph $\GG$ can be unfolded to a forest $\FF$
and that the size of $\GG$ is bounded in the size of $\FF$. Moreover if $T \in \FF$ is
of maximal depth, then the size of $\FF$ is linearly bounded in the size of $T$ as the
number of trees in $\FF$ depends only on $\RS$. 
Suppose further that $T$ is not degenerated to a branch. Then a given derivation may only be 
abstractable by different paths in $T$, as exemplified by Example~\ref{ex:dg:2}.
\begin{example} \label{ex:dg:2}
Consider the TRS $\RS = \{\m{f} \to \m{c}(\m{g},\m{h}), \m{g} \to \m{a}, \m{h} \to \m{a}\}$.
Thus $\WDP(\RR)$ consists of three dependency pairs:
$1\colon \m{f}^\sharp \to \m{c_1}(\m{g}^\sharp,\m{h}^\sharp)$, 
$2\colon \m{g}^\sharp \to \m{c_2}$, and
$3\colon \m{h}^\sharp \to \m{c_3}$.
Let $\PP \defsym \WDP(\RR)$, then 
clearly $\PP = \WDG(\RR) = \PG{\WDG(\RR)}$. Consider the following
derivation
\begin{equation*}
\m{f}^\sharp 
\rsrew{\PP} \m{c_1}(\m{g}^\sharp, \m{h}^\sharp)
\rsrew{\PP} \m{c_1}(\m{c_2}, \m{h}^\sharp)
\rsrew{\PP} \m{c_1}(\m{c_2}, \m{c_3}) \tpkt
\end{equation*}
This derivation is composed from the paths $(\{1\},\{2\})$ and 
$(\{1\},\{3\})$.
\end{example}

Fortunately, we can circumvent these obstacles. 
Let $\PP$ denote the set of weak or weak innermost dependency pairs
of a TRS $\RS$. We make the following easy observation.
\begin{lemma} \label{l:1}
Let $\GG$ denote a weak or weak innermost dependency graph. 
Let ${\CC} \subseteq {\GG}$ and let $D \colon s \rssrew{\CC/\UU(\PP)} t$ denote a
derivation based on $\CC$ with $s \in \TBC$. 
Then $D$ has the following form:
$s = s_{0} \rsrew{\CC/\UU(\PP)} s_{1} \rsrew{\CC/\UU(\PP)} \dots \rsrew{\CC/\UU(\PP)} s_{n} = t$
where each $s_i \in \TBC$. 
\end{lemma}
\begin{proof}
It is easy to see that $D$ has the presented form and that 
for each $i \in \{0,\dots,n\}$ there exists a context $C$
such that $s_i = C[u^\sharp_1,\dots,u^\sharp_r]$ and $C$ consists of compound
symbols only. This establishes the lemma. 
\end{proof}

Motivated by Example~\ref{ex:dg:1} we observe 
that a weak (innermost) dependency pair containing an $m$-ary ($m > 1$) 
compound symbol can only induces $m$ \emph{independent} derivations. 
Hence, we can reorder derivations to achieve the structure of derivation~\eqref{eq:dg:3}. 
This is formally proven via the next two lemmas.

\begin{lemma} \label{l:2}
Let $\GG$ denote a weak or weak innermost dependency graph and
let $\KK$ and $\LL$ denote two different nodes in $\PG{\GG}$ such that
there is no edge from $\KK$ to $\LL$. Let $s_0 \in \TBC$ and suppose the existence of a derivation $D$ of the following form:
$s_{0} \rew_{\KK/\UU(\PP)}^{n} s_{n} \rssrew{\UU(\PP)} t_{0} \rew_{\LL/\UU(\PP)}^{m} t_m$.
Then there exists a derivation $D'$ which has the form
$t'_{0} \rew_{\LL/\UU(\PP)}^{m} t'_{m} \rssrew{\UU(\PP)} s'_{0} \rew_{\KK/\UU(\PP)}^{n} s'_n$
with $t'_0 \in \TBC$.
\end{lemma}

\begin{proof}
Consider the following two dependency pairs:
$1\colon u_k^\sharp \to \COM(v_{k1}^\sharp,\ldots,v_{kr}^\sharp)$
and $2\colon u_l^\sharp \to \COM(v_{l1}^\sharp,\ldots,u_{lr}^\sharp)$.
Here the dependency pair $1$ belongs to $\KK$ and denotes the last 
dependency pair employed in $D$ before the path leaves $\KK$ into $\LL$,
while $2$ denotes the first pair in $\LL$. The assumption that there is no edge
from $\KK$ to $\LL$ can be reformulated as follows:
\begin{description}
\item[$(\dagger)$] No context $C$ and no substitutions 
$\sigma, \tau \colon \VS \to \TA(\FS,\VS)$ exist such that 
\begin{equation*}
  \COM(v_{k1}^\sharp\sigma,\ldots,v_{kr}^\sharp\sigma) \rssrew{\UU(\PP)} C[u_l^\sharp\tau] \tkom
\end{equation*}
holds.
\end{description}  
To prove the lemma, we proceed by induction on $n$. 
It suffices to consider the step case $n>1$. By assumption 
the last rewrite step in the subderivation $D_0 \colon s_0 \rew_{\KK/\UU(\PP)}^{n} s_{n}$
employs dependency pair $1$. Let $p \in \Pos(s_n)$
denote the position of the reduct 
$\COM(v_{k1}^\sharp\tau,\ldots,v_{kr}^\sharp\tau)$ in $s_{n}$.
By assumption there exists a derivation $s_n \rssrew{\UU(\PP)} t_{0}$. Let
$q \in \Pos(s_n)$ denote the position of the redex in  $s_{n}$ that is
contracted as first step in this reduction. 
Without loss of generality we can assume that both positions are parallel 
to each other.
Otherwise one of the following cases applies. Either $p < q$ or $p \geqslant q$. 
But clearly the first case contradicts the assumption $(\dagger)$. Hence,
assume the second. 
But this is also impossible. Lemma~\ref{l:1} yields
that $\atpos{s_n}{q} \in \TBC$, which contradicts that $q$ is redex with respect
to $\UU(\PP)$.
Repeating this argument we see that position $p$ has exactly one descendant
in $t_{0}$. A similar argument shows that all redex positions in the
subderivation $D_1 \colon t_0 \rew_{\LL/\UU(\PP)}^{m} t_{m}$ are parallel to (descendants of) 
$p$. Hence, we can move the last rewrite step $s_{n-1} \rsrew{\KK} s_{n}$ 
in the derivation $D_0$ after the derivation $D_1$. Note that in each of
the terms $(t_i)_{i=1,\dots,m}$ the position $p$ exists and denotes the
term $\COM(v_{k1}^\sharp\tau,\ldots,v_{kr}^\sharp\tau)$. 
Hence, the replacement of $\COM(v_{k1}^\sharp\tau,\ldots,v_{kr}^\sharp\tau)$ 
everywhere by $u_k^\sharp\sigma$ does not affect the validity 
of the rewrite sequence. Furthermore the set $\TBC$ is
closed under this operation.
Now, the induction hypothesis becomes applicable to 
derive the existence of the sought derivation $D'$.
\end{proof}

Let $\GG$ denote a weak or weak innermost dependency graph and 
let $D \colon s \rew^\ell t$ denote a derivation, such that $s \in \TBS$.
Here $\rew$ denotes either $\rsrew{\PP/\UU(\PP)}$ or $\irew{\PP/\UU(\PP)}$.
We say that $D$ is \emph{based on $(\PP_1,\ldots,\PP_m)$} 
in $\PG{\GG}$ if $D$ is of the form
\begin{equation*}
s~\parenito_{\PP_1/\UU(\PP)}^{\ell_1} \cdots 
  \parenito_{\PP_m/\UU(\PP)}^{\ell_m}~t \tkom  
\end{equation*}
with $\ell_1,\ldots,\ell_m \geqslant 0$. 
We arrive at the main lemma of this section.

\begin{lemma} \label{l:lpar08:3}
Let $\PP$ denote a set of weak or weak innermost dependency pairs,
let $s \in \TBS$ and let $D \colon s \rew^\ell t$ denote a maximal derivation, 
where $\rew$ denotes $\rsrew{\PP/\UU(\PP)}$ or
$\irew{\PP/\UU(\PP)}$ respectively. Suppose that $D$ is based on 
$(\PP_1,\ldots,\PP_m)$ and $\PP_1 \in \Src$. Then there exists a derivation $D' \colon s \rew^\ell t$
based on $(\PP'_1,\ldots,\PP'_{m'})$, with $\PP_1' \in \Src$ such that 
all $\PP'_i$ ($i \in \{1,\dots,m'\}$) are pairwise distinct.
\end{lemma}
\begin{proof}
Without loss of generality, we restrict our attention to weak dependency pairs.
To prove the lemma, we consider a sequence $(\PP_1,\ldots,\PP_m)$,
where there exist indices $i$, $j$ and $k$ with $i < j < k$ and $\PP_i = \PP_k$.
By induction on $j-i$ we show that this path is transformable into a sequence
$(\PP'_1,\ldots,\PP'_{m'})$ of the required form. It suffices to
prove the step case. Moreover, we can assume without loss of generality 
that $k = j+1$. Consider the two dependency pairs:
$1\colon l_j^\sharp  \to \COM(u_{j1}^\sharp,\ldots,u_{jr}^\sharp)$
and $ 2\colon l_k^\sharp \to \COM(u_{k1}^\sharp,\ldots,u_{kr}^\sharp)$.
Dependency pair $1$ belongs to $\PP_j$ and denotes the last 
dependency pair employed in $D$ before the sequence leaves $\PP_j$ into $\PP_{k}$,
while $2$ denotes the first pair in $\PP_k$. We consider two cases:
\begin{enumerate}
\item Assume there exist a context $C$ and substitutions 
$\sigma,\tau \colon \VS \to \TA(\FS,\VS)$
such that the following holds:
$\COM(u_{j1}^\sharp\sigma,\ldots,u_{jr}^\sharp\sigma) \rew^\ast C[l_k^\sharp\tau]$.
Thus by definition of weak dependency graphs the node in $\WDG(\RS)$ representing 
dependency pair $1$ is connected to the node representing dependency pair $2$. 
In particular every node in the SCCs represented by $\PP_i = \PP_k$ 
is connected to every node in the SCC represented by $\PP_j$. This implies that
$\PP_i = \PP_j = \PP_k$ contradicting the assumption.
\item Otherwise, there is no edge between $\PP_j$ and $\PP_k$ in
the graph $\PG{\GG}$ and by the assumptions on $(\PP_1,\dots,\PP_\ell)$
we find a derivation of the following form:
$D_0 \colon s_{j_1} \rew_{\PP_j/\UU(\PP)}^p s_{j_p} \rssrew{\UU(\PP)} s_{k_1}  
\rew_{\PP_k/\UU(\PP)}^q s_{k_q}$.
Due to Lemma~\ref{l:2} there exists a derivation 
$D_1 \colon s'_{k_1} \rew_{\PP_k/\UU(\PP)}^{q} s'_{k_q} \rssrew{\UU(\PP)} s'_{j_1} 
\rew_{\PP_j/\UU(\PP)}^{p} s'_{j_p}$
so that the number of (weak) dependency pair steps is unchanged. 
The sequence $(\PP_1,\ldots,\PP_{j},\PP_k,\dots,\PP_m)$ is reorderable
into $(\PP_1,\ldots,\PP_{k},\PP_j,\dots,\PP_m)$ without affecting the 
length $\ell$ of the $\rsrew{\PP/\UU(\PP)}$-rewrite sequence. 
By assumption $k=j+1$, hence the induction hypothesis becomes applicable and we conclude
the existence of a path $(\PP'_1,\ldots,\PP'_{m'})$ fulfilling the assertions of the lemma.
\end{enumerate}
\end{proof}

Finally, we arrive at the main contribution of this paper.
\begin{theorem}
\label{t:dg}
Let $\RR$ be a TRS,
let $\PP$ be the set of weak or weak innermost dependency pairs,
let $A$ denote the maximum arity of compound symbols and let
$K$ denote the number of SCCs in the weak (innermost) dependency graph $\GG$.
Suppose $t \in \TBS$ is (innermost) terminating and define
\begin{equation*}
 L(t) \defsym \max\{ \dheight(t, \parenito_{\PP_m/\SS}) \mid 
\text{$(\PP_1,\ldots,\PP_m)$ is a path in $\PG{\GG}$ such that $\PP_1 \in \Src$}\} \tkom
\end{equation*}
where $\SS = \PP_1 \cup \cdots \cup \PP_{m-1} \cup \UU(\PP_1 \cup \cdots \cup \PP_m)$.
Then $\dheight(t, \parenito_{\PP/\UU(\PP)}) \leqslant K^2 \cdot L(t)$.
\end{theorem}
\begin{proof}
Let $(\PP_1,\ldots,\PP_m)$ be a path in $\PG{\PP}$ such that $\PP_1 \in \Src$ and
let $D \colon t \to^\ell u$, denote a maximal derivation based on this path. 
(Here $\rew$ denotes $\rsrew{\PP/\UU(\PP)}$ or $\irew{\PP/\UU(\PP)}$.)
Lemma~\ref{l:lpar08:3} yields that $D$ has the following form:
\begin{equation} \label{eq:dg:4}
t=t_0 \to_{\PP_1/\UU(\PP_1)}^{\ell_1} 
t_{\ell_1} \to_{\PP_2/\UU(\PP_1) \cup \UU(\PP_2)}^{\ell_2} 
\cdots \to_{\PP_m/\UU(\PP_1) \cup \cdots \cup \UU(\PP_m)}^{\ell_m} t_n = u \tkom
\end{equation}
where $t_0 \in \TBS$ and $t_i \in \TBC$ for all $i \geqslant 1$.
It suffices to estimate $\ell_j$ for all $j=1,\dots,m$ suitably.
Let $j$ be arbitrary, but fixed. Consider the subderivation $D'$ of~\eqref{eq:dg:4}
where $m$ is replaced by $j$. Clearly $D'$ is contained 
in the following derivation:
\begin{equation*}
t \to_{\PP_1 \cup \cdots \cup \PP_{i-1} \cup \UU(\PP_1) \cup \cdots \cup \UU(\PP_{j-1})}^{\ast} 
\cdot \to_{\PP_j/\UU(\PP_1) \cup \cdots \cup \UU(\PP_j)}^{\ell_j} t_{\ell_j}  
\end{equation*}
Hence Lemma~\ref{l:shift} is applicable, thus
$\ell_j \leqslant \dheight(t,\rsrew{\PP_j/\PP_1 \cup \cdots \cup \PP_{j-1} \cup 
    \UU(\PP_1) \cup \cdots \cup \UU(\PP_{j})})$.
As 
$\UU(\PP_1) \cup \cdots \cup \UU(\PP_{j}) \subseteq \UU(\PP_1 \cup \cdots \cup \PP_{j})$
we conclude $\ell_j \leqslant L(t)$ and obtain
$\ell = \ell_1 + \ell_2 + \cdots + \ell_m \leqslant K \cdot L(t)$.

Above we argued that any connected component in $\PG{\PP}$ is a tree. Clearly
the number of nodes in this tree is less than $K$. Thus
an arbitrary derivation can at most be based on $K$-many 
different paths. As the length of a derivation $D$ based on a specific path can 
be estimated by $K \cdot L(t)$, we conclude that the length of an arbitrary derivation is
less than $K^2 \cdot L(t)$.
This completes the proof of the theorem.
\end{proof}

Theorem~\ref{t:dg} together with Proposition~\ref{p:relative} 
form a suitable analog of Theorem~\ref{t:GAO02}: Let $\PP$ be the set of weak
or weak innermost dependency pairs. Suppose for every path 
$\WW_i \defsym (\PP_{i1},\ldots,\PP_{im})$ in $\PP$ there exist an SLI $\AA_j$ compatible 
with the usable rules of $\bigcup_{j \in \{i_1,\dots,i_m\}} \PP_j$. Assume 
the existence of a safe and $\Slow$-collapsible reduction pairs
$(\gtrsim_i,\succ_i)$ such that 
$\UU(\bigcup_{j \in \{i_1,\dots,i_m\}} \PP_j) \cup \bigcup_{j \in \{i_1,\dots,i_{m-1}\}} \PP_j$ 
is compatible with $\gtrsim_i$ and 
$\PP_{im}$ compatible with $\succ_i$, such that
$\PP_{im}$ is non-duplicating. Then for any $t \in \TB$ 
the derivation height $\dheight(t,\parenito)$ with respect to 
(innermost) rewriting is majorised by $\Slow(t^{\sharp},\succ_i)$ and $\size{t}$.

\begin{corollary}
\label{c:dg}
Let $\RS$ be a TRS, let $\PP$ be the set of weak (innermost) 
dependency pairs, and let $\GG$ denote the weak (innermost) dependency graph.
Suppose for every path $\WW_i = (\PP_{i1},\ldots,\PP_{im})$ in $\PG{\GG}$ 
there exist an SLI $\AA_i$ and linear (quadratic) restricted interpretations $\BB_j$ 
such that 
$(\geqord{\BB_i},\gord{\BB_i})$
forms a safe reduction pair with
(i) $\UU(\PP_{i1} \cup \cdots \cup \PP_{im}) \subseteq {\gord{\AA_i}}$
(ii) $\PP_{i1} \cup \cdots \cup \PP_{im-1} \cup \UU(\PP_{i1} \cup \cdots \cup \PP_{im})
 \subseteq {\geqord{\BB_i}}$, 
(iii) ${\PP_{im}} \subseteq {\gord{\BB_{im}}}$, and
(iv) $\PP_{im}$ is non-duplicating.
Then the runtime complexity of a TRS $\RR$ is linear or quadratic, respectively. 
\end{corollary}
\begin{proof}
Observe that the assumptions imply that any basic term $t \in \TB$
is terminating with respect to $\RS$: Any infinite
derivation with respect to $\RS$ starting in $t$ can be translated into
an infinite derivation with respect to $\UU(\RS) \cup \PP$ (see~\cite[Lemma~16]{HM:2008}).
Moreover, as the number of paths in $\PG{\GG}$ is finite, there exists 
a component $\PP_j$ that represents an infinite rewrite sequence. This is
a contradiction.
Without loss of generality, we assume $\PP = \WDP(\PP)$ and $\GG = \WDG(\PP)$.
Note that the reduction pair $(\geqord{\BB_i},\gord{\BB_i})$ is safe and collapsible.
Hence for all $i$, the length of any $\rsrew{\PP_{im}/\SS}$-rewrite sequence
is less than $p_i(\size{t})$, where $p_m$ denotes 
a linear (or quadratic) polynomial, depending on $\size{t}$ only. (Here
$\SS = \PP_{i1} \cup \cdots \cup \PP_{im-1} \cup \UU(\PP_{i1} \cup \cdots \cup \PP_{im})$.)
In analogy to the operator $L$, we define 
$M(t) \defsym  \max\{ \dheight(t, \rsrew{\PP_{im} \cup \SS}) \mid 
\text{$(\PP_{i1},\ldots,\PP_{im})$ is a path in $\PG{\GG}$ such that $\PP_{i1} \in \Src$}\}$.
An application of Proposition~\ref{p:relative} yields 
$M(t) = \bO(p_i(\size{t}))$. Following the pattern of the proof of the theorem, 
we establish the existence of a polynomial $p$ such that 
$\dheight(t,\rsrew{\PP \cup \UU(\PP)}) \leqslant p(\size{t})$
holds for any basic term $t$. Finally, the corollary follows by 
an application of Proposition~\ref{p:dp:usable}.
\end{proof}

As mentioned above, in the dependency graph refinement for 
termination analysis it suffices to guarantee for each cycle $\CC$ that
there exist no $\CC$-minimal rewrite sequences. For that one only needs to
find a reduction pair $(\gtrsim,\succ)$ such that
${\RS} \subseteq {\gtrsim}$, ${\CC} \subseteq {\gtrsim}$ and 
${{\CC} \cap {\succ}} \not= {\varnothing}$. Thus, considering Theorem~\ref{t:dg} it is
tempting to think that it should suffice to replace strongly connected
components by cycles and the stronger conditions should apply. 
However this intuition is deceiving as shown by the next example.
%
\begin{example}
Consider the TRS $\RS$ of
$\m{f}(\m{s}(x),\m{0}) \to \m{f}(x, \m{s}(0))$ and
$\m{f}(x,\m{s}(y)) \to \m{f}(x, y)$.
$\WDP(\RS)$ consists of
$1\colon \m{f}^\sharp(\m{s}(x),\m{0}) \to \m{f}^\sharp(x, \m{s}(x))$
and $2\colon \m{f}^\sharp(x,\m{s}(y)) \to \m{f}^\sharp(x,y)$,
and the weak dependency graph $\WDG(\RS)$ 
contains two cycles $\{1,2\}$ and $\{2\}$.
There are two linear restricted interpretations $\AA$ and $\BB$ such that
$\{1,2\} \subseteq {\geqslant_\A} \cup {>_\A}$,
$\{1\} \subseteq {>_\AA}$, and
$\{2\} \subseteq {>_\BB}$.
Here, however, we must not conclude linear runtime complexity,
because the runtime complexity of $\RS$ is at least quadratic.
\end{example}

\section{Conclusion} \label{lpar08:Conclusion}

In this section we provide (experimental) evidence on the applicability of
the technique for complexity analysis established in this paper. 
We briefly consider the
efficient implementation of the techniques provided by Theorem~\ref{t:dg}
and Corollary~\ref{c:dg}. 
Firstly, in order to approximate (weak) dependency graphs, 
we adapted (innermost) dependency graph estimations
using the functions $\m{TCAP}$ ($\m{ICAP}$)~\cite{Giesl:2005frocos}.
Secondly, note that a graph including $n$ nodes 
may contain an exponential number of paths.
However, to apply Corollary~\ref{c:dg} it is sufficient to handle only paths in
the following set. 
Note that if we can assume that $\PG{\GG}$ is a tree, then this set contains at most $n^2$ paths.
\begin{equation}
\label{eq:maximal}
\{ (\PP_1,\ldots,\PP_k) \mid 
\text{$(\PP_1,\ldots,\PP_m)$ is a maximal path and $k \leqslant m$} \} \tkom
\end{equation}

\begin{example}
Consider the following non-total terminating $\RS$.\\[-4ex]
\begin{alignat*}{4}
13\colon && \mp(\mf(\mf(x))) & \rew \mq(\mf(\mg(x))) & \hspace{5ex}
15\colon &&  \mp(\mg(\mg(x))) & \rew \mq(\mg(\mf(x)))\\
14\colon && \mq(\mf(\mf(x))) & \rew \mp(\mf(\mg(x))) &
16\colon && \mq(\mg(\mg(x))) & \rew \mp(\mg(\mf(x)))
\end{alignat*}  
%
The weak dependency pairs $\WDP(\RR)$ are given as follows:
\begin{alignat*}{4}
17\colon && \mp^\sharp(\mf(\mf(x))) & \rew \mq^\sharp(\mf(\mg(x))) & \hspace{5ex}
19\colon &&  \mp^\sharp(\mg(\mg(x))) & \rew \mq^\sharp(\mg(\mf(x)))\\
18\colon && \mq^\sharp(\mf(\mf(x))) & \rew \mp^\sharp(\mf(\mg(x))) &
20\colon && \mq^\sharp(\mg(\mg(x))) & \rew \mp^\sharp(\mg(\mf(x)))
\end{alignat*} 
Hence for $\PG{\WDG(\RR)}$ the set~\eqref{eq:maximal} consists of 4 trivial paths:
$(\{17\})$, $(\{18\})$, $(\{19\})$, and $(\{20\})$. Exemplarily we treat one of the
trivial paths.
\begin{itemize}
\item Consider $17\colon \mp^\sharp(\mf(\mf(x))) \rew \mq^\sharp(\mf(\mg(x)))$.
Observe that $\UU(\{17\}) = \varnothing$. In order to orient this weak dependency 
pair it suffices to employ the following polynomial interpretation $\mp^\sharp_\BB(x) = 1$
and $\mq^\sharp_\BB(x) = \mf_\BB(x) = \mg_\BB(x) = 0$. 
\end{itemize}
In a similar fashion, we can treat the remaining three paths. It is not difficult
to argue that this implies that the runtime complexity function of $\RR$ is constant. Note that none
of the other techniques in the analysis of runtime complexities can obtain
this (optimal) bound.
\end{example}

\begin{table}
\begin{center}
\caption{Results for Linear Runtime Complexities}
\label{table1}
\medskip
\begin{tabular}{%
@{\hspace{0pt}}l@{~~}r@{~~~~}r@{~~}r@{~~}r@{~~~~~~}r@{~~}r@{~~}r@{\hspace{0pt}}}
\hline
& \multicolumn{4}{c}{\emph{full rewriting}}
& \multicolumn{3}{c}{\emph{innermost rewriting}} \\
& direct
& Prop.\ref{p:dp:usable} & Prop.\ref{p:main} & Cor.\ref{c:dg}
& Prop.\ref{p:dp:usable} & Prop.\ref{p:main} & Cor.\ref{c:dg}
\\
\hline\\[-2ex]
S
& 139
& 139 & 93 & 108 
& 144 & 102 & 117 
\\
&   
&      & (147\makebox[0mm][l]{)} & (162\makebox[0mm][l]{)}
&      & (166\makebox[0mm][l]{)} & (181\makebox[0mm][l]{)}
\\
& \emph{7}
& \emph{9} & \emph{14} & \emph{6} 
& \emph{8} & \emph{16} & \emph{6} 
\\
F
& 1591
& 1582 & 1646 & 1630
& 1577 & 1637 & 1619
\\
& \emph{2474}
& \emph{4789} & \emph{456} & \emph{607}
& \emph{4699} & \emph{462} & \emph{536} 
\\
T
& 8
& 17 & 0 & 1
& 17 & 0 & 3
\\
\hline\\[-3em]
\end{tabular}
\end{center}
\end{table}
Moreover, to deal efficiently with polynomial interpretations, the issuing
constraints are encoded in \emph{propositional logic} in a similar
spirit as in~\cite{FGMSTZ07}.  
Assignments are found by employing a state-of-the-art SAT solver, in our case
Mini\-Sat.%
\footnote{\url{http://minisat.se/}.}
Furthermore, SLIs are handled by linear programming.
Based on these ideas we implemented a complexity analyser.
These techniques have also been incorportated into the 
\emph{Tyrolean Complexity Tool} (\tct\ for short) that 
incorporates the most powerful techniques to analyse the complexity of 
rewrite systems that are currently at hand.%
\footnote{\url{http://cl-informatik.uibk.ac.at/software/tct/}.}
For compilation of the here presented experimental data we used the
latter implementation.

As suitable test bed we used the rewrite systems in the Termination Problem Data Base
version 4.0.%
\footnote{See~\url{http://termcomp.uibk.ac.at}.}
This test bed comprises 1739 TRSs.
The presented tests were performed on a server with 8 Dual-Core 2.6 GHz AMD\textsuperscript\textregistered\ 
Opteron\texttrademark\ Processor 8220 CPUs, for a total of 16 cores. 64 GB of RAM are
available. For each system we used a timeout of 60 seconds, the times 
in the tables are given in seconds. 
In interpreting defined and dependency pair symbols,
we restrict to polynomials whose coefficients are in the range $\{0,1,\dots,5\}$.
Table~\ref{table1} (\ref{table2}) shows the experimental results for
linear (quadratic) runtime complexities based on linear (quadratic) 
restricted interpretations.%
\footnote{For full experimental evidence see~\url{http://www.jaist.ac.jp/~hirokawa/08b/} or
~\url{http://cl-informatik.uibk.ac.at/software/tct/}.}
Text written in~\textit{italics} below the number of successes or failures
indicates the total time 
(in seconds) of success cases or failure cases, respectively.
The columns marked ``Prop.~\ref{p:main}'' and ``Cor.~\ref{c:dg}''
refer to the applicability of the respective results. 
Moreover ``S'', ``F'', ``T'' denotes
\emph{success}, \emph{failure}, or \emph{timeout} respectively.
For the sake of comparison, in the parentheses we indicate the number of
successes by the method of the column \emph{or} by Proposition~\ref{p:dp:usable}.

\begin{table}
\begin{center}
\caption{Results for Quadratic Runtime Complexities}
\label{table2}
\medskip
\begin{tabular}{%
@{\hspace{0pt}}l@{~~}r@{~~~~}r@{~~}r@{~~}r@{~~~~~~}r@{~~}r@{~~}r@{\hspace{0pt}}}
\hline\\[-2ex]
& \multicolumn{4}{c}{\emph{full rewriting}}
& \multicolumn{3}{c}{\emph{innermost rewriting}} \\
& direct
& Prop.\ref{p:dp:usable} & Prop.\ref{p:main} & Cor.\ref{c:dg}
& Prop.\ref{p:dp:usable} & Prop.\ref{p:main} & Cor.\ref{c:dg}
\\
\hline\\[-2ex]
S
& 182
& 182 & 93 & 108
& 183 & 102 & 117
\\
&  
&   & (186\makebox[0mm][l]{)} & (202\makebox[0mm][l]{)}
&   & (193\makebox[0mm][l]{)} & (208\makebox[0mm][l]{)}
\\
& \emph{152}
& \emph{329} & \emph{97} & \emph{74}
& \emph{324} & \emph{98} & \emph{54}
\\
F
& 524
& 473 & 1636 & 1624
& 492 & 1627 & 1613
\\
& \emph{5469}
& \emph{5436} & \emph{793} & \emph{850}
& \emph{5500} & \emph{825} & \emph{781} 
\\
T
& 864
& 951 & 10 & 7
& 924 & 10 & 9
\\
\hline\\[-3em]
\end{tabular}
\end{center}
\end{table}

In concluding, we observe that the experimental data shows that the 
here introduced dependency graph refinement for complexity
analysis extends the analytic power of the methods introduced in~\cite{HM:2008}
(see also Chapter~\ref{ijcar08a}). 
Note the significant difference between those TRSs that can be handled
by Propositions~\ref{p:dp:usable} and~\ref{p:main} in contrast to those that
can be handled either by Proposition~\ref{p:dp:usable} or by 
Corollary~\ref{c:dg}.
Moreover observe the gain in power in relation to direct methods, 
compare also~\cite{BCMT:2001,AM:2008}.
}

\backmatter

\begin{thebibliography}{150}
\providecommand{\natexlab}[1]{#1}
\providecommand{\url}[1]{\texttt{#1}}
\expandafter\ifx\csname urlstyle\endcsname\relax
  \providecommand{\doi}[1]{doi: #1}\else
  \providecommand{\doi}{doi: \begingroup \urlstyle{rm}\Url}\fi

\bibitem[Aehlig and Schwichtenberg(2002)]{AS:2002}
K.~Aehlig and H.~Schwichtenberg.
\newblock A syntactical analysis of non-size-increasing polynomial time
  computation.
\newblock \emph{ACM Trans. Comput. Log.}, 3:\penalty0 383--401, 2002.

\bibitem[Aehlig et~al.(2004)Aehlig, Berger, Hofmann, and
  Schwichtenberg]{ABHS:2004}
K.~Aehlig, U.~Berger, M.~Hofmann, and H.~Schwichtenberg.
\newblock An arithmetic for non-size-increasing polynomial-time computation.
\newblock \emph{Theor.~Comput.~Sci.}, 318\penalty0 (1--2):\penalty0 3--27,
  2004.

\bibitem[Amadio(2005)]{Amadio:2005}
R.~Amadio.
\newblock Synthesis of max-plus quasi-interpretations.
\newblock \emph{Fundam. Inform.}, 65\penalty0 (1-2):\penalty0 29--60, 2005.

\bibitem[Anderson et~al.(2005)Anderson, Khoo, Andrei, and Luca]{AKAL05}
H.~Anderson, S-C. Khoo, S.~Andrei, and B.~Luca.
\newblock Calculating polynomial runtime properties.
\newblock In \emph{Procedings of the 3rd ASIAN Symposium on Programming
  Languages and Systems}, volume 3780 of \emph{LNCS}, pages 230--246. Springer
  Verlag, 2005.

\bibitem[Arai(1998{\natexlab{a}})]{Arai:1998}
T.~Arai.
\newblock Variations on a {T}heme by {W}eiermann.
\newblock \emph{J. Symb. Logic}, 63:\penalty0 897--925, 1998{\natexlab{a}}.

\bibitem[Arai(1998{\natexlab{b}})]{Arai:1998:trs}
T.~Arai.
\newblock Some results on cut-elimination, provable well-orderings, induction,
  and reflection.
\newblock \emph{Ann.~Pure Appl.~Logic}, 95:\penalty0 93--184,
  1998{\natexlab{b}}.
\newblock Chapter 8.

\bibitem[Arai and Moser(2004)]{AraiMoser:2004}
T.~Arai and G.~Moser.
\newblock A note on a term rewriting characterization of {P}{T}{I}{M}{E}.
\newblock In \emph{Proceedings of the 7th International Workshop on
  Termination}, pages 10--13. number AIB-2004-07 of Aachener
  Informatik-Berichte, 2004.
\newblock Extended abstract.

\bibitem[Arai and Moser(2005)]{fsttcs:2005}
T.~Arai and G.~Moser.
\newblock Proofs of termination of rewrite systems for polytime functions.
\newblock In \emph{Proceedings of the 25th Conference on Foundations of
  Software Technology and Theoretical Computer Science}, volume 3821 of
  \emph{LNCS}, pages 529--540. Springer Verlag, 2005.

\bibitem[Arts and Giesl(2000)]{ArtsGiesl:2000}
T.~Arts and J.~Giesl.
\newblock Termination of term rewriting using dependency pairs.
\newblock \emph{Theor.~Comput.~Sci.}, 236\penalty0 (1--2):\penalty0 133--178,
  2000.

\bibitem[Arts and Giesl(2001)]{ArtsGiesl:2001}
T.~Arts and J.~Giesl.
\newblock A collection of examples for termination of term rewriting using
  dependency pairs.
\newblock Technical Report AIB-2001-09, RWTH Aachen, 2001.

\bibitem[Avanzini and Moser(2008{\natexlab{a}})]{AM:2008}
M.~Avanzini and G.~Moser.
\newblock Complexity analysis by rewriting.
\newblock In \emph{Proceedings of the 9th International Symposium on Functional
  and Logic Programming}, volume 4989 of \emph{LNCS}, pages 130--146. Springer
  Verlag, 2008{\natexlab{a}}.

\bibitem[Avanzini and Moser(2008{\natexlab{b}})]{TR:2008}
M.~Avanzini and G.~Moser.
\newblock Complexity analysis by rewriting.
\newblock Technical report, Computational Logic, November 2008{\natexlab{b}}.

\bibitem[Avanzini et~al.(2007)Avanzini, Hirokawa, Middeldorp, and
  Moser]{AHMM:2007}
M.~Avanzini, N.~Hirokawa, A.~Middeldorp, and G.~Moser.
\newblock Towards an automatic runtime complexity analysis of scheme programs
  by rewriting.
\newblock Technical report, Computational Logic, December 2007.

\bibitem[Avanzini et~al.(2008)Avanzini, Moser, and Schnabl]{AMS:2008}
M.~Avanzini, G.~Moser, and A.~Schnabl.
\newblock Automated implicit computational complexity analysis (system
  description).
\newblock In \emph{Proceedings of 4th International Joint Conference on
  Automated Reasoning}, volume 5195 of \emph{LNCS}, pages 132--139. Springer
  Verlag, 2008.

\bibitem[Baader and Nipkow(1998)]{BaaderNipkow:1998}
F.~Baader and T.~Nipkow.
\newblock \emph{Term {R}ewriting and {A}ll {T}hat}.
\newblock Cambridge University Press, 1998.

\bibitem[Bachmair(1987)]{Bachmair:1987}
L.~Bachmair.
\newblock \emph{Proof methods for equational theories}.
\newblock PhD thesis, University of Illinois, 1987.

\bibitem[Bachmann(1955)]{Bachmann:1955}
H.~Bachmann.
\newblock \emph{Transfinite Zahlen}.
\newblock Springer Verlag, 1955.

\bibitem[Baillot and K.¨Terui(2009)]{BT:2009}
P.~Baillot and K.¨Terui.
\newblock Light types for polynomial time computation in lambda calculus.
\newblock \emph{Inf. Comput.}, 207\penalty0 (1):\penalty0 41--62, 2009.

\bibitem[Beckmann and Weiermann(1996)]{BeckmannWeiermann:1996}
A.~Beckmann and A.~Weiermann.
\newblock A term rewriting characterization of the polytime functions and
  related complexity classes.
\newblock \emph{Arch.~{M}ath.~{L}og.}, 36:\penalty0 11--30, 1996.

\bibitem[Beklemishev(2004)]{BeklemGPS:2004}
L.~Beklemishev.
\newblock Provability algebras and proof-theoretic ordinals - part {I}.
\newblock \emph{Ann.\ Pure Appl.\ Logic}, 128:\penalty0 103--124, 2004.

\bibitem[Bellantoni and Cook(1992)]{BellantoniCook:1992}
S.~Bellantoni and S.~Cook.
\newblock A new recursion-theoretic characterization of the polytime functions.
\newblock \emph{Comput. Complex.}, 2\penalty0 (2):\penalty0 97--110, 1992.

\bibitem[Ben-Amram(2002)]{B02}
A.M. Ben-Amram.
\newblock General size-change termination and lexicographic descent.
\newblock In \emph{The Essence of Computation: Complexity, Analysis,
  Transformation. Essays Dedicated to Neil D.~Jones}, volume 2566 of
  \emph{LNCS}, pages 3--17, 2002.

\bibitem[Ben-Amram et~al.(2008)Ben-Amram, Jones, and Kristiansen]{BJK:2008}
A.M. Ben-Amram, N.~Jones, and L.~Kristiansen.
\newblock Linear, polynomial or exponential? complexity inference in polynomial
  time.
\newblock In \emph{Proceedings of the 4th Conference on Computability in
  Europe}, volume 5028 of \emph{LNCS}, pages 67--76. Springer Verlag, 2008.

\bibitem[Bird(1998)]{Bird:1998}
R.~Bird.
\newblock \emph{Introduction to Functional Programming}.
\newblock Prentice Hall Series in Computer Science, second edition, 1998.

\bibitem[Blanqui et~al.(2006)Blanqui, Jouannaud, and Rubio]{Jouannaud:2006}
F.~Blanqui, J.-P. Jouannaud, and A.~Rubio.
\newblock Higher-order termination: From {K}ruskal to computability.
\newblock In \emph{Proceedings of the International Conference on Logic for
  Programming, Artificial Intelligence and Reasoning}, volume 4246 of
  \emph{LNAI}, pages 1--14. Springer Verlag, 2006.

\bibitem[Bonfante et~al.(2001)Bonfante, Cichon, Marion, and Touzet]{BCMT:2001}
G.~Bonfante, A.~Cichon, J.-Y. Marion, and H.~Touzet.
\newblock Algorithms with polynomial interpretation termination proof.
\newblock \emph{J.~Funct.~Program.}, 11\penalty0 (1):\penalty0 33--53, 2001.

\bibitem[Bonfante et~al.(2005)Bonfante, Marion, and Moyen]{BMM:2005:rta}
G.~Bonfante, J.-Y. Marion, and J.-Y. Moyen.
\newblock Quasi-interpretations and small space bounds.
\newblock In \emph{Proceedings of the 16th International Conference on
  Rewriting Techniques and Applications}, volume 3467 of \emph{LNCS}, pages
  150--164. Springer Verlag, 2005.

\bibitem[Bonfante et~al.(2007)Bonfante, Marion, and P{\'e}choux]{BMP:2007}
G.~Bonfante, J.-Y. Marion, and R.~P{\'e}choux.
\newblock Quasi-interpretation synthesis by decomposition.
\newblock In \emph{Proceedings of the 4th International Colloquium on
  Theoretical Aspects of Computing}, volume 4711 of \emph{LNCS}, pages
  410--424. Springer Verlag, 2007.

\bibitem[Bonfante et~al.(2009)Bonfante, Marion, and Moyen]{BMM:2009:tcs}
G.~Bonfante, J.-Y. Marion, and J.-Y. Moyen.
\newblock Quasi-interpretations: A way to control resources.
\newblock \emph{Theor.~Comput.~Sci.}, 2009.
\newblock To appear.

\bibitem[Buchholz(1995)]{Buchholz:1995}
W.~Buchholz.
\newblock Proof-theoretical analysis of termination proofs.
\newblock \emph{Ann.~Pure Appl.~Logic}, 75:\penalty0 57--65, 1995.

\bibitem[Buchholz(2003)]{Buchholz:2003}
W.~Buchholz.
\newblock Ordinal notations and fundamental sequences.
\newblock Unpublished manuscript; availabe at {\tt
  www.mathematik.uni-muenchen.de/\~{\mbox{}}buchholz/}, 2003.

\bibitem[Buss(1998)]{Buss:1998}
S.-R. Buss, editor.
\newblock \emph{Handbook of Proof Theory}, volume 137.
\newblock Elsevier Science, 1998.

\bibitem[Caseiro(1997)]{Caseiro:1997}
V.-H. Caseiro.
\newblock An equational characterization of the poly-time functions on any
  constructor data structure.
\newblock Technical report, Departement of Informatics, University of Oslo,
  1997.
\newblock \url{http://www.ifi.uio.no/~ftp/publications}.

\bibitem[Caviness and Johnson(2004)]{CJ:2004}
B.-F. Caviness and J.R. Johnson, editors.
\newblock \emph{Quantifier Elimination and Cylindrical Algebraic
  Decomposition}.
\newblock Springer Verlag, 2004.

\bibitem[Choppy et~al.(1989)Choppy, Kaplan, and Soria]{CKS:1989}
C.~Choppy, S.~Kaplan, and M.~Soria.
\newblock Complexity analysis of term-rewriting systems.
\newblock \emph{Theor. Comput. Sci.}, 67\penalty0 (2--3):\penalty0 261--282,
  1989.

\bibitem[Cichon(1992)]{Cichon:1992}
E.-A. Cichon.
\newblock Termination orderings and complexity characterisations.
\newblock In P.~Aczel, H.~Simmons, and S.S. Wainer, editors, \emph{Proof
  {T}heory}, pages 171--193, 1992.

\bibitem[Cichon and Lescanne(1992)]{CichonLescanne:1992}
E.-A. Cichon and P.~Lescanne.
\newblock Polynomial interpretations and the complexity of algorithms.
\newblock In \emph{Proceedings of the 11th International Conference on
  Automated Deduction}, pages 139--147. Springer Verlag, 1992.

\bibitem[Cichon and Marion(1999)]{CichonMarion:1999}
E.-A. Cichon and J.-Y. Marion.
\newblock Light {L}{P}{O}.
\newblock Technical report 99-R-138, 1999.

\bibitem[Cichon and Wainer(1983)]{CichonWainer:1983}
E.-A. Cichon and S.-S. Wainer.
\newblock The slow growing and the {G}rzegorczyk hierarchies.
\newblock \emph{J. Symb. Logic}, 48:\penalty0 399--408, 1983.

\bibitem[Cichon and Weiermann(1997)]{CichonWeiermann:1997}
E.-A. Cichon and A.~Weiermann.
\newblock Term rewriting theory for the primitive recursive functions.
\newblock \emph{Ann.~Pure Appl.~Logic}, 83\penalty0 (3):\penalty0 199--223,
  1997.

\bibitem[Cobham(1965)]{Cobham:1965}
A.~Cobham.
\newblock The intrinsic computational difficulty of functions.
\newblock In Y.~Bar-Hillel, editor, \emph{Logic, Methodology and Philosophy of
  Science, proceedings of the second International Congress}, Jerusalem, 1964,
  1965. North-Holland.

\bibitem[Contejean et~al.(2003)Contejean, March{\'e}, Monate, and
  Urbain]{Cime:2003}
E.~Contejean, C.~March{\'e}, B.~Monate, and X.~Urbain.
\newblock Proving termination of rewriting with {C}i{ME}.
\newblock In \emph{6th International Workshop on Termination}, pages 71--73,
  2003.
\newblock Technical Report DSIC-II/15/03, Universidad Polit{\'e}cnica de
  Valencia.

\bibitem[Contejean et~al.(2005)Contejean, March{\'e}, Tom{\'a}s, and
  Urbain]{contejean:2005}
E.~Contejean, C.~March{\'e}, A.-P. Tom{\'a}s, and X.~Urbain.
\newblock Mechanically proving termination using polynomial interpretations.
\newblock \emph{Journal of Automated Reasoning}, 34\penalty0 (4):\penalty0
  325--363, 2005.

\bibitem[Dershowitz(1982)]{Dershowitz:1982}
N.~Dershowitz.
\newblock Orderings for term rewriting systems.
\newblock \emph{Theor.~Comput.~Sci.}, 17\penalty0 (3):\penalty0 279--301, 1982.

\bibitem[Dershowitz(1987)]{Dershowitz:1987}
N.~Dershowitz.
\newblock Termination of {R}ewriting.
\newblock \emph{J.~Symb.\ Comput.}, 3\penalty0 (1--2):\penalty0 69--116, 1987.

\bibitem[Dershowitz and Moser(2007)]{DershowitzMoser:2007}
N.~Dershowitz and G.~Moser.
\newblock The {H}ydra battle revisited.
\newblock In \emph{Rewriting, Computation and Proof}, volume 4600 of
  \emph{LNCS}, pages 1--27. Springer Verlag, 2007.
\newblock Essays Dedicated to Jean-Pierre Jouannaud on the Occasion of His 60th
  Birthday.

\bibitem[Dershowitz and Okada(1988)]{DershowitzOkada:1988}
N.~Dershowitz and M.~Okada.
\newblock Proof-theoretic techniques for term rewriting theory.
\newblock In \emph{Proceedings of the 13th Annual IEEE Symposium on Logic in
  Computer Science}, pages 104--111, 1988.

\bibitem[Endrullis et~al.(2008)Endrullis, Waldmann, and Zantema]{EWZ:2008}
J.~Endrullis, J.~Waldmann, and H.~Zantema.
\newblock Matrix interpretations for proving termination of term rewriting.
\newblock \emph{J.~Autom. Reason-}, 40\penalty0 (3):\penalty0 195--220, 2008.

\bibitem[Feferman(1996)]{Feferman:1996}
S.~Feferman.
\newblock Three conceptual problems that bug me.
\newblock Lecture text for the $7^{th}$ Scandinavian Logic Symposium,
  \url{ftp://math.stanford.edu/pub/papers/feferman/}, 1996.

\bibitem[Friedman(2001)]{Friedman:2001}
H.~Friedman.
\newblock Lecture notes on term rewriting and computational complexity.
\newblock \url{http://www.math.ohio-state.edu/~friedman/manuscripts.html},
  2001.

\bibitem[Fuhs et~al.(2007)Fuhs, Giesl, Middeldorp, Schneider-Kamp, Thiemann,
  and Zankl]{FGMSTZ07}
C.~Fuhs, J.~Giesl, A.~Middeldorp, P.~Schneider-Kamp, R.~Thiemann, and H.~Zankl.
\newblock {SAT} solving for termination analysis with polynomial
  interpretations.
\newblock In \emph{Proceedings of the 10th International Conference on Theory
  and Applications of Satisfiability Testing}, volume 4501 of \emph{LNCS},
  pages 340--354, 2007.

\bibitem[Fuhs et~al.(2009)Fuhs, Giesl, Pl{\"u}cker, Schneider-Kamp, and
  Falke]{FGPSF:2009}
C.~Fuhs, J.~Giesl, M.~Pl{\"u}cker, P.~Schneider-Kamp, and S.~Falke.
\newblock Proving termination of integer term rewriting.
\newblock In \emph{Proceedings of the 20th International Conference on
  Rewriting Techniques and Applications}, volume 5595 of \emph{LNCS}, pages
  32--47. Springer Verlag, 2009.

\bibitem[Geser(1990)]{Geser:1990}
A.~Geser.
\newblock \emph{Relative Termination}.
\newblock PhD thesis, Universit{\"a}t Passau, 1990.

\bibitem[Geser et~al.(2004)Geser, Hofbauer, and Waldmann]{Geser2004}
A.~Geser, D.~Hofbauer, and J.~Waldmann.
\newblock Match-bounded string rewriting systems.
\newblock \emph{Appl.~Algebra Eng.~Commun.~Comput.}, 15:\penalty0 149--171,
  2004.

\bibitem[Geser et~al.(2005)Geser, Hofbauer, Waldmann, and Zantema]{Geser2005}
A.~Geser, D.~Hofbauer, J.~Waldmann, and H.~Zantema.
\newblock On tree automata that certify termination of left-linear term
  rewriting systems.
\newblock In \emph{Proceedings of the 16th International Conference on
  Rewriting Techniques and Applications}, volume 3467 of \emph{LNCS}, pages
  353--367. Springer Verlag, 2005.

\bibitem[Geser et~al.(2007)Geser, Hofbauer, Waldmann, and Zantema]{GHWZ07}
A.~Geser, D.~Hofbauer, J.~Waldmann, and H.~Zantema.
\newblock On tree automata that certify termination of left-linear term
  rewriting systems.
\newblock \emph{Inf.~and Comput.}, 205\penalty0 (4):\penalty0 512--534, 2007.

\bibitem[Giesl et~al.(2002)Giesl, Arts, and Ohlebusch]{GAO:2002}
J.~Giesl, T.~Arts, and E.~Ohlebusch.
\newblock Modular termination proofs for rewriting using dependency pairs.
\newblock \emph{J.~Symb.\ Comput.}, 34:\penalty0 21--58, 2002.

\bibitem[Giesl et~al.(2004)Giesl, Thiemann, Schneider-Kamp, and
  Falke]{Aprove:2004}
J.~Giesl, R.~Thiemann, P.~Schneider-Kamp, and S.~Falke.
\newblock Automated termination proofs with {A}prove.
\newblock In \emph{Proceedings of the 15th International Conference on Rewrite
  Techniques and Applications}, volume 3091 of \emph{LNCS}, pages 210--220.
  Springer Verlag, 2004.

\bibitem[Giesl et~al.(2005)Giesl, Thiemann, and
  Schneider-Kamp]{Giesl:2005frocos}
J.~Giesl, R.~Thiemann, and P.~Schneider-Kamp.
\newblock Proving and disproving termination of higher-order functions.
\newblock In \emph{Proceedings of the 5th International Workshop on Frontiers
  of Combining Systems}, volume 3717 of \emph{LNAI}, pages 216--231. Springer
  Verlag, 2005.

\bibitem[Giesl et~al.(2006{\natexlab{a}})Giesl, Schneider-Kamp, and
  Thiemann]{Aprove1.2}
J.~Giesl, P.~Schneider-Kamp, and R.~Thiemann.
\newblock {A}{P}ro{V}{E} 1.2: {A}utomatic termination proofs in the dependency
  pair framework.
\newblock In \emph{Proceedings of the 3rd International Joint Conference on
  Automated Reasoning}, volume 4130 of \emph{LNCS}, pages 281--286. Springer
  Verlag, 2006{\natexlab{a}}.

\bibitem[Giesl et~al.(2006{\natexlab{b}})Giesl, Thiemann, and
  Schneider-Kamp]{Giesl:2006}
J.~Giesl, R.~Thiemann, and P.~Schneider-Kamp.
\newblock Proving and disproving termination in the dependency pair framework.
\newblock In F.~Baader, P.~Baumgartner, R.~Nieuwenhuis, and A.~Voronkov,
  editors, \emph{Deduction and Applications}, volume 05431 of \emph{Dagstuhl
  Seminar Proceedings}, Germany, 2006{\natexlab{b}}. Internationales
  Begegnungs- und Forschungszentrum für Informatik (IBFI), Schloss Dagstuhl.

\bibitem[Giesl et~al.(2006{\natexlab{c}})Giesl, Thiemann, Schneider-Kamp, and
  Falke]{GTSF06}
J.~Giesl, R.~Thiemann, P.~Schneider-Kamp, and S.~Falke.
\newblock Mechanizing and improving dependency pairs.
\newblock \emph{J.~Autom. Reason-}, 37\penalty0 (3):\penalty0 155--203,
  2006{\natexlab{c}}.

\bibitem[Girard(1987)]{Gir87}
J.-Y. Girard.
\newblock \emph{Proof Theory and Logical Complexity}, volume~1 of \emph{Studies
  in Proof Theory, Monographs}.
\newblock Bibliopolis, Napoli, Italy, 1987.

\bibitem[Girard(1981)]{Girard:1981}
J.-Y. Girard.
\newblock {$\Pi^1_2$}-logic {I}:{D}ilators.
\newblock \emph{Ann. Math. Logic}, 21:\penalty0 75--219, 1981.

\bibitem[G{\'o}mez and Liu(2002)]{GL02}
G.~G{\'o}mez and Y.~A. Liu.
\newblock Automatic time-bound analysis for a higher-order language.
\newblock In \emph{Proceedings of the ACM SIGPLAN 2002 Workshop on Partial
  Evaluation and Semantics-Based Program Manipulation}, pages 75--86. ACM,
  2002.

\bibitem[Goubault-Larrecq(2001)]{GL:2001}
J.~Goubault-Larrecq.
\newblock Well-founded recursive relations.
\newblock In \emph{Proceedings of the 10th International EACSL Conference on
  Computer Science Logic}, volume 2142 of \emph{LNCS}, pages 484--498. Springer
  Verlag, 2001.

\bibitem[Hamana(2005)]{Hamana:2005}
M.~Hamana.
\newblock Universal algebra for termination of higher-order rewriting.
\newblock In \emph{Proceedings of the 16th International Conference on
  Rewriting Techniques and Applications}, volume 3467 of \emph{LNCS}, pages
  135--149. Springer Verlag, 2005.

\bibitem[Hein(2009)]{Hein:2009}
J.L. Hein.
\newblock \emph{Discrete Structues, {L}ogic, and {C}omputability}.
\newblock Jones and Bartlett Publishers, LLC, third edition, 2009.

\bibitem[Hirokawa and Middeldorp(2003)]{HirokawaMiddeldorp:2003}
N.~Hirokawa and A.~Middeldorp.
\newblock Tsukuba termination tool.
\newblock In \emph{Proceedings of the 14th International Conference on
  Rewriting Techniques and Applications}, volume 2706 of \emph{LNCS}, pages
  311--320. Springer Verlag, 2003.

\bibitem[Hirokawa and Middeldorp(2005{\natexlab{a}})]{HirokawaMiddeldorp:2005}
N.~Hirokawa and A.~Middeldorp.
\newblock Automating the dependency pair method.
\newblock \emph{Inf.~and Comput.}, 199\penalty0 (1,2):\penalty0 172--199,
  2005{\natexlab{a}}.

\bibitem[Hirokawa and Middeldorp(2006)]{HirokawaMiddeldorp:2006}
N.~Hirokawa and A.~Middeldorp.
\newblock Predictive labeling.
\newblock In \emph{Proceedings of the 17th International Conference on
  Rewriting Techniques and Applications}, volume 4098 of \emph{LNCS}, pages
  313--327. Springer Verlag, 2006.

\bibitem[Hirokawa and Middeldorp(2007)]{HirokawaMiddeldorp:2007}
N.~Hirokawa and A.~Middeldorp.
\newblock Tyrolean termination tool: Techniques and features.
\newblock \emph{Inf.~and Comput.}, 205:\penalty0 474--511, 2007.

\bibitem[Hirokawa and Middeldorp(2005{\natexlab{b}})]{TTT:2005}
N.~Hirokawa and A.~Middeldorp.
\newblock Tyrolean termination tool.
\newblock In \emph{Proceedings of the 16th International Conference on
  Rewriting and Applications}, volume 3467 of \emph{LNCS}, pages 175--184.
  Springer Verlag, 2005{\natexlab{b}}.
\newblock \url{http://colo6-c703.uibk.ac.at/ttt/}.

\bibitem[Hirokawa and Moser(2008{\natexlab{a}})]{HM:2008}
N.~Hirokawa and G.~Moser.
\newblock Automated complexity analysis based on the dependency pair method.
\newblock In \emph{Proceedings of the 4th International Joint Conference on
  Automated Reasoning}, volume 5195 of \emph{LNAI}, pages 364--380. Springer
  Verlag, 2008{\natexlab{a}}.

\bibitem[Hirokawa and Moser(2008{\natexlab{b}})]{HM:2008b}
N.~Hirokawa and G.~Moser.
\newblock Complexity, graphs, and the dependency pair method.
\newblock In \emph{Proceedings of the 15th International Conference on Logic
  for Programming Artificial Intelligence and Reasoning}, volume 5330 of
  \emph{LNAI}, pages 652--666. Springer Verlag, 2008{\natexlab{b}}.

\bibitem[Hofbauer(1991)]{Hofbauer:1991}
D.~Hofbauer.
\newblock \emph{Termination Proofs and Derivation Lengths in Term Rewriting
  Systems}.
\newblock PhD thesis, Technische Universit{\"a}t Berlin, 1991.

\bibitem[Hofbauer(1992)]{Hofbauer:1992}
D.~Hofbauer.
\newblock Termination proofs by multiset path orderings imply primitive
  recursive derivation lengths.
\newblock \emph{Theor.~Comput.~Sci.}, 105:\penalty0 129--140, 1992.

\bibitem[Hofbauer(2001)]{Hofbauer:2001}
D.~Hofbauer.
\newblock Termination proofs by context-dependent interpretations.
\newblock In \emph{Proceedings of the 12th International Conference on
  Rewriting Techniques and Applications}, volume 2051 of \emph{LNCS}, pages
  108--121. Springer Verlag, 2001.

\bibitem[Hofbauer and Lautemann(1989)]{HofbauerLautemann:1989}
D.~Hofbauer and C.~Lautemann.
\newblock Termination proofs and the length of derivations.
\newblock In \emph{Proceedings of the 3rd International Conference on Rewriting
  Techniques and Applications}, volume 355 of \emph{LNCS}, pages 167--177.
  Springer Verlag, 1989.

\bibitem[Hofbauer and Waldmann(2004)]{HW:2004}
D.~Hofbauer and J.~Waldmann.
\newblock Deleting string rewriting systems preserve regularity.
\newblock \emph{Theor.~Comput.~Sci.}, 327:\penalty0 301--317, 2004.

\bibitem[Hofbauer and Waldmann(2006)]{Hofbauer2006}
D.~Hofbauer and J.~Waldmann.
\newblock Termination of string rewriting with matrix interpretations.
\newblock In \emph{Proceedings of the 17th International Conference on
  Rewriting Techniques and Applications}, volume 4098 of \emph{LNCS}, pages
  328--342. Springer Verlag, 2006.

\bibitem[Hofmann(1999)]{Hofmann:1999}
M.~Hofmann.
\newblock Linear types and non-size increasing polynomial time compuations.
\newblock In \emph{Proceedings of the 14th Annual IEEE Symposium on Logic in
  Computer Science}, pages 464--473. IEEE Computer Society Press, 1999.

\bibitem[Hofmann(2002)]{Hofmann:2002}
M.~Hofmann.
\newblock The strength of non-size increasing computation.
\newblock In \emph{Proceedings of the 29th Annual ACM Symposium on Principles
  of Programming Languages}, pages 260--269. ACM Press, 2002.

\bibitem[Huet and Oppen(1980)]{HuetOppen:1980}
G.~Huet and D.-C. Oppen.
\newblock Equations and rewrite rules: {A} survey.
\newblock In R.~Book, editor, \emph{Formal Language Theory: Perspectives and
  Open Problems}, pages 349--405. Academic Press, New York, 1980.

\bibitem[Jech(2002)]{Jech}
T.~Jech.
\newblock \emph{Set Theory}.
\newblock Springer Verlag, 2002.

\bibitem[Jones and Kristiansen(2009)]{NK:2009}
N.~Jones and L.~Kristiansen.
\newblock A flow calculus of mwp-bounds for complexity analysis.
\newblock \emph{ACM Trans. Comput. Log.}, 10\penalty0 (4), 2009.
\newblock To appear.

\bibitem[Jouannaud and Rubio(2007)]{JR:2007}
J.-P. Jouannaud and A.~Rubio.
\newblock Polymorphic higher-order recursive path orderings.
\newblock \emph{J.~ACM}, 54\penalty0 (1), 2007.

\bibitem[Jouannaud and Rubio(1998)]{Jouannaud:1998}
J.-P. Jouannaud and A.~Rubio.
\newblock Rewrite orderings for higher-order terms in $\eta$-long
  $\beta$-normal forms and the recursive path ordering.
\newblock \emph{Theor.~Comput.~Sci.}, 208\penalty0 (1-2):\penalty0 33--58,
  1998.

\bibitem[Jouannaud and Rubio(1999)]{Jouannaud:1999}
J.-P. Jouannaud and A.~Rubio.
\newblock The higher-order recursive path ordering.
\newblock In \emph{14th Annual IEEE Symposium on Logic in Computer Science},
  pages 402--411. IEEE Computer Society Press, 1999.

\bibitem[Jouannaud and Rubio(2006)]{Jouannaud:2006b}
J.-P. Jouannaud and A.~Rubio.
\newblock Higher-order orderings for normal rewriting.
\newblock In \emph{Proceedings of the 17th International Conference on
  Rewriting Techniques and Applications}, volume 4098 of \emph{LNCS}, pages
  387--399. Springer Verlag, 2006.

\bibitem[Kennaway et~al.(1996)Kennaway, Klop, Sleep, and de~Vries]{KKSV:1996}
R.~Kennaway, J.-W. Klop, R.~Sleep, and F.~de~Vries.
\newblock Comparing curried and uncurried rewriting.
\newblock \emph{J.~Symb.\ Comput.}, 21\penalty0 (1):\penalty0 15--39, 1996.

\bibitem[Kirby and Paris(1982)]{KirbyParis:1982}
L.~Kirby and J.~Paris.
\newblock Accessible independence results for {P}eano arithmetic.
\newblock \emph{Bulletin London Mathematical Society}, 4:\penalty0 285--293,
  1982.

\bibitem[Koprowski(2006)]{TPA:2006}
A.~Koprowski.
\newblock {TPA}: {T}ermination proved automatically.
\newblock In \emph{Proceedings of the 17th International Conference on
  Rewriting Techniques and Applications}, volume 4098 of \emph{LNCS}, pages
  257--266. Springer Verlag, 2006.
\newblock \url{http://www.win.tue.nl/tpa/}.

\bibitem[Koprowski and Middeldorp(2007)]{KoprowskiMiddeldorp:2007}
A.~Koprowski and A.~Middeldorp.
\newblock Predictive labeling with dependency pairs using {S}{A}{T}.
\newblock In \emph{Proceedings of the 18th International Conference on
  Rewriting Techniques and Applications}, volume 4603 of \emph{LNAI}, pages
  410--425. Springer Verlag, 2007.

\bibitem[Koprowski and Waldmann(2008)]{KW:2008}
A.~Koprowski and J.~Waldmann.
\newblock Arctic termination ...below zero.
\newblock In \emph{Proceedings of the 19th International Conference on
  Rewriting Techniques and Applications}, volume 5117 of \emph{LNCS}, pages
  202--216, 2008.

\bibitem[Koprowski and Zantema(2006)]{KoprowskiZantema:2006}
A.~Koprowski and H.~Zantema.
\newblock Automation of {R}ecursive {P}ath {O}rdering for {I}nfinite {L}abelled
  {R}ewrite {S}ystems.
\newblock In \emph{Proceedings of the 3rd International Joint Conference on
  Automated Reasoning}, volume 4130 of \emph{LNAI}, pages 332--346. Springer
  Verlag, 2006.

\bibitem[Korovin and Voronkov(2003)]{KorovinVoronkov:2003}
K.~Korovin and A.~Voronkov.
\newblock Orienting rewrite rules with the {K}nuth-{B}endix order.
\newblock \emph{Inf.~and Comput.}, 183\penalty0 (2):\penalty0 165--186, 2003.

\bibitem[Korp and Middeldorp(2007)]{KM:2007}
M.~Korp and A.~Middeldorp.
\newblock Proving termination of rewrite systems using bounds.
\newblock In \emph{Proceedings of the 18th International Conference on
  Rewriting Techniques and Applications}, volume 4533 of \emph{LNCS}, pages
  273--287, 2007.

\bibitem[Korp et~al.(2008)Korp, Sternagel, Zankl, and Middeldorp]{TTTT:2008}
M.~Korp, C.~Sternagel, H.~Zankl, and A.~Middeldorp.
\newblock Tyrolean termination tool 2.
\newblock Availabe at \url{http://colo6-c703.uibk.ac.at/ttt2/}, 2008.

\bibitem[Leivant(1994)]{Leivant:1994}
D.~Leivant.
\newblock Predicative recurrence and computatinal complexity {I}: {W}ord
  recurrence and poly-time.
\newblock In P.~Clote and J.~Remmel, editors, \emph{Feasible Mathematics {II}},
  pages 320--343. Birkh{\"a}user, 1994.

\bibitem[Leivant and Marion(1993)]{LeivantMarion:1993}
D.~Leivant and J.-Y. Marion.
\newblock Lamba calculus characterization of poly-time.
\newblock \emph{Fund.~Inform.}, 19:\penalty0 167--184, 1993.

\bibitem[Leivant and Marion(1997)]{LeivantMarion:1997}
D.~Leivant and J.-Y. Marion.
\newblock Predicative functional recurrence and poly-space.
\newblock In \emph{Proceedings of the 7th International Joint Conference on
  Theory and Practice of Software Development}, volume 1214 of \emph{LNCS},
  pages 369--380. Springer Verlag, 1997.

\bibitem[Lepper(2001)]{Lepper:2001a}
I.~Lepper.
\newblock Derivation lengths and order types of {K}nuth-{B}endix orders.
\newblock \emph{Theor.~Comput.~Sci.}, 269:\penalty0 433--450, 2001.

\bibitem[Lepper(2002)]{Lepper:2002}
I.~Lepper.
\newblock \emph{Simplification Orders in Term Rewriting}.
\newblock PhD thesis, WWU M{\"u}nster, 2002.
\newblock \url{http://wwwmath.uni-muenster.de/logik/publ/diss/9.html}.

\bibitem[Lepper(2004)]{Lepper:2004}
I.~Lepper.
\newblock Simply terminating rewrite systems with long derivations.
\newblock \emph{Arch.~{M}ath.~{L}og.}, 43:\penalty0 1--18, 2004.

\bibitem[Lescanne(1995)]{Lescanne:1995}
P.~Lescanne.
\newblock Termination of rewrite systems by elementary interpretations.
\newblock \emph{Formal Aspects of Computing}, 7\penalty0 (1):\penalty0 77--90,
  1995.

\bibitem[Lucas(2004)]{Muterm:2004}
S.~Lucas.
\newblock {MU-TERM}: {A} tool for proving termination of context-sensitive
  rewriting.
\newblock In \emph{Proceedings of the 15th International Conference on
  Rewriting Techniques and Applications}, volume 3091 of \emph{LNCS}, pages
  200--209. Springer Verlag, 2004.
\newblock Available at~\url{http://www.dsic.upv.es/~rgutierrez/muterm/}.

\bibitem[Lucas and Pe{\~n}a(2007)]{LP:2007}
S.~Lucas and R.~Pe{\~n}a.
\newblock Termination and complexity bounds for {SAFE} programs.
\newblock In \emph{Proceedings of the 7th Spanish Conference on Programming and
  Computer Languages}, pages 233--242, 2007.

\bibitem[Marion(2003)]{Marion:2003}
J.-Y. Marion.
\newblock Analysing the implicit complexity of programs.
\newblock \emph{Inf.~and Comput.}, 183:\penalty0 2--18, 2003.

\bibitem[Marion and Moyen(2000)]{MarionMoyen:2000}
J.-Y. Marion and J.-Y. Moyen.
\newblock Efficient first order functional program interpreter with time bound
  certifications.
\newblock In \emph{Proceedings of the 7th International Conference on Logic for
  Programming and Automated Reasoning}, volume 1955 of \emph{LNCS}, pages
  25--42, 2000.

\bibitem[Marion and P{\'e}choux(2006)]{Marion:2006}
J.-Y. Marion and R.~P{\'e}choux.
\newblock Resource analysis by sup-interpretation.
\newblock In \emph{Proceedings of the 8th International Symposium on Functional
  and Logic Programming}, volume 3945 of \emph{LNCS}, pages 163--176. Springer
  Verlag, 2006.

\bibitem[Martin(1987)]{Martin:1987}
U.~Martin.
\newblock How to chose weights in the {K}nuth-{B}endix ordering.
\newblock In \emph{Proceedings of the 2nd International Conference on Rewriting
  Techniques and Applications}, volume 256 of \emph{LNCS}, pages 42--53.
  Springer Verlag, 1987.

\bibitem[Matiyasevich(1970)]{Matiyasevich:1970}
Y.~Matiyasevich.
\newblock Enumerable sets are diophantine.
\newblock \emph{Soviet Mathematics (Dokladi)}, 11\penalty0 (2):\penalty0
  354--357, 1970.

\bibitem[Middeldorp et~al.(1996)Middeldorp, Ohsaki, and
  Zantema]{Middeldorp:1996}
A.~Middeldorp, H.~Ohsaki, and H.~Zantema.
\newblock Transforming termination by self-labelling.
\newblock In \emph{Proceedings of the 13th International Conference on
  Automated Deduction}, volume 1104 of \emph{LNCS}, pages 373--387. Springer
  Verlag, 1996.

\bibitem[Moser(2006)]{Moser:2006lpar}
G.~Moser.
\newblock Derivational complexity of {K}nuth-{B}endix orders revisited.
\newblock In \emph{Proceedings of the 13th International Conference on Logic
  for Programming Artificial Intelligence and Reasoning}, volume 4246 of
  \emph{LNAI}, pages 75--89. Springer Verlag, 2006.

\bibitem[Moser(2009)]{Moser:2009}
G.~Moser.
\newblock The {H}ydra {B}attle and {C}ichon's {P}rinciple.
\newblock \emph{Appl.~Algebra Eng.~Commun.~Comput.}, 20\penalty0 (2):\penalty0
  133--158, 2009.
\newblock \url{doi:10.1007/s00200-009-0094-4}.

\bibitem[Moser and Schnabl(2008)]{MS:2008}
G.~Moser and A.~Schnabl.
\newblock Proving quadratic derivational complexities using context dependent
  interpretations.
\newblock In \emph{Proceedings of the 19th International Conference on Rewrite
  Techniques and Applications}, volume 5117 of \emph{LNCS}, pages 276--290.
  Springer Verlag, 2008.

\bibitem[Moser and Schnabl(2009)]{MoserSchnabl:2009}
G.~Moser and A.~Schnabl.
\newblock The derivational complexity induced by the dependency pair method.
\newblock In \emph{Proceedings of the 20th International Conference on
  Rewriting Techniques and Applications}, volume 5595 of \emph{LNCS}, pages
  276--290. Springer Verlag, 2009.

\bibitem[Moser and Weiermann(2003)]{MoserWeiermann:2003}
G.~Moser and A.~Weiermann.
\newblock Relating derivation lengths with the slow-growing hierarchy directly.
\newblock In \emph{Proceedings of the 14th International Conference on
  Rewriting Techniques and Applications}, volume 2706 of \emph{LNCS}, pages
  296--310. Springer Verlag, 2003.

\bibitem[Moser et~al.(2008)Moser, Schnabl, and Waldmann]{MSW:2008}
G.~Moser, A.~Schnabl, and J.~Waldmann.
\newblock Complexity analysis of term rewriting based on matrix and context
  dependent interpretations.
\newblock In \emph{Proceedings of the 28th Foundations of Software Technology
  and Theoretical Computer Science}, pages 304--315. Schloss Dagstuhl -
  Leibniz-Zentrum fuer Informatik, Germany, 2008.
\newblock Creative-Commons-NC-ND licensed.

\bibitem[Niggl(2005)]{Niggl:2005}
K.-H. Niggl.
\newblock Control structures in programs and computational complexity.
\newblock \emph{Ann.~Pure Appl.~Logic}, 133\penalty0 (1-3):\penalty0 247--273,
  2005.

\bibitem[Niggl and Wunderlich(2006)]{NW:2006}
K.-H. Niggl and H.~Wunderlich.
\newblock Certifying polynomial time and linear/polynomial space for imperative
  programs.
\newblock \emph{SIAM J. Comput.}, 35\penalty0 (5):\penalty0 1122--1147, 2006.

\bibitem[Oitavem(2002)]{Oitavem:2002}
I.~Oitavem.
\newblock A term rewriting characterization of the functions computable in
  polynomal space.
\newblock \emph{Arch.~{M}ath.~{L}og.}, 41:\penalty0 35--47, 2002.

\bibitem[Otto et~al.(2009)Otto, Brockschmidt, v.~Essen, and Giesl]{OBEG:2009}
C.~Otto, M.~Brockschmidt, C.~v.~Essen, and J.~Giesl.
\newblock Termination analysis of java bytecode by term rewriting.
\newblock In \emph{Proceedings of the 10th International Workshop on
  Termination}, pages 64--68, 2009.

\bibitem[P{\'e}ter(1967)]{Peter:1967}
R.~P{\'e}ter.
\newblock \emph{Recursive Functions}.
\newblock Academic Press, 1967.

\bibitem[Pfenning(2001)]{Pfenning:2001}
F.~Pfenning.
\newblock \emph{Computation and Deduction}.
\newblock Cambridge University Press, 2001.

\bibitem[Robbin(1965)]{Robbin:1965}
J.-W. Robbin.
\newblock \emph{Subrecursive Hierarchies}.
\newblock PhD thesis, Princeton University, 1965.

\bibitem[Rose(1984)]{RoseSubrecursion}
H.E. Rose.
\newblock \emph{Subrecursion: Functions and Hierarchies}.
\newblock Oxford University Press, 1984.

\bibitem[Rosendahl(1989)]{R89}
M.~Rosendahl.
\newblock Automatic complexity analysis.
\newblock In \emph{Proceedings of the 4th International Conference on
  Functional Programming Languages and Computer Architecture}, pages 144--156,
  1989.

\bibitem[Schmidt(1979)]{Schmidt:1979}
D.~Schmidt.
\newblock Well-partial orderings and their maximal order types.
\newblock Fakult{\"a}t f{\"u}r Mathematik der Ruprecht-Karl-Universit{\"a}t
  Heidelberg, 1979.
\newblock Habilitationsschrift.

\bibitem[Schnabl(2007)]{Schnabl:2007}
A.~Schnabl.
\newblock Context {D}ependent {I}nterpretations.
\newblock Master's thesis, {U}niversit{\"a}t {I}nnsbruck, 2007.
\newblock Available at \url{http://cl-informatik.uibk.ac.at/~aschnabl/}.

\bibitem[Schneider-Kamp et~al.(2007)Schneider-Kamp, Thiemann, Annov, Codish,
  and Giesl]{SK:2007}
P.~Schneider-Kamp, R.~Thiemann, E.~Annov, M.~Codish, and J.~Giesl.
\newblock Proving termination using recursive path orders and {S}{A}{T}
  solving.
\newblock In \emph{Proceedings of the 6th International Symposium on Frontiers
  of Combining Systems}, volume 4720 of \emph{LNCS}, pages 267--282. Springer
  Verlag, 2007.

\bibitem[Sch{\"u}tte(1977)]{Schuette:1977}
K.~Sch{\"u}tte.
\newblock \emph{Proof {T}heory}.
\newblock Springer Verlag, Berlin and New York, 1977.

\bibitem[Schwichtenberg(2006)]{Schwichtenberg:2006}
H.~Schwichtenberg.
\newblock An arithmetic for polynomial-time computation.
\newblock \emph{Theor.~Comput.~Sci.}, 357\penalty0 (1):\penalty0 202--214,
  2006.

\bibitem[Sperber et~al.(2007)Sperber, Dybvig, Flatt, and {A.~v.~Stratten et
  al.}]{R6RS}
M.~Sperber, R.~K. Dybvig, M.~Flatt, and {A.~v.~Stratten et al.}
\newblock $\text{Revised}^6$ report on the algorithmic language {S}cheme.
\newblock Available at \url{www.r6rs.org}., 2007.

\bibitem[Steinbach and K{\"u}hler(1990)]{SteinbachKuehler:1990}
J.~Steinbach and U.~K{\"u}hler.
\newblock Check your ordering - termination proofs and open problems.
\newblock Technical Report SEKI-Report SR-90-25, University of Kaiserslautern,
  1990.

\bibitem[Te{R}e{S}e(2003)]{Terese}
Te{R}e{S}e.
\newblock \emph{Term Rewriting Systems}, volume~55 of \emph{Cambridge Tracks in
  Theoretical Computer Science}.
\newblock Cambridge University Press, 2003.

\bibitem[Thiemann(2007)]{Thiemann:2007}
R.~Thiemann.
\newblock \emph{The {DP} Framework for Proving Termination of Term Rewriting}.
\newblock PhD thesis, University of Aachen, Department of Computer Science,
  2007.
\newblock available as Technical Report AIB-2007-17.

\bibitem[Touzet(1998)]{Touzet:1998}
H.~Touzet.
\newblock Encoding the {H}ydra battle as a rewrite system.
\newblock In \emph{Proceedings of the 23rd International Symposium on
  Mathematical Foundations of Computer Science}, LNCS 1450, pages 267--276.
  Springer Verlag, 1998.

\bibitem[Touzet(2002)]{Touzet:2002}
H{\'e}l{\'e}ne Touzet.
\newblock A characterisation of multiply recursive functions with higman's
  lemma.
\newblock \emph{Inf.~and Comput.}, 178\penalty0 (2):\penalty0 534--544, 2002.

\bibitem[Toyama(2004)]{Toyama:2004}
Y.~Toyama.
\newblock Termination of {S}-expression rewriting systems: Lexicographic path
  ordering for higher-order terms.
\newblock In \emph{Proceedings of the 15th International Conference on
  Rewriting Techniques and Applications}, volume 3091 of \emph{LNCS}, pages
  40--54. Springer Verlag, 2004.

\bibitem[Toyama(2008)]{Toyama:2008}
Y.~Toyama.
\newblock Termination proof of {S}-expression rewriting systems with recursive
  path relations.
\newblock In \emph{Proceedings of the 19th International Conference on
  Rewriting Techniques and Applications}, volume 5117 of \emph{LNCS}, pages
  381--391. Springer Verlag, 2008.

\bibitem[Waldmann(2004)]{Matchbox:2004}
J.~Waldmann.
\newblock {M}atchbox: {A} tool for match-bounded string rewriting.
\newblock In \emph{Proceedings of the 15th International Conference on
  Rewriting Techniques and Applications}, volume 3091 of \emph{LNCS}, pages
  85--94. Springer Verlag, 2004.
\newblock \url{http://dfa.imn.htwk-leipzig.de/matchbox/}.

\bibitem[Weiermann(1995{\natexlab{a}})]{Weiermann:1995}
A.~Weiermann.
\newblock Investigations on slow versus fast growing: How to majorize slow
  growing functions nontrivially by fast growing ones.
\newblock \emph{Ar\-ch.\ {M}ath.\ {L}ogic}, 34:\penalty0 313--330,
  1995{\natexlab{a}}.

\bibitem[Weiermann(1995{\natexlab{b}})]{Weiermann:1995:tcs}
A.~Weiermann.
\newblock Termination proofs for term rewriting systems with lexicographic path
  ordering imply multiply recursive derivation lengths.
\newblock \emph{Theor.~Comput.~Sci.}, 139:\penalty0 355--362,
  1995{\natexlab{b}}.

\bibitem[Weiermann(2001)]{Weiermann:2001}
A.~Weiermann.
\newblock Some interesting connections between the slow growing hierarchy and
  the {A}ckermann function.
\newblock \emph{J.~Symb.\ Logic}, 66\penalty0 (2):\penalty0 609--628, 2001.

\bibitem[Zankl and Middeldorp(2007)]{ZM:2007}
H.~Zankl and A.~Middeldorp.
\newblock Satisfying {KBO} constraints.
\newblock In \emph{Proceedings of the 18th International Conference on
  Rewriting Techniques and Applications}, volume 4533 of \emph{LNCS}, pages
  389--403. Springer Verlag, 2007.

\bibitem[Zantema(2005)]{Torpa:2005}
H.~Zantema.
\newblock Termination of string rewriting proved automatically.
\newblock \emph{J.~Autom. Reason-}, 34\penalty0 (2):\penalty0 105--109, 2005.

\bibitem[Zantema(1994)]{Zantema:1994}
H.~Zantema.
\newblock Termination of term rewriting: interpretation and type elimination.
\newblock \emph{J.~Symb.\ Comput.}, 17\penalty0 (1):\penalty0 23--50, 1994.

\bibitem[Zantema(1995)]{Zantema:1995}
H.~Zantema.
\newblock Termination of term rewriting by semantic labelling.
\newblock \emph{Fund.~Inform.}, 24:\penalty0 89--105, 1995.

\end{thebibliography}

\end{document}